\RequirePackage{fix-cm}
\documentclass[11pt]{article}
\usepackage[fontsize=10.8pt]{scrextend}

\usepackage[round]{natbib}
\RequirePackage[OT1]{fontenc}
\RequirePackage{amsthm,amsmath}
\allowdisplaybreaks[4]
\usepackage{amssymb}
\usepackage{fullpage}	
\usepackage[shortlabels]{enumitem}
\usepackage{amsthm}
\usepackage{bbm}
\usepackage{bm}
\usepackage{float}
\usepackage{graphicx}
\usepackage[ruled,linesnumbered]{algorithm2e}
\usepackage{xr-hyper}
\usepackage{hyperref}[]
\hypersetup{
	colorlinks=true,
	linkcolor=blue,
	filecolor=magenta,      
	urlcolor=cyan,
	citecolor=blue,
}
\usepackage{cleveref}
\usepackage{subcaption}
\usepackage{xcolor}
\usepackage[hmargin={0.8in, 0.8in}, vmargin={0.75in, 0.75in}]{geometry}
\usepackage{booktabs}
\usepackage{array}
\newcolumntype{Z}{>{\setbox0=\hbox\bgroup}c<{\egroup}@{\hspace*{-\tabcolsep}}}

\usepackage{setspace}
\usepackage{diagbox}
\usepackage{tabu}
\usepackage{titlesec}
\usepackage{authblk}
\titlespacing*\section{0pt}{6pt plus 2pt minus 2pt}{0pt plus 2pt minus 2pt}
\titlespacing*\subsection{0pt}{4pt plus 2pt minus 2pt}{0pt plus 2pt minus 2pt}
\titlespacing*\subsubsection{0pt}{2pt plus 2pt minus 2pt}{0pt plus 2pt minus 2pt}

\newcommand{\R}{\mathbb{R}}
\newcommand{\N}{\mathbb{N}}

\def\bftau{\mathbb{\pmb{\tau}}}
\def\bftheta{\mathbb{\pmb{\theta}}}

\newtheorem{thm}{Theorem}[section]
\newtheorem{lemma}{Lemma}[section]
\newtheorem{ass}{Assumption}[section]
\newenvironment{thmbis}[1]
{\renewcommand{\theass}{\ref{#1}$^*$}%
	\addtocounter{ass}{-1}%
	\begin{ass}}
	{\end{ass}}

\newtheorem{pro}{Proposition}[section]

\newtheorem{defn}{Definition}[section]
\SetKwInput{KwIni}{Initialization}
\SetKwInput{KwPro}{Procedure}

\DeclareMathOperator*{\argmax}{arg\,max}

\date{\vspace{-5ex}}
\newcommand{\blind}{1}

\makeatletter
\newcommand*{\addFileDependency}[1]{
  \typeout{(#1)}
  \@addtofilelist{#1}
  \IfFileExists{#1}{}{\typeout{No file #1.}}
}
\newcommand*\showfontsize{\f@size{} point}
\makeatother

\newcommand*{\myexternaldocument}[1]{%
    \externaldocument{#1}%
    \addFileDependency{#1.tex}%
    \addFileDependency{#1.aux}%
}
\myexternaldocument{SNCP_supplement}

\begin{document}
\setlength{\abovedisplayskip}{4pt}
\setlength{\belowdisplayskip}{4pt}
\setlength{\abovedisplayshortskip}{2pt}
\setlength{\belowdisplayshortskip}{2pt}

\if1\blind
{
	\title{Segmenting Time Series via Self-Normalization}
	\author[1]{Zifeng Zhao}
    \author[2]{Feiyu Jiang}
    \author[3]{Xiaofeng Shao}
    
    \affil[1]{Mendoza College of Business, University of Notre Dame}
    \affil[2]{Department of Statistics and Data Science, Fudan University}
    \affil[3]{Department of Statistics, University of Illinois at Urbana Champaign}
		
	\date{}	
	\maketitle
	
} \fi

\if0\blind
{
	\title{Segmenting Time Series via Self-Normalization}	
	\author{}
	\date{}
	\maketitle
} \fi
\vspace{5em}

\begin{abstract}
We propose a novel and unified framework for change-point estimation in multivariate time series. The proposed method is fully nonparametric, robust to temporal dependence and avoids the demanding consistent estimation of long-run variance. One salient and distinct feature of the proposed method is its versatility, where it allows change-point detection for a broad class of parameters (such as mean, variance, correlation and quantile) in a unified fashion. At the core of our method, we couple the self-normalization (SN) based tests with a novel nested local-window segmentation algorithm, which seems new in the growing literature of change-point analysis. Due to the presence of an inconsistent long-run variance estimator in the SN test, non-standard theoretical arguments are further developed to derive the consistency and convergence rate of the proposed SN-based change-point detection method. Extensive numerical experiments and relevant real data analysis are conducted to illustrate the effectiveness and broad applicability of our proposed method in comparison with state-of-the-art approaches in the literature.
\end{abstract}

\noindent\textit{Keywords}: Binary segmentation; Change-point detection; Scanning; Studentization; Long-run variance; Temporal dependence

\section{Introduction}\label{sec:intro}
Change-point detection has been identified as one of the major challenges for modern data applications~\citep{Council2013}. There is a vast literature on change-point estimation and testing in statistics, in part due to its broad applications in bioinformatics, climate science, economics, finance, genetics, medical science, and signal processing among many other areas. See \cite{Csoergoe1997}, \cite{brodsky2013nonparametric} and \cite{tartakovsky2014sequential} for book-length treatments of the subject. We also refer to \cite{aue:13}, \cite{Casini2019structural} and \cite{truong2020} for excellent reviews.

In this paper, we study the problem of time series segmentation, also known as (offline) change-point estimation, where the task is to partition a sequence of potentially non-homogeneous ordered observations into piecewise homogeneous segments. Many change-point problems arise within a time series context~(e.g.\ climate, epidemiology, economics and financial data), where there is a natural temporal ordering in the observations. Although temporal dependence  is the norm rather than the exception for time series, most literature in change-point analysis assume and require independence of observations $\{Y_t\}_{t=1}^{n}$ over time for methodological and theoretical validity; see for example \cite{Olshen2004}, \cite{Killick2012}, \cite{Matteson2014}, \cite{Fryzlewicz2014}, and \cite{baranowski2019narrowest} among others. One stream of literature addresses temporal dependence via the assumption of parametric models, see \cite{Davis2006} and \cite{Yau2016} for change-point detection in AR process and \cite{Fryzlewicz2014a} in ARCH process. However, parametric approaches generally require stronger conditions and potential violation of parametric assumptions can inevitably cast doubts on the estimation result.

Existing nonparametric approaches for change-point estimation in temporally dependent observations primarily focus on first or second-order moments, see \cite{Bai1998}, \cite{Eichinger2018} for change-point estimation in mean, \cite{Aue2009}, \cite{Preuss2015} in (auto)-covariance, {and \cite{Cho2012}, \cite{Casini2021change} in spectral density function (thus second-order properties)}. However, for many applications, the key interest can go beyond mean or covariance. For example, detecting potential changes in extreme quantiles is critical for monitoring systemic risk~(i.e.\ Value-at-Risk) in finance and for studying evolving behavior of severe weather systems such as hurricanes in climate science. Moreover, existing nonparametric methods are mostly designed for detecting only one specific type of change~(e.g.\ mean or variance) and cannot be universally used for examining changes in different aspects of the data, which may limit its applications and cause inconvenience of implementation for practitioners. Additionally, existing nonparametric procedures typically involve certain tuning or smoothing parameters, such as the bandwidth parameter involved in the consistent estimation of the long-run variance, and how to choose these tuning parameters is important yet highly challenging in practice.  

To fill in the gap in the literature, we propose a new multiple change-point estimation framework that is fully nonparametric, robust to temporal dependence, enjoys effortless tuning, and works universally for various parameters of interest for a multivariate time series $\{Y_t\}_{t=1}^n$ where $Y_t \in \mathbb{R}^p$ with a fixed dimension $p\geq 1$. Specifically, denote $F_t$ as the cumulative distribution function (CDF) of $Y_t$, the proposed procedure allows change-point detection for any $\theta$ such that $\theta=\theta(F_t)$, where $\theta(\cdot)$ is a functional that takes value in $\mathbb{R}^d$ with $d\geq 1.$ This is a broad framework that covers important quantities such as mean, variance, quantile, (auto)-correlation and (auto)-covariance among others, see \cite{Kunsch1989} and \cite{shao2010self}.

As in the standard change-point literature, we assume the change happens in a piecewise constant fashion. Specifically, we assume $\{Y_t\}_{t=1}^n$ is a piecewise stationary time series and there exist $m_o\geq 0$ unknown number of change-points $0 <k_1<\cdots<k_{m_o} <n$ that partition $\{Y_t\}_{t=1}^n$ into $m_o+1$ stationary segments. Define $k_0=0$ and $k_{m_o+1}=n$, the $i$th segment contains stationary observations $\{Y_t\}_{t=k_{i-1}+1}^{k_i}$ that share common behavior characterized by $\theta_i$~(e.g.\ mean, variance, correlation, quantile), where we require $\theta_i\not= \theta_{i+1}$ for $i=1,\cdots,m_o$ due to the structural break. Our primary interest is to recover the unknown number and locations of the change-points. 

To achieve broad applicability and robustness against temporal dependence, our proposed multiple change-point estimation method is built upon self-normalization (SN, hereafter), a nascent inference technique for time series~\citep{shao2010self,shao2015self}. We note that since its first proposal in \cite{shao2010self}, SN has been extended to retrospective change-point testing by \cite{shao2010testing}, \cite{Hoga2018}, \cite{Betken2018},  \cite{Zhang2018}, and \cite{Dette2020b}, and to sequential change-point monitoring by \cite{dette2020a} and \cite{Chan2021}. However, the primary focus of these papers is to construct SN-based change-point testing procedures (either retrospective or sequential) but not change-point estimation. Compared to change-point testing, change-point estimation is a much more challenging task both methodologically and theoretically: it further requires the estimation of the unknown number and locations of change-points, which involves substantially different techniques and analysis. 

Indeed, the use of SN for time series segmentation (i.e.\ multiple change-point estimation) seems largely unexplored, with the exception of \cite{jiang2020time, jiang2022modelling} for piecewise linear and quantile trend models designed for COVID-19 time series. One notable reason for the scarcity of SN-based time series segmentation algorithms is that, unlike the classical CUSUM-based change-point test, the SN-based change-point testing cannot be easily extended to multiple change-point estimation by combining with the standard binary segmentation algorithm~\citep{vostrikova1981detecting}. Such a combination simply fails due to the potential inflation of the self-normalizer under the presence of multiple change-points. We discuss this point in more details later in Section \ref{sec:multiple} and provide further illustration via both theory and numerical experiments in Section \ref{sec: bs_powerloss} of the supplementary material.

To bypass this difficulty, we propose a novel nested local-window segmentation algorithm, which is then combined with an SN test to achieve multiple change-point estimation. We name the procedure SNCP. Through a series of carefully designed nested local-windows, the proposed procedure can isolate each true change-point adaptively and thus achieves respectable detection power and estimation accuracy. The statistical and computational efficiency of the nested local-window segmentation algorithm is further illustrated via extensive numerical comparison with popular segmentation algorithms such as SaRa in \cite{Niu2012}, WBS in \cite{Fryzlewicz2014} and SBS in \cite{Kovacs2020}.

In addition to methodological advances, new theoretical arguments based on the partial influence functions~\citep{pires2002partial} are further developed to establish the consistency and convergence rate of the proposed change-point estimation procedure, which seems to be the first in the SN literature. The proof is non-standard and built on a subtle analysis of the behavior of SN-based test statistic around change-points. It differs from existing techniques in the change-point literature due to the presence of the self-normalizer~(an inconsistent long-run variance estimator) and is of independent interest.

To our best knowledge, the proposed method~(SNCP) is the first to address multiple change point estimation for a general parameter in the time series setting. One salient and distinct feature of SNCP is its versatility: it allows the user to examine potential change in virtually any parameter of interest in an effortless fashion. This is valuable as in practice, the ground truth is unknown and it is important to examine the behavior change of the data via different angles. In addition, due to its versatility and robustness to temporal dependence, SNCP can serve as a numerically credible and theoretically valid benchmark for almost all algorithms designed for multiple change-point estimation in a fixed-dimensional time series, which is of interest to both practical applications and academic research.

The rest of the paper is organized as follows. We first provide background of SN and introduce the SN-based detection method for single change-point estimation in  Section~\ref{sec:single}. Building upon a novel nested local-window segmentation algorithm, Section \ref{sec:multiple} proposes a unified SN-based framework~(SNCP) for multiple change-point estimation and further studies its theoretical properties. Extensive numerical experiments are conducted in Section \ref{sec:simulation} to demonstrate the promising performance of SNCP when compared with state-of-the-art methods for change-point estimation in mean, variance, quantile of univariate time series, and correlation and covariance matrix of multivariate time series. Section \ref{sec:conclusion} concludes. Technical  proofs and additional simulation and real data application results can be found in the supplement.

Some notations used throughout the paper are defined as follows. Let $D[0,1]$ denote the space of functions on $[0,1]$ which are right continuous with left limits, endowed with the Skorokhod topology \citep{Billingsley1968}. We use $\Rightarrow$ to denote weak convergence in $D[0,1]$ or more generally in $\R^m$-valued function space $D^m[0,1]$, where $m\in \N$. We use $\overset{\mathcal{D}}{\longrightarrow}$ to denote convergence in distribution. We use $\|\cdot\|_2$ to denote the $l_2$ norm of a vector and use $\|\cdot\|$ to denote the spectral norm of a matrix.

\section{Single Change-point Estimation}\label{sec:single}
In this section, we provide some background on the SN test and propose an SN test based method for single change-point estimation, which serves as a building block for the proposed multiple change-point estimation procedure in Section \ref{sec:multiple}. Model assumptions and consistency results are discussed in details to provide intuition and foundation for more involved results in Section \ref{sec:multiple}. For ease of presentation, in the following we assume $d=1$, in other words, the parameter of interest $\theta$ is univariate, and postpone the results for the multivariate case of $d>1$ to Section \ref{subsec:vector}.

\subsection{An SN-based estimation procedure}\label{subsec:general}
We start with single change-point estimation in a general parameter $\theta=\theta(F_{t})$  for a univariate time series $\{Y_t\}_{t=1}^n$, where $F_t$ denotes the CDF of $Y_t$ and $\theta(\cdot)$ is a general functional. Under the no change-point scenario, $\{Y_t\}_{t=1}^n$ is a stationary time series. Under the single change-point alternative, we follow the framework of \cite{dette2020a} and assume $\{Y_t\}_{t=1}^n$ is generated by
\begin{align}\label{eq:generalcase}
	Y_t=\begin{cases}
		Y_t^{(1)},&1\le t\le k_1\\
		Y_t^{(2)},& k_1+1\le t\le n,\\
	\end{cases}
\end{align}
where $\{Y_t^{(i)}\}_{t\in\mathbb{Z}}$ is a stationary time series with $Y_t^{(i)}\sim F^{(i)}$ for $i=1,2$. Thus we have $F_t=F^{(1)}\mathbf{1}(t\leq k_1)+F^{(2)}\mathbf{1}(t> k_1)$. Denote $\theta_1=\theta(F^{(1)})$ and $\theta_2=\theta(F^{(2)})$, we have $\delta=\theta_2-\theta_1\neq 0$ and the change-point $k_1=\lfloor n\tau_1\rfloor$ with $\tau_1\in(0,1)$. Note that the dependence between $\{Y_t^{(1)}\}$ and $\{Y_t^{(2)}\}$ is deliberately left unspecified, as the validity of the proposed method does not rely on the specification of the dependence~(see Assumption \ref{ass_influence}(i) for more details).

To detect the existence and further estimate the location of the (potential) single change-point $k_1=\lfloor n\tau_1\rfloor$, we propose an SN-based testing approach. Specifically, we define
\begin{align}\label{statgeneral}
	SN_n= \max_{k=1,\cdots,n-1}T_n(k),  \quad T_n(k)= D_n(k)^2/V_{n}(k),
\end{align}
where
\begin{flalign}\label{statistic}
	\begin{split}
		D_n(k)=&\frac{k(n-k)}{n^{3/2}}(\widehat{\theta}_{1,k}-\widehat{\theta}_{k+1,n}),\\
		V_{n}(k)=&\sum_{i=1}^{k}\frac{i^2(k-i)^2}{n^2k^2}(\widehat{\theta}_{1,i}-\widehat{\theta}_{i+1,k})^2+\sum_{i=k+1}^{n}\frac{(n-i+1)^2(i-k-1)^2}{n^2(n-k)^2}(\widehat{\theta}_{i,n}-\widehat{\theta}_{k+1,i-1})^2,
	\end{split}
\end{flalign}
and for any $1\leq a<b\leq n$, $\widehat{\theta}_{a,b}=\theta(\widehat{F}_{a,b})$ where $\widehat{F}_{a,b}$ is the empirical distribution of $\{Y_t\}_{t=a}^b$. In other words, $\widehat{\theta}_{a,b}$ denotes the nonparametric estimator of $\theta$ based on the subsample $\{Y_t\}_{t=a}^b$.

When $\theta(\cdot)$ is the mean functional, i.e., $\theta(F_t)=\int x F_t(dx)$, the newly defined contrast-based test $SN_n$ in \eqref{statgeneral} reduces to the CUSUM-based SN test statistic in \cite{shao2010testing} (cf. equation (4) therein). 
However, for a nonlinear functional $\theta(\cdot)$, such as variance, correlation and quantile, $SN_n$ is not equivalent to the CUSUM-based counterpart and is preferred due to its contrast nature. We refer to \cite{Zhang2018} for more discussion.

Built upon the test statistic defined in \eqref{statgeneral}, the SN-based change-point detection procedure proceeds as follows. For a pre-specified threshold $K_n$, we declare no change-point if $SN_n\leq K_n$. Given that $SN_n$ exceeds the threshold, we estimate the single change-point location via
$$\widehat{k}=\arg\max_{k=1,\cdots,n-1}T_n(k).$$

This SN-based procedure provides a general and unified change-point estimation framework, as it can be implemented for any functional $\theta(\cdot)$ with a nonparametric estimator based on the empirical distribution.

\subsection{Assumptions and theoretical results}\label{subsec:theory_SCP}

To establish the consistency of the SN-based estimation procedure under the general functional setting \eqref{eq:generalcase}, the key is to track the asymptotic behavior of $\widehat{\theta}_{a,b}$ for $1\leq a < b\leq n$. To achieve this, we operate under the framework of approximately linear functional, which covers important quantities such as mean, variance, covariance, correlation and quantile~\citep{Kunsch1989,shao2010self}.

Specifically, we assume the subsample estimator $\widehat{\theta}_{a,b}$ admits the following expansion on the stationary time series $\{Y_t^{(i)}\}, i=1,2,$ where
\begin{flalign}\label{inf1}
	\begin{split}
		\widehat{\theta}_{a,b}=&\theta_1+\frac{1}{b-a+1}\sum_{t=a}^{b}\xi_1(Y_t^{(1)})+r^{(1)}_{a,b},\quad \text{for}~b\leq k_1,\\
		\widehat{\theta}_{a,b}=&\theta_2+\frac{1}{b-a+1}\sum_{t=a}^{b}\xi_2(Y_t^{(2)})+r^{(2)}_{a,b},\quad \text{for}~a>k_1.
	\end{split}
\end{flalign}
In other words, $\widehat{\theta}_{a,b}$ is approximately linear when the subsample $\{Y_t\}_{t=a}^b$ is stationary. Note that $\xi_1(Y_t^{(1)})$ and $\xi_2(Y_t^{(2)})$ are indeed the influence functions of the functional $\theta(\cdot)$~\citep{Hampel1986}, which is the leading term for asymptotic behavior of $\widehat{\theta}_{a,b}$, and $r^{(1)}_{a,b}, r^{(2)}_{a,b}$ are the remainder terms.

To further regulate the behavior of $\widehat{\theta}_{a,b}$ when the subsample $\{Y_t\}_{t=a}^b$ is a mixture of two stationary segments, we utilize the concept of \textit{partial influence functions} originated from the robust statistics literature~\citep{pires2002partial}. Specifically, for $a\leq k_1<b$, we assume
\begin{align}\label{inf3}
	\widehat{\theta}_{a,b}=&\theta(\omega_{a,b})+\frac{1}{b-a+1}\left[\sum_{t=a}^{k_1}\xi_1(Y_t^{(1)},\omega_{a,b})+\sum_{t=k_1+1}^{b}\xi_2(Y_t^{(2)},\omega_{a,b})\right]+r_{a,b}(\omega_{a,b}),
\end{align}
where $\omega_{a,b}=\left(\omega_{a,b}^{(1)},\omega_{a,b}^{(2)}\right)^\top=\left(\frac{k_1-a+1}{b-a+1},\frac{b-k_1}{b-a+1}\right)^{\top}$ denotes the proportion of each stationary segment in  $\{Y_t\}_{t=a}^b$, $\theta(\omega_{a,b})$ denotes $\theta(\cdot)$ evaluated at the mixture distribution  $F^{\omega_{a,b}}=\omega_{a,b}^{(1)}F^{(1)}+\omega_{a,b}^{(2)}F^{(2)}$ and $r_{a,b}(\omega_{a,b})$ is the remainder term. The terms $\xi_1(Y_t^{(1)},\omega_{a,b})$ and $\xi_2(Y_t^{(2)},\omega_{a,b})$ are related to the partial influence functions of the functional $\theta(\cdot)$ evaluated at the mixture distribution $F^{\omega_{a,b}}$. See detailed discussion later.

Note that the expansion \eqref{inf3} generalizes \eqref{inf1} under the single change-point scenario. Specifically, define $\omega_{a,b}=(1,0)^{\top}$ and $(0,1)^{\top}$ for $b\leq k_1$ and $a>k_1$ respectively, \eqref{inf1} can be viewed as a special case of \eqref{inf3} where the mixture distribution is \textit{pure} such that $\xi_1(Y_t^{(1)})=\xi_1(Y_t^{(1)},(1,0)^\top)$, $r_{a,b}^{(1)}=r_{a,b}((1,0)^\top)$ and $\xi_2(Y_t^{(2)})=\xi_2(Y_t^{(2)},(0,1)^\top)$, $r_{a,b}^{(2)}=r_{a,b}((0,1)^\top)$ respectively.

We now work out the explicit formulation of the expansion \eqref{inf3} under the framework of partial influence function~\citep{pires2002partial}. Denote the mixture weight $\omega=(\omega^{(1)},\omega^{(2)})^{\top}$ such that $\omega^{(i)}\in [0,1]$, $i=1,2$ and $\omega^{(1)}+\omega^{(2)}=1$. Denote $\theta(\omega,F^{(1)},F^{(2)}):=\theta(\omega^{(1)} F^{(1)}+\omega^{(2)}F^{(2)})$ as the functional $\theta(\cdot)$ evaluated at the mixture $F^{\omega}:=\omega^{(1)} F^{(1)}+\omega^{(2)}F^{(2)}$. Definition \ref{def:PIF} defines the partial influence function as in \cite{pires2002partial}.
\begin{defn}\label{def:PIF}
	The partial influence functions of the functional $\theta(F^\omega)=\theta(\omega,F^{(1)},F^{(2)})$ with relation to $F^{(1)}$ and $F^{(2)}$, respectively, are given by
	\begin{flalign*}
		IF_1\left(y,\theta(\omega,F^{(1)},F^{(2)})\right)=&\lim\limits_{\epsilon\to0}\epsilon^{-1}\Big[\theta\left(\omega,(1-\epsilon)F^{(1)}+\epsilon\delta_y,F^{(2)}\right)-\theta(\omega,F^{(1)},F^{(2)})\Big],\\
		IF_2\left(y,\theta(\omega,F^{(1)},F^{(2)})\right)=&\lim\limits_{\epsilon\to0}\epsilon^{-1}\Big[\theta\left(\omega,F^{(1)},(1-\epsilon)F^{(2)}+\epsilon\delta_y\right)-\theta(\omega,F^{(1)},F^{(2)})\Big],
	\end{flalign*}
	provided the limits exist, where $\delta_y$ is the Dirac mass at $y$. 
\end{defn}
\noindent To understand the partial influence functions, define $\zeta=\omega^{(1)}\epsilon$, by Definition \ref{def:PIF}, we have
\begin{flalign*}
	IF_1\left(y,\theta(\omega,F^{(1)},F^{(2)})\right)=&\omega^{(1)}\lim\limits_{\zeta\to0}\zeta^{-1}\Big[\theta\left((\delta_y-F^{(1)})\zeta+F^{\omega}\right)-\theta\left(F^{\omega}\right)\Big]=\omega^{(1)}\xi_1(y,\omega),
\end{flalign*}
where $\xi_1(y,\omega)$ is the G\^ateaux derivative of $\theta\big(F^{\omega}\big)$ in the direction $\delta_y-F^{(1)}$. Similarly, $IF_2\left(y,\theta(\omega,F^{(1)},F^{(2)})\right)=\omega^{(2)}\xi_2(y,\omega)$, where $\xi_2(y,\omega)$ is the G\^ateaux derivative of $\theta\big(F^{\omega}\big)$ in the direction $\delta_y-F^{(2)}$.

To establish the expansion \eqref{inf3}, note that $\widehat{\theta}_{a,b}=\theta(\widehat{F}_{a,b})$, where $\widehat{F}_{a,b}$ denotes the empirical CDF based on the subsample $\{Y_t\}_{t=a}^b$. The key observation is that $\widehat{F}_{a,b}=\omega^{(1)}_{a,b} \widehat{F}_{a,k_1}+\omega^{(2)}_{a,b}\widehat{F}_{k_1+1,b}$ with $\omega_{a,b}=\left(\frac{k_1-a+1}{b-a+1},\frac{b-k_1}{b-a+1}\right)^{\top}$. In other words, $\widehat{F}_{a,b}$ can be viewed as a mixture of two empirical CDFs $\widehat{F}_{a,k_1}$ and $\widehat{F}_{k_1+1,b}$ based on stationary segments with CDF $F^{(1)}$ and $F^{(2)}$ respectively. Thus, by the results in \cite{pires2002partial}, we have  
\begin{flalign*}
	\theta(\widehat{F}_{a,b})=&\theta\left(F^{\omega_{a,b}}\right)+\frac{1}{k_1-a+1}\sum_{t=a}^{k_1}IF_1\left(Y_t^{(1)},\theta(\omega_{a,b},F^{(1)},F^{(2)})\right)\\&+\frac{1}{b-k_1}\sum_{t=k_1+1}^{b}IF_2\left(Y_t^{(2)},\theta(\omega_{a,b},F^{(1)},F^{(2)})\right)+R( \widehat{F}_{a,b}-F^{\omega_{a,b}}),
\end{flalign*}
where $R( \widehat{F}_{a,b}-F^{\omega_{a,b}})$ denotes the remainder term. The expansion \eqref{inf3} follows immediately by substituting the partial influence functions with the G\^ateaux derivatives $\xi_1(y,\omega_{a,b})$ and $\xi_2(y,\omega_{a,b})$. 

We proceed by imposing the following Assumptions \ref{ass_influence}-\ref{ass_theta} on the approximately linear functional $\theta(\cdot)$, which are further verified in Section \ref{sec:verify} of the supplement for {the smooth function model~(including mean, variance, (auto)-covariance, (auto)-correlation) and in Section \ref{sec:verify_quantile} of the supplement for quantile. We refer to Remark 1 in Section \ref{subsec:theory} for more detailed discussion on the verification of assumptions.}

\begin{ass}\label{ass_influence}
	$(\mathrm{i})$ For some $\sigma_1>0$ and $\sigma_2>0$, we have
	$$\frac{1}{\sqrt{n}} \sum_{t=1}^{[nr]}\Big(\xi_1(Y_t^{(1)}), \xi_2(Y_t^{(2)})\Big)\Rightarrow (\sigma_1 B^{(1)}(r), \sigma_2 B^{(2)}(r)), $$
	where  $B^{(1)}(\cdot)$ and $B^{(2)}(\cdot)$ are standard Brownian motions. 
	
	\noindent$(\mathrm{ii})$	$\sup_{k< k_1}\left|\sum_{t=k+1}^{k_1}\xi_1\left(Y_t^{(1)},\omega_{k+1,n}\right)+\sum_{t=k_1+1}^{n}\xi_2\left(Y_t^{(2)},\omega_{k+1,n}\right)\right|=O_p(n^{1/2}),$\\
	\text{\hspace{0.38cm} } $\sup_{k>k_1}\left|\sum_{t=1}^{k_1}\xi_1(Y_t^{(1)},\omega_{1,k})+\sum_{t=k_1+1}^{k}\xi_2(Y_t^{(2)},\omega_{1,k})\right|=O_p(n^{1/2}).$
\end{ass}

\begin{ass}\label{ass_remainder}
	$\sup_{1\leq k\leq n}k|r_{1,k}(\omega_{1,k})|+\sup_{1\leq k\leq n}(n-k+1)|r_{k,n}(\omega_{k,n})|=o_p(n^{1/2}).$
\end{ass}

Assumption \ref{ass_influence} regulates the behavior of the (partial) influence function $\xi_1(\cdot)$ and $\xi_2(\cdot)$. Specifically, Assumption \ref{ass_influence}(i) requires the invariance principle to hold for each stationary segment. Note that the dependence of the two Brownian motions $B^{(1)}(\cdot)$ and $B^{(2)}(\cdot)$ are left unspecified as we do not require a specific dependence structure on $\{Y_t^{(1)}\}$ and $\{Y_{t}^{(2)}\}$. Assumption \ref{ass_influence}(ii) are tailored to regulate $\widehat{\theta}_{a,b}$ estimated on a mixture of two stationary segments. Assumption \ref{ass_remainder} requires that the remainder term is asymptotically negligible and is a commonly used assumption in the SN literature~\citep{shao2010self, shao2015self}.

\begin{ass}\label{ass_theta}
	Denote $\theta(\omega)=\theta(\omega^{(1)} F^{(1)}+\omega^{(2)}F^{(2)})$, where $\omega=(\omega^{(1)},\omega^{(2)})^{\top}$ is the mixture weight with $\omega^{(i)}\in[0,1], i=1,2$ and $\omega^{(1)}+\omega^{(2)}=1$. There exist some constants $0<C_1<C_2<\infty$ such that for any mixture weight $\omega$, we have
	\begin{align*}
		C_1\omega^{(2)}|\theta_1-\theta_2|\leq|\theta_1-\theta(\omega)|\leq C_2\omega^{(2)}|\theta_1-\theta_2| \text{ and }
		C_1\omega^{(1)}|\theta_1-\theta_2|\leq|\theta_2-\theta(\omega)|\leq C_2\omega^{(1)}|\theta_1-\theta_2|.
	\end{align*}
\end{ass}

Assumption \ref{ass_theta} regulates the smoothness of $\theta(\omega)$. Intuitively, it means that the functional $\theta(\cdot)$ can distinguish the mixture distribution $w^{(1)}F^{(1)}+w^{(2)}F^{(2)}$ from $F^{(1)}$ and $F^{(2)}$. For mean functional, we have $\theta(\omega)=\omega^{(1)}\theta_1+\omega^{(2)}\theta_2$, thus we can set $C_1=C_2=1$ as $\theta(\omega)$ is linear in $\omega$. 
\begin{ass}\label{ass_K}
	$n\delta^2\to\infty$ as $n\to\infty$, and $K_n$ satisfies $K_n=(n\delta^2)^{\kappa}$ for some $\kappa\in (\frac{1}{2},1)$.
\end{ass}

Assumption \ref{ass_K} quantifies the asymptotic order of the change size $\delta$ and the threshold $K_n$. Under Assumptions \ref{ass_influence}-\ref{ass_K}, Theorem \ref{thm_onechange} gives the consistency results of the SN-based change-point estimation method for approximately linear functionals.
\begin{thm}\label{thm_onechange}
	$(\mathrm{i})$ Under the no change-point scenario, suppose Assumptions \ref{ass_influence}(i) and \ref{ass_remainder} hold, we have
	$SN_n \overset{\mathcal{D}}{\longrightarrow} G=\sup_{r\in [0,1]}\{B(r)-rB(1)\}^2/V^{}(r),$ where $B(\cdot)$ denotes a standard Brownian motion and $V(r)=\int_0^r[B(s)-(s/r) B(r)]^2 ds + \int_{r}^{1} [B(1)-B(s)-(1-s)/(1-r)\{B(1)-B(r)\}]^2 ds$. \\
	\noindent $(\mathrm{ii})$ Under the one change-point scenario, suppose Assumptions \ref{ass_influence}-\ref{ass_K} hold, we have
	$$\lim\limits_{n\to\infty}P(T_n(\widehat{k})>K_n\quad \text{and}\quad|\widehat{k}-k_1|\leq \iota_n)= 1,$$
	for any sequence $\iota_n$ such that $\iota_n/n \to 0$ and $\iota_n^{-2}\delta^{-2}n\to0$ as $n\to\infty$. 
\end{thm}

Theorem \ref{thm_onechange}(i) indicates that the asymptotic distribution of $SN_n$ for a general functional $\theta(\cdot)$ coincides with the asymptotic distribution of the CUSUM-based SN test for mean~(see Theorem 3.1 in \cite{shao2010testing}). This implies that the same threshold $K_n$ can be used to control false positives~(i.e.\ Type-I error) for change-point detection in various parameters and thus greatly simplifies the implementation of the proposed method. In practice, we recommend to set $K_n$ as the 90\% or 95\% quantile of $G$, which can be obtained via simulation as $G$ is pivotal. See \cite{shao2010testing} for tabulated critical values of $G$.

Theorem \ref{thm_onechange}(ii) gives the convergence rate of the estimated change-point $\widehat{k}$, providing a unified theoretical justification of the SN-based method for a broad class of functionals. Due to the presence of the self-normalizer $V_n(k)$, which is complex and further varies by $k$, nonstandard technical arguments different from existing techniques in the change-point literature are developed to establish the consistency result. It involves a simultaneous analysis of the contrast statistic $D_n(k)$ and the self-normalizer $V_n(k)$. {In general, the localization error rate of SNCP is not optimal (at least for change in mean). However, a simple local refinement procedure can be performed to help achieve the optimal rate. We refer to the discussion following \Cref{thm1} in \Cref{subsec:theory} for more details on this matter. }


The traditional CUSUM based estimation procedure in the change-point literature typically admits the form $\max_{k=1,\cdots,n-1}|D_n(k)|/\widehat{\sigma}_n$, where theoretical results are derived under the assumption that $\widehat{\sigma}_n$ is a consistent estimator of the long-run variance~(LRV), leading to less involved technical analysis than the proposed SN-based estimation. However, in practice, the construction of a consistent $\widehat{\sigma}_n$ involves a bandwidth tuning parameter that is difficult to select, especially under the presence of change-points. For example, in the mean case, using a data-driven bandwidth with the estimation-optimal bandwidth formula in \cite{andrew1991} could lead to non-monotonic power under the change-point alternative and large size distortion under the null, see \cite{cv2007} and \cite{shao2010testing}. {\cite{Casini2021theory} and \cite{casini2021minimax} further provide a comprehensive theoretical analysis of such phenomenon based on Edgeworth expansion.} Additionally, different construction of $\widehat{\sigma}_n$ is required for different functional $\theta(\cdot)$, which can be highly involved and non-trivial for parameters such as correlation and quantile, making the practical implementation challenging.

{In contrast, thanks to the self-normalizer $V_n(k)$, the proposed SN-based procedure avoids the challenging estimation of LRV and provides a robust framework that works universally for a broad class of functionals under temporal dependence.}





\section{Multiple Change-point Estimation}\label{sec:multiple}

In this section, we further extend the proposed SN-based test to multiple change-point estimation. As in standard change-point literature, we assume $\{Y_t\}_{t=1}^n$ is a piecewise stationary time series and there exist $m_o\geq 0$ unknown number of change-points $0 <k_1<\cdots<k_{m_o} <n$ that partition $\{Y_t\}_{t=1}^n$ into $m_o+1$ stationary segments. Define $k_0=0$ and $k_{m_o+1}=n$, the $i$th segment contains stationary observations $\{Y_t\}_{t=k_{i-1}+1}^{k_i}$ that share common behavior characterized by $\theta_i$, for $i=1,\cdots,m_o+1$.

More specifically, we operate under the following data generating process for $\{Y_t\}_{t=1}^{n}$ such that
\begin{equation}
	Y_t=Y_t^{(i)}, ~ k_{i-1}+1\leq t\leq k_{i}, ~ \text{ for } i=1,\cdots, m_o+1,
\end{equation}
where  $\{Y_t^{(i)}\}_{t\in\mathbb{Z}}$ is a stationary time series with CDF $F^{(i)}$ and we require $\theta_i=\theta(F^{(i)}) \not= \theta_{i+1}=\theta(F^{(i+1)})$ for $i=1,\cdots,m_o$ due to the structural break. Our primary interest is to recover the unknown number and locations of the change-points.

To proceed, we first introduce some notations. For $1\leq t_1<k<t_2\leq n$, we define
\begin{align}\label{eq:SN_subsample}
	T_n(t_1,k,t_2)={D_{n}(t_1,k,t_2)^2}/{V_{n}(t_1,k,t_2)},
\end{align}
where $D_n(t_1,k,t_2)=\frac{(k-t_1+1)(t_2-k)}{(t_2-t_1+1)^{3/2}}(\widehat{\theta}_{t_1,k}-\widehat{\theta}_{k+1,t_2})$, $
V_n(t_1,k,t_2)=L_n(t_1,k,t_2)+R_n(t_1,k,t_2),$ and
\begin{flalign*}
	L_n(t_1,k,t_2)=&\sum_{i=t_1}^{k}\frac{(i-t_1+1)^2(k-i)^2}{(t_2-t_1+1)^{2}(k-t_1+1)^2}(\widehat{\theta}_{t_1,i}-\widehat{\theta}_{i+1,k})^2,\\
	R_n(t_1,k,t_2)=&\sum_{i=k+1}^{t_2}\frac{(t_2-i+1)^2(i-1-k)^2}{(t_2-t_1+1)^{2}(t_2-k)^2}(\widehat{\theta}_{i,t_2}-\widehat{\theta}_{k+1,i-1})^2.
\end{flalign*}
Note that $T_n(t_1,k,t_2)$ is essentially the proposed SN test defined on subsample $\{Y_t\}_{t=t_1}^{t_2}$. Set $t_1=1$ and $t_2=n$, $T_n(t_1,k,t_2)=T_n(1,k,n)$ reduces to the \textit{global} SN test defined in \eqref{statgeneral} of Section \ref{subsec:general}. 


The key observation is that, due to the presence of the self-normalizer $V_n$, the \textit{global} test statistic $T_n(1,k,n)$ may experience severe power loss under multiple change-point scenarios. The intuition is as follows. Suppose $k$ is a true change-point and  $\{Y_t\}_{t=1}^{n}$ has other change-points besides $k$. Intuitively, $V_{n}(1,k,n)$ may observe significant inflation as $L_{n}(1,k,n)$ and $R_{n}(1,k,n)$ are based on contrast statistics and their values could significantly inflate due to the existence of other change-points besides $k$. This can in turn cause $T_{n}(1,k,n)$ to suffer severe deflation and thus a loss of power. Consequently, a naive combination of the standard binary segmentation~\citep{vostrikova1981detecting} and the SN test cannot serve as a viable option for multiple change-point estimation~(see both theoretical evidence and numerical illustration of this phenomenon in Section \ref{sec: bs_powerloss} of the supplement).

\subsection{The nested local-window segmentation algorithm}\label{subsec: nested_alg}
To bypass this issue, we combine the SN test with a novel nested local-window segmentation algorithm, where for each $k$, instead of one global SN test $T_n(1,k,n)$, we compute a maximal SN test based on a collection of nested windows covering $k$. Specifically, fix a small $\epsilon\in (0,1/2)$ such as $\epsilon=0.05, 0.1$, define the window size $h=\lfloor n\epsilon\rfloor$. For each $k=h,\cdots, n-h$, we define its nested window set $H_{1:n}(k)$ where
\begin{align*}
	H_{1:n}(k)=\biggl\{ (t_1,t_2)\bigg\vert t_1=k-j_1 h+1, j_1=1,\ldots, \lfloor k/h\rfloor;t_2=k+j_2 h, j_2=1,\ldots, \lfloor (n-k)/h \rfloor  \biggl\}.
\end{align*}
Note that for $k<h$ and $k>n-h$, by definition, we have $H_{1:n}(k)=\varnothing$.

For each $k=1,\cdots,n$, based on its nested window set $H_{1:n}(k)$, we define a maximal SN test statistic $T_{1,n}(k)$ such that
\begin{align*}
	T_{1,n}(k)&=\max\limits_{(t_1,t_2)\in H_{1:n}(k)}T_n(t_1,k,t_2),
\end{align*}
where we set $\max\limits_{(t_1,t_2)\in \varnothing}T_n(t_1,k,t_2):=0.$ Note that unlike the standard binary segmentation, the test statistic $T_{1,n}(k)$ is calculated based on a set of nested \textit{local}-window observations $\{Y_t\}_{t=t_1}^{t_2}$ surrounding the time point $k$ instead of directly based on the \textit{global} observations $\{Y_t\}_{t=1}^{n}$.

This mechanism is precisely designed to alleviate the inflation of the self-normalizer $V_n$ for the SN test under multiple change-point scenarios. With a sufficiently small window size $\epsilon$, for any change-point $k$, there exists some local-window $(\tilde{t}_1,\tilde{t}_2)$ which contains $k$ as the only change-point, thus the maximal statistic $T_{1,n}(k)$ remains effective thanks to $T_n(\tilde{t}_1,k,\tilde{t}_2)$. In the literature, there exists \textit{pure} local-window based segmentation algorithms, e.g.\ SaRa in \cite{Niu2012} for change in mean, LRSM in \cite{Yau2016} for change in AR models. The pure local-window approach only considers the smallest local-window $(k-h+1,k+h)$ when constructing change-point tests for $k$ given a window size $h$. {Such an approach is also employed in the literature of ``piecewise smooth" change, see \cite{wuzhao2007}, \cite{Bibinger2017} and \cite{Casini2021change}. }

Compared to the \textit{pure} local-window approach, the constructed nested window set $H_{1:n}(k)$ makes our algorithm more adaptive as it helps $T_{1,n}(k)$ retain more power when $k$ is far away from other change-points by utilizing larger windows that cover $k$. We refer to Section \ref{subsec:bs_sara} of the supplement for more detailed discussion of this point and numerical evidence of {the substantial advantage in detection power and estimation accuracy of the proposed nested local-window approach over the pure local-window approach. In addition, since the nested local-window algorithm examines a set of expanding windows instead of a single window, its performance is more robust to the choice of the bandwidth $h$. This is confirmed by numerical experiments in Section \ref{subsec: sensitivity} of the supplement, where we conduct sensitivity analysis of $h$ and it is seen that performance of the nested local-window is robust and stable w.r.t.\ the choice of $h$.}

{Note that the nested window-based SN statistic $T_{1,n}(k)$ can be viewed as a discretized version of the SN test statistic $\widetilde{T}_{1,n}(k)=\max_{1\leq t_1<k<t_2\leq n}T_n(t_1,k,t_2)$,} which is related to the scan statistics~\citep{Chan2013} and multiscale statistics~\citep{Frick2014}. However, $\widetilde{T}_{1,n}(k)$ is computationally impractical, thus we instead approximate $\widetilde{T}_{1,n}(k)$ by $T_{1,n}(k)$ computed on the nested window set $H_{1:n}(k)$. 

Based on the maximal test statistic $T_{1,n}(k)$ and a prespecified threshold $K_n$, the SN-based multiple change-point estimation~(SNCP) proceeds as follows. Starting with the full sample $\{Y_t\}_{t=1}^{n}$, we calculate $T_{1,n}(k), k=1,\cdots, n.$ Given that $\max_{k=1,\ldots,n} T_n(k)\leq K_n$, SNCP declares no change-point. Otherwise, SNCP sets $\widehat{k}=\arg\max_{k=1,\ldots,n} T_{1,n}(k)$ and we recursively perform SNCP on the subsample $\{Y_t\}_{t=1}^{\widehat{k}}$ and $\{Y_{t}\}_{t=\widehat{k}+1}^{n}$ until no change-point is declared.

Denote $W_{s,e}=\bigl\{(t_1,t_2)\big\vert s\leq t_1<t_2\leq e \bigl\}$ and $H_{s:e}(k)=H_{1:n}(k)\bigcap W_{s,e}$, which is the nested window set of $k$ on the subsample $\{Y_{t}\}_{t=s}^e$. Define the subsample maximal SN test statistic as $T_{s,e}(k)=\max\limits_{(t_1,t_2)\in H_{s:e}(k)}T_n(t_1,k,t_2)$. Algorithm \ref{alg_SNMCP} gives the formal description of SNCP.

\begin{algorithm}[h]
	\caption{SNCP for multiple change-point estimation}
	\label{alg_SNMCP}
	\KwIn{Time series $\{Y_t\}_{t=1}^{n}$, threshold $K_n$, window size $h=\lfloor n\epsilon \rfloor$.}
	\KwOut{Estimated change-points set $\widehat{\bf{k}}=(\widehat{k}_1,\cdots,\widehat{k}_{\widehat{m}})$}
	\KwIni{SNCP($1,n,K_n,h$), $\widehat{\bf{k}}=\varnothing$}
	\KwPro{SNCP($s,e,K_n,h)$}
	\eIf{$e-s+1< 2h$}{Stop}
	{$\widehat{k}^*=\argmax_{k=s,\cdots,e} T_{s,e}(k)$\;
		\eIf{$T_{s,e}(\widehat{k}^*)\leq K_n$ }{Stop}
		{$\widehat{\bf{k}}=\widehat{\bf{k}}\cup\widehat{k}^*$\;
			SNCP($s, \widehat{k}^*,K_n,h)$\; SNCP($\widehat{k}^*+1,e,K_n,h)$\;}
	}
\end{algorithm}


\textbf{Comparison with popular segmentation algorithms in the literature}: We remark that it is possible to combine the proposed SN test statistic with other segmentation algorithms designed for multiple change-point estimation, such as wild binary segmentation (WBS)~\citep{Fryzlewicz2014} or its variants including narrowest-over-threshold (NOT)~\citep{baranowski2019narrowest} and seeded binary segmentation (SBS)~\citep{Kovacs2020}. WBS and NOT use randomly generated intervals for searching multiple change-points, whereas SBS employs deterministic intervals. However, theoretical guarantees for such procedures can be challenging to establish as the above-mentioned segmentation algorithms are mainly used for change-point estimation in a sequence of independent data. Nevertheless, in Section~\ref{subsec: wbs_not} of the supplement, we provide an extensive numerical comparison between the proposed nested local-window segmentation algorithm (SNCP) and the combinations of the SN test with WBS, NOT and SBS, where the performance of SNCP is seen to be very competitive in terms of both statistical accuracy and computational efficiency.

\subsection{Assumptions and theoretical results}\label{subsec:theory}
In this section, we study the theoretical properties of the proposed SNCP for multiple change-point estimation. We operate under the  classical infill framework where we assume $k_i/n\to \tau_i \in (0,1)$ for $i=1,\ldots,m_o$ as $n\to \infty.$ Define $\tau_0=0$ and $\tau_{m_o+1}=1$, we further assume that $\min_{1\leq i\leq m_o+1}(\tau_i-\tau_{i-1})=\epsilon_o>\epsilon$, where $\epsilon$ is the window size parameter used in SNCP, which imposes an implicit upper bound for $m_o$ such that $m_o\leq 1/\epsilon.$ This is a common assumption in the literature for change-point testing and estimation under temporal dependence, see \cite{andrews1993}, \cite{Bai2003}, \cite{Davis2006} and \cite{Yau2016}. In practice, we set $\epsilon$ to be a small constant such as $\epsilon=0.05, 0.10, 0.15$, which can be based on prior information about the minimum spacing between consecutive change-points.

In Section~\ref{subsec: sensitivity} of the supplement, we conduct an extensive sensitivity analysis of SNCP w.r.t.\ the window size $\epsilon$ and the threshold $K_n$, and the result indicates SNCP is rather robust to the choices of $(\epsilon, K_n)$ as long as $\epsilon_o>\epsilon$, the violation of which could lead to unsatisfactory segmentation results. This suggests that the assumption $\epsilon_o>\epsilon$ is necessary both theoretically and empirically, and hence the proposed SNCP may not be suitable for time series with frequent change-points where $\epsilon_o$ is vanishing with $\epsilon_o=o(1)$; see \cite{Fryzlewicz2020} for a recent contribution to detecting frequent change-points.

Denote the true parameter for the $i$th segment by $\theta_i$ and denote the change size by $\delta_i=\theta_{i+1}-\theta_i$ for $i=1,\ldots,m_o.$ For ease of presentation, we assume that $\delta_i=c_i\delta$ for $i=1,\ldots,m_o$, where $c_i\neq 0$ is a fixed constant. Thus, the overall change size is controlled by $\delta.$ 

We assume the following expansions for the empirical functional $\widehat{\theta}_{a,b}=\theta(\widehat{F}_{a,b})$, which is a natural extension of the expansions \eqref{inf1} and \eqref{inf3} from the single change-point setting in Section \ref{subsec:theory_SCP} to the multiple change-point setting. Specifically, for $\widehat{\theta}_{a,b}$ computed exclusively on the $i$th stationary segments with $i=1,\cdots,m_o+1$, we assume
\begin{align}\label{theta_null}
	\widehat{\theta}_{a,b}=\theta_i+\frac{1}{b-a+1}\sum_{t=a}^{b}\xi_i(Y_t^{(i)})+r^{(i)}_{a,b},\quad \text{for}~k_{i-1}+1\leq a < b\leq k_{i},
\end{align}
where $\xi_i(Y_t^{(i)})$ is the influence function of the functional $\theta(\cdot)$ for the $i$th segment and $r^{(i)}_{a,b}$ denotes the remainder term. For $\widehat{\theta}_{a,b}$ computed based on a mixture of stationary segments, we further assume
\begin{align}\label{theta}
	\widehat{\theta}_{a,b}=&\theta(\omega_{a,b})+\frac{1}{b-a+1}\left[\sum_{t=a}^{k_{i}}\xi_{i}(Y_t^{(i)},\omega_{a,b})+\sum_{l=1}^{j-i}\sum_{t=k_{l+i-1}+1}^{k_{l+i}}\xi_{i+l}(Y_t^{(l+i)},\omega_{a,b})+\sum_{t=k_{j}+1}^{b}\xi_{j+1}(Y_t^{(j+1)},\omega_{a,b})\right]\nonumber\\
	&+r_{a,b}(\omega_{a,b}):=\theta_{a,b}+\bar{\xi}_{a,b}(\omega_{a,b})+r_{a,b}(\omega_{a,b}),
\end{align}
where $(k_i,k_{i+1},\cdots,k_j)$ with $i\leq j$ denotes the $j-i+1$ true change-points between $a$ and $b$ such that $k_{i-1}+1\leq a\leq k_i$ and $k_j+1\leq b\leq k_{j+1}$, and
\begin{align*}
	\omega_{a,b}=\left(\omega_{a,b}^{(1)},\cdots, \omega_{a,b}^{(m_o+1)} \right)^\top=\left(\overbrace{0,\cdots,0}^{\mbox{of}~i-1},\frac{k_{i}-a+1}{b-a+1},\frac{k_{i+1}-k_{i}}{b-a+1},\cdots,\frac{k_j-k_{j-1}}{b-a+1},\frac{b-k_j}{b-a+1},\overbrace{0,\cdots,0}^{\mbox{of}~m_o-j}\right)^\top,
\end{align*} 
denotes the proportion of each stationary segment in $\{Y_t\}_{t=a}^b$, $\theta(\omega_{a,b})$ denotes $\theta(\cdot)$ evaluated at the mixture distribution $F^{\omega_{a,b}}=\sum_{i=1}^{m_o+1}\omega_{a,b}^{(i)}F^{(i)}$ and $r_{a,b}(\omega_{a,b})$ denotes the remainder term.

Similar to the single change-point scenario, the expansion \eqref{theta_null} of $\widehat{\theta}_{a,b}$ with $k_{i-1}+1\leq a < b\leq k_{i}$ can be viewed as a special case of \eqref{theta} where the mixture distribution is pure and $\omega_{a,b}$ is defined as $\omega_{a,b}^{(i)}=1$ and $\omega_{a,b}^{(i')}=0, i'\neq i$. We proceed by making the following assumptions.
\begin{ass}\label{ass_no1}
	$\mathrm{(i)}$ 	For some $\sigma_i>0$, $i=1,\cdots,m_o+1$, 
	$$\frac{1}{\sqrt{n}}\sum_{t=1}^{\lfloor nr\rfloor}\Big(\xi_{1}(Y_t^{(1)}),\cdots,\xi_{m_o+1}(Y_t^{(m_o+1)})\Big)\Rightarrow (\sigma_1 B^{(1)}(r), \cdots,\sigma_{m_o+1}B^{(m_o+1)}(r)),$$	
	where $B^{(i)}(\cdot)$, $i=1,\cdots,m_o+1$ are standard Brownian motions.
	
	\noindent $\mathrm{(ii)} ~\sup_{1\leq a<b\leq n}|(b-a+1)\bar{\xi}_{a,b}(\omega_{a,b})|=O_p(n^{1/2}).$
\end{ass}

\begin{ass}\label{ass_no2}
	$\sup_{1\leq a<b\leq n}|(b-a+1)r_{a,b}(\omega_{a,b})|= o_p(n^{1/2}).$
\end{ass}
Assumptions \ref{ass_no1} and \ref{ass_no2} are natural extensions of Assumptions \ref{ass_influence} and \ref{ass_remainder} to the multiple change-point setting and can also be verified for smooth function models and quantile under mild conditions. We refer to Sections~\ref{sec:verify} and \ref{sec:verify_quantile} of the supplement for more details.

\begin{ass}\label{ass_no3}
	For $1\leq a<b\leq n$, $\theta_{a,b}=\theta(\omega_{a,b})$ can be expressed almost linearly such that
	$\sup_{1\leq a<b\leq n}\Big|\theta_{a,b}-(\theta_{1},\cdots,\theta_{m_o+1})\omega_{a,b}\Big|=
	\sup_{1\leq a<b\leq n}\Big|\theta_{a,b}-\sum_{i=1}^{m_o+1}\omega_{a,b}^{(i)}\theta_i\Big|=o(n^{-1/2})$.
\end{ass}
Assumption \ref{ass_no3} imposes a relatively strong technical condition on the functional $\theta(\cdot)$ such that $\theta_{a,b}\approx \sum_{i=1}^{m_o+1}\omega_{a,b}^{(i)}\theta_i$. Assumption \ref{ass_no3} holds trivially for mean change and is typically satisfied when $\theta(\cdot)$ is the only quantity that changes, which is a common assumption in testing-based change-point estimation literature. For example, Assumption \ref{ass_no3} holds for variance, (auto)-covariance change with constant mean~\citep{Aue2009,Cho2012} and (auto)-correlation change with constant mean and variance~\citep{Wied2012}. Numerical experiments conducted in Section \ref{subsec: multipara_main} and Sections \ref{subsec:poweracf}-\ref{subsec:powermeanvariance} of the supplement indicate that SNCP is robust and continues to perform well when Assumption \ref{ass_no3} can not be easily verified.

An alternative Assumption \ref{ass_no3star} is provided in Section \ref{subsec:Alter_assno3} of the supplement, which is a natural extension of Assumption \ref{ass_theta} to the multiple change-point setting and further includes Assumption \ref{ass_no3} as a special case. We defer Assumption \ref{ass_no3star} to the supplement as it is a more involved technical assumption.

{\textbf{Remark 1} (Verification of assumptions): Assumptions \ref{ass_no1}-\ref{ass_no3} are high-level assumptions made on a general functional $\theta(\cdot)$ to facilitate presentation. In Sections \ref{sec:verify} and \ref{sec:verify_quantile} of the supplement, under mild conditions, we provide verification of Assumptions \ref{ass_no1}-\ref{ass_no3} for commonly used functionals including the smooth function model and quantile. In general, the assumptions can be verified for mean change, variance and (auto)-covariance change with constant mean or with concurrent small-scale mean change, (auto)-correlation change with constant mean and variance or with concurrent small-scale mean and variance change, and quantile change with density functions that are smooth and bounded. In particular, the verification of Assumption \ref{ass_no2} for quantile is highly nontrivial and of independent interest. It essentially provides a uniform Bahadur representation for quantiles in subsamples. Our result allows for change-points and temporal dependence, and thus  generalizes the ones in \cite{Wu2005} and \cite{dette2020a}. }

For $u\in(\epsilon, 1-\epsilon)$, define the scaled limit of $H_{1:n}(k)$ by $H_{\epsilon}(u)=\bigl\{ (u_1,u_2)\big\vert u_1=u-j_1\epsilon, j_1=1,\cdots, \lfloor u/\epsilon\rfloor; u_2=u+j_2\epsilon, j_2=1,\cdots, \lfloor (1-u)/\epsilon \rfloor  \bigr\}$ and define $\Delta(u_1, u, u_2)=B(u)-B(u_1)-\frac{u-u_1}{u_2-u_1}\{B(u_2)-B(u_1)\}$, where $B(\cdot)$ is a standard Brownian motion. Theorem \ref{thm1} gives the consistency result of SNCP for multiple change-point estimation.
\begin{thm}\label{thm1}
	(i) Under the no change-point scenario, and Assumptions \ref{ass_no1}(i) and \ref{ass_no2}, we have 
	\begin{align}\label{eq:limiting_null}
		\max_{k=1,\cdots,n} T_{1,n}(k) \overset{\mathcal{D}}{\longrightarrow} G_\epsilon= \sup_{u\in(\epsilon,1-\epsilon)}\max_{(u_1,u_2)\in H_{\epsilon}(u)}{D(u_1,u,u_2)^2}/{V(u_1,u,u_2)},
	\end{align}
	where $D(u_1,u,u_2)=\frac{1}{\sqrt{u_2-u_1}}\Delta(u_1, u, u_2)$ and 
	$V(u_1,u,u_2)=\frac{1}{(u_2-u_1)^2}\left( \int_{u_1}^{u} \Delta(u_1, s, u)^2 ds + \int_{u}^{u_2} \Delta(u, s, u_2)^2  ds  \right)$.
	
	\noindent (ii) Under the multiple change-point scenario,
	suppose Assumption \ref{ass_K}, Assumptions \ref{ass_no1}, \ref{ass_no2} and \ref{ass_no3}~(or \ref{ass_no3star}) hold	and suppose $\epsilon < \epsilon_o$, we have
	$$\lim\limits_{n\to\infty}P(\widehat{m}=m_o\quad \text{and}\quad \max_{1\leq i\leq m_o}|\widehat{k}_i-k_i|\leq \iota_n)= 1,$$
	for any sequence $\iota_n$ such that $\iota_n/n\to 0$ and $\iota_n^{-2}\delta^{-2}n\to0$ as $n\rightarrow\infty$. 
\end{thm}

Theorem \ref{thm1}(i) characterizes the asymptotic behavior of SNCP under no change-point and thus provides a natural choice of threshold $K_n$. In practice, we set $K_n$ as a high quantile, e.g.\ 90\% or 95\% quantile of $G_\epsilon$ to control the Type-I error of SNCP. For a given window size $\epsilon$, $G_\epsilon$ is a pivotal distribution and its critical values can be obtained via simulation. Theorem \ref{thm1}(ii) indicates that SNCP can correctly identify the number of change-points $m_o$ with an increasing threshold $K_n$ of a proper order. Note that the localization error rate of SNCP is the same as the single change-point scenario in Theorem \ref{thm_onechange}.

{Theorem \ref{thm1}(ii) assumes all changes have the same order $\delta$ and requires $\iota_n^{-2}\delta^{-2}n\to0$ to achieve consistency. In fact, this can be relaxed to allow multiscale changes and we then require $\iota_n^{-2}\delta^2_{max}\delta_{min}^{-4}n\to0$, where $\delta_{max}$ and $\delta_{min}$ denotes the maximum and minimum change size. This multiscale condition matches the one required by \cite{lavielle2000least} for multiple change-point estimation in mean under temporal dependence (cf.\ Theorem 3 therein).}

{\textbf{Remark 2} (Localization error rate and local refinement): Set the change size $\delta=D_0 n^{-c}$ with $c\in [0,1/2)$ and $D_0\not=0$, Theorem \ref{thm1}(ii) implies that $n^{1/2+c}=o(\iota_n)$. Under the fixed change size~($c=0$), it implies that the convergence rate $\iota_n/n$ of SNCP is at best $1/\sqrt{n}$, which is slower than the optimal rate $1/n$ for change-point estimation in mean under temporal dependence, see \cite{bai1994} and \cite{lavielle2000least}.\footnote{{For multiple change-point estimation of univariate mean in a sequence of independent sub-Gaussian observations, this is further shown as the minimax optimal localization rate, see \cite{wangEJS2020}, \cite{Verzelen2020} and references therein.}} We note that the derived rate is technically difficult to be further improved due to the complex nature of the self-normalizer $V_n(k)$. On the other hand, the derived rate applies to a general functional, which seems not well studied in the literature. Nevertheless, in Section \ref{sec:local_refine} of the supplement, we further propose a simple and intuitive local refinement procedure, which provably improves the localization error rate of SNCP to $1/n$ for the mean functional. The key observation is that by Theorem \ref{thm1}, SNCP can asymptotically isolate each single change-point and thus a simple CUSUM statistic can be used within a well-designed local interval around each estimated change-point $\widehat{k}_i$ by SNCP to achieve further refinement. We refer to Sections \ref{sec:local_proc}-\ref{sec:local_simu} for more detailed theoretical and numerical results of the procedure.
	
}

\subsection{Extension to vector-valued functionals}\label{subsec:vector}
In this section, we discuss the extension of SNCP to a vector-valued functional, where $\bftheta(\cdot)\in \mathbb{R}^d$ with $d>1.$ A natural example is change-point detection in mean or covariance matrix of multivariate time series, see for example \cite{Aue2009}. Additionally, for a univariate time series, we may be interested in detecting any structural break among multiple parameters of interest, such as examining mean and variance together or examining multiple quantile levels simultaneously.

Note that the dimension of the underlying time series $\{Y_t\}_{t=1}^n$ may or may not equal to that of $\bftheta$~(i.e.\ $d$). For change-point estimation in mean of multivariate time series, we have $\bftheta=E(Y_t)$ and the dimension of $Y_t$ is $d$. However, for change-point estimation in covariance matrix~($\bftheta=\text{Cov}(Y_t)$) or multiple parameters~(e.g.\ $Y_t\in \mathbb{R}$ and $\bftheta=(E(Y_t), \text{Var}(Y_t))^\top$), the dimension of $Y_t$ can be smaller than $d$. We examine the performance of SNCP for all three cases via numerical experiments in Section \ref{sec:simulation}.

To accommodate the vector-valued functional, we modify the SN test statistic in \eqref{eq:SN_subsample} such that
\begin{align}\label{eq:SN_subsample_vector}
	T_n^*(t_1,k,t_2)={D_{n}^*(t_1,k,t_2)^\top}{V_{n}^*(t_1,k,t_2)}^{-1}D_{n}^*(t_1,k,t_2),
\end{align}
where $\widehat{\bftheta}_{a,b}=\bftheta(\widehat{F}_{a,b})$ with $\widehat{F}_{a,b}$ being the empirical distribution of $\{Y_t\}_{t=a}^b$ and
\begin{flalign*}
	D_n^*(t_1,k,t_2)=&\frac{(k-t_1+1)(t_2-k)}{(t_2-t_1+1)^{3/2}}(\widehat{\bftheta}_{t_1,k}-\widehat{\bftheta}_{k+1,t_2}), \quad
	V_n^*(t_1,k,t_2)=L_n^*(t_1,k,t_2)+R_n^*(t_1,k,t_2),\\
	L_n^*(t_1,k,t_2)=&\sum_{i=t_1}^{k}\frac{(i-t_1+1)^2(k-i)^2}{(t_2-t_1+1)^{2}(k-t_1+1)^2}(\widehat{\bftheta}_{t_1,i}-\widehat{\bftheta}_{i+1,k})(\widehat{\bftheta}_{t_1,i}-\widehat{\bftheta}_{i+1,k})^\top,\\
	R_n^*(t_1,k,t_2)=&\sum_{i=k+1}^{t_2}\frac{(t_2-i+1)^2(i-1-k)^2}{(t_2-t_1+1)^{2}(t_2-k)^2}(\widehat{\bftheta}_{i,t_2}-\widehat{\bftheta}_{k+1,i-1})(\widehat{\bftheta}_{i,t_2}-\widehat{\bftheta}_{k+1,i-1})^\top.
\end{flalign*}
With a pre-specified threshold $K_n$, SNCP proceeds as in Algorithm \ref{alg_SNMCP} where the only difference is that we replace $T_{s,e}(k)$ with $T^*_{s,e}(k)=\max\limits_{(t_1,t_2)\in H_{s:e}(k)}T^*_n(t_1,k,t_2)$ as defined in \eqref{eq:SN_subsample_vector}.

\textbf{Limiting distribution under no change-point scenario}: We first derive the limiting null distribution of $\max_{k=1,\cdots,n} T_{1,n}^*(k)$, which is pivotal and thus provides natural choices of the threshold $K_n$. We assume the subsample estimator $\widehat{\bftheta}_{a,b}$ for the parameter of interest ${\bftheta}\in\mathbb{R}^d$ admits the following expansion
$$
\widehat{\bftheta}_{a,b}={\bftheta}_0+\frac{1}{b-a+1}\sum_{t=a}^{b}\xi(Y_t)+r_{a,b},
$$
where ${\bftheta}_0$ denotes the true value of $\bftheta$, $\xi(Y_t)\in\mathbb{R}^d$ denotes the influence function of $\bftheta$ and $r_{a,b}\in\mathbb{R}^d$ is the remainder term. We further impose the following mild assumptions.
\begin{ass}\label{ass_multi1}
	For some positive definite matrix $\Sigma\in\mathbb{R}^{d\times d}$, we have 
	$$
	\frac{1}{\sqrt{n}}\sum_{t=1}^{\lfloor nr\rfloor} \xi(Y_t)\Rightarrow \Sigma^{1/2} \mathcal{B}_d(r),
	$$
	where $\mathcal{B}_d(\cdot)$ is a $d$-dimensional Brownian motion with independent entries.
\end{ass}
Assumption \ref{ass_multi1} is a standard functional central limit theorem~(FCLT) result commonly assumed in the SN literature under the no change-point scenario, and can be verified under mild moment and weak dependence conditions, see for example, \cite{shao2010self} (Assumption 2.1), \cite{shao2010testing} (Assumption 3.1) and \cite{dette2020a} (Assumption 3.1).

\begin{ass}\label{ass_multi2}
	The remainder term $r_{a,b}$ is asymptotically negligible such that
	$$
	\sup_{1\leq a<b\leq n}(b-a+1)\|r_{a,b}\|_2=o_p(n^{1/2}).
	$$
\end{ass}
\begin{pro}\label{pro_multi_parameter}
	Under the no change-point scenario, given Assumptions \ref{ass_multi1} and \ref{ass_multi2}, we have 
	$$
	\max_{k=1,\cdots,n} T_{1,n}^*(k)\overset{\mathcal{D}}{\longrightarrow} G_{\epsilon,d}^*= \sup_{u\in(\epsilon,1-\epsilon)}\max_{(u_1,u_2)\in H_{\epsilon}(u)}{D^*(u_1,u,u_2)}^{\top}{V^*(u_1,u,u_2)}^{-1}{D^*(u_1,u,u_2)},
	$$
	where 
	$D^*(u_1,u,u_2)=\frac{1}{\sqrt{u_2-u_1}}\mathbf{\Delta}(u_1,u,u_2)$ and $V^*(u_1,u,u_2)=\frac{1}{(u_2-u_1)^2}\Big(\int_{u_1}^{u}\mathbf{\Delta}(u_1,s,u)\mathbf{\Delta}(u_1,s,u)^{\top}ds+\int_{u}^{u_2}\mathbf{\Delta}(u,s,u_2)\mathbf{\Delta}(u_1,s,u)^{\top}ds\Big)$
	with $\mathbf{\Delta}(u_1,u,u_2)=\mathcal{B}_d(u)-\mathcal{B}_d(u_1)-\frac{u-u_1}{u_2-u_1}[\mathcal{B}_d(u_2)-\mathcal{B}_d(u_1)]$.
\end{pro}
The proof of Proposition \ref{pro_multi_parameter} is straightforward and follows the same argument as the proof of Theorem 2.1 in \cite{shao2010self} and the continuous mapping theorem, hence omitted. For a given dimension $d$ and window size $\epsilon$, the limiting distribution $G_{\epsilon,d}^*$ is pivotal and its critical values can be obtained via simulation. Table \ref{tab: criticalvalue} tabulates the critical values of $G_{\epsilon,d}^*$ for $\epsilon=0.05$ and $d=1,\cdots, 10$. Note that for $d=1$, $G_{\epsilon,d}^*$ coincides with the univariate limiting distribution $G_\epsilon$ derived in Theorem \ref{thm1}(i). 

\textbf{Consistency of SNCP}: To ease presentation and facilitate understanding, we first establish the consistency of SNCP for change-point estimation in mean of multivariate time series. We then provide further discussions on how to extend the consistency result to a general vector-valued functionals.

Specifically, we operate under the following data generating process for $\{Y_t \in \mathbb{R}^d\}_{t=1}^n$ such that
$$
Y_t=X_t+\theta_{i},\quad k_{i-1}+1\leq t\leq k_{i}, \quad \text{for}~ i=1,\cdots,m_o+1,
$$
where $\{X_t\}_{t=1}^n$ is a $d$-dimensional stationary time series with $E(X_t)=0$, $k_0:=0<k_1<\cdots<k_{m_o}<k_{m_o+1}:=n$ denote the (potential) change-points, and $\theta_i\in\mathbb{R}^d$ denotes the mean of the $i$th segment. We assume that, for $i=1,\cdots, m_o$,  $\theta_{i+1}-\theta_i=\eta_i\delta$ where $\eta_i\in\mathbb{R}^d/\{\bf{0}\}$ is a nonzero vector. Thus, the overall change size is controlled by $\delta.$

Same as in Section \ref{subsec:theory}, we use the infill framework where we assume $k_i/n\to \tau_i \in (0,1)$ for $i=1,\ldots,m_o$ as $n\to \infty.$ Define $\tau_0=0$ and $\tau_{m_o+1}=1$, we again require that $\min_{1\leq i\leq m_o+1}(\tau_i-\tau_{i-1})=\epsilon_o>\epsilon$, where $\epsilon$ is the window size parameter used in SNCP.

\begin{thm}\label{thm_onechange_multi}
	Suppose $\{X_t\}_{t=1}^n$ satisfies the invariance principle such that
	$n^{-1/2}\sum_{t=1}^{\lfloor nr\rfloor}X_t\Rightarrow \Sigma_X^{1/2} \mathcal{B}_d(r)$, where $\Sigma_X$ is  a positive definite matrix.
	
	\noindent $\mathrm{(i)}$ Under the no change-point scenario, we have 
	$\max_{k=1,\cdots,n}T^*_{1,n}(k)\overset{\mathcal{D}}{\longrightarrow} G_{\epsilon,d}^* $.

	\noindent $\mathrm{(ii)}$ Under the multiple change-point scenario, suppose Assumption \ref{ass_K} hold and suppose $\epsilon<\epsilon_o$, we have
	$$\lim\limits_{n\to\infty}P(\widehat{m}=m_o\quad \text{and}\quad\max_{1\leq i\leq m_o}|\widehat{k}_i-k_i|\leq \iota_n)= 1,$$
	for any sequence $\iota_n$ such that $\iota_n/n \to 0$ and $\iota_n^{-2}\delta^{-2}n\to0$ as $n\to\infty$.
\end{thm}

Compared to the univariate result in Theorem~\ref{thm1}(ii), it can be seen that the same localization rate is obtained in Theorem~\ref{thm_onechange_multi}(ii) for the multivariate mean case. However, compared to the univariate proof, the technical argument needed for Theorem~\ref{thm_onechange_multi} is substantially different, which is indeed much more challenging as it requires the analysis of a random matrix and its inverse, since the self-normalizer $V_n^*(t_1,k,t_2)$ is a random matrix in $\mathbb{R}^{d\times d}$ due to the vector nature of the functional $\bftheta(\cdot)$.

It is easy to see that the result of Theorem \ref{thm_onechange_multi} can be directly used to establish consistency of SNCP for change-point estimation in covariance matrix of $\{Y_t\in \mathbb{R}^d\}_{t=1}^n$~(assuming constant mean $E(Y_t)$), as the problem can be transformed into multivariate mean change-point estimation for the $(d+d^2)/2$-dimensional time series $\{(Y_{ti}\cdot Y_{tj})_{i\leq j}\}_{t=1}^n$, see for example \cite{Aue2009}.

\textbf{Remark 3}~(Extension to general vector-valued functionals): To further extend the consistency result in Theorem \ref{thm_onechange_multi} to a general vector-valued functional $\bftheta(\cdot)$, we need an additional assumption on the (approximate) linearity of $\bftheta$, similar to Assumption \ref{ass_no3} of the univariate case. Combined with Assumption \ref{ass_multi1}~(FCLT) and \ref{ass_multi2}~(asymptotic negligibility of reminder terms), the same argument used for the multivariate mean in Theorem \ref{thm_onechange_multi} can then be applied to establish consistency of SNCP for the general functional $\bftheta$. We omit the details to conserve space.

\section{Simulation Studies}\label{sec:simulation}
In this section, we conduct extensive numerical experiments to demonstrate the promising performance of SNCP for a wide range of change-point detection problems under temporal dependence. Under the unified framework of SNCP, we consider change-point estimation for four different settings: mean, covariance matrix, multi-parameter and correlation. In the supplement, we further consider change-point estimation for variance, autocorrelation and quantile.

For comparison, we further implement several state-of-the-art nonparametric change-point detection methods in the literature that are explicitly designed to accommodate temporal dependence. Specifically, {\bf (A)} For mean change, we compare with the classical CUSUM with binary segmentation~\citep{Csoergoe1997}~(hereafter CUSUM) and \cite{Bai2003}~(hereafter BP), which are designed for detecting mean change in time series and uses a model selection approach to simultaneously detect all change-points. {\bf (B)} For covariance matrix change, we compare with the CUSUM method in \cite{Aue2009}~(hereafter AHHR). {\bf (C)} For correlation change, we compare with \cite{Galeano2017}~(hereafter GW), which is essentially a combination of binary segmentation and the correlation change test proposed in \cite{Wied2012}. {\bf (D)} For variance change and autocorrelation change, we compare with \cite{Cho2012}~(hereafter MSML) and \cite{Korkas2017}~(hereafter KF). Both methods are designed for detecting second-order structural change in time series based on wavelet representation. {\bf (E)} For multi-parameter change and quantile change, to our best knowledge, there is no existing nonparametric method that works under temporal dependence. For illustration, we compare with the energy statistics based segmentation in \cite{Matteson2014}~(hereafter ECP) for multi-parameter change and with the multiscale quantile segmentation in \cite{Vanegas2020}~(hereafter MQS) for quantile change. Both ECP and MQS require temporal independence. {All methods are implemented using the recommended setting in the corresponding \texttt{R} packages or papers. We refer to Section \ref{subsec:dgp} of the supplement for implementation details of these methods.} 

\textbf{Implementation details of SNCP}: Throughout Sections \ref{sec:simulation}, we set the window size $\epsilon$ of SNCP to be $\epsilon=0.05$. We denote SNCP for mean as SNM, for covariance matrix as SNCM,  for multi-parameter as SNMP,  for correlation as SNC,  for variance as SNV,  for autocorrelation as SNA, and for quantile as SNQ. In addition, SNM90 denotes SNM using 90\% quantile~(i.e.\ critical value at $\alpha=0.1$) of the limiting null distribution $G^*_{\epsilon,d}$ as the threshold $K_n$, and similarly for other types of change and levels of critical value. For the power analysis in Sections \ref{sec:simulation} and real data applications in Section \ref{sec:applications}, the threshold $K_n$ for SNCP is set at 90\% quantile of $G^*_{\epsilon,d}$~(i.e. $\alpha=0.1$), which can be found in Table \ref{tab: criticalvalue} for $d=1,2,\cdots,10.$

We remark that the performance of SNCP is robust w.r.t. the window size $\epsilon$ and the quantile level $\alpha$ as the limiting distribution $G^*_{\epsilon,d}$, and thus the threshold $K_n$, adapt to the effect of $\epsilon$ and $\alpha$. We refer to Section \ref{subsec: sensitivity} of the supplement for a detailed sensitivity analysis.

\begin{table}[H]
	\centering
	\caption{Critical values of the limiting null distribution $G_{\epsilon,d}^*$ with $\epsilon=0.05$.}
	\label{tab: criticalvalue}
	{\small 
		\begin{tabular}{c|cccccccccc}
			\hline\hline
			\diagbox[height=1.3em, width=10em]{$1-\alpha$}{ $d$} & 1 & 2 & 3 & 4 & 5 & 6 & 7 & 8 & 9 & 10 \\ \hline
			90\% & 141.9 & 208.2 & 275.0 & 344.4 & 415.9 & 492.5 & 568.4 & 651.4 & 740.3 & 823.5 \\ 
			95\% & 165.5 & 237.5 & 309.1 & 387.5 & 464.5 & 541.7 & 624.1 & 713.3 & 808.6 & 898.9 \\
			\hline\hline
		\end{tabular}
	}
\end{table}

\textbf{Error measures of change-point estimation}: To assess the accuracy of change-point estimation, we use the Hausdorff distance and adjusted Rand index~(ARI). The Hausdorff distance is defined as follows. Denote the set of true (relative) change-points as $\bftau_o$ and the set of estimated (relative) change-points as $\hat{\bftau}$, we define $d_1(\bftau_o, \hat{\bftau})=\max_{\tau_1 \in \hat{\bftau}}\min_{\tau_2 \in \bftau_o}|\tau_1-\tau_2| \text{~~  and  ~~} d_2(\bftau_o, \hat{\bftau})=\max_{\tau_1 \in \bftau_o}\min_{\tau_2 \in \hat{\bftau}}|\tau_1-\tau_2|,$ where $d_1$ measures the over-segmentation error of $\hat{\bftau}$ and $d_2$ measures the under-segmentation error of $\hat{\bftau}$. The Hausdorff distance is $d_H(\bftau_o, \hat{\bftau})=\max(d_1(\bftau_o, \hat{\bftau}), d_2(\bftau_o, \hat{\bftau}))$. The ARI is originally proposed in \cite{Morey1984} as a measure of similarity between two different partitions of the same observations for evaluating the accuracy of clustering. Under the change-point setting, we calculate the ARI between partitions of the time series given by $\hat{\bftau}$ and $\bftau_o$. Ranging from 0 to 1, a higher ARI indicates more coherence between the two partitions by $\hat{\bftau}$ and $\bftau_o$ and thus more accurate change-point estimation.

\subsection{No change}\label{subsec:size}
We first investigate the performance of SNCP under the null, where the time series is stationary with no change-point. We report the performance of SNM and SNV observed in extensive numerical experiments. The performance of SNCP for other functionals is similar and thus omitted.

We simulate a stationary univariate time series $\{Y_t\}_{t=1}^n$ from an AR(1) process $Y_t=\rho Y_{t-1}+\epsilon_t,$ where $\{\epsilon_t\}$ is \textit{i.i.d.}\ standard normal $N(0,1)$. We set $n=1024,  4096$\footnote{$n$ is deliberately set as power of 2 as MSML in \cite{Cho2012} can only handle such sample size.} and vary $\rho\in \{-0.8,-0.5,0,0.5,0.8\}$ to examine robustness of SNCP against false positives~(i.e.\ Type-I error) under different direction and strength of temporal dependence. {Section \ref{subsec:size_supplement} of the supplement further provides the simulation results for $n=512$.} For each combination of ($n, \rho$), we repeat the simulation 1000 times.

The numerical result is summarized in Table \ref{tab: null_gaussian_error}, where we report the proportion of $\hat{m}=0$, $\hat{m}=1$ and $\hat{m}\geq 2$ among 1000 experiments. In general, the observation is as follows. SNCP gives satisfactory performance under moderate temporal dependence with $|\rho|\leq 0.5$ for all sample sizes and its performance further improves as the sample size $n$ increases.

BP performs well under $\rho=-0.8,-0.5,0$ but exhibits severe over-rejection  under positive temporal dependence for $\rho=0.5, 0.8$ and the performance does not improve as $n$ increases. KF and MSML perform well under $\rho=0,0.5,0.8$ but produce high proportion of false positives under negative temporal dependence for $\rho=-0.5, -0.8$ and the performance does not improve as $n$ increases. Overall, SNCP provides reasonably accurate size under different direction and strength of temporal dependence and achieves the target size as the sample size $n$ increases.

\begin{table}[]
	\caption{Performance under no change-point scenario with $m_o=0$.}
	\label{tab: null_gaussian_error}
	{\small
		\centering
		\begin{tabular}{c|rrr|rrr|rrr|rrr|rrr}
			\hline
			\hline  $n= 1024 $ & \multicolumn{3}{c}{$\rho= -0.8 $}& \multicolumn{3}{c}{$\rho= -0.5 $}& \multicolumn{3}{c}{$\rho= 0 $}& \multicolumn{3}{c}{$\rho= 0.5 $}& \multicolumn{3}{c}{$\rho= 0.8 $} \\
			\hline
			$\hat{m}$ & $0$ & $1$ & $\geq 2$ & $0$ & $1$ & $\geq 2$ & $0$ & $1$ & $\geq 2$ & $0$ & $1$ & $\geq 2$ & $0$ & $1$ & $\geq 2$  \\
			\hline
			SNM90 & 0.99 & 0.01 & 0.00 & 0.96 & 0.04 & 0.00 & 0.93 & 0.06 & 0.00 & 0.87 & 0.12 & 0.01 & 0.60 & 0.30 & 0.10 \\ 
			BP & 1.00 & 0.00 & 0.00 & 1.00 & 0.00 & 0.00 & 0.99 & 0.01 & 0.00 & 0.35 & 0.12 & 0.53 & 0.00 & 0.00 & 1.00 \\ 
			\hline
			SNV90 & 0.80 & 0.18 & 0.02 & 0.90 & 0.09 & 0.01 & 0.90 & 0.09 & 0.01 & 0.86 & 0.12 & 0.01 & 0.73 & 0.22 & 0.05 \\ 
			KF & 0.18 & 0.20 & 0.63 & 0.76 & 0.14 & 0.10 & 0.96 & 0.03 & 0.01 & 0.95 & 0.04 & 0.01 & 0.94 & 0.04 & 0.02 \\ 
			MSML & 0.48 & 0.33 & 0.19 & 0.84 & 0.15 & 0.01 & 0.92 & 0.08 & 0.00 & 0.92 & 0.08 & 0.00 & 0.90 & 0.09 & 0.00 \\ 
			\hline \hline
			$n= 4096 $ & \multicolumn{3}{c}{$\rho= -0.8 $}& \multicolumn{3}{c}{$\rho= -0.5 $}& \multicolumn{3}{c}{$\rho= 0 $}& \multicolumn{3}{c}{$\rho= 0.5 $}& \multicolumn{3}{c}{$\rho= 0.8 $} \\
			\hline
			$\hat{m}$ & $0$ & $1$ & $\geq 2$ & $0$ & $1$ & $\geq 2$ & $0$ & $1$ & $\geq 2$ & $0$ & $1$ & $\geq 2$ & $0$ & $1$ & $\geq 2$  \\
			\hline
			SNM90 & 0.94 & 0.06 & 0.00 & 0.89 & 0.10 & 0.00 & 0.89 & 0.10 & 0.01 & 0.88 & 0.11 & 0.01 & 0.84 & 0.14 & 0.02 \\ 
			BP & 1.00 & 0.00 & 0.00 & 1.00 & 0.00 & 0.00 & 1.00 & 0.00 & 0.00 & 0.49 & 0.13 & 0.38 & 0.00 & 0.00 & 1.00 \\ 
			\hline
			SNV90 & 0.88 & 0.12 & 0.00 & 0.90 & 0.10 & 0.01 & 0.91 & 0.08 & 0.00 & 0.90 & 0.09 & 0.01 & 0.85 & 0.13 & 0.02 \\ 
			KF & 0.02 & 0.01 & 0.97 & 0.54 & 0.17 & 0.29 & 0.90 & 0.06 & 0.04 & 0.92 & 0.05 & 0.04 & 0.88 & 0.06 & 0.06 \\ 
			MSML & 0.38 & 0.27 & 0.36 & 0.80 & 0.18 & 0.02 & 0.92 & 0.08 & 0.00 & 0.92 & 0.08 & 0.00 & 0.90 & 0.10 & 0.00 \\ 
			\hline \hline
		\end{tabular}
	}
\end{table}

\subsection{Change in mean}\label{subsec:powermean}
For mean change, we first simulate a stationary $d$-dimensional time series $\{X_t=(X_{t1},\cdots,X_{td})\}_{t=1}^n$ from a VAR(1) process with $X_t=\rho \mathbf I_d X_{t-1}+\epsilon_t,$ where $\{\epsilon_t\}$ is \textit{i.i.d.}\ standard $d$-variate normal $N(0,\mathbf I_d)$, and $\mathbf I_d$ denotes the $d$-dimensional identity matrix. We then generate time series $\{Y_t\}_{t=1}^n$ with piecewise constant mean based on $\{X_t\}_{t=1}^n$.
\begin{align*}
	&\text{(M1)}: n=600, &\rho=&0.2, &Y_t&=\begin{cases} 
		0+ X_{t}, & t\in [1,100], [201,300], [401,500],\\
		{2/\sqrt{d}+ X_{t}}, & t\in [101,200], [301,400], [501,600].
	\end{cases}\\
	&\text{(M2)}: n=1000, &\rho=&0.5, &Y_t&=\begin{cases} 
		{-3/\sqrt{d}+ X_{t}}, & t\in [1,75], [526,575],\\
		0+ X_{t}, & t\in [76,375], [426,525], [576,1000],\\
		{3/\sqrt{d}+X_{t}}, & t\in[376, 425].
	\end{cases}\\
	&\text{(M3)}: n=2000, &\rho=&-0.7, &Y_t&=\begin{cases} 
		{0.4/\sqrt{d}+ X_{t}}, & t\in [1,1000], [1501,2000],\\
		0+ X_{t}, & t\in [1001,1500].\\
	\end{cases}
\end{align*}
(M1) has evenly spaced change-points with moderate temporal dependence, (M2) features abrupt changes where shortest segments have only 50 or 75 time points with change-points mainly located at the first half of the time series, and (M3) has longer segments with small-scale changes. Typical realizations of (M1)-(M3) for $d=1$ can be found in Figure \ref{fig: dgp} of the supplementary material.

Note that the change size in (M1)-(M3) is normalized by $\sqrt{d}$ to keep the signal-to-noise ratio~(SNR) the same across time series of different dimensions. This enables us to isolate and examine the effect of dimension $d$ on estimation. Intuitively, a larger $d$ makes the estimation more difficult as the quality of finite sample approximation by FCLT worsens for higher dimension.

We set the dimension $d=1,5,10$. Note that BP only works for $d=1$~(i.e.\ univariate time series) and thus is not included in the comparison for $d=5,10$. The estimation results for $d=1$ and $d=5$ are summarized in Table \ref{tab: mean_gaussian_error}, where we report the distribution of $\hat{m}-m_o$, average ARI, over- and under-segmentation errors $d_1$, $d_2$ and Hausdorff distance $d_H$ among 1000 experiments. The estimation result for $d=10$ can be found in Table \ref{tab: mean_gaussian_error_add2} of the supplement.

\textbf{Univariate time series $d=1$}: For (M1), all methods perform well overall, though  CUSUM tends to greatly over-estimate the number of change-points $m_o$, as reflected by the distribution of $\hat{m}-m_o$. For (M2), SNM tends to slightly under-estimate $m_o$~(missing a short segment) while BP and CUSUM severely over-estimate $m_o$ and provide much less accurate estimation with noticeably larger Hausdorff distance $d_H$ and smaller ARI. For (M3), which corresponds to strong negative dependence, BP experiences severe power loss and have large under-segmentation error $d_2$. In summary, BP and CUSUM are prone to produce false positives under positive dependence, and BP may lose power under strong negative dependence. SNM is robust but may experience power loss when detecting short segment changes.

\textbf{Multivariate time series $d=5, 10$}: For (M1) and (M3), the estimation accuracy of SNM is remarkably robust to the increasing dimension, where the ARI and $d_H$ achieved by SNM only worsen slightly from $d=1$ to $d=5$. This also holds true for $d=10$~(see Table \ref{tab: mean_gaussian_error_add2} of the supplement). For (M2), with abrupt changes and strong positive temporal dependence, SNM is less robust to the increasing dimension and gives more false positives for $d=5,10$, however, its performance is still decent as measured by ARI and $d_H$. On the contrary, for all three models (M1)-(M3), the performance of CUSUM worsens significantly from $d=1$ to $d=5$~(and even more so for $d=10$).

\begin{table}[!htp]
	\centering
	\caption{Performance of SNM, BP, CUSUM under change in mean for $d=1$ and $5$.}
	\label{tab: mean_gaussian_error}
	{\small 
		\begin{tabu}{c|c|rrrrrrr|r|r|r|rr}
			\hline
			\hline & & \multicolumn{7}{c}{$\hat{m}-m_o$} & &&& \\
			\hline
			Method & Model & $\leq -3$ & $-2$ & $-1$ & $0$ & $1$ & $2$ & $\geq 3$ & ARI & $d_1 \times 10^2$ & $d_2 \times 10^2$ & $d_H \times 10^2$ & time\\
			\hline
			SNM &  & 0 & 0 & 9 & 974 & 17 & 0 & 0 & 0.960 & 0.87 & 0.90 & 1.01 & 1.75 \\ 
			BP & $(M1)$ & 0 & 0 & 0 & 847 & 142 & 11 & 0 & 0.974 & 1.48 & 0.50 & 1.48 & 9.10 \\ 
			CUSUM &  & 0 & 0 & 0 & 438 & 414 & 119 & 29 & 0.944 & 4.43 & 0.53 & 4.43 & 0.05 \\ 
			\hline
			SNM &  & 0 & 11 & 196 & 749 & 43 & 1 & 0 & 0.970 & 1.33 & 1.77 & 2.67 & 3.55 \\ 
			BP & $(M2)$ & 0 & 0 & 0 & 425 & 226 & 203 & 146 & 0.863 & 11.68 & 0.19 & 11.68 & 34.04 \\ 
			CUSUM &  & 2 & 0 & 15 & 365 & 341 & 190 & 87 & 0.821 & 10.63 & 2.86 & 10.76 & 0.06 \\ 
			\hline
			SNM &  & 0 & 0 & 1 & 986 & 13 & 0 & 0 & 0.969 & 1.11 & 0.80 & 1.14 & 10.59 \\ 
			BP & $(M3)$ & 0 & 371 & 6 & 623 & 0 & 0 & 0 & 0.616 & 0.33 & 19.03 & 19.03 & 179.75 \\ 
			CUSUM &  & 0 & 0 & 0 & 947 & 53 & 0 & 0 & 0.965 & 1.32 & 0.88 & 1.32 & 0.09 \\ 
			\hline
			\hline & & \multicolumn{7}{c}{$\hat{m}-m_o$} & &&& \\
			\hline
			Method & Model & $\leq -3$ & $-2$ & $-1$ & $0$ & $1$ & $2$ & $\geq 3$ & ARI & $d_1 \times 10^2$ & $d_2 \times 10^2$ & $d_H \times 10^2$ & time\\\hline
			SNM & $(M1)$ & 0 & 0 & 13 & 946 & 41 & 0 & 0 & 0.953 & 1.16 & 1.12 & 1.37 & 12.48 \\ 
			CUSUM & $d=5$ & 167 & 0 & 0 & 230 & 336 & 189 & 78 & 0.783 & 5.18 & 11.41 & 15.71 & 0.04 \\  
			\hline
			SNM & $(M2)$ & 0 & 11 & 175 & 628 & 166 & 18 & 2 & 0.937 & 4.59 & 1.93 & 5.68 & 22.88 \\ 
			CUSUM & $d=5$ & 63 & 5 & 5 & 98 & 161 & 213 & 455 & 0.626 & 18.02 & 5.77 & 20.88 & 0.07 \\ 
			\hline
			SNM & $(M3)$ & 0 & 0 & 4 & 993 & 3 & 0 & 0 & 0.968 & 0.93 & 0.96 & 1.03 & 60.00 \\ 
			CUSUM & $d=5$ & 0 & 70 & 0 & 928 & 2 & 0 & 0 & 0.896 & 1.02 & 4.50 & 4.52 & 0.07 \\
			\hline\hline
		\end{tabu}
	}
\end{table}

\subsection{Change in covariance matrix}\label{subsec:powercovariance}
For covariance matrix change, we adopt the simulation settings in \cite{Aue2009} and detect change in covariance matrices of a four-dimensional time series $\{Y_t=(Y_{t1},\cdots,Y_{t4})\}_{t=1}^n$ with $n=1000$. Thus, the number of parameters in the covariance matrix is $d=(4\times 5)/2=10$. Denote $\Sigma_\rho$ as an exchangeable covariance matrix with unit variance and equal covariance $\rho$, we consider
\begin{align*}
	&\text{(C0)}:Y_t={0.3}\mathbf{I}_4Y_{t-1}+\mathbf{e}_t, ~ \mathbf{e}_t\overset{i.i.d.}{\sim} N(0, \Sigma_{0.5}), ~t\in[1,1000].\\
	&\text{(C1)}: Y_t=\begin{cases} 
		L_0F_{t}+\mathbf{e}_t, ~~ \mathbf{e}_t\overset{i.i.d.}{\sim} N(0, \mathbf{I}_{4}), & t\in [1,333],\\
		\sqrt{3}L_0F_{t}+\mathbf{e}_t, ~~ \mathbf{e}_t\overset{i.i.d.}{\sim} N(0, \mathbf{I}_{4}), & t\in [334,667],\\
		L_0F_{t}+\mathbf{e}_t, ~~ \mathbf{e}_t\overset{i.i.d.}{\sim} N(0, \mathbf{I}_{4}), & t\in [668,1000].
	\end{cases}\\
	&\text{(C2)}: Y_t=\begin{cases} 
		L_0F_{t}+\mathbf{e}_t, ~~ \mathbf{e}_t\overset{i.i.d.}{\sim} N(0, \mathbf{I}_{4}), & t\in [1,333],\\
		\sqrt{3}L_0F_{t}+\mathbf{e}_t, ~~ \mathbf{e}_t\overset{i.i.d.}{\sim} N(0, \mathbf{I}_{4}), & t\in [334,667],\\
		3L_0F_{t}+\mathbf{e}_t, ~~ \mathbf{e}_t\overset{i.i.d.}{\sim} N(0, \mathbf{I}_{4}), & t\in [668,1000].
	\end{cases}
\end{align*}
Here, $\{F_t\}_{t=1}^n$ is a two-dimensional stationary VAR(1) process with the transition matrix $0.3\mathbf{I}_{2}$ and $L_0=[1,1,0,0;0,0,1,1]$ denotes the factor loading matrix. (C1) and (C2) generate covariance changes in the dynamic factor model, which is widely used in the time series literature. We refer to Section \ref{subsec: covmatrix} of the supplement for additional simulation settings with covariance changes in VAR models. The estimation result is reported in Table \ref{tab: sn_4dcov}. For monotonic changes (C2), both methods perform well though AHHR tends to over-estimate the number of change-points, while for non-monotonic changes (C1), AHHR seems to over-estimate and experience power loss at the same time and is outperformed by SNCM. For (C0), both methods give decent performance under moderate temporal dependence with SNCM achieving the target size more accurately.


\begin{table}[h]
	\centering
	\caption{Performance of SNCM and AHHR under change in covariance matrix.} 
	\label{tab: sn_4dcov}
	{\small
		\begin{tabular}{c|c|rrrrrrr|r|r|r|r|r}
			\hline
			\hline & & \multicolumn{7}{c}{$\hat{m}-m_o$} & &&&& \\
			\hline
			Method & Model & $\leq -3$ & $-2$ & $-1$ & $0$ & $1$ & $2$ & $\geq 3$ & ARI & $d_1 \times 10^2$ & $d_2 \times 10^2$ & $d_H \times 10^2$ & time\\
			\hline
			SNCM & $(C1)$ & 0 & 1 & 19 & 951 & 29 & 0 & 0 & 0.923 & 2.13 & 2.46 & 2.78 & 56.44 \\ 
			AHHR & & 0 & 221 & 0 & 687 & 82 & 10 & 0 & 0.721 & 2.45 & 12.37 & 13.50 & 0.44 \\ 
			\hline
			SNCM & $(C2)$ & 0 & 0 & 59 & 902 & 39 & 0 & 0 & 0.898 & 2.53 & 3.95 & 4.37 & 55.17 \\ 
			AHHR &  & 0 & 0 & 1 & 792 & 168 & 32 & 7 & 0.896 & 4.97 & 2.34 & 5.00 & 0.56 \\ 
			\hline
			Method & Model & $\hat m=0$ & $\hat m=1$ & $\hat{m} \geq 2$ & \\\cline{1-5}
			SNCM & $(C0)$ & 916 & 80 & 4 \\
			AHHR & & 932 & 59 & 9 \\\hline\hline
		\end{tabular}
	}
\end{table}

\subsection{Change in multi-parameter}\label{subsec: multipara_main}
As discussed before, one notable advantage of SNCP is its universal applicability, where it treats change-point detection for a broad class of parameters in a unified fashion. To conserve space, we refer to Sections \ref{subsec:powervariance}, \ref{subsec:poweracf}, \ref{subsec:powercorrelation} and \ref{subsec:powerquantile} of the supplement for extensive numerical evidence of the favorable performance of SNCP for change-point detection in variance, auto-correlation, correlation and quantile.

In this section, we further consider change-point estimation for multi-parameter of a univariate time series, where we aim to detect any structural break among multiple parameters of interest. This can be useful for practical scenarios where one does not know the exact nature of the change but wishes to detect any change among a group of parameters of interest. For example, if one is interested in central tendency of the time series, SNMP can be used to simultaneously detect change in mean and median, while if the user suspects there is change in the dispersion/volatility of the data, SNMP can be used to detect change jointly in variance and high quantiles.

In some sense, this is related to change-point detection in distribution~\citep[e.g.\ ECP,][]{Matteson2014}, where the focus is to detect any change in the marginal distribution of a univariate time series. In theory, algorithms that target distributional change can capture all potential changes in the data. However, it only informs users the existence of a change but is unable to narrow down the specific type of change (e.g.\ is the detected change in central tendency or in volatility?). This can be less informative in real data analysis when the practitioner is particularly concerned about one certain behavior change of the data and may also lead to potential power loss compared to methods that target a specific type of change. In addition, existing methods on distributional change typically require the temporal independence assumption, such as ECP in \cite{Matteson2014}.

We consider two simulation settings with $n=1000$, and compare the performance of SNMP and ECP.
\begin{align*}
	\text{(MP1)}:  Y_t&=\begin{cases} 
		X_t, & t\in [1,333],\\
		F^{-1}(\Phi(X_t)), & t\in [334,667],\\
		X_t, & t\in[668, 1000].
	\end{cases}\quad 
	\text{(MP2)}:  Y_t=\begin{cases} 
		\epsilon_{t}, & t\in [1,333],\\
		1.6 \epsilon_{t}, & t\in [334,667],\\
		\epsilon_{t}, & t \in [668,1000].
	\end{cases}
\end{align*}
For (MP1), $\{X_t\}_{t=1}^n$ follows an AR(1) process with $X_t=\rho X_{t-1}+\sqrt{1-\rho^2}\epsilon_t$ where $\rho=0.2$ and $\{\epsilon_t\}$ is \textit{i.i.d.}\ $N(0,1)$, $\Phi(\cdot)$ denotes the CDF of $N(0,1)$, and $F(\cdot)$ denotes a mixture of a truncated normal and a generalized Pareto distribution such that $F^{-1}(q)=\Phi^{-1}(q)$ for $q\leq 0.5$ and $F^{-1}(q)\neq \Phi^{-1}(q)$ for $q>0.5.$ Thus, for (MP1), the change originates from upper quantiles. We refer to Section \ref{subsec:powerquantile} of the supplement for the detailed definition of $F(\cdot)$ and its motivation from financial applications. For (MP2), $\{\epsilon_{t}\}_{t=1}^n$ is \textit{i.i.d.}\ $N(0,1)$, thus we have temporal independence and the change is solely driven by variance.

The estimation result is summarized in Table \ref{tab: sn_multipara}. We compare the performance of SNCP based on individual parameters and their multi-parameter combination. For clarity, we specify the multi-parameter set that SNMP targets. For example, SNQ$_{90}$V denotes the SNMP that targets 90\% quantile and variance simultaneously. For (MP1), SNQ$_{90}$ and SNQ$_{95}$ perform well as the change originates from upper quantiles, and further improvement can be achieved by combining them into multi-parameter SNQ$_{90,95}$. Similarly, including variance in the multi-parameter set further improves the estimation accuracy. ECP provides decent performance but tends to over-estimate due to the temporal dependence of the time series. For (MP2), since the change is solely driven by variance, SNV gives the best performance, while quantile based detection, such as SNQ$_{90}$ experiences power loss. However, the multi-parameter detection based on SNQ$_{10,90}$ and SNQ$_{10,20,80,90}$ provide much improved performance over SNQ$_{90}$, though similar to ECP, they do experience certain power loss compared to SNV. Moreover, SNMP performs competently compared to SNV once variance is included in the multi-parameter set.

This numerical study clearly demonstrates the versatility of SNCP, where it can be effortlessly tailored to target various types of parameter change and their multi-parameter combination. Moreover, compared to detection based on an individual parameter, multi-parameter detection tends to enhance power and improve estimation accuracy when the underlying change affects several parameters in the considered multi-parameter set. We further illustrate this point in more details via real data analysis in Section \ref{subsec:finance}.

\begin{table}[H]
	\centering
	\caption{Performance of SNMP and ECP under change in multi-parameter.}
	\label{tab: sn_multipara}
	{\small
		\begin{tabu}{c|c|rrrrrrr|r|r|r|r|r}
			\hline
			\hline & & \multicolumn{7}{c}{$\hat{m}-m_o$} & &&& & \\
			\hline
			Method & Model & $\leq -3$ & $-2$ & $-1$ & $0$ & $1$ & $2$ & $\geq 3$ & ARI & $d_1 \times 10^2$ & $d_2 \times 10^2$ & $d_H \times 10^2$ & time\\
			\hline
			SNQ$_{90}$ &  & 0 & 10 & 132 & 805 & 50 & 3 & 0 & 0.839 & 3.25 & 7.26 & 7.85 & 17.74 \\ 
			SNQ$_{95}$ &  & 0 & 5 & 100 & 820 & 73 & 2 & 0 & 0.868 & 3.16 & 5.70 & 6.62 & 17.20 \\ 
			SNV &  & 0 & 2 & 110 & 832 & 54 & 2 & 0 & 0.869 & 2.45 & 5.47 & 6.06 &  12.20\\ 
			SNQ$_{90,95}$ &  & 0 & 3 & 82 & 850 & 62 & 3 & 0 & 0.878 & 3.01 & 4.88 & 5.67 & 39.56 \\ 
			SNQ$_{90}$V & $(MP1)$  & 0 & 0 & 56 & 869 & 70 & 5 & 0 & 0.891 & 3.04 & 3.95 & 4.77 & 30.96 \\ 
			SNQ$_{95}$V &  & 0 & 2 & 64 & 861 & 68 & 5 & 0 & 0.889 & 2.92 & 4.30 & 5.14 & 30.81 \\ 
			SNQ$_{90,95}$V &  & 0 & 2 & 48 & 882 & 66 & 2 & 0 & 0.894 & 2.95 & 3.79 & 4.58 & 49.72 \\ 
			ECP &  & 0 & 0 & 0 & 730 & 144 & 92 & 34 & 0.850 & 6.33 & 3.68 & 6.41 & 10.58 \\ 
			\hline
			SNV &  & 0 & 0 & 14 & 956 & 28 & 2 & 0 & 0.928 & 2.15 & 2.13 & 2.60 & 12.28 \\ 
			SNQ$_{90}$ & & 0 & 71 & 282 & 596 & 48 & 3 & 0 & 0.705 & 4.10 & 15.72 & 16.33 & 17.50 \\ 
			SNQ$_{10,90}$ &  & 0 & 13 & 165 & 788 & 32 & 2 & 0 & 0.826 & 3.00 & 8.36 & 8.84 & 39.62 \\ 
			SNQ$_{90}$V & $(MP2)$ & 0 & 0 & 32 & 929 & 39 & 0 & 0 & 0.913 & 2.42 & 2.95 & 3.45 & 30.92 \\ 
			SNQ$_{10,90}$V &  & 0 & 1 & 50 & 917 & 32 & 0 & 0 & 0.903 & 2.37 & 3.67 & 4.06 & 49.74 \\  
			SNQ$_{10,20,80,90}$ &  & 0 & 5 & 118 & 816 & 60 & 1 & 0 & 0.849 & 3.41 & 6.51 & 7.27 & 68.96\\
			ECP &  & 0 & 49 & 46 & 807 & 79 & 15 & 4 & 0.833 & 3.43 & 6.78 & 7.58 & 9.96 \\ 
			\hline\hline
		\end{tabu}
	}
\end{table}

{For each estimated change-point by SNMP, one may want to identify which features actually changed. One informal strategy is to further conduct a subsequent SN-test. Specifically, for each estimated change-point, based on a well-designed local interval, we can further conduct a single change-point SN test via \eqref{statgeneral} for each feature and determine if it is changed at this very change-point. Though this procedure is obviously subject to multiple testing issues, it can shed some light on which feature actually changed. We refer to Section \ref{sec:local_proc} for more details of this informal procedure.}

\section{Conclusion}\label{sec:conclusion}
In this paper, we present a novel and unified framework for time series segmentation in multivariate time series with rigorous theoretical guarantees. Our proposed method is motivated by the recent success of the SN method~\citep{shao2015self} and advances the methodological and theoretical frontier of statistics literature on change-point estimation by adapting the general framework of approximately linear functional in \cite{Kunsch1989}. Our method is broadly applicable to the estimation of piecewise stationary models defined in a general functional. In terms of statistical theory, the consistency and convergence rate of change-point estimation are established under the multiple change-points setting for the first time in the literature of SN-based change-point analysis.

For future research, it may be desirable to relax the piecewise constant assumption and allow the parameter to vary smoothly within each segment; see \cite{wuzhou2019} for such a formulation in nonparametric trend models and \cite{Casini2021change} in locally stationary time series.
\setlength{\bibsep}{0.2pt plus 1ex}
	
	\newpage
	
	\begin{center}
		{\Large Supplementary Material}
	\end{center}

\setcounter{section}{0}
\renewcommand{\thetable}{S.\arabic{table}}
\renewcommand{\thefigure}{S.\arabic{figure}}
\renewcommand{\thesection}{S.\arabic{section}}
\renewcommand{\theass}{S.\arabic{ass}}
\renewcommand{\thethm}{S.\arabic{thm}}
\renewcommand{\thelemma}{S.\arabic{lemma}}
\renewcommand{\theequation}{S.\arabic{equation}}

The supplementary material is organized as follows. Section~\ref{sec: bs_powerloss} illustrates the failure of combining the proposed SN-based test with the classical binary segmentation or a pure local-window based segmentation algorithm. Section \ref{sec:addsim} contains additional simulation results. Section \ref{sec:applications} illustrates the effectiveness and practical significance of SNCP via meaningful real data applications in climate science and finance. Section~\ref{sec:verify} provides detailed verification for technical assumptions of SNCP for the smooth function model, which includes a wide class of parameters such as mean, variance, (auto)-covariance and (auto)-correlation. { Section~\ref{sec:verify_quantile} further provides detailed verification for technical assumptions of SNCP for quantiles.} Section~\ref{sec:consistencygeneral} contains the consistency proof of SNCP for a general univariate functional. In Section~\ref{sec:consistencymean}, we further provide the proof for the consistency of SNCP for detecting changes in multivariate mean. {\Cref{sec:local_refine} proposes a simple local refinement procedure for SNCP, which improves the localization error rate of SNCP to the optimal $O_p(n^{-1})$ rate for the mean functional.}

There are 12 subsections in Section~\ref{sec:addsim}. In particular, Section~\ref{subsec: sensitivity} conducts sensitivity analysis w.r.t.\ to the choice of the window size $\epsilon$ and the critical value level $\alpha$ for SNCP; Section~\ref{subsec: wbs_not} provides extensive numerical comparison between the proposed nested local-window segmentation algorithm and other popular state-of-the-art segmentation algorithms~(WBS, SBS, NOT and fused-LASSO) for detecting changes in univariate and multivariate mean; Section~\ref{subsec:size_supplement} contains additional results for no change; Section~\ref{subsec:add_powermean} presents additional numerical comparison between SNCP and the conventional CUSUM for the multivariate mean case; Section~\ref{subsec:powervariance}, \ref{subsec:poweracf}, \ref{subsec:powercorrelation}, \ref{subsec:powerquantile} and \ref{subsec:powermeanvariance} conduct numerical comparison between SNCP and other popular change-point detection methods for variance, auto-correlation, correlation, and quantile changes, respectively; Section~\ref{subsec: covmatrix} contains additional results for changes in covariance matrix; Section~\ref{subsec: multipara} provides additional simulation results for changes in multi-dimensional parameters; Section~\ref{subsec:dgp} contains the implementation details of comparison methods and typical realizations of DGP used in simulation.

In terms of notation, throughout the supplement, we let $X_n\in\mathbb{R}^d$ with dimension $d>0$ be a set of random vector  defined in a probability space $(\Omega,\mathbb{P},\mathcal{F})$. For a corresponding set of constants $a_n$, we say $X_n=O_p^s(a_n)$ if for any $\varepsilon>0$, there exists a finite $M>0$ and a finite $N>0$ such that for all $n>N$,
$$
\mathbb{P}(\|X_n/a_n\|>M)+\mathbb{P}(\|X_n/a_n\|<1/M)<\varepsilon,
$$ 
where $\|\cdot \|$ denotes the $L_2$ norm, i.e. we say $X_n=O_p^s(1)$  if both $\|X_n\|$ and $\|X_n\|^{-1}$ are bounded (from above) in probability. In addition, we let $C$ be a generic constant that may vary from line to line. 

\section{Failure of SN with binary segmentation and the pure local-window based segmentation}\label{sec: bs_powerloss}

\subsection{Theoretical evidence}\label{subsec: bs_powerloss_theory}
In this section, we provide theoretical evidence to demonstrate that a simple combination of the proposed SN test statistic and the classical binary segmentation can suffer severe power loss and inconsistency under the multiple change-point scenario.

For simplicity, we focus on the univariate mean case with two change-points. Suppose $\{Y_t\}_{t=1}^{n}$ is generated by:
\[Y_t=\left\{\begin{array}{cc}
	\delta+X_t,&1\le t\le k_1\\
	X_t,& k_1+1\le t\le k_2\\
	\delta+X_t,& k_2+1\le t\leq n,
\end{array}
\right.\]
where $\delta>0$ is a constant, $\{X_t\}_{t\in\mathbb{Z}}$ is a stationary time series, and $k_i=\lfloor n\tau_i\rfloor$, $i=1,2$ with $0<\tau_1<\tau_2<1$ denotes the two change-points.

In the following, we explicitly derive the asymptotic limit of the SN test statistic $SN_n^*=\max_{k=1,\cdots,n-1} T_n(k)$ $=\max_{k=1,\cdots,n-1} D_n(k)^2/V_n(k)$ based on the entire sample $\{Y_t\}_{t=1}^n$ and show that the asymptotic order of $SN_n^*$ is $O_p(1)$~(see Section 2 of the main text for detailed definition of $SN_n^*$). Note that the binary segmentation algorithm uses $SN_n^*$ to detect the existence of potential change-points and thus $SN_n^*=O_p(1)$ indicates the power loss and asymptotic inconsistency of the binary segmentation algorithm.

Denote $\bar{X}_{a,b}=\frac{1}{b-a+1}\sum_{t=a}^{b}X_t$. By simple calculation, the contrast statistic $D_n(k)$ takes the form
\begin{align*}
	D_n(k)=&\left\{\begin{aligned}
		&\frac{k(n-k)}{n^{3/2}}(\bar{X}_{1,k}-\bar{X}_{k+1,n}+\frac{k_2-k_1}{n-k}\delta),&1\le k\le k_1,\\
		&	\frac{k(n-k)}{n^{3/2}}(\bar{X}_{1,k}-\bar{X}_{k+1,n}+\frac{k_1}{k}\delta-\frac{n-k_2}{n-k}\delta),&k_1+1\leq k\leq k_2,\\
		&	\frac{k(n-k)}{n^{3/2}}(\bar{X}_{1,k}-\bar{X}_{k+1,n}+\frac{k_1-k_2}{k}\delta),& k_2+1\le k\leq n-1.
	\end{aligned}\right.&&
\end{align*}
Similarly, we can derive the explicit form of the self-normalizer $V_n(k)$. For $1\le k\le k_1$,
\begin{align*}
	V_n(k)&=\sum_{i=1}^{k}\frac{i^2(k-i)^2}{n^2k^2}(\bar{X}_{1,i}-\bar{X}_{i+1,k})^2
	\\&+\sum_{i=k+1}^{k_1+1}\frac{(n-i+1)^2(i-k-1)^2}{n^2(n-k)^2}(\bar{X}_{k+1,i-1}-\bar{X}_{i,n}+\frac{k_2-k_1}{n-i+1}\delta)^2\\&
	+\sum_{i=k_1+2}^{k_2+1}\frac{(n-i+1)^2(i-k-1)^2}{n^2(n-k)^2}(\bar{X}_{k+1,i-1}-\bar{X}_{i,n}+\frac{k_1-k}{i-k-1}\delta-\frac{n-k_2}{n-i+1}\delta)^2
	\\&+\sum_{i=k_2+2}^{n}\frac{(n-i+1)^2(i-k-1)^2}{n^2(n-k)^2}(\bar{X}_{k+1,i-1}-\bar{X}_{i,n}-\frac{k_2-k_1}{i-k-1}\delta)^2.
\end{align*}	
For $k_1+1\leq k\leq k_2$,
\begin{align*}
	V_n(k)&=		
	\sum_{i=1}^{k_1}\frac{i^2(k-i)^2}{n^2k^2}(\bar{X}_{1,i}-\bar{X}_{i+1,k}+\frac{k-k_1}{k-i}\delta)^2\\&+\sum_{i=k_1+1}^{k}\frac{i^2(k-i)^2}{n^2k^2}(\bar{X}_{1,i}-\bar{X}_{i+1,k}+\frac{k_1}{i}\delta)^2
	\\&+\sum_{i=k+1}^{k_2+1}\frac{(n-i+1)^2(i-k-1)^2}{n^2(n-k)^2}(\bar{X}_{k+1,i-1}-\bar{X}_{i,n}-\frac{n-k_2}{n-i+1}\delta)^2\\&
	+\sum_{i=k_2+2}^{n}\frac{(n-i+1)^2(i-k-1)^2}{n^2(n-k)^2}(\bar{X}_{k+1,i-1}-\bar{X}_{i,n}-\frac{k_2-k}{i-k-1}\delta)^2.
\end{align*}
For $ k_2+1\le k\le n-1$,
\begin{align*}
	V_n(k)&=	
	\sum_{i=1}^{k_1}\frac{i^2(k-i)^2}{n^2k^2}(\bar{X}_{1,i}-\bar{X}_{i+1,k}+\frac{k_2-k_1}{k-i}\delta)^2\\&+\sum_{i=k_1+1}^{k_2}\frac{i^2(k-i)^2}{n^2k^2}(\bar{X}_{1,i}-\bar{X}_{i+1,k}+\frac{k_1}{i}\delta-\frac{k-k_2}{k-i}\delta)^2
	\\&+\sum_{i=k_2+1}^{k}\frac{i^2(k-i)^2}{n^2k^2}(\bar{X}_{1,i}-\bar{X}_{i+1,k}-\frac{k_2-k_1}{i}\delta)^2\\&
	+\sum_{i=k+1}^{n}\frac{(n-i+1)^2(i-k-1)^2}{n^2(n-k)^2}(\bar{X}_{k+1,i-1}-\bar{X}_{i,n})^2.
\end{align*}

Thus, by the FCLT that $n^{1/2}(\bar{X}_{\lfloor nr_1\rfloor+1,\lfloor nr_2\rfloor}-EX_t)\Rightarrow \sigma_X(B(r_2)-B(r_1))$, we have that
$$
\left\{ n^{-1/2}D_n(\lfloor n\tau\rfloor), n^{-1} V_n(\lfloor n\tau\rfloor)\right\}
\Rightarrow  \Big\{ \sigma_X\delta D^f(\tau), \sigma_X^2\delta^2V^f(\tau)\Big\},
$$
and 
$$
\Big\{T_n(\lfloor n\tau\rfloor)=D_n(\lfloor n\tau\rfloor)^2/ V_n(\lfloor n\tau\rfloor)\Big\}
\Rightarrow  \Big\{ T^f(\tau)\Big\},
$$
where $ T^f(\tau)= D^f(\tau)^2/ V^f(\tau)$, and
\begin{flalign*}
	D^{f}(\tau)=&\left\{\begin{aligned}
		&\tau(\tau_2-\tau_1), &0\leq \tau< \tau_1,\\
		&(1-\tau)\tau_1-\tau(1-\tau_2), &\tau_1\leq \tau< \tau_2,\\
		&(1-\tau)(\tau_1-\tau_2),& \tau_2\leq \tau\leq 1,
	\end{aligned}\right.\\
	V^f(\tau)=&\left\{\begin{aligned}&\int_{\tau}^{\tau_1}\frac{(\tau_2-\tau_1)^2(s-\tau)^2}{(1-\tau)^2}ds +\int_{\tau_1}^{\tau_2}\frac{(1-s)^2(s-\tau)^2}{(1-\tau)^2}(\frac{\tau_1-\tau}{s-\tau}-\frac{1-\tau_2}{1-s})^2ds+\int_{\tau_2}^{1}\frac{(1-s)^2(\tau_2-\tau_1)^2}{(1-\tau)^2}ds,\\ &\hspace{14.3cm} 0\leq \tau< \tau_1,\\
		&\int_{0}^{\tau_1}\frac{s^2(\tau-\tau_1)^2}{\tau^2}ds +\int_{\tau_1}^{\tau}\frac{(s-\tau)^2\tau_1^2}{\tau^2}ds +\int_{\tau}^{\tau_2}\frac{(1-\tau_2)^2(s-\tau)^2}{(1-\tau)^2}ds+\int_{\tau_2}^{1}\frac{(1-s)^2(\tau_2-\tau)^2}{(1-\tau)^2}ds,  \\ &\hspace{14.3cm}\tau_1\leq \tau< \tau_2,\\
		&\int_{0}^{\tau_1}\frac{s^2(\tau_2-\tau_1)^2}{\tau^2}ds +\int_{\tau_1}^{\tau_2}\frac{(\tau-s)^2s^2}{\tau^2}(\frac{\tau_1}{s}-\frac{\tau-\tau_2}{\tau-s})^2ds +\int_{\tau_2}^{\tau}\frac{(\tau-s)^2(\tau_2-\tau_1)^2}{\tau^2}ds, \\ &\hspace{14.3cm} \tau_2\leq \tau\leq 1.
	\end{aligned}\right.
\end{flalign*}

In other words, the asymptotic limit of the SN test statistic $SN_n^*$ is a deterministic constant $\max_{\tau \in (0,1)} T^f(\tau)$. This interesting phenomenon is caused by the existence of the \textit{two} change-points, which inflates the self-normalizer $V_n(k)$ and thus deflates the SN test statistic $T_n(k)$.

Together, this implies that $SN_n^*=O_p(1)$. Hence the probability of detecting change-points is less than 1 even when $n\rightarrow\infty$, indicating the power loss and asymptotic inconsistency for a simple combination of the SN test and binary segmentation.

As explained in the main text, unlike the classical binary segmentation, which evaluates the SN test based on the whole sample, the proposed nested local-window segmentation bypasses this power loss issue due to inflated self-normalizer by evaluating the SN test statistic on a set of carefully designed nested local-windows.

Another popular segmentation algorithm in the change-point literature is the pure local-window based approach, see for example, \cite{Niu2012}, \cite{Yau2016} and \cite{Niu2016}. Compared to the proposed nested local-window segmentation algorithm in our paper, the pure local-window approach \textit{only} considers one single local-window around each time point in the data.

Specifically, denote the window size as $h$, following the notation in Section \ref{sec:multiple} of the main text, for each $k=h,\cdots,n-h$, the pure local-window approach computes the SN-based test for time point $k$ via
$$T'_{1,n}(k)=T_n(k-h+1,k,k+h).$$
In other words, the pure local-window approach only computes the SN-based statistic on the smallest local-window $(k-h+1,k+h)$. The change-point estimator is then obtained by comparing the so-called local-window maximizer~(see \cite{Niu2012} for detailed definition) of $\{T_{1,n}(k)\}_{k=h}^{n-h}$ with a properly chosen threshold. Following the same argument as the one for Theorem \ref{thm1} in the main text, we can easily show that 
\begin{align*}
	\max_{k=h,\cdots,n-h} T'_{1,n}(k) \overset{\mathcal{D}}{\longrightarrow} G_\epsilon'= \sup_{u\in(\epsilon,1-\epsilon)}{D(u-\epsilon,u,u+\epsilon)^2}/{V(u-\epsilon,u,u+\epsilon)}.
\end{align*}
Thus, we can use the 90\% or 95\% quantile of the limiting distribution $G_\epsilon'$ as the threshold for the pure local-window approach, which controls the asymptotic false positive detection rate.

In comparison, the proposed nested local-window segmentation algorithm computes the SN-based test for each time point $k=h,\cdots,n-h$ via
$$T_{1,n}(k)=\max\limits_{(t_1,t_2)\in H_{1:n}(k)}T_n(t_1,k,t_2),$$
based on a series of expanding nested local-windows surrounding $k$ indexed by $H_{1:n}(k)=\bigl\{ (t_1,t_2)\big\vert t_1=k-j_1 h+1, j_1=1,\ldots, \lfloor k/h\rfloor;t_2=k+j_2 h, j_2=1,\ldots, \lfloor (n-k)/h \rfloor  \bigl\}.$ Note that $(k-h+1,k+h)$ is the smallest local-window in $H_{1:n}(k)$. As discussed in the main text, such a strategy is expected to achieve higher power than the pure local-window approach, especially for the case where the change-point $k$ is far away from other change-points by utilizing larger nested windows that cover $k$ other than $(k-h+1,k+h)$. We further verify this claim in Section \ref{subsec:bs_sara} via numerical study.


%

\subsection{Numerical evidence}\label{subsec:bs_sara}
In this section, we demonstrate the power loss of the simple combination between the SN test and the classical binary segmentation or the pure local-window approach via a small simulation example. To illustrate, we simulate $\{Y_t\}_{t=1}^n$ from 
\begin{align*}
&\text{(M4)}: n=2000, &\rho=&0.7, &Y_t&=\begin{cases} 
	0.8+ X_{t}, & t\in [1,1000], [1501,2000],\\
	0+ X_{t}, & t\in [1001,1500].\\
\end{cases}\\
&\text{(M5)}: n=2000, &\rho=&0.7, &Y_t&=\begin{cases} 
	0+ X_{t}, & t\in [1,1000],\\
	0.8+ X_{t}, & t\in [1001,1500],\\
	1.6+ X_{t}, & t\in [1501,2000],
\end{cases}
\end{align*}
where $\{X_t\}_{t=1}^n$ is a stationary AR(1) process such that $X_t=\rho X_{t-1}+\epsilon_t,$ where $\{\epsilon_t\}$ is \textit{i.i.d.} $N(0,1)$. Note that the main difference between (M4) and (M5) is that the mean change in (M4) is non-monotonic while the mean change in (M5) is monotonic.

We apply the proposed nested local-window based SNM for change-point detection in mean. We further apply the simple combination between the SN test and the classical binary segmentation~(SNBS) or the pure local-window approach~(SNLocal). The estimation result is summarized in Table \ref{tab: power_loss_BS}. As can be seen clearly, SNLocal has severe power loss compared to SNM in both (M4) and (M5), indicating the advantage of the proposed nested local-window segmentation over the pure local-window approach. In addition, under the non-monotonic change in (M4), SNBS almost completely loses power while its performance is comparable to SNM under monotonic change in (M5). In summary, this result suggests the necessity of the proposed nested local-window segmentation algorithm in SNCP for change-point detection.

\begin{table}[H]
\centering
\caption{Estimation result under change in mean for (M4)-(M5).}
\label{tab: power_loss_BS}
{\small
	\begin{tabu}{c|r|rrrrrrr|r|r|r|r|r}
		\hline
		\hline & & \multicolumn{7}{c}{$\hat{m}-m_o$} & &&& \\
		\hline
		Method & Model & $\leq -3$ & $-2$ & $-1$ & $0$ & $1$ & $2$ & $\geq 3$ & ARI & $d_1 \times 10^2$ & $d_2 \times 10^2$ & $d_H \times 10^2$ & time\\
		\hline
		SNM &  & 0 & 19 & 237 & 688 & 51 & 5 & 0 & 0.825 & 3.57 & 9.03 & 10.20 & 11.54 \\ 
		SNBS & $(M4)$ & 0 & 981 & 19 & 0 & 0 & 0 & 0 & 0.009 & 0.14 & 49.66 & 49.66 & 0.61 \\
		SNLocal & & 0 & 563 & 329 & 94 & 12 & 2 & 0 & 0.250 & 3.25 & 39.10 & 39.52 & 0.10\\
		\hline
		SNM &  & 0 & 11 & 268 & 669 & 49 & 3 & 0 & 0.824 & 3.60 & 9.09 & 10.23 & 11.35 \\ 
		SNBS & $(M5)$ & 0 & 2 & 56 & 722 & 215 & 5 & 0 & 0.800 & 7.24 & 6.33 & 8.38 & 1.02 \\
		SNLocal & & 0 & 528 & 369 & 85 & 17 & 1 & 0 & 0.273 & 3.14 & 38.17 & 38.56 & 0.10 \\
		\hline\hline
	\end{tabu}
}
\end{table}

\section{Additional Simulation Results}\label{sec:addsim}
\subsection{Sensitivity analysis}\label{subsec: sensitivity}
In this section, we conduct sensitivity analysis of SNCP w.r.t. the window size $\epsilon$ and the critical value level $\alpha$. Specifically, we vary $\epsilon=0.05,0.08,0.10,0.12,0.15$ and vary the critical value level $\alpha=0.1,0.05$, and study how $(\epsilon, \alpha)$ influences the performance of SNCP. For clarity of presentation, in the following, we use the quantile level $q=(1-\alpha)\times 100\%=90,95$ to refer to the critical value level $\alpha$.

Recall that the window size $\epsilon$ reflects one's belief of minimum (relative) spacing between two consecutive change-point and the quantile value level $q$ balances one's tolerance of type-I and type-II errors. For consistency of SNCP, we require $\epsilon<\epsilon_o$, which is the minimum spacing between change-points.

We consider two simulation settings (SA1) and (SA2) for change in mean. Specifically, we first simulate a stationary unit-variance univariate time series $\{X_t\}_{t=1}^n$ from a unit-variance AR(1) process with $X_t=\rho X_{t-1}+\sqrt{1-\rho^2}\epsilon_t$ where $\{\epsilon_t\}$ is \textit{i.i.d.} standard normal $N(0,1)$. We then generate univariate time series $\{Y_t\}_{t=1}^n$ with piecewise constant mean based on $\{X_t\}_{t=1}^n.$
\begin{align*}
&\text{(SA1)}: n=1200, ~~~\rho=0, &Y_t&=\begin{cases} 
	0+ X_{t}, & t\in [1,150], [301,450], [601,750], [901,1050]\\
	\delta+ X_{t}, & t\in [151,300], [451,600], [751,900], [1051,1200].
\end{cases}\\
&\text{(SA2)}: n=1200, ~~~\rho=0.5, &Y_t&=\begin{cases} 
	0+ X_{t}, & t\in [1,150], [301,450], [601,750], [901,1050]\\
	\delta+ X_{t}, & t\in [151,300], [451,600], [751,900], [1051,1200].
\end{cases}
\end{align*}
Note that for both (SA1) and (SA2), \textit{all} change-points are evenly located with the minimum spacing $\epsilon_o=150/1200=0.125.$

For (SA1), the noise $\{X_t\}_{t=1}^n$ is \textit{i.i.d.} Gaussian random variables as $\rho=0$. We further vary $\delta=1, \sqrt{2}$ to generate two scenarios with low and high signal-to-noise ratios~(SNR). For (SA2), the noise $\{X_t\}_{t=1}^n$ is a stationary AR(1) process with moderate temporal dependence $\rho=0.5.$ We set $\delta=\sqrt{3}, \sqrt{6}$ to generate scenarios with low and high SNR. Note that compared to (SA1), $\delta$ in (SA2) is multiplied by a factor of $\sqrt{3}$ to compensate the long-run variance~(LRV) of $\{X_t\}_{t=1}^n$, which is $\sqrt{3}$. Thus, (SA1) and (SA2) have the same level of SNR.

The estimation results under (SA1) and (SA2) are summarized in Table \ref{tab: sn_sensitivity_indep} and Table \ref{tab: sn_sensitivity_dep}. The general findings are as follows. We focus on the result of (SA1) as the result of (SA2) is similar.

\textbf{Robustness w.r.t.\ the window size $\epsilon$}: The performance of SNCP is reasonably robust across all window sizes $\epsilon =0.05,0.08,0.1,0.12 <\epsilon_o=0.125$, as evidenced by the stable values of ARI and Hausdorff distance $d_H$ achieved across different $\epsilon$. This is especially true for the high SNR scenario.

On the other hand, SNCP fails to detect changes with the window size $\epsilon=0.15$, which exceeds the minimum spacing $\epsilon_o=0.125$. This is consistent with the discussion in Section \ref{subsec: nested_alg} of the main text. As for $\epsilon=0.15$, even the smallest local-window centered around any true change-point contains at least two change-points, this significantly lowers the power of SN-tests due to inflated self-normalizer. The drastic contrast between the performance of SNCP with $\epsilon<\epsilon_o$ and $\epsilon=0.15$ is partially due to the fact that in (SA1), all change-points are evenly spaced with the same spacing $\epsilon_o=0.125$, thus the assumption $\epsilon<\epsilon_o$ is violated all at once for all change-points.

Note that though SNCP with $\epsilon=0.05$ may not always deliver the best performance among all window sizes, it does offer one of the best performance under both low and high SNR scenarios. Thus, we recommend setting $\epsilon=0.05$ as it guards against the violation of $\epsilon<\epsilon_o$ to the best extent.

\textbf{Robustness w.r.t.\ the quantile level $q$}: The choice of the quantile level $q$ is less essential for SNCP and it is more about the trade-off between type-I and type-II error in finite sample. As can be seen in Table \ref{tab: sn_sensitivity_indep}, for low SNR, given the same window size $\epsilon$, the quantile level $q=90$ provides better performance due to higher power, while for high SNR, the difference between $q=90$ and $95$ is minimal. Of course, setting $q=90$ will incur higher type-I error when there is no change-point.

Finally, comparing the estimation results in Table \ref{tab: sn_sensitivity_indep}~(SA1) and Table \ref{tab: sn_sensitivity_dep}~(SA2), it can be seen that given the same SNR, the robustness of SNCP w.r.t. the window size $\epsilon$ and the quantile level $q$ remain the same with or without temporal dependence.

\begin{table}[H]
\centering
\caption{Sensitivity analysis under (SA1) with $\delta={1}$ and $\sqrt{2}$.}
\label{tab: sn_sensitivity_indep}
{\small 
	\begin{tabular}{c|c|rrrrrrr|r|r|r|r|r}
		\hline
		\hline & & \multicolumn{7}{c}{$\hat{m}-m_o$} & &&& \\
		\hline
		$(q,\epsilon)$ & Model & $\leq -3$ & $-2$ & $-1$ & $0$ & $1$ & $2$ & $\geq 3$ & ARI & $d_1 \times 10^2$ & $d_2 \times 10^2$ & $d_H \times 10^2$ & time\\\hline
		90, 0.05 &  & 0 & 2 & 89 & 899 & 10 & 0 & 0 & 0.930 & 1.07 & 2.09 & 2.14 & 3.34 \\ 
		95, 0.05 &  & 0 & 19 & 180 & 796 & 5 & 0 & 0 & 0.916 & 1.02 & 3.31 & 3.33 & 3.34 \\ 
		90, 0.08 &  & 0 & 26 & 169 & 805 & 0 & 0 & 0 & 0.913 & 1.00 & 3.30 & 3.30 & 1.20 \\ 
		95, 0.08 &  & 8 & 75 & 282 & 635 & 0 & 0 & 0 & 0.886 & 0.95 & 5.38 & 5.38 & 1.20 \\ 
		90, 0.10 & $(SA1)$ & 0 & 1 & 46 & 953 & 0 & 0 & 0 & 0.931 & 1.04 & 1.59 & 1.59 & 0.71 \\ 
		95, 0.10 & $\delta=1$ & 0 & 16 & 102 & 882 & 0 & 0 & 0 & 0.921 & 1.02 & 2.44 & 2.44 & 0.71 \\ 
		90, 0.12 &  & 0 & 6 & 148 & 846 & 0 & 0 & 0 & 0.930 & 0.77 & 2.37 & 2.37 & 0.47 \\ 
		95, 0.12 &  & 1 & 15 & 173 & 811 & 0 & 0 & 0 & 0.924 & 0.76 & 2.80 & 2.80 & 0.47 \\ 
		90, 0.15 &  & 1000 & 0 & 0 & 0 & 0 & 0 & 0 & 0.002 & 0.00 & 49.95 & 49.95 & 0.27 \\ 
		95, 0.15 &  & 1000 & 0 & 0 & 0 & 0 & 0 & 0 & 0.001 & 0.00 & 50.01 & 50.01 & 0.27 \\ 
		\hline
		Method & Model & $\leq -3$ & $-2$ & $-1$ & $0$ & $1$ & $2$ & $\geq 3$ & ARI & $d_1 \times 10^2$ & $d_2 \times 10^2$ & $d_H \times 10^2$ & time\\
		\hline
		90, 0.05 &  & 0 & 0 & 3 & 980 & 17 & 0 & 0 & 0.963 & 0.68 & 0.63 & 0.72 & 3.76 \\ 
		95, 0.05 &  & 0 & 0 & 7 & 982 & 11 & 0 & 0 & 0.963 & 0.65 & 0.68 & 0.73 & 3.76 \\ 
		90, 0.08 &  & 0 & 0 & 13 & 987 & 0 & 0 & 0 & 0.961 & 0.60 & 0.75 & 0.75 & 1.34 \\ 
		95, 0.08 &  & 0 & 0 & 39 & 961 & 0 & 0 & 0 & 0.958 & 0.59 & 1.06 & 1.06 & 1.34 \\ 
		90, 0.10 & $(SA1)$ & 0 & 0 & 1 & 999 & 0 & 0 & 0 & 0.959 & 0.65 & 0.66 & 0.66 & 0.80 \\ 
		95, 0.10 & $\delta=\sqrt{2}$ & 0 & 0 & 5 & 995 & 0 & 0 & 0 & 0.958 & 0.64 & 0.70 & 0.70 & 0.80 \\ 
		90, 0.12 &  & 0 & 0 & 38 & 962 & 0 & 0 & 0 & 0.956 & 0.59 & 0.99 & 0.99 & 0.53 \\ 
		95, 0.12 &  & 0 & 0 & 39 & 961 & 0 & 0 & 0 & 0.956 & 0.59 & 1.00 & 1.00 & 0.53 \\ 
		90, 0.15 &  & 1000 & 0 & 0 & 0 & 0 & 0 & 0 & 0.000 & 0.00 & 50.01 & 50.01 & 0.31 \\ 
		95, 0.15 &  & 1000 & 0 & 0 & 0 & 0 & 0 & 0 & 0.000 & 0.00 & 50.00 & 50.00 & 0.31 \\  
		\hline\hline
	\end{tabular}
}
\end{table}

\begin{table}[H]
\centering
\caption{Sensitivity analysis under (SA2) with $\delta=\sqrt{3}$ and $\sqrt{6}$.}
\label{tab: sn_sensitivity_dep}
{\small 
	\begin{tabular}{c|c|rrrrrrr|r|r|r|r|r}
		\hline
		\hline & & \multicolumn{7}{c}{$\hat{m}-m_o$} & &&& \\
		\hline
		$(q,\epsilon)$ & Model & $\leq -3$ & $-2$ & $-1$ & $0$ & $1$ & $2$ & $\geq 3$ & ARI & $d_1 \times 10^2$ & $d_2 \times 10^2$ & $d_H \times 10^2$ & time\\\hline
		90, 0.05 &  & 0 & 1 & 89 & 882 & 27 & 1 & 0 & 0.933 & 1.12 & 2.06 & 2.18 & 3.12 \\ 
		95, 0.05 &  & 0 & 26 & 175 & 788 & 11 & 0 & 0 & 0.918 & 1.01 & 3.32 & 3.38 & 3.12 \\ 
		90, 0.08 &  & 0 & 33 & 195 & 772 & 0 & 0 & 0 & 0.911 & 0.97 & 3.65 & 3.65 & 1.10 \\ 
		95, 0.08 &  & 16 & 85 & 302 & 597 & 0 & 0 & 0 & 0.880 & 0.92 & 5.92 & 5.92 & 1.10 \\ 
		90, 0.10 & $(SA2)$ & 0 & 4 & 58 & 938 & 0 & 0 & 0 & 0.931 & 1.02 & 1.75 & 1.75 & 0.66 \\ 
		95, 0.10 & $\delta=\sqrt{3}$ & 0 & 24 & 120 & 856 & 0 & 0 & 0 & 0.919 & 0.99 & 2.74 & 2.74 & 0.66 \\ 
		90, 0.12 &  & 0 & 4 & 130 & 866 & 0 & 0 & 0 & 0.933 & 0.77 & 2.17 & 2.17 & 0.44 \\ 
		95, 0.12 &  & 1 & 14 & 159 & 826 & 0 & 0 & 0 & 0.926 & 0.76 & 2.69 & 2.69 & 0.44 \\ 
		90, 0.15 &  & 1000 & 0 & 0 & 0 & 0 & 0 & 0 & 0.002 & 0.00 & 49.98 & 49.98 & 0.25 \\ 
		95, 0.15 &  & 1000 & 0 & 0 & 0 & 0 & 0 & 0 & 0.001 & 0.00 & 50.01 & 50.01 & 0.25 \\  
		\hline
		$(q,\epsilon)$ & Model & $\leq -3$ & $-2$ & $-1$ & $0$ & $1$ & $2$ & $\geq 3$ & ARI & $d_1 \times 10^2$ & $d_2 \times 10^2$ & $d_H \times 10^2$ & time\\
		\hline
		90, 0.05 &  & 0 & 0 & 4 & 964 & 31 & 1 & 0 & 0.965 & 0.72 & 0.60 & 0.77 & 2.88 \\ 
		95, 0.05 &  & 0 & 0 & 11 & 968 & 21 & 0 & 0 & 0.965 & 0.66 & 0.68 & 0.79 & 2.88 \\ 
		90, 0.08 &  & 0 & 0 & 16 & 984 & 0 & 0 & 0 & 0.963 & 0.58 & 0.77 & 0.77 & 1.03 \\ 
		95, 0.08 &  & 0 & 2 & 48 & 950 & 0 & 0 & 0 & 0.958 & 0.57 & 1.17 & 1.17 & 1.03 \\ 
		90, 0.10 & $(SA2)$ & 0 & 0 & 2 & 998 & 0 & 0 & 0 & 0.961 & 0.62 & 0.65 & 0.65 & 0.61 \\ 
		95, 0.10 & $\delta=\sqrt{6}$ & 0 & 0 & 8 & 992 & 0 & 0 & 0 & 0.961 & 0.62 & 0.72 & 0.72 & 0.61 \\ 
		90, 0.12 &  & 0 & 0 & 34 & 966 & 0 & 0 & 0 & 0.958 & 0.59 & 0.95 & 0.95 & 0.41 \\ 
		95, 0.12 &  & 0 & 0 & 34 & 966 & 0 & 0 & 0 & 0.958 & 0.59 & 0.95 & 0.95 & 0.41 \\ 
		90, 0.15 &  & 1000 & 0 & 0 & 0 & 0 & 0 & 0 & 0.000 & 0.00 & 50.01 & 50.01 & 0.24 \\ 
		95, 0.15 &  & 1000 & 0 & 0 & 0 & 0 & 0 & 0 & 0.000 & 0.00 & 50.00 & 50.00 & 0.24 \\ 
		\hline\hline
	\end{tabular}
}
\end{table}

\subsection{Comparison with state-of-the-art segmentation algorithms}\label{subsec: wbs_not}
In this subsection, we further demonstrate the promising performance of the proposed nested local-window segmentation algorithm~(i.e.\ SCNP) by comparing it with  state-of-the-art segmentation algorithms in the change-point literature. In particular, we consider the wild binary segmentation~(WBS) in \cite{Fryzlewicz2014}, the narrowest over threshold~(NOT) in \cite{baranowski2019narrowest} and their variants including seeded binary segmentation (SBS) and seeded NOT (SNOT) in  \cite{Kovacs2020}. We also compare with the change-point estimators by least squares with total variation penalty (i.e.\ the fused LASSO penalty) in \cite{harchaoui2010multiple}, which is denoted by LASSO.  

WBS, NOT, SBS and SNOT are generic segmentation algorithms that can be combined with a specific change-point test statistic to achieve multiple change-point detection and are robust to non-monotonic changes.  Thus, we combine the SN-based test statistic $T_n(t_1,k,t_2)$ proposed in the main text~(equation \eqref{eq:SN_subsample}) with WBS, NOT, SBS and SNOT to construct multiple change-point detection procedures SN-WBS, SN-NOT, SN-SBS and SN-SNOT, respectively.

For SN-WBS and SN-NOT, we set the number of random intervals used in the algorithm at $M=1000$. Furthermore, we consider two versions of SN-WBS and SN-NOT, where we set the minimal length of the $M$ random intervals to be $4$~(default value in WBS and NOT) or $\lfloor n\epsilon\rfloor$~(same as the minimum nested local window in SNCP with $\epsilon=0.05$), and denote the corresponding procedures by the superscript $^1$ or $^2$, respectively.

For SN-SBS and SN-SNOT, following \cite{Kovacs2020}, we set the decay rate for the seeded intervals to be $(1/2)^{1/4}$ or $(1/2)^{1/16}$, and denote the corresponding procedures by the superscript $^1$ or $^2$, respectively. Note that for SN-SBS and SN-SNOT, the minimal length of the seeded intervals is only set at $\lfloor n\epsilon\rfloor$ with $\epsilon=0.05$, instead of $4$, to avoid overwhelming number of short seeded intervals that are not suitable for the use of self-normalization.  

\subsubsection{Univariate mean change}
For simulation comparison, we first consider the change in univariate mean setting, specifically models (M1), (M2) and (M3) with $d=1$, as in Section \ref{subsec:powermean} of the main text.

The estimation result is summarized in Table \ref{tab: compare_wbs_not}. As can be seen, for all three models (M1)-(M3), the proposed SNM gives comparable (or more favorable in (M2)) performance as SN-WBS, SN-NOT, SN-SBS and SN-SNOT and outperforms LASSO. Furthermore, SNM is computationally more efficient than SN-WBS and SN-NOT, and comparable with SN-SBS, SN-SNOT and LASSO. This further confirms the value of the proposed nested local-window segmentation algorithm. Moreover, in unreported simulation experiments, similar findings are confirmed universally under other simulation settings such as change in variance, covariance, auto-correlation and quantile.

Note that the performance of SN-WBS$^2$ and SN-NOT$^2$ are in general better than that of SN-WBS$^1$ and SN-NOT$^1$, indicating the benefit of incorporating a minimum spacing $\lfloor n\epsilon\rfloor$. SN-SBS$^1$ and SN-SBS$^2$ have similar performance, so do SN-SNOT$^1$ and SN-SNOT$^2$, indicating SBS and SNOT are robust to its tuning parameter decay rate, which is also observed in \cite{Kovacs2020}.

\begin{table}[H]
\centering
\caption{Performance of SNM, SN-WBS, SN-NOT, SN-SBS, SN-SNOT and LASSO under change in univariate mean with $d=1$.}
\label{tab: compare_wbs_not}
{\small 
	\begin{tabular}{c|c|rrrrrrr|r|r|r|r|r}
		\hline
		\hline & &  \multicolumn{7}{c}{$\hat{m}-m_o$} & &&& \\
		\hline
		Method & Model & $\leq -3$ & $-2$ & $-1$ & $0$ & $1$ & $2$ & $\geq 3$ & ARI & $d_1 \times 10^2$ & $d_2 \times 10^2$ & $d_H \times 10^2$ & time\\
		\hline
		SNM    &     & 0 & 0 & 9 & 974 & 17 & 0 & 0 & 0.960 & 0.87 & 0.90 & 1.01 & 1.75 \\ 
		SN-WBS$^1$ &   & 0 & 1 & 22 & 805 & 153 & 16 & 3 & 0.953 & 1.65 & 1.18 & 2.07 & 9.26 \\ 
		SN-NOT$^1$ &   & 0 & 0 & 26 & 856 & 111 & 7 & 0 & 0.958 & 1.58 & 1.23 & 2.00 & 9.26 \\ 
		SN-WBS$^2$ &  & 0 & 0 & 7 & 936 & 56 & 1 & 0 & 0.962 & 1.06 & 0.82 & 1.17 & 10.04 \\ 
		SN-NOT$^2$ & $(M1)$  & 0 & 0 & 7 & 942 & 51 & 0 & 0 & 0.968 & 0.98 & 0.77 & 1.09 & 10.04 \\ 
		SN-SBS$^1$ & & 0 & 0 & 21 & 854 & 117 & 7 & 1 & 0.956 & 1.53 & 1.07 & 1.86 & 0.98 \\ 
		SN-SNOT$^1$ &   & 0 & 0 & 22 & 870 & 103 & 4 & 1 & 0.962 & 1.45 & 1.10 & 1.80 & 0.98 \\
		SN-SBS$^2$ &   & 0 & 0 & 20 & 837 & 128 & 14 & 1 & 0.956 & 1.63 & 1.05 & 1.95 & 4.05 \\ 
		SN-SNOT$^2$ &   & 0 & 0 & 20 & 853 & 114 & 12 & 1 & 0.962 & 1.55 & 1.06 & 1.88 & 4.05 \\ 
		LASSO &   & 0 & 0 & 0 & 5 & 19 & 67 & 909 & 0.954 & 2.08 & 0.21 & 2.08 & 2.12 \\ 
		\hline
		SNM    &    & 0 & 11 & 196 & 749 & 43 & 1 & 0 & 0.970 & 1.33 & 1.77 & 2.67 & 3.55 \\ 
		SN-WBS$^1$ &   & 0 & 30 & 149 & 537 & 223 & 51 & 10 & 0.928 & 5.45 & 1.92 & 6.45 & 18.63 \\ 
		SN-NOT$^1$ &   & 0 & 33 & 177 & 555 & 194 & 34 & 7 & 0.928 & 5.42 & 2.09 & 6.48 & 18.63 \\ 
		SN-WBS$^2$ &  & 0 & 11 & 137 & 724 & 111 & 16 & 1 & 0.955 & 2.98 & 1.30 & 3.71 & 20.20 \\ 
		SN-NOT$^2$  &$(M2)$   & 0 & 14 & 137 & 725 & 108 & 14 & 2 & 0.956 & 2.95 & 1.30 & 3.69 & 20.20 \\ 
		SN-SBS$^1$ &   & 0 & 7 & 88 & 694 & 178 & 30 & 3 & 0.943 & 4.28 & 1.05 & 4.73 & 1.86 \\ 
		SN-SNOT$^1$ &   & 0 & 6 & 99 & 693 & 169 & 30 & 3 & 0.943 & 4.25 & 1.12 & 4.74 & 1.86 \\ 
		SN-SBS$^2$ &   & 0 & 5 & 80 & 651 & 207 & 44 & 13 & 0.933 & 5.25 & 0.99 & 5.65 & 7.71 \\ 
		SN-SNOT$^2$ &  & 0 & 5 & 80 & 651 & 207 & 46 & 11 & 0.934 & 5.23 & 0.96 & 5.63 & 7.71 \\ 
		LASSO &    & 0 & 0 & 0 & 0 & 1 & 0 & 999 & 0.728 & 12.83 & 0.23 & 12.83 & 3.71 \\ 	\hline
		SNM    &     & 0 & 0 & 1 & 986 & 13 & 0 & 0 & 0.969 & 1.11 & 0.80 & 1.14 & 10.59 \\ 
		SN-WBS$^1$ &   & 0 & 0 & 1 & 985 & 14 & 0 & 0 & 0.967 & 1.15 & 0.88 & 1.17 & 50.26 \\ 
		SN-NOT$^1$ &   & 0 & 0 & 1 & 985 & 14 & 0 & 0 & 0.975 & 0.96 & 0.69 & 0.98 & 50.26 \\ 
		SN-WBS$^2$ &   & 0 & 0 & 1 & 983 & 16 & 0 & 0 & 0.966 & 1.24 & 0.88 & 1.27 & 55.39 \\ 
		SN-NOT$^2$ &$(M3)$   & 0 & 0 & 1 & 983 & 16 & 0 & 0 & 0.973 & 1.08 & 0.69 & 1.10 & 55.39 \\ 
		SN-SBS$^1$ &   & 0 & 0 & 0 & 982 & 16 & 2 & 0 & 0.966 & 1.29 & 0.86 & 1.29 & 4.81 \\ 
		SN-SNOT$^1$ &   & 0 & 0 & 0 & 982 & 16 & 2 & 0 & 0.974 & 1.11 & 0.66 & 1.11 & 4.81 \\ 
		SN-SBS$^2$ &   & 0 & 0 & 1 & 986 & 13 & 0 & 0 & 0.968 & 1.11 & 0.87 & 1.13 & 19.92 \\ 
		SN-SNOT$^2$ &   & 0 & 0 & 3 & 984 & 13 & 0 & 0 & 0.975 & 0.91 & 0.72 & 0.99 & 19.92 \\ 
		LASSO &   & 0 & 0 & 64 & 331 & 296 & 202 & 107 & 0.912 & 2.11 & 4.42 & 5.72 & 8.64 \\ 
		\hline\hline
	\end{tabular}
}
\end{table} 


As shown in Table \ref{tab: compare_wbs_not}, LASSO, as an $L_1$ penalty based method, tends to significantly overestimate the number of change-points in (M1)-(M3). Therefore, we further consider a simulation setting with a large number~(11) of change-points by adopting the block model in \cite{harchaoui2010multiple}. Figure \ref{fig: block_model} gives typical realizations of the block model. Here, the block signals are corrupted with Gaussian errors $\mathcal{N}(0,\sigma^2)$ at two noise levels: medium-noise with $\sigma=0.1$ and high-noise with $\sigma=0.5$, and at two temporal dependence levels: independence with $\rho=0$ and AR(1)-dependence with $\rho=0.5$. The estimation result is summarized in Table \ref{tab: block_model}. As can be seen, SNM still outperforms LASSO in all settings with little efficiency loss.

\begin{figure}[H]
\begin{subfigure}{0.3\textwidth}
	\includegraphics[width=1.1\textwidth]{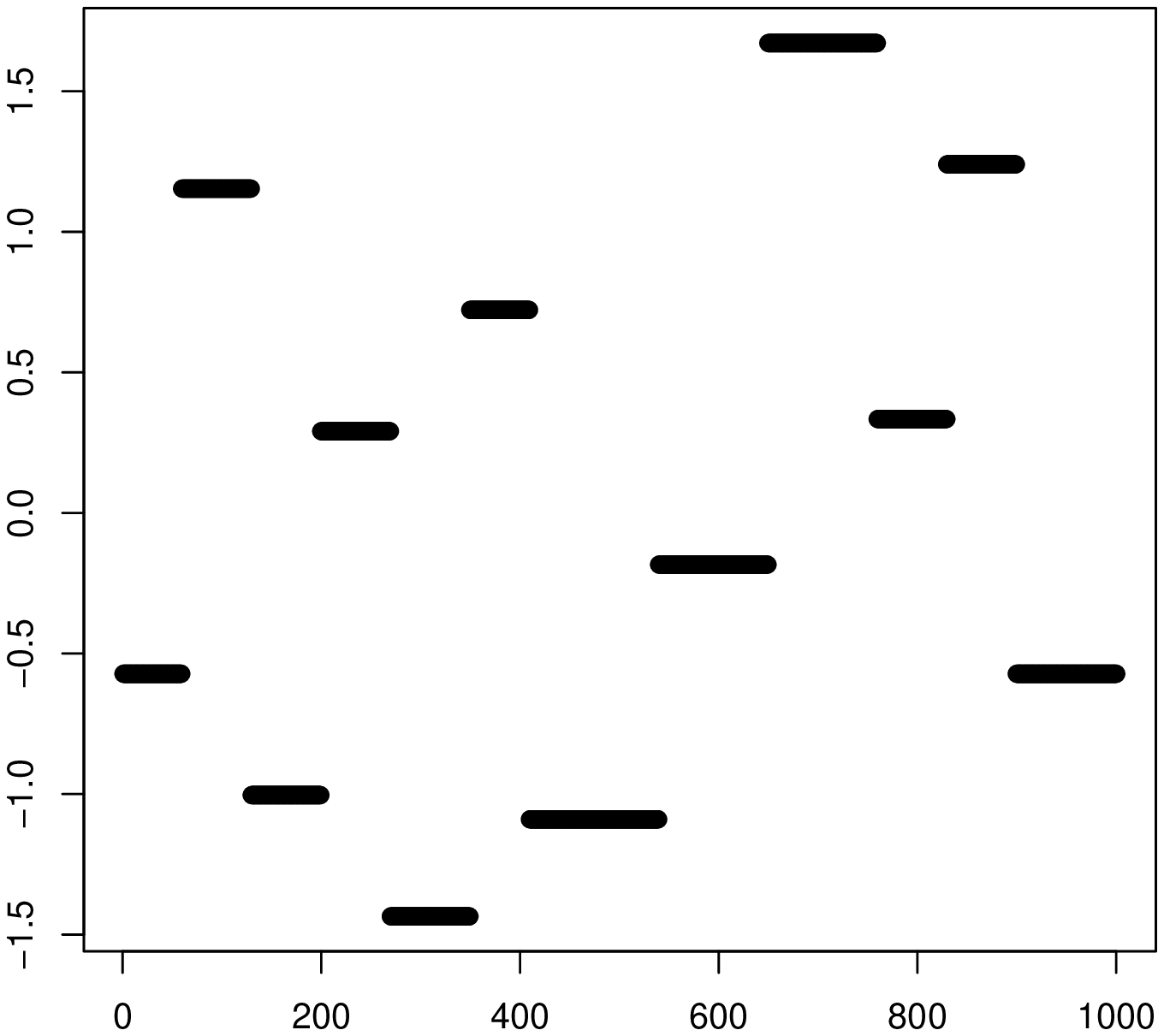}
	\vspace{-0.2cm}
\end{subfigure}
~
\begin{subfigure}{0.3\textwidth}
	\includegraphics[width=1.1\textwidth]{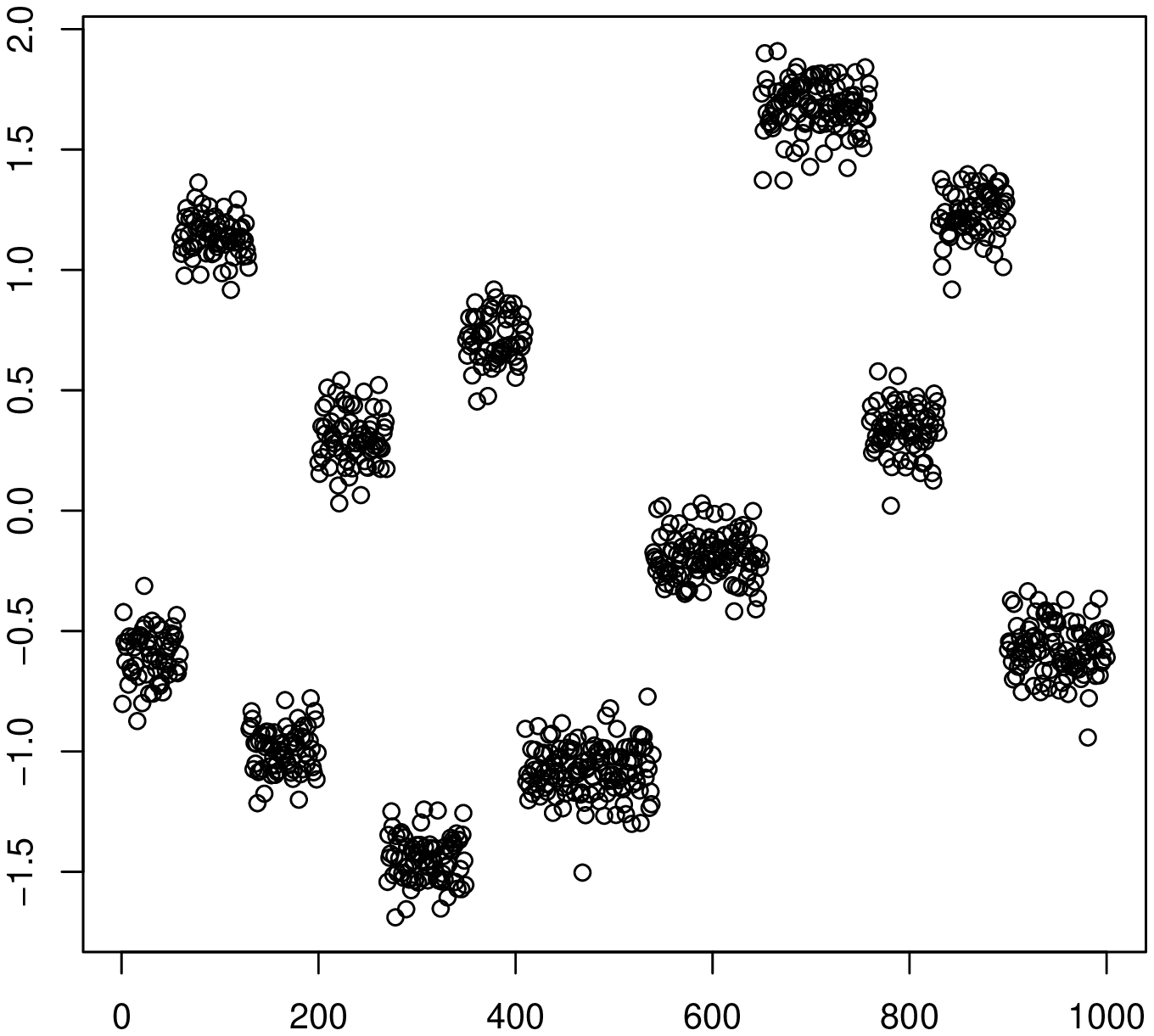}
	\vspace{-0.2cm}
\end{subfigure}
~
\begin{subfigure}{0.3\textwidth}
	\includegraphics[width=1.1\textwidth]{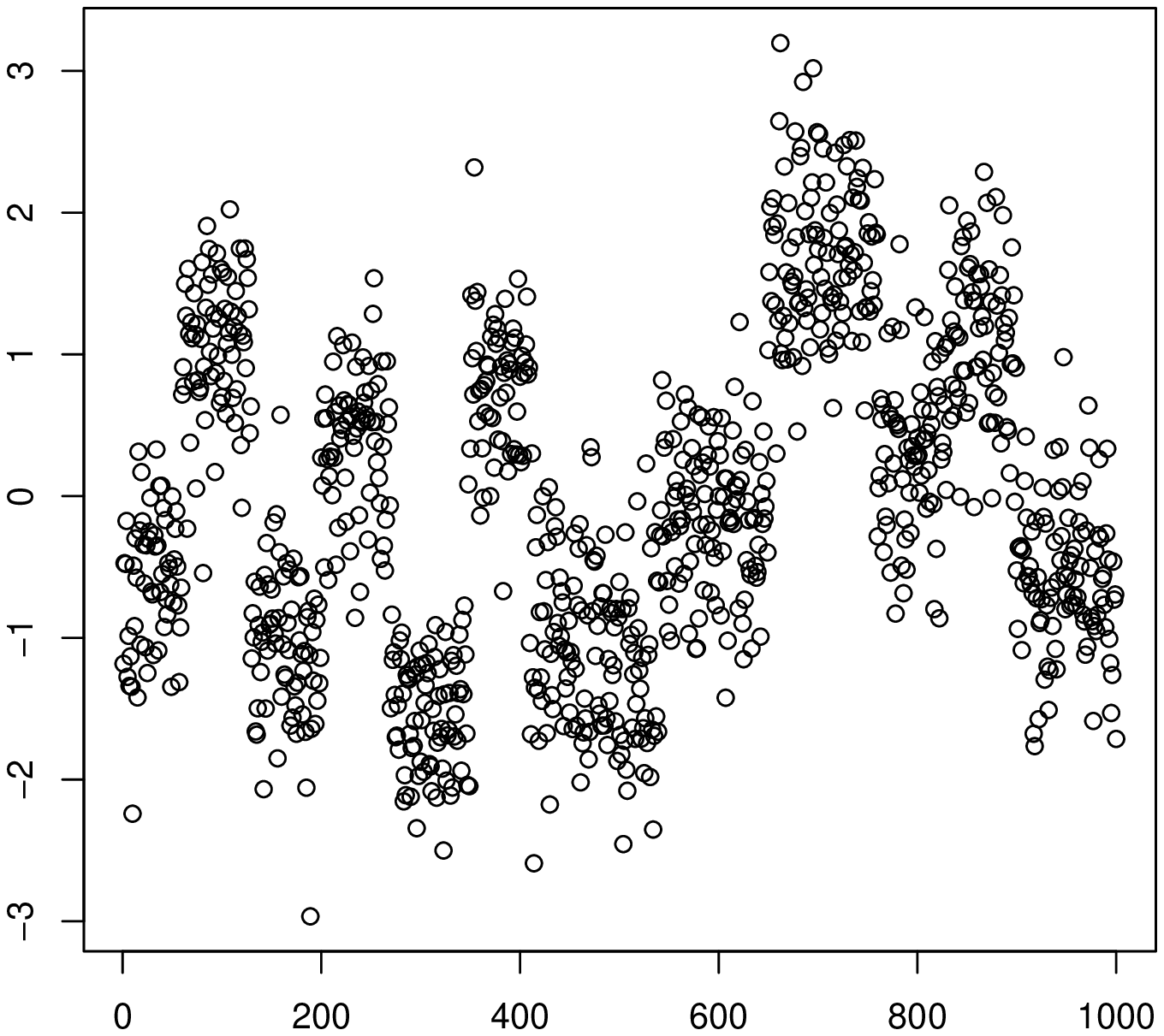}
	\vspace{-0.2cm}
\end{subfigure}		
\caption{Typical realization of the block model. Left: no noise ($\sigma=0$); Middle: medium noise ($\sigma=0.1$); Right: high noise ($\sigma=0.5$).}
\label{fig: block_model}	
\end{figure}

\begin{table}[H]
\centering
\caption{Performance of SNM and LASSO under block models with different noise levels $\sigma$ and temporal dependence levels $\rho$.}
\label{tab: block_model}
{\small 
	\begin{tabular}{rrrrrrrrrrrrrrr}
		\hline
		\hline & & & \multicolumn{7}{c}{$\hat{m}-m_o$} & &&& \\
		\hline
		Method & $\sigma$ &$\rho$ & $\leq -3$ & $-2$ & $-1$ & $0$ & $1$ & $2$ & $\geq 3$ & ARI & $d_1 \times 10^2$ & $d_2 \times 10^2$ & $d_H \times 10^2$ & time\\
		\hline
		SNM & 0.1 & 0.0 & 0 & 0 & 0 & 999 & 1 & 0 & 0 & 0.989 & 0.11 & 0.10 & 0.11 & 3.82 \\ 
		LASSO & 0.1 & 0.0 & 3 & 22 & 60 & 134 & 199 & 194 & 388 & 0.627 & 3.29 & 12.99 & 12.99 & 3.36 \\
		\hline
		SNM & 0.1 & 0.5 & 0 & 0 & 0 & 996 & 4 & 0 & 0 & 0.988 & 0.14 & 0.11 & 0.14 & 3.90 \\ 
		LASSO  & 0.1 & 0.5 & 4 & 30 & 67 & 133 & 173 & 157 & 436 & 0.627 & 3.30 & 12.99 & 12.99 & 3.40 \\
		\hline 
		SNM & 0.5 & 0.0 & 0 & 0 & 45 & 954 & 1 & 0 & 0 & 0.973 & 0.32 & 0.62 & 0.63 & 3.86 \\ 
		LASSO  & 0.5 & 0.0 & 26 & 51 & 122 & 162 & 179 & 166 & 294 & 0.647 & 3.29 & 11.93 & 11.93 & 3.34 \\ 
		\hline
		SNM & 0.5 & 0.5 & 10 & 110 & 450 & 429 & 1 & 0 & 0 & 0.916 & 0.57 & 4.54 & 4.55 & 3.88 \\ 
		LASSO & 0.5 & 0.5 & 21 & 33 & 81 & 127 & 141 & 153 & 444 & 0.652 & 3.31 & 11.23 & 11.23 & 3.35 \\ 
		\hline\hline
	\end{tabular}
}
\end{table}

\subsubsection{Multivariate mean change}

We further compare the proposed nested local-window segmentation algorithm~(i.e. SNM) with SN-WBS, SN-NOT, SN-SBS and SN-SNOT for multivariate mean change. Specifically, we consider models (M1), (M2) and (M3) with $d=5$ and $d=10$, as in Section \ref{subsec:powermean} of the main text. To conserve space, we only report the performance of SN-WBS$^2$, SN-NOT$^2$, SN-SBS$^2$ and SN-NOT$^2$. The performance of SN-WBS$^1$, SN-NOT$^1$, SN-SBS$^1$ and SN-NOT$^1$ are similar but slightly worse.

Table \ref{tab: mean_gaussian_error_add_not_wbs} gives the estimation result for $d=5$ and Table \ref{tab: mean_gaussian_error_add_not_wbs2} summarizes the result for $d=10$. In general, the observation is the same as the one for univariate mean. Specifically, SNM provides comparable (or more favorable) performance as SN-WBS, SN-NOT, SN-SBS and SN-SNOT. For (M1) and (M2), SNM indeed provides notably better performance.

\begin{table}[H]
\centering
\caption{Performance of SNM, SN-WBS, SN-NOT, SN-SBS, SN-SNOT under change in multivariate mean with $d=5$.}
\label{tab: mean_gaussian_error_add_not_wbs}
{\small 
	\begin{tabular}{c|c|rrrrrrr|r|r|r|rr}
		\hline
		\hline & & \multicolumn{7}{c}{$\hat{m}-m_o$} & &&& \\
		\hline
		Method & Model & $\leq -3$ & $-2$ & $-1$ & $0$ & $1$ & $2$ & $\geq 3$ & ARI & $d_1 \times 10^2$ & $d_2 \times 10^2$ & $d_H \times 10^2$ & time\\
		\hline
		SNM & & 0 & 0 & 13 & 946 & 41 & 0 & 0 & 0.953 & 1.16 & 1.12 & 1.37 & 12.48 \\ 
		SN-WBS$^2$  &    & 0 & 0 & 12 & 802 & 178 & 8 & 0 & 0.950 & 1.94 & 1.10 & 2.15 & 146.87 \\ 
		SN-NOT$^2$  &$(M1)$     & 0 & 0 & 11 & 826 & 156 & 7 & 0 & 0.944 & 2.16 & 1.41 & 2.36 & 146.87 \\
		SN-SBS$^2$   &  & 0 & 1 & 34 & 648 & 255 & 54 & 8 & 0.934 & 2.96 & 1.78 & 3.68 & 50.01 \\ 
		SN-SNOT$^2$  &  & 0 & 1 & 37 & 694 & 224 & 39 & 5 & 0.926 & 3.11 & 2.28 & 3.84 & 50.01 \\ 
		\hline
		SNM && 0 & 11 & 175 & 628 & 166 & 18 & 2 & 0.937 & 4.59 & 1.93 & 5.68 & 22.88 \\ 
		SN-WBS$^2$  &    & 0 & 6 & 71 & 380 & 338 & 152 & 53 & 0.854 & 11.85 & 1.65 & 12.17 & 291.21 \\ 
		SN-NOT$^2$  & $(M2)$     & 0 & 7 & 72 & 395 & 341 & 140 & 45 & 0.852 & 11.89 & 1.85 & 12.21 & 291.21 \\ 
		SN-SBS$^2$ &  & 0 & 1 & 13 & 106 & 259 & 244 & 377 & 0.737 & 20.89 & 1.23 & 20.97 & 97.32 \\ 
		SN-SNOT$^2$ &  & 0 & 1 & 12 & 108 & 260 & 246 & 373 & 0.734 & 20.90 & 1.47 & 20.98 & 97.32 \\ 
		\hline
		SNM &  & 0 & 0 & 4 & 993 & 3 & 0 & 0 & 0.968 & 0.93 & 0.96 & 1.03 & 60.00 \\ 
		SN-WBS$^2$ &   & 0 & 0 & 4 & 992 & 4 & 0 & 0 & 0.966 & 1.01 & 1.02 & 1.10 & 829.59 \\ 
		SN-NOT$^2$ &$(M3)$   & 0 & 0 & 4 & 992 & 4 & 0 & 0 & 0.967 & 1.02 & 1.04 & 1.12 & 829.59 \\ 
		SN-SBS$^2$ &   & 0 & 0 & 14 & 985 & 1 & 0 & 0 & 0.962 & 1.02 & 1.33 & 1.35 & 258.03 \\ 
		SN-SNOT$^2$ &  & 0 & 0 & 13 & 986 & 1 & 0 & 0 & 0.967 & 0.91 & 1.21 & 1.23 & 258.03\\ 
		\hline\hline
	\end{tabular}
}
\end{table}

\begin{table}[H]
\centering
\caption{Performance of SNM, SN-WBS, SN-NOT, SN-SBS, SN-SNOT under change in multivariate mean with $d=10$.}
\label{tab: mean_gaussian_error_add_not_wbs2}
{\small 
	\begin{tabular}{c|c|rrrrrrr|r|r|r|rr}
		\hline
		\hline & & \multicolumn{7}{c}{$\hat{m}-m_o$} & &&& \\
		\hline
		Method & Model & $\leq -3$ & $-2$ & $-1$ & $0$ & $1$ & $2$ & $\geq 3$ & ARI & $d_1 \times 10^2$ & $d_2 \times 10^2$ & $d_H \times 10^2$ & time\\
		\hline
		SNM &   & 0 & 0 & 9 & 835 & 142 & 13 & 1 & 0.945 & 1.92 & 1.19 & 2.08 & 24.79 \\ 
		SN-WBS$^2$  &   & 0 & 0 & 10 & 520 & 324 & 118 & 28 & 0.925 & 3.93 & 1.37 & 4.20 & 315.41 \\ 
		SN-NOT$^2$  &$(M1)$   & 0 & 1 & 15 & 574 & 296 & 99 & 15 & 0.905 & 4.15 & 2.20 & 4.45 & 315.41 \\ 
		SN-SBS$^2$  &  & 0 & 5 & 51 & 377 & 349 & 163 & 55 & 0.892 & 5.18 & 3.14 & 6.59 & 98.97 \\ 
		SN-SNOT$^2$ &  & 0 & 5 & 66 & 441 & 348 & 110 & 30 & 0.871 & 5.29 & 4.18 & 6.79 & 98.97 \\ 
		\hline
		SNM &  & 0 & 4 & 22 & 283 & 330 & 250 & 111 & 0.802 & 15.32 & 1.38 & 15.50 & 48.95 \\ 
		SN-WBS$^2$  &   & 0 & 0 & 0 & 7 & 52 & 104 & 837 & 0.564 & 27.58 & 1.42 & 27.58 & 651.82 \\ 
		SN-NOT$^2$  &$(M2)$   & 0 & 0 & 1 & 7 & 50 & 135 & 807 & 0.559 & 27.56 & 2.15 & 27.56 & 651.82 \\ 
		SN-SBS$^2$  &  & 0 & 0 & 0 & 0 & 0 & 0 & 1000 & 0.373 & 34.93 & 0.98 & 34.93 & 205.36 \\ 
		SN-SNOT$^2$ &  & 0 & 0 & 0 & 0 & 0 & 0 & 1000 & 0.369 & 34.92 & 1.39 & 34.92 & 205.36 \\
		\hline
		SNM & & 0 & 0 & 16 & 983 & 1 & 0 & 0 & 0.966 & 0.90 & 1.26 & 1.29 & 136.20 \\ 
		SN-WBS$^2$  &   & 0 & 0 & 25 & 975 & 0 & 0 & 0 & 0.961 & 0.96 & 1.58 & 1.58 & 1870.29 \\ 
		SN-NOT$^2$  & $(M3)$  & 0 & 0 & 26 & 974 & 0 & 0 & 0 & 0.962 & 0.96 & 1.59 & 1.59 & 1870.29 \\  
		SN-SBS$^2$ &  & 0 & 0 & 46 & 954 & 0 & 0 & 0 & 0.955 & 0.95 & 2.09 & 2.09 & 565.24 \\ 
		SN-SNOT$^2$ &  & 0 & 0 & 45 & 955 & 0 & 0 & 0 & 0.957 & 0.95 & 2.05 & 2.05 & 565.24 \\ 
		\hline\hline
	\end{tabular}
}
\end{table}

\subsection{No change}\label{subsec:size_supplement}
Table \ref{tab: null_gaussian_error_small} reports the performance of SNCP and comparison methods under the null for $n=512$. Since the sample size $n=512$ is small, we set the window size $\epsilon=0.1$ for SNCP to ensure we have sufficient observations in each local-window. All implementations are the same as Table \ref{tab: null_gaussian_error} in the main text. As can be seen, SNCP in general gives decent performance (i.e. achieving the target size at 10\%) except for $\rho=0.8$, which is understandable as the effective sample size is low for strong positive temporal dependence.

\begin{table}[ht]
\caption{Performance under no change-point scenario with $m_o=0$.}
\label{tab: null_gaussian_error_small}
\centering
\small
\begin{tabular}{rrrrrrrrrrrrrrrr}
	\hline\hline 
	$n= 512 $ & \multicolumn{3}{c}{$\rho= -0.8 $}& \multicolumn{3}{c}{$\rho= -0.5 $}& \multicolumn{3}{c}{$\rho= 0 $}& \multicolumn{3}{c}{$\rho= 0.5 $}& \multicolumn{3}{c}{$\rho= 0.8 $} \\
	\hline
	$\hat{m}$ & $0$ & $1$ & $\geq 2$ & $0$ & $1$ & $\geq 2$ & $0$ & $1$ & $\geq 2$ & $0$ & $1$ & $\geq 2$ & $0$ & $1$ & $\geq 2$  \\
	\hline
	SN90M & 0.99 & 0.01 & 0.00 & 0.96 & 0.04 & 0.00 & 0.92 & 0.07 & 0.00 & 0.89 & 0.10 & 0.01 & 0.72 & 0.23 & 0.06 \\ 
	BP & 1.00 & 0.00 & 0.00 & 1.00 & 0.00 & 0.00 & 0.98 & 0.01 & 0.00 & 0.26 & 0.12 & 0.62 & 0.00 & 0.00 & 1.00 \\ 
	\hline
	SN90V & 0.84 & 0.14 & 0.01 & 0.91 & 0.09 & 0.00 & 0.93 & 0.07 & 0.00 & 0.87 & 0.12 & 0.01 & 0.75 & 0.21 & 0.04 \\ 
	KF & 0.32 & 0.25 & 0.44 & 0.69 & 0.21 & 0.11 & 0.81 & 0.12 & 0.06 & 0.83 & 0.12 & 0.05 & 0.83 & 0.12 & 0.05 \\ 
	MSML & 0.48 & 0.34 & 0.19 & 0.76 & 0.21 & 0.04 & 0.83 & 0.15 & 0.01 & 0.84 & 0.14 & 0.01 & 0.84 & 0.14 & 0.01 \\ 
	\hline\hline
\end{tabular}
\end{table}

\subsection{Change in multivariate mean}\label{subsec:add_powermean}

Table \ref{tab: mean_gaussian_error_add2} reports the performance of SNM and CUSUM under change in multivariate mean with $d=10$. Compared to Table \ref{tab: mean_gaussian_error} in the main text (with $d=1$ and $5$), it can be seen that the performance of SNM and CUSUM both deteriorate due to the increasing dimension $d$. However, the deterioration of CUSUM is much more notable while SNM still gives decent performance as measured by ARI and $d_H$.

\begin{table}[H]
\centering
\caption{Performance of SNM and CUSUM under change in multivariate mean with $d=10$.}
\label{tab: mean_gaussian_error_add2}
{\small 
	\begin{tabular}{c|c|rrrrrrr|r|r|r|rr}
		\hline
		\hline & & \multicolumn{7}{c}{$\hat{m}-m_o$} & &&& \\
		\hline
		Method & Model & $\leq -3$ & $-2$ & $-1$ & $0$ & $1$ & $2$ & $\geq 3$ & ARI & $d_1 \times 10^2$ & $d_2 \times 10^2$ & $d_H \times 10^2$ & time\\
		\hline
		SNM & $(M1)$  & 0 & 0 & 9 & 835 & 142 & 13 & 1 & 0.945 & 1.92 & 1.19 & 2.08 & 24.79 \\ 
		CUSUM &  & 242 & 9 & 2 & 186 & 287 & 192 & 82 & 0.695 & 4.79 & 15.47 & 19.19 & 0.04 \\ 
		\hline
		SNM & $(M2)$ & 0 & 4 & 22 & 283 & 330 & 250 & 111 & 0.802 & 15.32 & 1.38 & 15.50 & 48.95 \\ 
		CUSUM &  & 49 & 13 & 15 & 19 & 57 & 157 & 690 & 0.505 & 23.29 & 5.64 & 25.54 & 0.10 \\
		\hline
		SNM & $(M3)$ & 0 & 0 & 16 & 983 & 1 & 0 & 0 & 0.966 & 0.90 & 1.26 & 1.29 & 136.20 \\ 
		CUSUM & & 0 & 617 & 1 & 381 & 1 & 0 & 0 & 0.367 & 0.47 & 31.34 & 31.35 & 0.07 \\ 
		\hline\hline
	\end{tabular}
}
\end{table}

\subsection{Change in variance}\label{subsec:powervariance}
For variance change, we consider four univariate time series $\{Y_t\}_{t=1}^n$ with piecewise constant variance. (V1) is an AR(1) process with moderate temporal dependence adapted from \cite{Cho2012}. (V2) is an ARMA(1,1) process taken from \cite{Korkas2017}. (V3) is an AR(2) process with strong positive dependence taken from \cite{Cho2012}. (V4) is an AR(1) process with longer segments and only one small-scale change. Typical realizations of (V1)-(V4) can be found in Figure \ref{fig: dgp} of the supplementary material.
\begin{align*}
&\text{(V1)}: n=1024, ~~Y_t=\begin{cases} 
	0.5Y_{t-1}+\epsilon_t, & t\in [1,400],\\
	0.5Y_{t-1}+2\epsilon_t, &t\in[401,750],\\
	0.5Y_{t-1}+\epsilon_t, &t\in[751,1024].
\end{cases}\\
&\text{(V2)}: n=1024, ~~Y_t=\begin{cases} 
	0.7Y_{t-1}+\epsilon_t+0.6\epsilon_{t-1}, & t\in [1,125],\\
	0.3Y_{t-1}+\epsilon_t+0.3\epsilon_{t-1}, & t\in [126,532],\\
	0.9Y_{t-1}+\epsilon_t, &t\in[533,704],\\
	0.1Y_{t-1}+\epsilon_t-0.5\epsilon_{t-1}, &t\in[705,1024].
\end{cases}\\
&\text{(V3)}: n=1024, ~~Y_t=\begin{cases} 
	0.9Y_{t-1}+\epsilon_t, & t\in [1,512],\\
	1.69Y_{t-1}-0.81Y_{t-2}+\epsilon_t, &t\in[513,768],\\
	1.32Y_{t-1}-0.81Y_{t-2}+\epsilon_t, &t\in[769,1024].
\end{cases}\\
&\text{(V4)}: n=2048, ~~Y_t=\begin{cases} 
	-0.7Y_{t-1}+\epsilon_t, & t\in [1,1024],\\
	-0.7Y_{t-1}+\sqrt{2}\epsilon_t, &t\in[1025,2048].
\end{cases}
\end{align*}
The error process $\{\epsilon_t\}$ is \textit{i.i.d.} standard normal $N(0,1)$.

The estimation result is summarized in Table \ref{tab: variance_gaussian_error}. For (V1), under moderate temporal dependence, all methods give decent performance with some degree of over-estimation. For (V2), due to the complex dependence, all methods experience power loss, especially for KF and SNV, with MSML giving the best performance. For (V3), due to the strong positive dependence, SNV and KF again experience power loss, with SNV giving noticeably larger estimation error. For (V4), under strong negative dependence, KF and MSML severely over-estimate the number of change-points while SNV gives robust and best performance. In summary, SNV performs quite well  compared to MSML and KF, though it may exhibit some lack of power under strong positive dependence $(\rho\geq 0.9)$.


\begin{table}[H]
\centering
\caption{Performance of SNV, KF, MSML under change in variance.}
\label{tab: variance_gaussian_error}
{\small
	\begin{tabular}{c|r|rrrrrrr|r|r|r|rZ}
		\hline
		\hline & & \multicolumn{7}{c}{$\hat{m}-m_o$} & &&& \\
		\hline
		Method & Model & $\leq -3$ & $-2$ & $-1$ & $0$ & $1$ & $2$ & $\geq 3$ & ARI & $d_1 \times 10^2$ & $d_2 \times 10^2$ & $d_H \times 10^2$ & time\\
		\hline
		SNV &  & 0 & 0 & 5 & 938 & 53 & 4 & 0 & 0.942 & 2.09 & 1.49 & 2.26 & 17.30 \\ 
		KF & $(V1)$ & 0 & 0 & 0 & 963 & 35 & 2 & 0 & 0.958 & 1.70 & 1.09 & 1.70 & 0.27 \\ 
		MSML &  & 0 & 0 & 0 & 892 & 102 & 6 & 0 & 0.957 & 2.84 & 0.96 & 2.84 & 0.02 \\ 
		\hline
		SNV &  & 12 & 95 & 335 & 538 & 18 & 2 & 0 & 0.762 & 2.94 & 17.35 & 17.50 & 17.73 \\ 
		KF & $(V2)$ & 0 & 121 & 397 & 463 & 17 & 2 & 0 & 0.730 & 3.30 & 17.53 & 17.73 & 0.34 \\ 
		MSML &  & 0 & 19 & 161 & 674 & 137 & 7 & 2 & 0.792 & 5.49 & 8.49 & 9.46 & 0.02 \\ 
		\hline
		SNV &  & 0 & 37 & 315 & 542 & 94 & 12 & 0 & 0.733 & 6.76 & 13.65 & 15.99 & 17.62 \\ 
		KF & $(V3)$ & 0 & 0 & 129 & 543 & 258 & 60 & 10 & 0.853 & 6.52 & 7.44 & 9.94 & 0.30 \\ 
		MSML &  & 0 & 0 & 1 & 638 & 295 & 58 & 8 & 0.883 & 8.23 & 3.53 & 8.26 & 0.02 \\ 
		\hline
		SNV &  & 0 & 0 & 38 & 898 & 58 & 6 & 0 & 0.870 & 3.66 & 4.03 & 5.56 & 41.28 \\ 
		KF & $(V4)$ & 0 & 0 & 0 & 315 & 219 & 246 & 220 & 0.769 & 22.35 & 1.33 & 22.35 & 1.02 \\ 
		MSML &  & 0 & 0 & 0 & 439 & 345 & 154 & 62 & 0.834 & 17.68 & 1.58 & 17.68 & 0.03 \\ 
		\hline\hline
	\end{tabular}
}
\end{table}

\subsection{Change in autocorrelation}\label{subsec:poweracf}
For autocorrelation change, we generate three univariate time series $\{Y_t\}_{t=1}^n$ with piecewise constant autocorrelation. (A1) and (A2) are AR(1) processes taken from \cite{Cho2012}. (A3) is an ARMA(1,1) process adapted from \cite{Korkas2017}. Typical realizations of (A1)-(A3) can be found in Figure \ref{fig: dgp} of the supplementary material.
\begin{align*}
&\text{(A1)}: n=1024, ~~Y_t=\begin{cases} 
0.5Y_{t-1}+\epsilon_t, & t\in [1,400],\\
0.9Y_{t-1}+\epsilon_t, &t\in[401,750],\\
0.3Y_{t-1}+\epsilon_t, &t\in[751,1024].
\end{cases}\\
&\text{(A2)}: n=1024, ~~Y_t=\begin{cases} 
0.75Y_{t-1}+\epsilon_t, & t\in [1,50],\\
-0.5Y_{t-1}+\epsilon_t, &t\in[51,1024].
\end{cases}\\
&\text{(A3)}: n=1024, ~~Y_t=\begin{cases} 
-0.9Y_{t-1}+\epsilon_t+0.7\epsilon_{t-1}, & t\in [1,512],\\
0.9Y_{t-1}+\epsilon_t, & t\in [513,768],\\
\epsilon_t-0.7\epsilon_{t-1}, &t\in[769,1024].
\end{cases}
\end{align*}
The error process $\{\epsilon_t\}$ is \textit{i.i.d.} standard normal $N(0,1)$.

The estimation result is summarized in Table \ref{tab: acf_gaussian_error}. For (A1), SNA gives the best performance while both KF and MSML seem to suffer power loss. For (A2), the change-point location is close to the boundary with $\tau_1=50/1024<0.05=\epsilon$, violating the assumption of SNA. However, SNA still delivers arguably the best performance as measured by ARI. For (A3), all methods tend to over-estimate with SNA providing the most robust performance. In summary, SNA performs favorably compared to MSML and KF for detecting autocorrelation changes in the time series.
On the other hand, SNA is computationally more expensive than KF and MSML, since the latter two methods are built on fast wavelet transformation. 
\begin{table}[H]
\centering
\caption{Performance of SNA, KF, MSML under change in autocorrelation.}
\label{tab: acf_gaussian_error}
{\small
\begin{tabular}{c|c|rrrrrrr|r|r|r|r|r}
	\hline
	\hline & & \multicolumn{7}{c}{$\hat{m}-m_o$} & &&& \\
	\hline
	Method & Model & $\leq -3$ & $-2$ & $-1$ & $0$ & $1$ & $2$ & $\geq 3$ & ARI & $d_1 \times 10^2$ & $d_2 \times 10^2$ & $d_H \times 10^2$ & time\\
	\hline
	SNA &  & 0 & 2 & 69 & 907 & 20 & 2 & 0 & 0.895 & 2.17 & 4.22 & 4.56 & 18.84 \\ 
	KF & $(A1)$ & 0 & 135 & 107 & 728 & 27 & 3 & 0 & 0.650 & 5.34 & 15.99 & 16.37 & 0.29 \\ 
	MSML &  & 0 & 131 & 71 & 654 & 139 & 4 & 1 & 0.699 & 5.94 & 12.99 & 15.06 & 0.01 \\ 
	\hline
	SNA &  & 0 & 0 & 59 & 741 & 170 & 29 & 1 & 0.724 & 11.23 & 4.39 & 14.18 & 19.03 \\ 
	KF & $(A2)$ & 0 & 0 & 143 & 668 & 118 & 58 & 13 & 0.548 & 13.10 & 11.87 & 20.25 & 0.40 \\ 
	MSML &  & 0 & 0 & 7 & 769 & 201 & 20 & 3 & 0.673 & 12.47 & 2.02 & 12.82 & 0.01 \\
	\hline
	SNA &  & 0 & 0 & 0 & 831 & 150 & 16 & 3 & 0.940 & 4.55 & 0.57 & 4.55 & 18.82 \\ 
	KF & $(A3)$ & 0 & 0 & 0 & 638 & 222 & 118 & 22 & 0.885 & 9.57 & 1.08 & 9.57 & 0.30 \\ 
	MSML &  & 0 & 0 & 0 & 665 & 258 & 69 & 8 & 0.870 & 9.35 & 2.26 & 9.35 & 0.02 \\ 
	\hline\hline
\end{tabular}
}
\end{table}

\subsection{Change in correlation}\label{subsec:powercorrelation}
For correlation change, we generate two bivariate time series $\{Y_t=(Y_{t1},Y_{t2})\}_{t=1}^n$ with $n=1000$ using piecewise constant correlation. Denote 
$\Sigma_r=[1,r;r,1]$, we define
\begin{align*}
&\text{(R0)}: Y_t=0.5\mathbf{I}_2 Y_{t-1}+\mathbf{e}_t, ~ \mathbf{e}_t\overset{i.i.d.}{\sim} N(0, \Sigma_{0.5}), t\in[1,1000].\\
&\text{(R1)}: Y_t=\begin{cases} 
0.5\mathbf{I}_2Y_{t-1}+2\mathbf{e}_t, ~~ \mathbf{e}_t\overset{i.i.d.}{\sim} N(0, \Sigma_{0.8}), & t\in [1,333],\\
0.5\mathbf{I}_2Y_{t-1}+\mathbf{e}_t, ~~ \mathbf{e}_t\overset{i.i.d.}{\sim} N(0, \Sigma_{0.2}), & t\in [334,667],\\
0.5\mathbf{I}_2Y_{t-1}+\mathbf{e}_t, ~~ \mathbf{e}_t\overset{i.i.d.}{\sim} N(0, \Sigma_{0.8}), & t\in [668,1000].
\end{cases}\\
&\text{(R2)}: Y_t=\begin{cases} 
0.5\mathbf{I}_2Y_{t-1}+\sqrt{2}\mathbf{e}_t, ~~ \mathbf{e}_t\overset{i.i.d.}{\sim} N(0, \Sigma_{0.8}), & t\in [1,333],\\
0.5\mathbf{I}_2Y_{t-1}+\mathbf{e}_t, ~~ \mathbf{e}_t\overset{i.i.d.}{\sim} N(0, \Sigma_{0.2}), & t\in [334,667],\\
0.5\mathbf{I}_2Y_{t-1}+2\mathbf{e}_t, ~~ \mathbf{e}_t\overset{i.i.d.}{\sim} N(0, \Sigma_{0.2}), & t\in [668,1000].
\end{cases}
\end{align*}
For (R1), the correlation changes from 0.8 to 0.2 and back to 0.8. At the first change-point $t=333$, the variance of the time series also changes. For (R2), the correlation only changes once from 0.8 to 0.2 at $t=333$. At $t=667$, the covariance matrix changes but the correlation remains the same.


The estimation result is reported in Table \ref{tab: correlation}. For (R1), SNC gives noticeably better performance with much higher ARI and lower Hausdorff distance, while GW seems to over-estimate and experience power loss at the same time. The power loss of GW is due to its inability to detect the first change-point, where correlation change comes with large variance change. For (R2), SNC again outperforms GW. Note that GW systematically over-estimates the number of change-points as it mistakenly detects the sole variance change at $t=667$ as correlation change. For (R0), both methods give decent performance with SNC achieving the target size perfectly. In summary, SNC performs favorably and retains size and power for detecting change in correlation when other quantities such as variance also experience structural breaks.

\begin{table}[H]
\centering
\caption{Performance of SNC, GW under change in correlation of bivariate time series.}
\label{tab: correlation}
{\small
\begin{tabular}{c|c|rrrrrrr|r|r|r|rr}
	\hline
	\hline & & \multicolumn{7}{c}{$\hat{m}-m_o$} & &&& \\
	\hline
	Method & Model & $\leq -3$ & $-2$ & $-1$ & $0$ & $1$ & $2$ & $\geq 3$ & ARI & $d_1 \times 10^2$ & $d_2 \times 10^2$ & $d_H \times 10^2$ & time\\
	\hline
	SNC & $(R1)$ & 0 & 0 & 9 & 941 & 50 & 0 & 0 & 0.937 & 2.06 & 1.76 & 2.35 & 22.43 \\ 
	GW &  & 0 & 376 & 0 & 440 & 159 & 17 & 8 & 0.558 & 3.88 & 20.03 & 22.68 & 0.28 \\  
	\hline
	Method & Model & $\leq -3$ & $-2$ & $-1$ & $0$ & $1$ & $2$ & $\geq 3$ & ARI & $d_1 \times 10^2$ & $d_2 \times 10^2$ & $d_H \times 10^2$ & time\\
	\hline
	SNC & $(R2)$ & 0 & 0 & 0 & 916 & 80 & 4 & 0 & 0.932 & 3.23 & 1.05 & 3.23 & 22.78 \\ 
	GW & & 0 & 0 & 0 & 11 & 471 & 328 & 190 & 0.522 & 31.15 & 1.24 & 31.15 & 0.49 \\\hline
	Method & Model & $\hat m=0$ & $\hat m=1$ & $\hat{m} \geq 2$ \\\cline{1-5} 
	SNC & $(R0)$ & 900 & 90 & 10 & \\ 
	GW &  & 892 & 80 & 28 & \vspace{-0.7mm} \\ 
	\cmidrule{1-5}\morecmidrules\cmidrule{1-5}
\end{tabular}
}
\end{table}

\subsection{Change in quantile}\label{subsec:powerquantile}
To our best knowledge, SNQ is the only nonparametric method available for detecting structural breaks in quantiles under temporal dependence. Note that there is a stream of literature on testing and estimating structural breaks in quantile regressions, see \cite{Qu2008}, \cite{Oka2011} and \cite{Aue2014}. However, in all these papers, a parametric linear quantile form is typically specified and  the error in the location-scale quantile regression model is assumed to be \textit{i.i.d.}, which does not cover the nonparametric time series setting we are focusing on here.

For illustration, we compare SNQ with the recently proposed multiscale quantile segmentation~(MQS) in \cite{Vanegas2020}, which is a nonparametric method designed for detecting piecewise constant quantile changes under a temporal independence assumption. In addition, we compare with ECP in \cite{Matteson2014}, which is a nonparametric method designed for detecting distributional changes based on the energy statistics. We emphasize that the comparison is not completely fair as the validity of MQS and ECP require temporal independence. 

A stylized fact of financial markets is that upper quantiles of negative log-returns~(i.e.\ losses) of stocks are subject to more changes than lower quantiles. Based on a mixture of a truncated normal distribution and a generalized Pareto distribution~(GPD), we design a univariate time series to resemble such phenomenon. GPD is a commonly used distribution for characterizing high quantiles of financial returns, see \cite{Embrechts1997}. Denote $\Phi(\cdot)$ as the CDF of $N(0,1)$, we use a standard normal distribution truncated at 0 with CDF $F_1(x)=2\Phi(x), x\leq 0$. The CDF of a GPD distribution takes the form $F_2(x)=1-(1+\xi(x-\mu)/\sigma)_+^{-1/\xi}$, where we set the location parameter $\mu=0$, scale parameter $\sigma=2$ and tail index $\xi=0.125$. Setting the mixture as $F(x)=0.5F_1(x)+0.5F_2(x)$, it is easy to see that $F(x)$ is a continuous distribution with $F^{-1}(q)=\Phi^{-1}(q)$ for $q\leq 0.5$.

To introduce temporal dependence, we first simulate a stationary univariate time series $\{X_t\}_{t=1}^n$ from an AR(1) process with $X_t=\rho X_{t-1}+\sqrt{1-\rho^2}\epsilon_t,$ where $\rho=0.2$ and $\{\epsilon_t\}$ is \textit{i.i.d.} $N(0,1)$. Thus $\{u_t=\Phi(X_t)\}_{t=1}^n$ is a stationary time series with uniform margins. Based on $F(x)$ and $\{u_t\}_{t=1}^n$, we define $\{Y_t\}_{t=1}^{n}$ such that
\begin{align*}
&\text{(Q1)}: n=1000, ~~Y_t=\begin{cases} 
\Phi^{-1}(u_t), & t\in [1,500],\\
F^{-1}(u_t), &t\in[501,1000].
\end{cases}
\end{align*}
A typical realization of (Q1) can be found in Figure \ref{fig: dgp} of the supplementary material.

We use SNQ and MQS to detect change-points in the 10\% quantile (no change-point) and 90\% quantile (one change-point) of $\{Y_t\}_{t=1}^n$ respectively. Note that ECP cannot be tailored to detect changes in a specific quantile level, thus we can only apply it to detect if there is any distributional change in the data (one change-point). The estimation result is reported in Table \ref{tab: quantile}. For (Q1) 90\% quantile, SNQ gives notably better performance with much higher ARI and lower Hausdorff distance, while MQS seems to severely over-estimate. For (Q1) 10\% quantile, the size of SNQ is close to the target size of 10\%, while MQS yields over-detection of change-points. As for ECP, it always estimates at least one change-point as it targets distributional change and tends to over-estimate due to intolerance to temporal dependence.

In summary, SNQ performs fairly well for quantile change detection while ignoring temporal dependence can lead to less favorable results for MQS. In addition, compared to algorithms that target distributional changes such as ECP, SNCP can be tailored to detect changes in a specific parameter and thus provides more information about the nature of change, e.g.\ whether the change stems from the lower 10\% quantile or the upper 90\% quantile.

\begin{table}[H]
\centering
\caption{Performance of SNQ, MQS under change in quantile.}
\label{tab: quantile}
{\small
\begin{tabu}{c|c|rrrrrrr|r|r|r|r|r}
	\hline
	\hline & & \multicolumn{7}{c}{$\hat{m}-m_o$} & &&& \\
	\hline
	Method & Model & $\leq -3$ & $-2$ & $-1$ & $0$ & $1$ & $2$ & $\geq 3$ & ARI & $d_1 \times 10^2$ & $d_2 \times 10^2$ & $d_H \times 10^2$ & time\\
	\hline
	SNQ &  & 0 & 0 & 4 & 903 & 84 & 9 & 0 & 0.911 & 3.99 & 1.94 & 4.19 & 25.64 \\ 
	MQS & $(Q1)$-90\%  & 0 & 0 & 0 & 594 & 221 & 115 & 70 & 0.746 & 15.19 & 4.01 & 15.19 & 509.27 \\ 
	ECP &  & 0 & 0 & 0 & 658 & 168 & 125 & 49 & 0.842 & 11.46 & 2.11 & 11.46 & 17.64\\
	\hline
	Method & Model & $\hat m=0$ & $\hat m=1$ & $\hat{m} \geq 2$ &  \\\cline{1-5}
	SNQ &  & 860 & 122 & 18 & \\
	MQS & $(Q1)$-10\% & 462 & 313 & 225 & \\
	ECP &  & 0 & 658 & 342  & \vspace{-0.7mm}\\
	\cmidrule{1-5}\morecmidrules\cmidrule{1-5}
\end{tabu}
}
\end{table}

\subsection{Simultaneous change in mean and variance}\label{subsec:powermeanvariance}
In this section, we further investigate the robustness of SNM and SNV under the scenario where both mean and variance change at the same time. We first simulate a stationary univariate time series $\{X_t\}_{t=1}^n$ from an AR(1) process such that $X_t=\rho X_{t-1}+\epsilon_t,$ where $\{\epsilon_t\}$ is \textit{i.i.d.} $N(0,1)$. We then generate two different time series $\{Y_t\}_{t=1}^n$ with piecewise constant mean and variance based on $\{X_t\}_{t=1}^n$.
\begin{align*}
&\text{(MV1)}: n=1024, ~~\rho=0.5, ~~Y_t=\begin{cases} 
0+X_t, & t\in [1,512],\\
1+1.5X_t, &t\in[513,1024].
\end{cases}\\
&\text{(MV2)}: n=1024, ~~\rho=-0.5, ~~Y_t=\begin{cases} 
0+X_t, & t\in [1,512],\\
1+1.5X_t, &t\in[513,1024].
\end{cases}
\end{align*}

We run (SNM, BP) and (SNV, KF, MSML) on $\{Y_t\}_{t=1}^n$. The estimation result is summarized in Table \ref{tab: mean_variance_gaussian_error}. In general, the observation is as follows. SNM and BP are robust against change in variance.  Additionally, BP suffers severe over-estimation under positive dependence. SNV, KF, MSML are robust against change in mean while MSML tends to over-estimate.

\begin{table}[H]
\centering
\caption{Estimation result under change in mean and variance.}
\label{tab: mean_variance_gaussian_error}
{\small
\begin{tabular}{c|c|rrrrrrr|r|r|r|r|r}
	\hline
	\hline & & \multicolumn{7}{c}{$\hat{m}-m_o$} & &&&& \\
	\hline
	Method & Model & $\leq -3$ & $-2$ & $-1$ & $0$ & $1$ & $2$ & $\geq 3$ & ARI & $d_1 \times 10^2$ & $d_2 \times 10^2$ & $d_H \times 10^2$ & time\\
	\hline
	SNM &  & 0 & 0 & 49 & 862 & 81 & 8 & 0 & 0.852 & 4.24 & 4.64 & 6.69 & 4.02 \\ 
	BP & $(MV1)$ & 0 & 0 & 0 & 267 & 174 & 229 & 330 & 0.715 & 24.54 & 2.20 & 24.54 & 37.12 \\ 
	\hline
	SNV &  & 0 & 0 & 27 & 887 & 80 & 6 & 0 & 0.887 & 3.87 & 3.18 & 5.22 & 17.80 \\ 
	KF & $(MV1)$ & 0 & 0 & 1 & 956 & 37 & 6 & 0 & 0.935 & 2.59 & 1.57 & 2.64 & 0.24 \\ 
	MSML &  & 0 & 0 & 0 & 743 & 250 & 7 & 0 & 0.908 & 6.62 & 1.35 & 6.62 & 0.02 \\  
	\hline
	SNM &  & 0 & 0 & 0 & 980 & 20 & 0 & 0 & 0.976 & 0.92 & 0.49 & 0.92 & 3.99 \\ 
	BP & $(MV2)$ & 0 & 0 & 0 & 1000 & 0 & 0 & 0 & 0.990 & 0.25 & 0.25 & 0.25 & 37.40 \\ 
	\hline
	SNV &  & 0 & 0 & 32 & 913 & 54 & 1 & 0 & 0.890 & 3.12 & 3.39 & 4.72 & 17.76 \\ 
	KF & $(MV2)$ & 0 & 0 & 0 & 809 & 135 & 51 & 5 & 0.906 & 6.47 & 1.62 & 6.47 & 0.29 \\ 
	MSML &  & 0 & 0 & 0 & 723 & 236 & 40 & 1 & 0.883 & 8.28 & 1.92 & 8.28 & 0.02 \\ 
	\hline\hline
\end{tabular}
}
\end{table}

\subsection{Change in covariance matrix}\label{subsec: covmatrix}
In this section, we consider two additional simulation settings for changes in covariance matrix of a VAR process. Specifically, we consider change in covariance matrices of a four-dimensional time series $\{Y_t=(Y_{t1},\cdots,Y_{t4})\}_{t=1}^n$ with $n=1000$. Denote $\Sigma_\rho$ as an exchangeable covariance matrix with unit variance and equal covariance $\rho$, we consider
\begin{align*}
&\text{(C3)}: Y_t=\begin{cases} 
0.3\mathbf{I}_{4}Y_{t-1}+\sqrt{2}\mathbf{e}_t, ~~ \mathbf{e}_t\overset{i.i.d.}{\sim} N(0, \Sigma_{0.5}), & t\in [1,333],\\
0.3\mathbf{I}_{4}Y_{t-1}+\mathbf{e}_t, ~~ \mathbf{e}_t\overset{i.i.d.}{\sim} N(0, \Sigma_{0.2}), & t\in [334,667],\\
0.3\mathbf{I}_{4}Y_{t-1}+\sqrt{2}\mathbf{e}_t, ~~ \mathbf{e}_t\overset{i.i.d.}{\sim} N(0, \Sigma_{0.5}), & t\in [668,1000].
\end{cases}\\
&\text{(C4)}: Y_t=\begin{cases} 
0.3\mathbf{I}_{4}Y_{t-1}+\mathbf{e}_t, ~~ \mathbf{e}_t\overset{i.i.d.}{\sim} N(0, \Sigma_{0.2}), & t\in [1,333],\\
0.3\mathbf{I}_{4}Y_{t-1}+\sqrt{2}\mathbf{e}_t, ~~ \mathbf{e}_t\overset{i.i.d.}{\sim} N(0, \Sigma_{0.5}), & t\in [334,667],\\
0.3\mathbf{I}_{4}Y_{t-1}+2\mathbf{e}_t, ~~ \mathbf{e}_t\overset{i.i.d.}{\sim} N(0, \Sigma_{0.5}), & t\in [668,1000].
\end{cases}
\end{align*}
(C3) and (C4) generate covariance changes in the VAR model, which is widely used in the time series literature. The estimation result is reported in Table \ref{tab: add_sn_4dcov}. For monotonic changes (C4), both methods perform well though AHHR tends to over-estimate the number of change-points, while for non-monotonic changes (C3), AHHR seems to over-estimate and experience power loss at the same time and is outperformed by SNCM.

\begin{table}[h]
\centering
\caption{Performance of SNCM and AHHR under change in covariance matrix.} 
\label{tab: add_sn_4dcov}
{\small
\begin{tabular}{c|c|rrrrrrr|r|r|r|r|r}
	\hline
	\hline & & \multicolumn{7}{c}{$\hat{m}-m_o$} & &&&& \\
	\hline
	Method & Model & $\leq -3$ & $-2$ & $-1$ & $0$ & $1$ & $2$ & $\geq 3$ & ARI & $d_1 \times 10^2$ & $d_2 \times 10^2$ & $d_H \times 10^2$ & time\\
	\hline
	SNCM & $(C3)$ & 0 & 0 & 51 & 912 & 37 & 0 & 0 & 0.904 & 2.36 & 3.66 & 4.06 & 56.50 \\ 
	AHHR &  & 0 & 525 & 8 & 342 & 104 & 18 & 3 & 0.408 & 2.84 & 27.95 & 29.32 & 0.37 \\ 
	\hline
	SNCM & $(C4)$ & 0 & 0 & 19 & 947 & 33 & 1 & 0 & 0.924 & 2.22 & 2.42 & 2.85 & 56.96 \\ 
	AHHR &  & 0 & 0 & 8 & 780 & 175 & 33 & 4 & 0.883 & 5.27 & 3.00 & 5.53 & 0.65 \\ \hline\hline
\end{tabular}
}
\end{table}

\subsection{Change in multi-parameter}\label{subsec: multipara}
In this section, we consider an additional simulation setting
\begin{align*}
\text{(MP3)}: n=1000,~~~ Y_t=\begin{cases} 
0.1Y_{t-1}+\epsilon_{t}, & t\in [1,333],\\
0.6Y_{t-1}+\epsilon_{t}, & t\in [334,667],\\
0.1Y_{t-1}+\epsilon_{t}, & t\in [668,1000].
\end{cases}
\end{align*}
Here $\{\epsilon_t\}_{t=1}^n$ is \textit{i.i.d.}\ $N(0,1)$. Thus the change in $\{Y_t\}_{t=1}^n$ is driven by autocorrelation and further affects the variance of the marginal distribution of $Y_t$.

The estimation result is summarized in Table \ref{tab: add_sn_multipara}. As can be seen, SNA gives decent performance as the change originates from autocorrelation. Including variance and quantile in the multi-parameter set improves the estimation accuracy of SNMP, but only by a small amount. On the other hand, ECP does not perform well, possibly due to the strong temporal dependence in the second segment of $\{Y_t\}_{t=1}^n$.

\begin{table}[H]
\centering
\caption{Performance of SNMP and ECP under change in multi-parameter.}
\label{tab: add_sn_multipara}
{\small
\begin{tabular}{c|c|rrrrrrr|r|r|r|r|r}
	\hline
	\hline & & \multicolumn{7}{c}{$\hat{m}-m_o$} & &&& \\
	\hline
	Method & Model & $\leq -3$ & $-2$ & $-1$ & $0$ & $1$ & $2$ & $\geq 3$ & ARI & $d_1 \times 10^2$ & $d_2 \times 10^2$ & $d_H \times 10^2$ & time\\
	\hline
	SNA &  & 0 & 7 & 86 & 870 & 37 & 0 & 0 & 0.873 & 2.72 & 5.29 & 5.80 & 11.32 \\ 
	SNVA & $(MP3)$ & 0 & 2 & 49 & 884 & 64 & 1 & 0 & 0.898 & 2.80 & 3.62 & 4.44 & 26.88 \\ 
	SNVAQ$_{90}$ &  & 0 & 1 & 70 & 841 & 83 & 5 & 0 & 0.881 & 3.27 & 4.49 & 5.50 & 49.87\\
	ECP &  & 0 & 126 & 24 & 245 & 212 & 211 & 182 & 0.512 & 12.22 & 17.61 & 21.59 & 12.28 \\ 
	\hline\hline
\end{tabular}
}
\end{table}

\newpage
\subsection{Implementation details of comparison methods and typical realizations of DGP used in simulation}\label{subsec:dgp}

All implementations of the comparison methods are set to the recommended settings by the corresponding papers or R packages. We believe the default settings of the papers or R packages accommodate the presence of serial correlations as the methods we compare with are \textit{explicitly} designed to handle temporal dependence. The only exceptions are ECP in \cite{Matteson2014} for distributional change and MQS in \cite{Vanegas2020} for quantile change, which can only handle temporal independence, as we cannot find methods for distributional or quantile change that can accommodate temporal dependence.


Note that for all simulation experiments in Section \ref{sec:simulation} of the main text and in Section \ref{sec:addsim} of the supplement, except CUSUM, which is simple enough to be coded by ourselves, all the other competing methods are implemented using source codes obtained from the authors' website or from the corresponding R packages.

For CUSUM, we estimate the long-run variance following the recommendation in \cite{Aue2009}, which uses a Bartlett kernel with bandwidth $\log_{10} n$. 
For BP in \cite{Bai2003} it is implemented via function \texttt{breakpoints()} in the R package \texttt{strucchange}. For MSML in \cite{Cho2012}, it is implemented via source code from Dr.\ Haeran Cho's website. For KF in \cite{Korkas2017}, it is implemented via function \texttt{wbs.lsw()} in the R package \texttt{wbsts}. For GW in \cite{Wied2012} and \cite{Galeano2017}, it is implemented via source code from Dr.\ Dominik Wied's website. For ECP in \cite{Matteson2014}, it is implemented via function \texttt{e.divisive()} in the R package \texttt{ecp}. For MQS in \cite{Vanegas2020}, it is implemented via function \texttt{mqse()} in the R package \texttt{mqs}.

\begin{figure}[H]
\begin{subfigure}{0.3\textwidth}
\includegraphics[angle=270, width=1.2\textwidth]{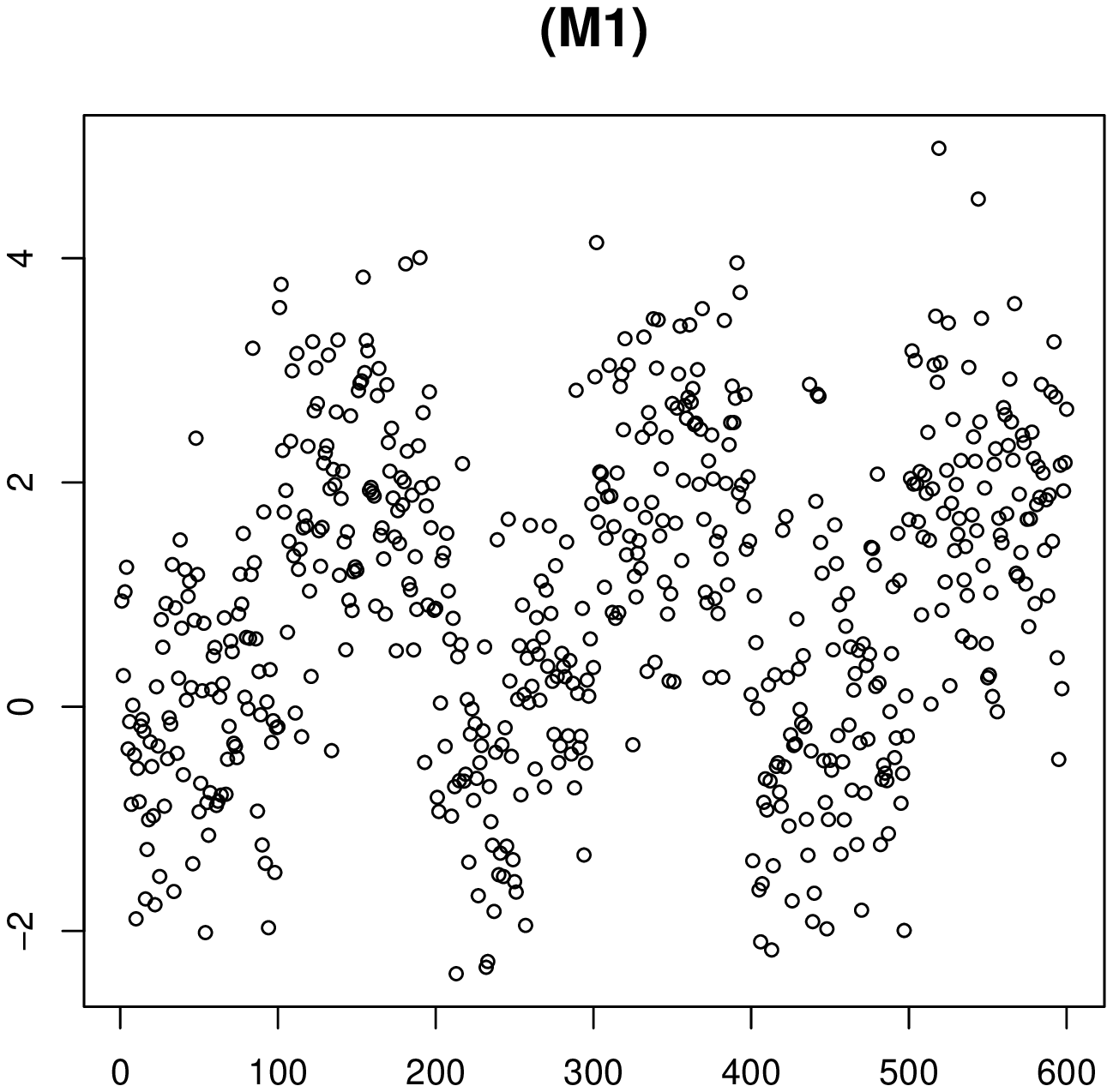}
\vspace{-0.5cm}
\end{subfigure}
~
\begin{subfigure}{0.3\textwidth}
\includegraphics[angle=270, width=1.2\textwidth]{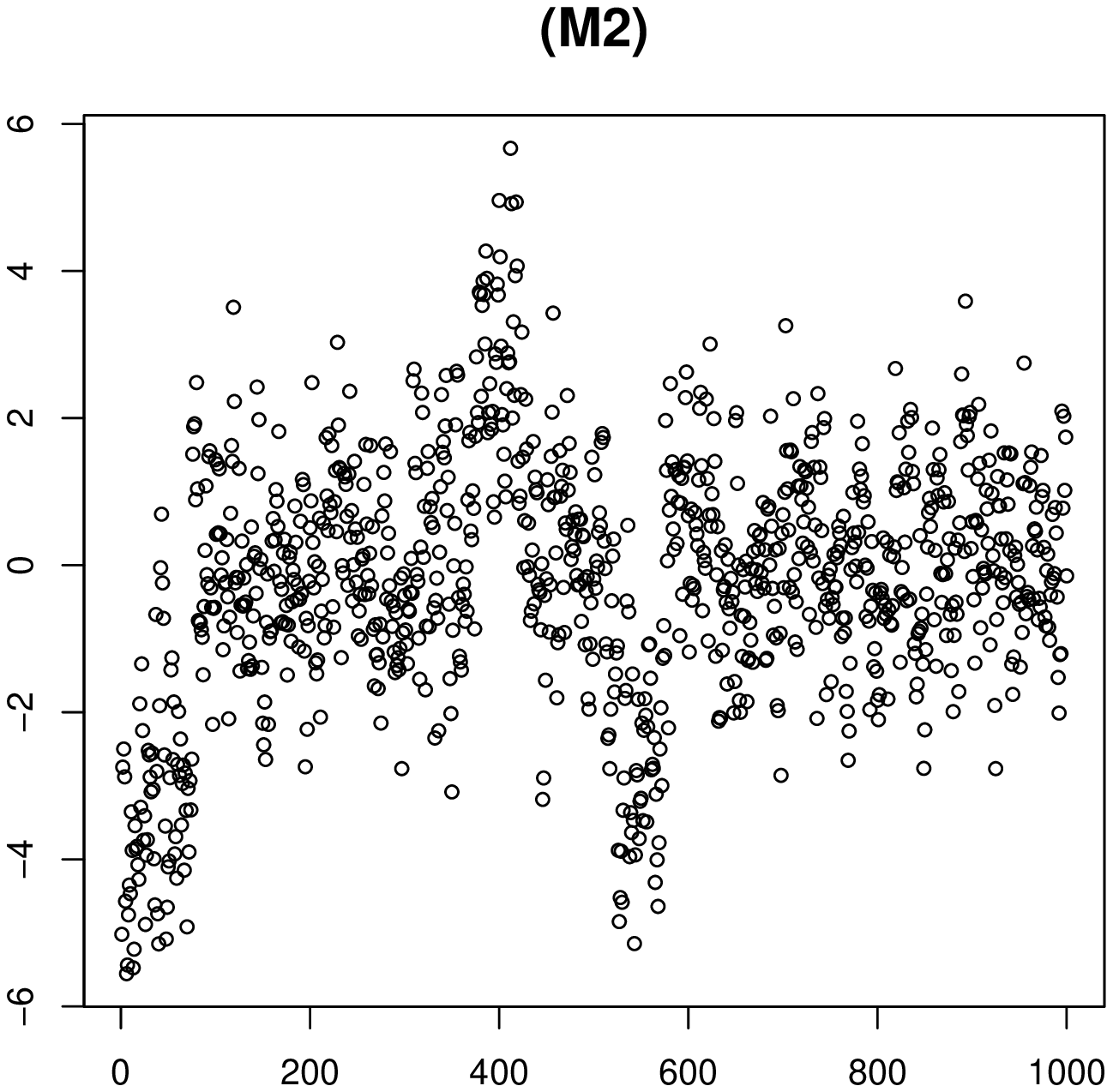}
\vspace{-0.5cm}
\end{subfigure}
~
\begin{subfigure}{0.3\textwidth}
\includegraphics[angle=270, width=1.2\textwidth]{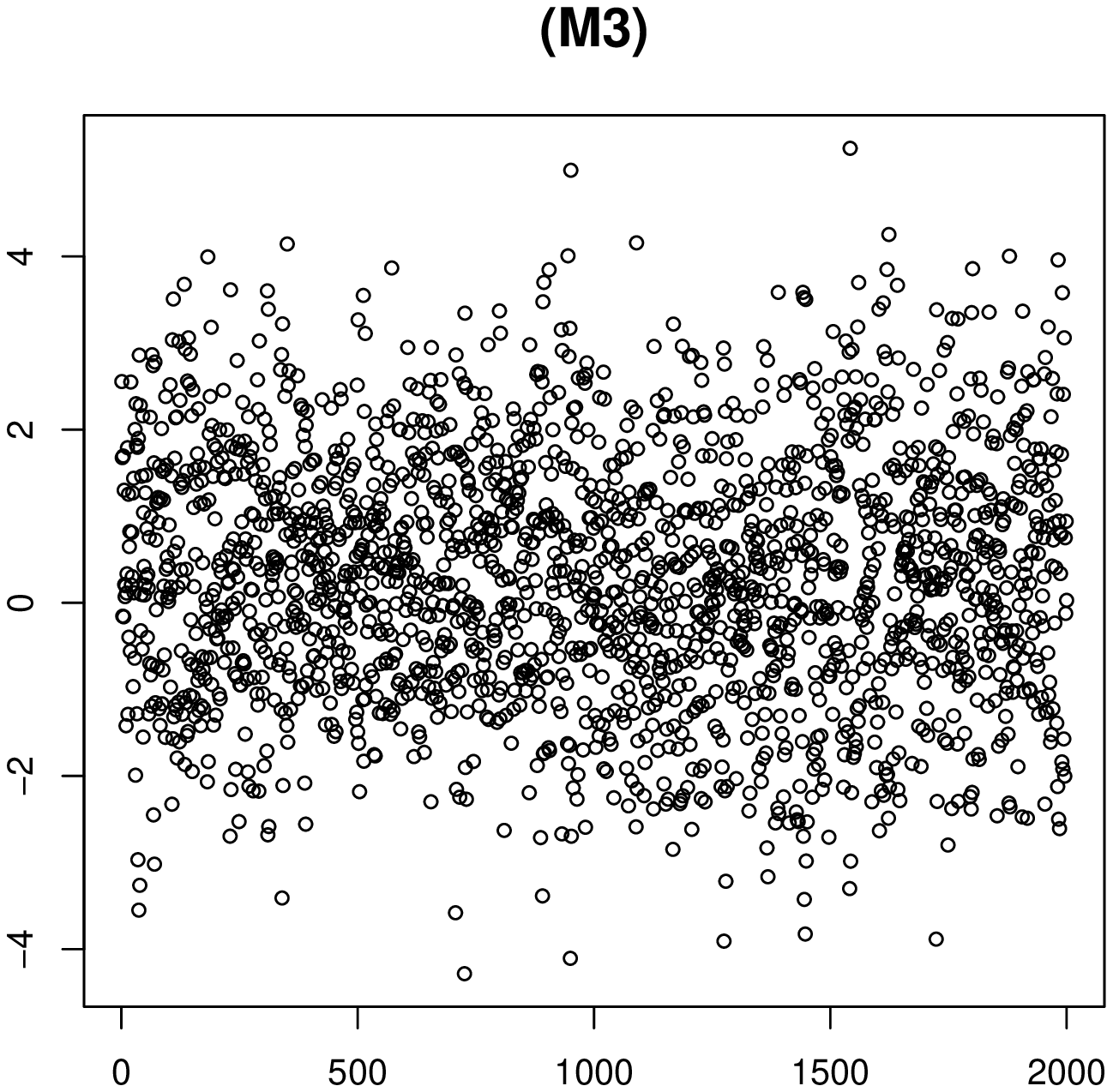}
\vspace{-0.5cm}
\end{subfigure}
~
\begin{subfigure}{0.3\textwidth}
\includegraphics[angle=270, width=1.2\textwidth]{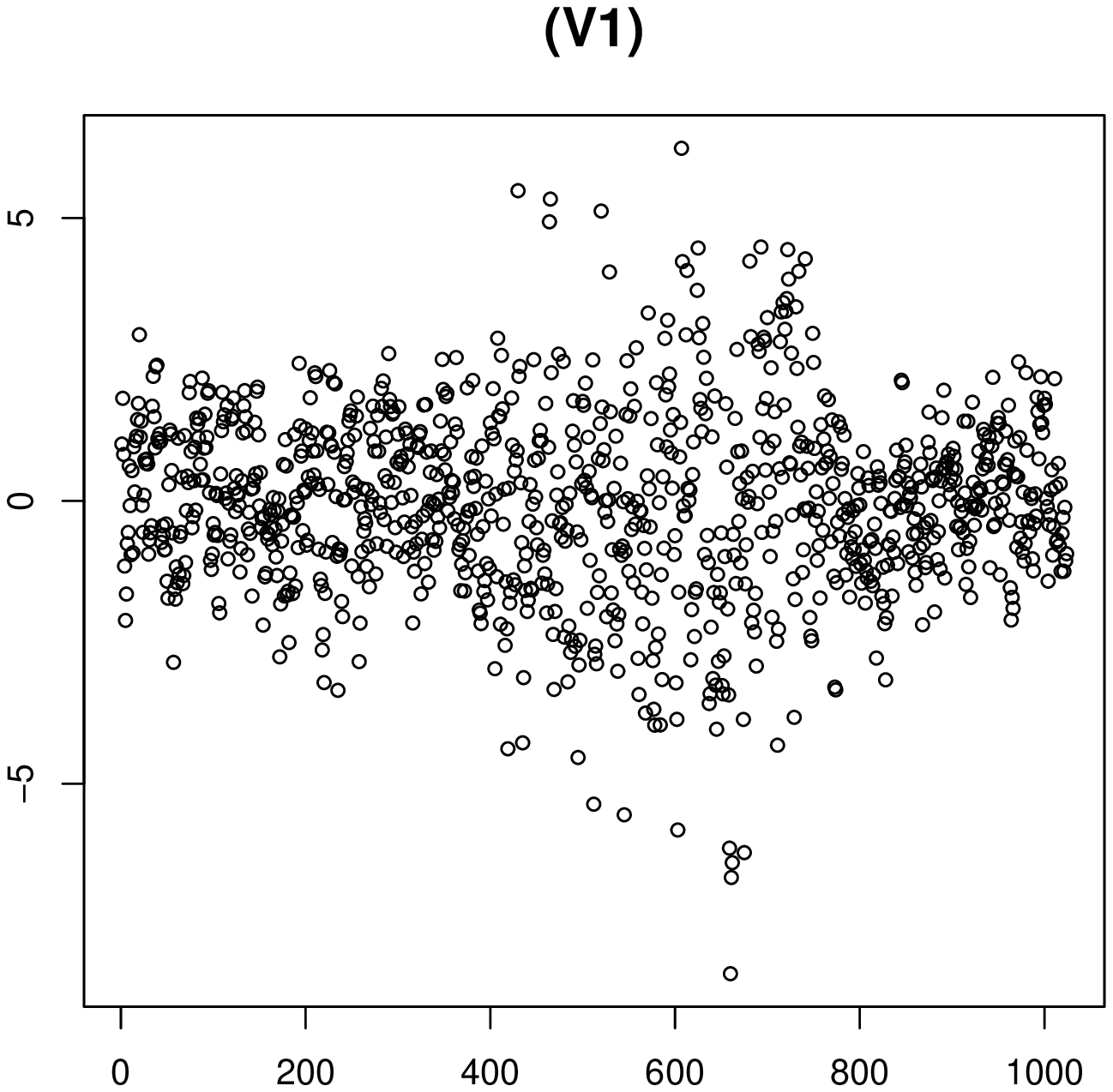}
\vspace{-0.5cm}
\end{subfigure}
~
\begin{subfigure}{0.3\textwidth}
\includegraphics[angle=270, width=1.2\textwidth]{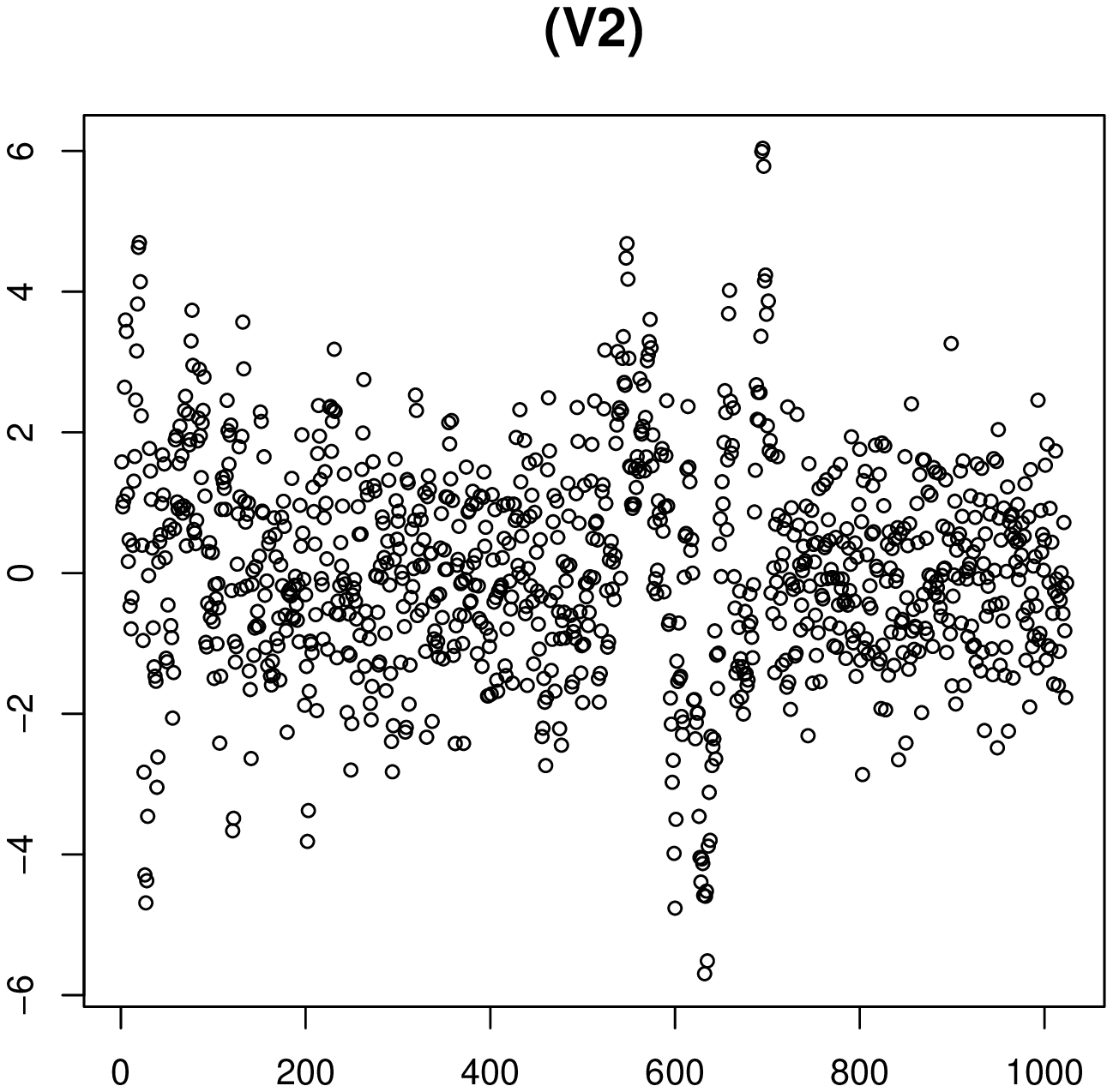}
\vspace{-0.5cm}
\end{subfigure}
~
\begin{subfigure}{0.3\textwidth}
\includegraphics[angle=270, width=1.2\textwidth]{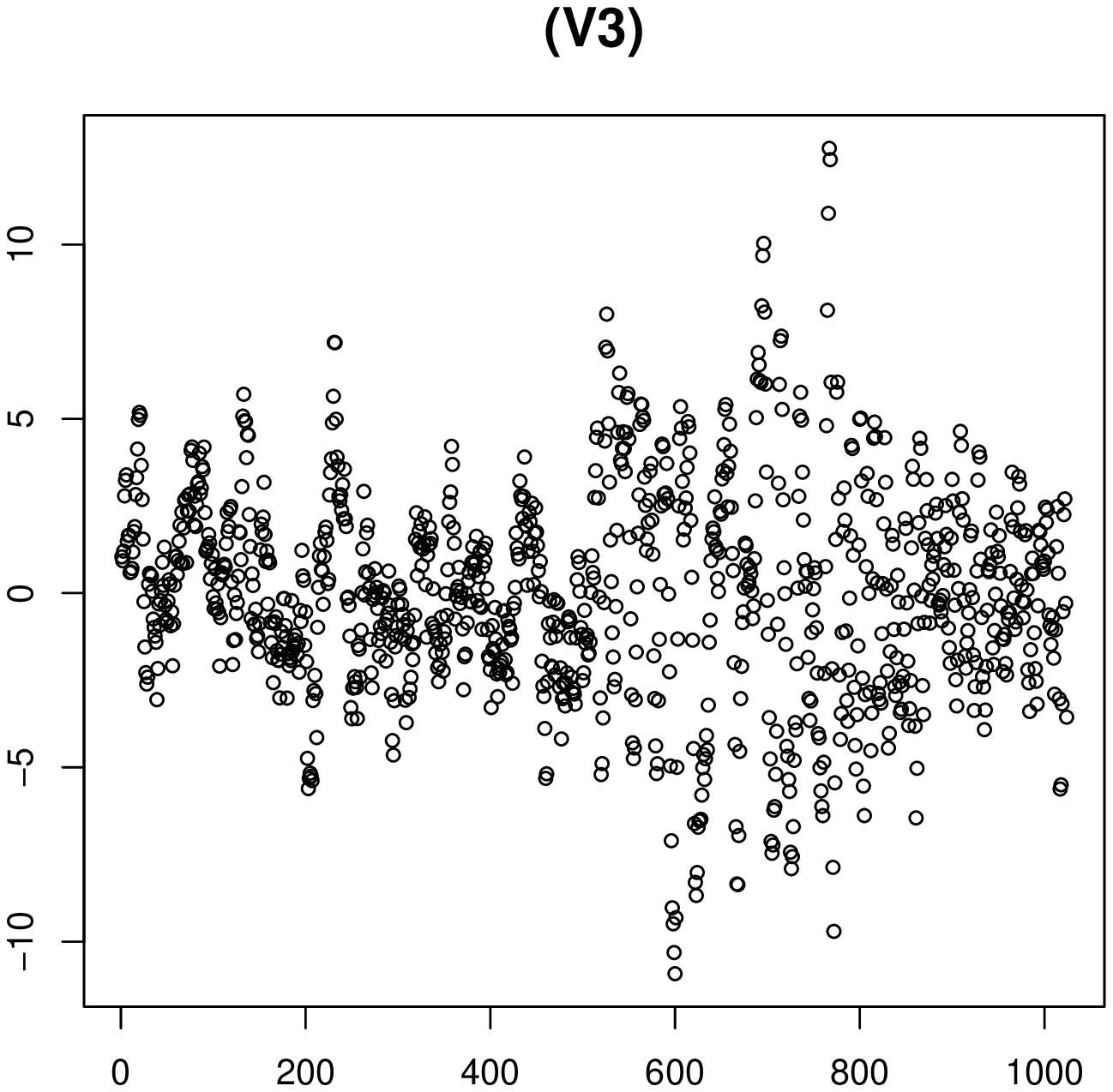}
\vspace{-0.5cm}
\end{subfigure}
~
\begin{subfigure}{0.3\textwidth}
\includegraphics[angle=270, width=1.2\textwidth]{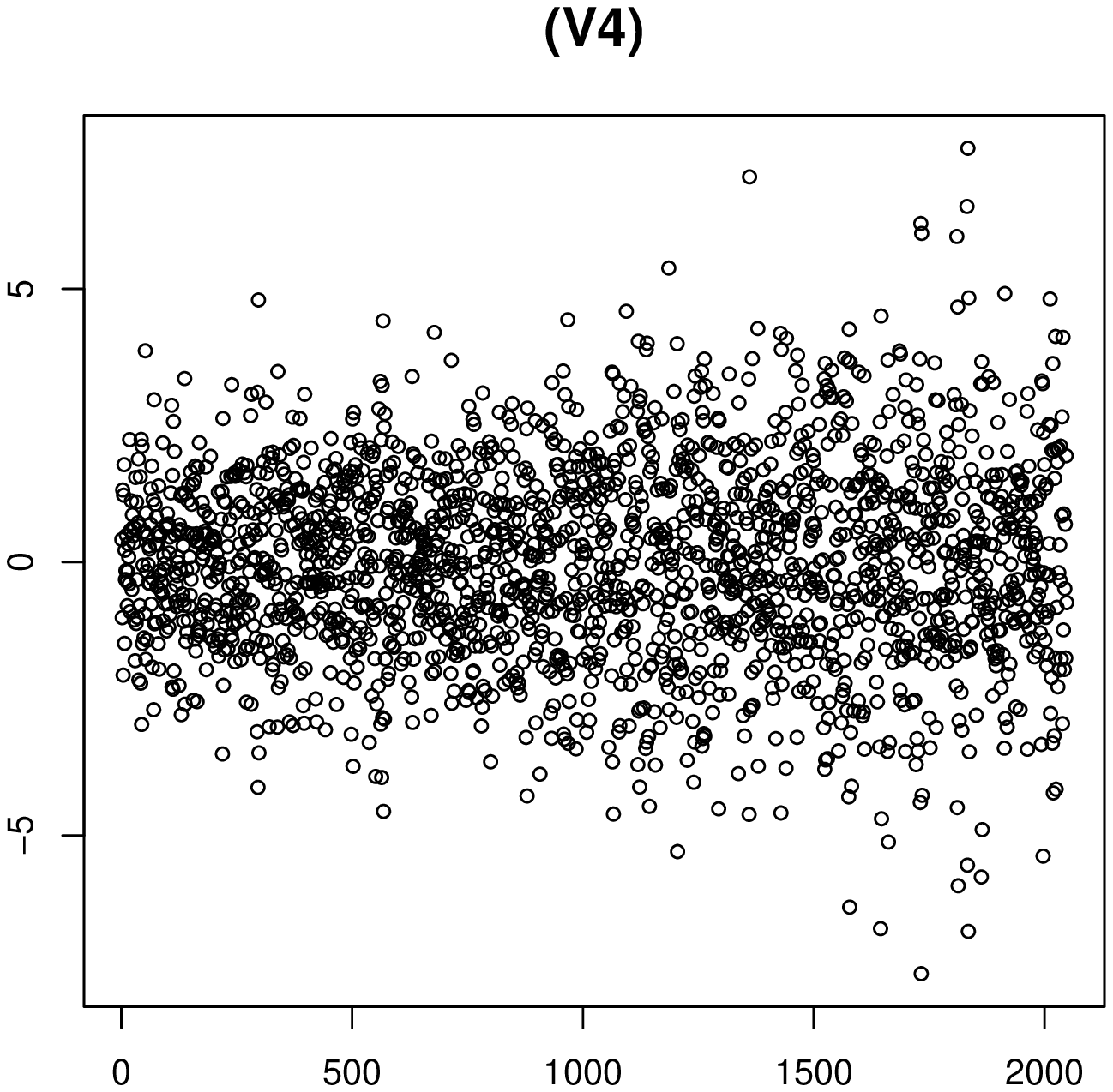}
\vspace{-0.5cm}
\end{subfigure}
~
\begin{subfigure}{0.3\textwidth}
\includegraphics[angle=270, width=1.2\textwidth]{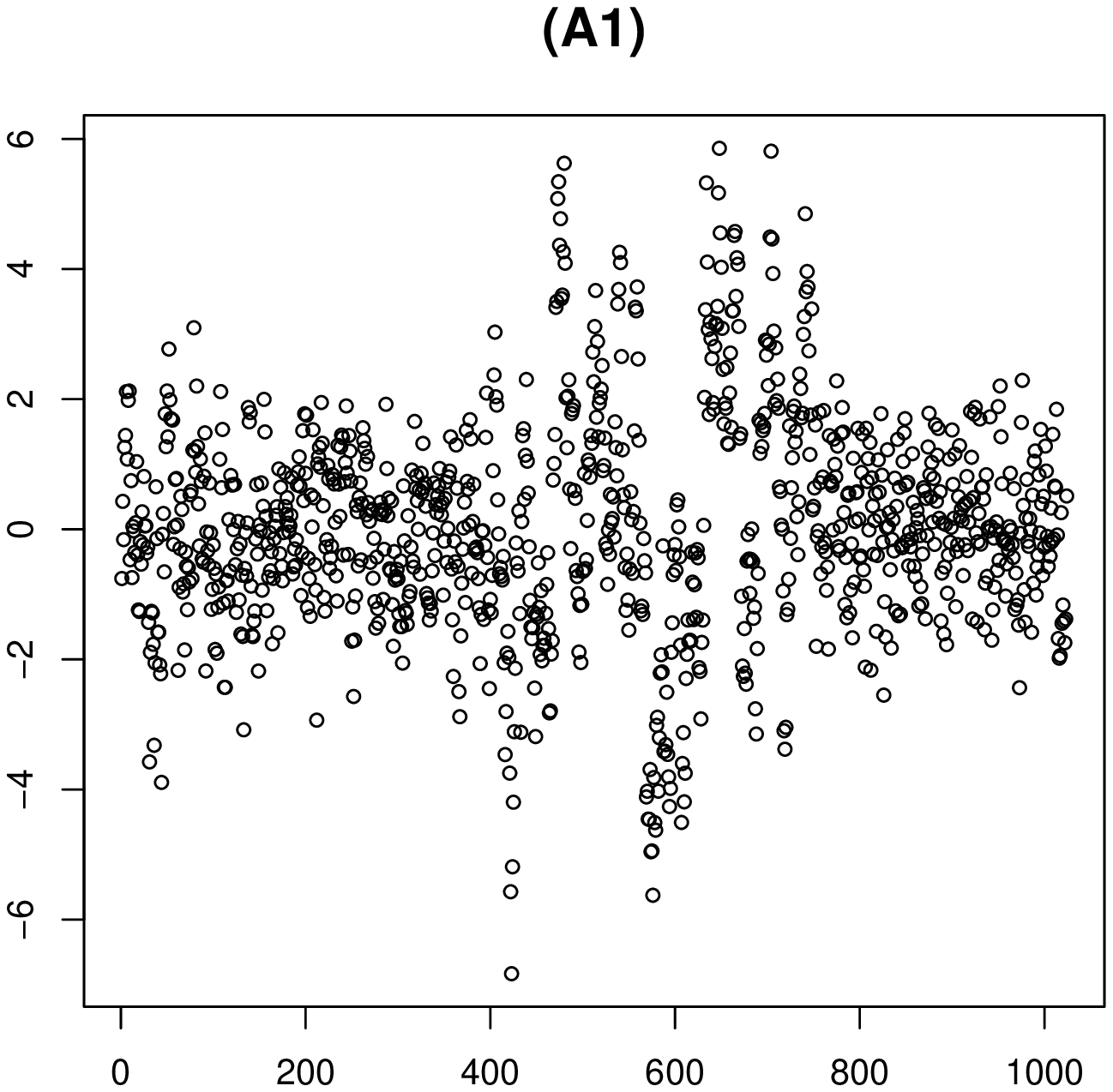}
\vspace{-0.5cm}
\end{subfigure}
~
\begin{subfigure}{0.3\textwidth}
\includegraphics[angle=270, width=1.2\textwidth]{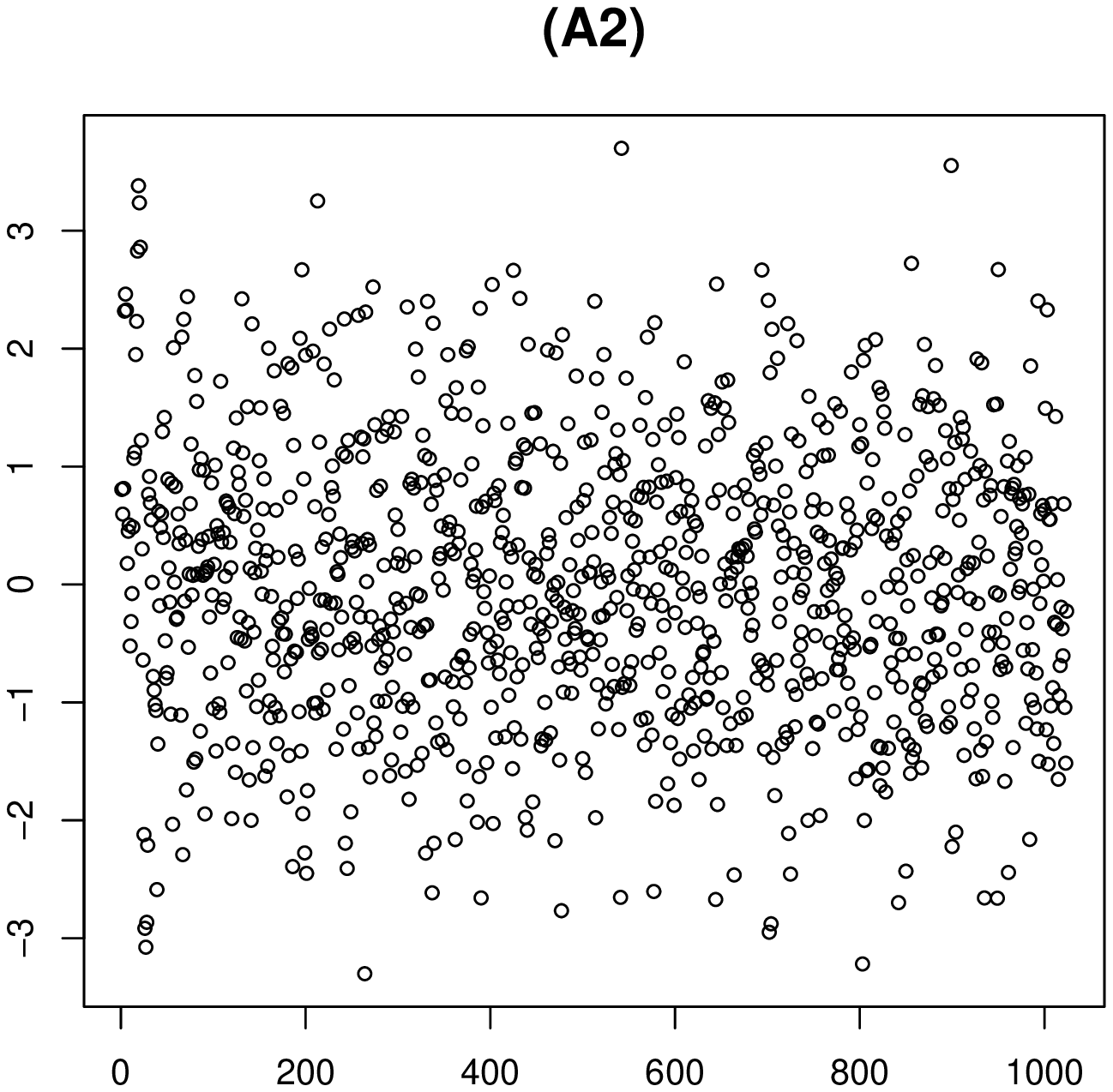}
\vspace{-0.5cm}
\end{subfigure}
~
\begin{subfigure}{0.3\textwidth}
\includegraphics[angle=270, width=1.2\textwidth]{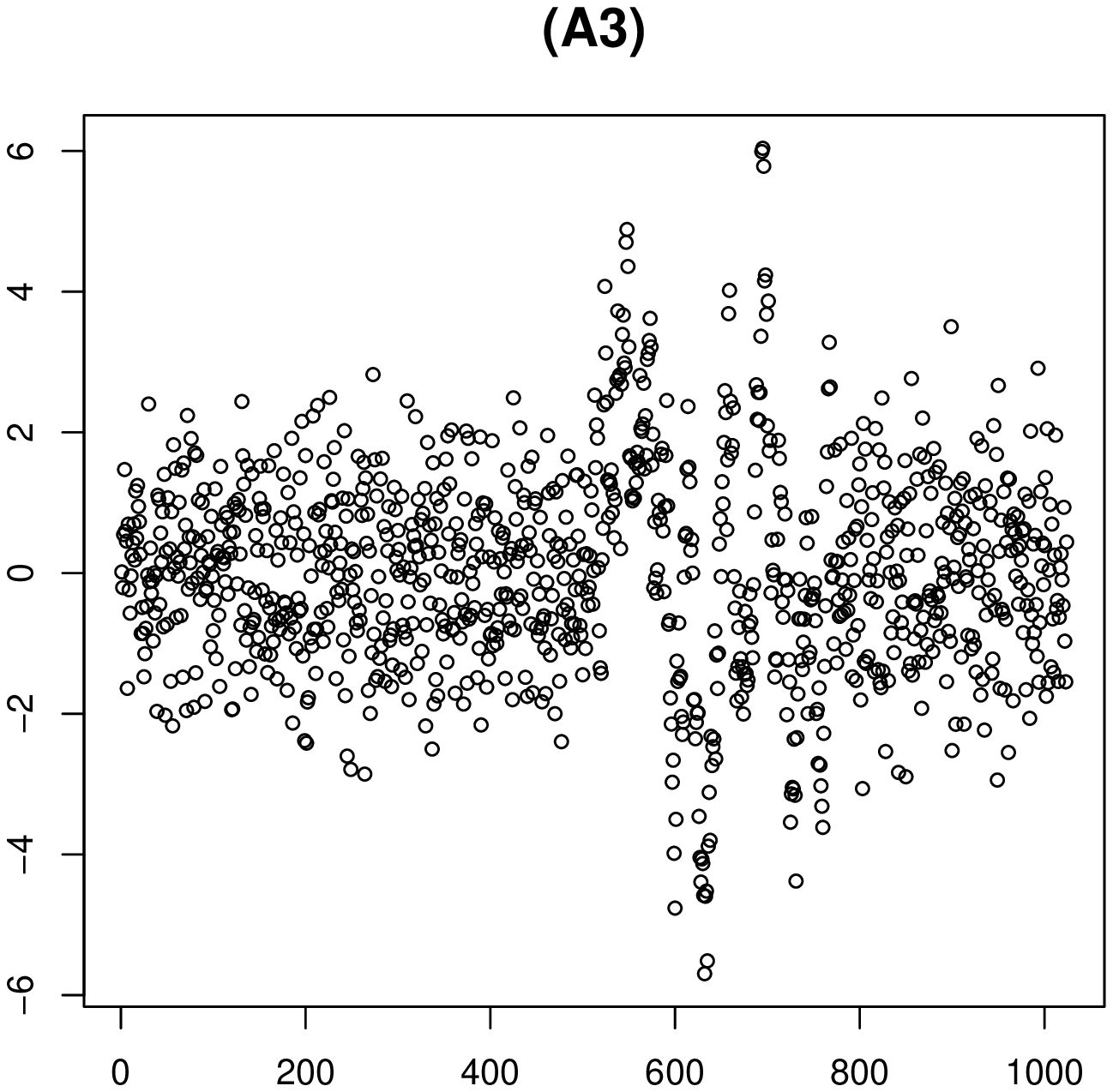}
\vspace{-0.5cm}	
\end{subfigure}
\hspace{0.6cm} 
\begin{subfigure}{0.3\textwidth}
\includegraphics[angle=270, width=1.2\textwidth]{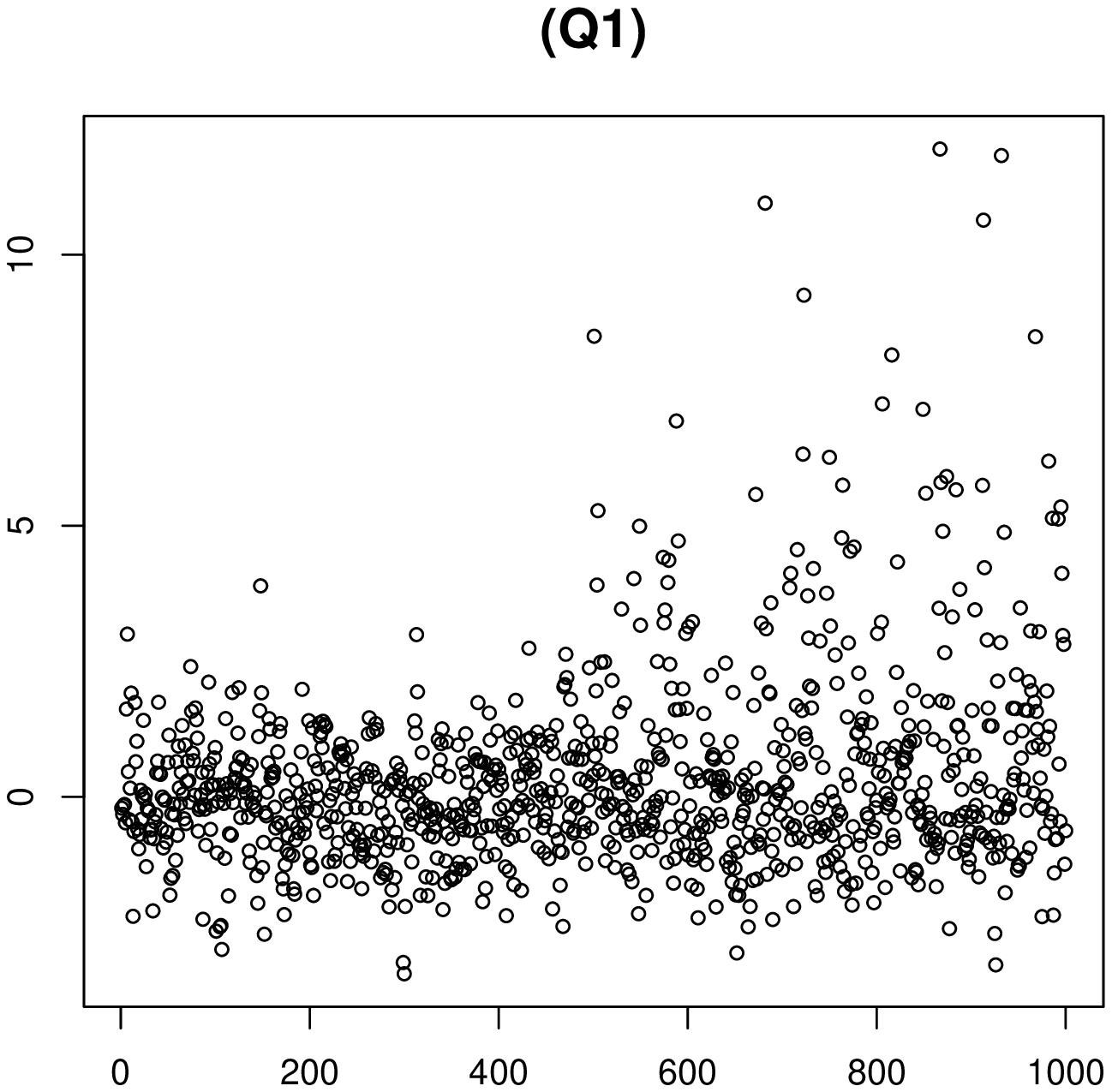}
\vspace{-0.5cm}	
\end{subfigure}
\caption{\it Typical realizations of DGP used in simulation.}
\label{fig: dgp}
\end{figure}
\clearpage

\section{Real Data Applications}\label{sec:applications}
\subsection{Change-point detection in climate data}
In this section, we analyze two climate datasets using SNCP to illustrate the evidence of climate change. The first dataset contains the annual mean temperature of central England from 1772 to 2019, covering $n=247$ years in total. We apply SNM and BP to detect possible mean change in the time series and the estimation result is given in Figure \ref{fig: climate}(a). SNM detects two change-points at 1919 and 1988 and BP gives two changes at 1910 and 1988. In addition, we apply ECP, which detects two change-points at 1926 and 1988. Based on the estimated change-points by SNM, the expected annual mean temperature is 9.15$^{\circ}$C from 1772-1919, 9.52$^{\circ}$C from 1920-1988, and 10.25$^{\circ}$C from 1989-now, which clearly indicates the trend of global warming. 

The second dataset contains the satellite-derived lifetime-maximum wind speeds of $n=2098$ tropical cyclones over the globe during 1981–2006, and we refer to \cite{Elsner2008} for a more detailed description. It is known in climate science that intensity of tropical cyclones increases with rising ocean temperature, one of the major consequences of global warming. As a result, researchers have been seeking empirical evidence on changes in tropical cyclone wind speeds, see \cite{Elsner2008}, \cite{Zhang2011}, \cite{Zhang2018} and references therein. Specifically, the change-point test in \cite{Zhang2018} indicates strong evidence of mean change in the time series of maximum wind speed, however, their test is incapable of locating the exact change-point. We apply SNM and BP to locate the possible change-points in mean. Surprisingly, BP fails to detect any change, while SNM detects one change-point at 1988 with an increase of maximum wind speed~(see Figure \ref{fig: climate}(b)). Interestingly, this aligns with the second estimated change-point in central England temperature, which provides further evidence of climate change. We also apply ECP, which detects one change-point at 1986, providing additional support for the finding based on SNM. The advantage of SNM here is that it not only detects the change but further indicates that the change is due to an increase in the mean level of maximum wind speed, which is more informative than ECP for the climate study.

\begin{figure}[]
\centerline{\includegraphics[angle=270, scale=0.5]{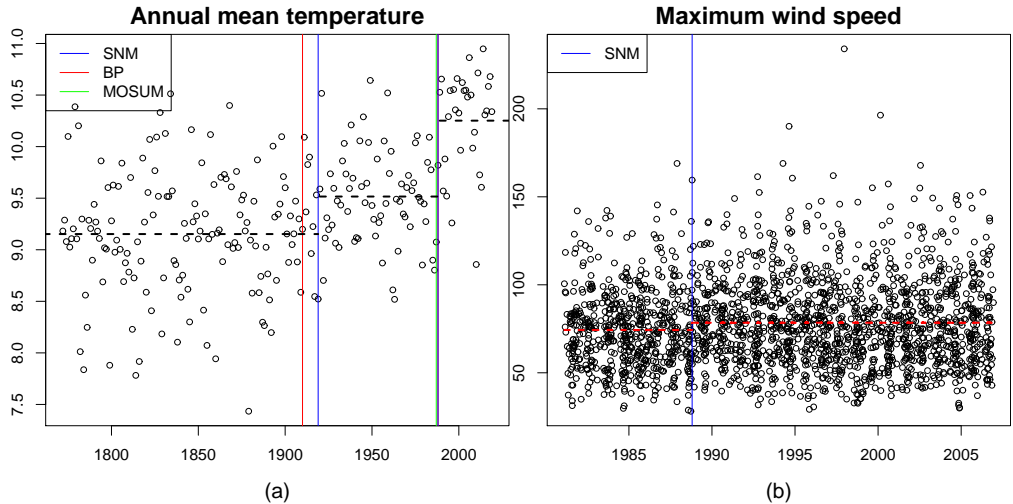}}
\vspace{-0.4cm}
\caption{\it (a) Estimated change-points for central England temperature. (b) Estimated change-points for cyclones maximum wind speed. Horizon dashed line indicates the sample mean of each segment by SNM.}
\label{fig: climate}
\end{figure}

\subsection{Change-point detection in financial data}\label{subsec:finance}
In this section, we study the behavior of financial markets using SNCP. A stylized fact in finance is that volatility of stock markets is much higher during crisis than normal periods. To validate this, we examine the volatility behavior of the S\&P 500 index~(hereafter SP500) for a period around the 2008 financial crisis. The data consist of daily (negative) log-returns of SP500 from June 2006 to December 2010 with $n=1024$ observations~(to facilitate the implementation of MSML). Two widely used measures for (unconditional) volatility in financial markets are variance and Value-at-Risk (which is high quantile such as the 90\% or 95\% quantile)\footnote{ In general, high quantile alone cannot be used to measure the dispersion/volatility of a distribution. However, it is known that the central tendency (i.e.\ mean and median) of the daily (negative) log-returns of stock indexes is stable over time (around 0), and thus high quantile is commonly used in practice to measure the volatility of financial markets.}, and they serve as basis for important applications such as portfolio optimization and systemic risk monitoring. Thus, it is crucial to study the structural breaks in these parameters.

We employ SNCP to detect changes in variance, 90\% and 95\% quantiles of daily (negative) log-returns of SP500, and their multi-parameter combination. Specifically, we implement seven SNCP based estimators: SNV, SNQ$_{90}$, SNQ$_{95}$, SNQ$_{90,95}$, SNQ$_{90}$V, SNQ$_{95}$V and SNQ$_{90,95}$V. Each estimator provides an assessment of volatility change in SP500. For comparison, we further apply MSML and KF, which target variance change and also apply ECP, which targets distributional change. The estimation result is summarized in Table \ref{tab: stocks} and is further visualized in Figure \ref{fig: stocks}. 

For variance change, SNV detects 4 change-points, MSML gives 6 changes and KF detects 5 changes. The estimated change-points by SNV, MSML and KF are close to each other. The three methods detect the inception of the financial crisis around July 2007, the acceleration of the crisis around September 2008, and the end of the crisis around May 2009. MSML and KF further detect changes in a relatively calm period from August 2006 to January 2007, which may be false positives or small-scale changes.

For high quantile change, both SNQ$_{90}$ and SNQ$_{95}$ detect 3 change-points at similar but \textit{different} dates, around the inception, acceleration and end of the financial crisis, providing further evidence for volatility changes in the stock market during crisis. Moreover, the multi-parameter estimation given by SNQ$_{90,95}$ points to 3 changes at similar locations and further nicely reconciles the difference between SNQ$_{90}$ and SNQ$_{95}$, thus suggesting the robustness of our findings. Interestingly, ECP also gives 3 changes around similar time as the three SNQ estimators. However, the result given by SNQ is more informative as it further narrows down the changes to high quantiles.

The multi-parameter SN based on both high-quantile and variance~(i.e.\ SNQ$_{90}$V, SNQ$_{95}$V, SNQ$_{90,95}$V) all detect 4 change-points at similar locations as SNV, which can be seen as evidence that the volatility changes due to variance detected by SNV is substantial and credible.

\begin{table}[h]
\caption{\it Estimated change-points by various SNCP estimators, MSML, KF and ECP for the S\&P 500 index from June 2006 to December 2010.}
\label{tab: stocks}
{\small
\begin{tabu}{cccccccc}
\hline\hline
&   Method  & CP1        & CP2        & CP3        & CP4        & CP5        & CP6        \\\hline
& SNV  &            &            & 07/17/2007 & 09/16/2008 & 12/05/2008 & 05/27/2009 \\
& MSML & 07/28/2006 & 02/23/2007 & 07/18/2007 & 09/02/2008 & 12/01/2008 & 04/20/2009 \\
& KF  & 08/01/2006 & 01/23/2007 & 07/23/2007 & 08/20/2008 &            & 04/20/2009 \\ \cmidrule{2-8}
& ECP  &            &            & 07/20/2007 & 09/17/2008 &  & 04/21/2009 \\
SP500& SNQ$_{90}$ & & & 06/12/2007 & 08/04/2008 & & 05/18/2009 \\
& SNQ$_{95}$ & & & 07/09/2007 & 09/17/2008 & & 04/30/2009 \\
& SNQ$_{90,95}$ & & & 07/09/2007 & 09/17/2008 & & 05/18/2009 \\
\cmidrule{2-8}
& SNQ$_{90}$V  &  &  & 07/17/2007 & 09/16/2008 & 12/05/2008 & 05/27/2009 \\
& SNQ$_{95}$V  &  &  & 07/09/2007 & 09/16/2008 & 12/08/2008 & 04/21/2009 \\
& SNQ$_{90,95}$V  &  &  & 07/09/2007 & 09/16/2008 & 12/05/2008 & 04/20/2009 \\	
\hline\hline
\end{tabu}
}
\end{table}

\begin{figure}[h]
\centerline{\includegraphics[angle=270, scale=0.5]{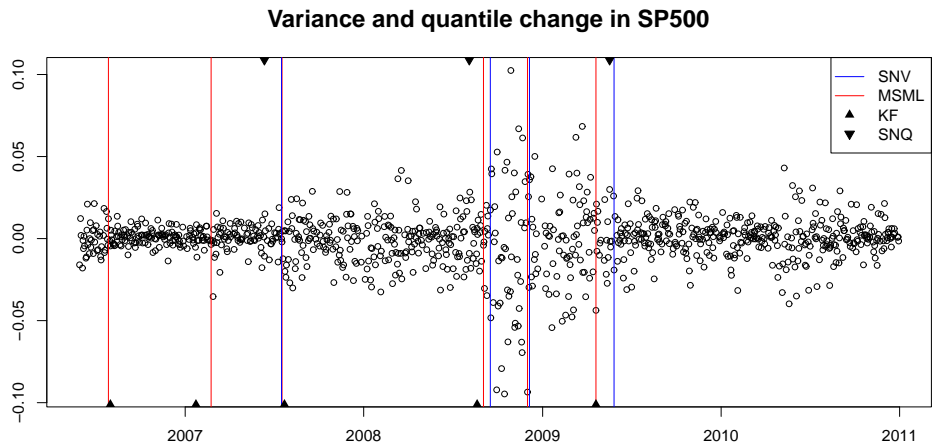}}
\caption{\it Estimated change-points in variance by SNV, MSML, KF and estimated change-points in 90\% quantile by SNQ for the S\&P 500 index from June 2006 to December 2010.}
\label{fig: stocks}
\end{figure}

The above analysis demonstrates the versatility and robustness of SNCP. It can be seamlessly applied to change-point detection for various parameters and provides reliable estimation results. In practice, the ground truth is unknown, thus it is important to examine the behavior change of the data via different angles. In this respect, the versatility of SNCP gives it a unique edge, as SNCP can provide practitioners the freedom to virtually examine any parameter of interest. With additional domain knowledge~(e.g.\ examining volatility of financial markets via both variance and high quantiles), the practitioners can further design multi-parameter based SNCP to improve estimation accuracy, reconcile possibly different estimations given by SNCP for individual parameters, and conduct robustness check of the estimation results. 

Another well-known hypothesis in the financial literature is that the international equity market correlation increases in volatile times and the correlation further increases due to the growing integration of the global economy~\citep{Longin2002,Poon2004}. 

To validate this hypothesis, we examine the correlation between daily (negative) log-returns of the S\&P 500 index~(U.S. market) and the DAX index~(German market) for a 12-year period from January 2000 to December 2012 with $n=2684$ observations. We apply SNC and GW to detect potential changes in correlation. The estimation result is visualized in Figure \ref{fig: stocks_correlation}. SNC detects three change-points at 11/06/2003, 01/10/2006 and 10/20/2008. In contrast, GW detects one change-point at 10/21/2008.

The three change-points detected by SNC partition the 12-year period into 4 segments with Jan. 2000-Nov. 2003~($\hat{\rho}=0.61$), Dec. 2003-Jan. 2006~($\hat{\rho}=0.40$), Feb. 2006-Oct. 2008~($\hat{\rho}=0.51$), and Nov. 2008-Dec. 2012~($\hat{\rho}=0.70$). SP500 and DAX exhibit strong correlation during the first segment, which contains the period of the early 2000 recession and Dot-com bubble from 2000-2002. The correlation decreases to 0.40 after the crisis but starts to build up following the inception of the 2008 financial crisis and remains at a high level 0.70 during the post-crisis period. This result provides empirical evidence for the hypothesis in \cite{Longin2002} and indicates that the systemic risk in the global financial market is increasing during the past decades.

We remark that since the ground truth is unknown, the estimation results given by SNC need to be interpreted and treated with caution. Here, we further provide some discussion about the discrepancy between the estimated change-points by SNC and GW. The magnitude of change in sample correlation at the first change-point~($\hat{\rho}=0.61$ to $0.40$) and the third change-point~($\hat{\rho}=0.51$ to $0.70$) estimated by SNC are rather significant, providing some  evidence for the validity of the detected change-points; whereas the magnitude of change is smaller at the second change-point~($\hat{\rho}=0.40$ to $0.51$), suggesting it could be due to false positives as the threshold $K_n$ of SNC is set at the 90\% critical level.

Note that around the first change-point~(11/06/2003) estimated by SNC, the bivariate time series also experienced a significant change in its variance. This matches the simulation setting of (R1) in Section \ref{subsec:powercorrelation} of the main text, where the simulation result suggests that GW tends to experience power loss for detecting correlation changes when there is large concurrent variance changes in the data. Indeed, the asymptotic validity of GW requires an approximately constant variance, see Assumptions (A4) and (A5) in \cite{Wied2012}.

\begin{figure}[H]
\centerline{\includegraphics[angle=270, scale=0.6]{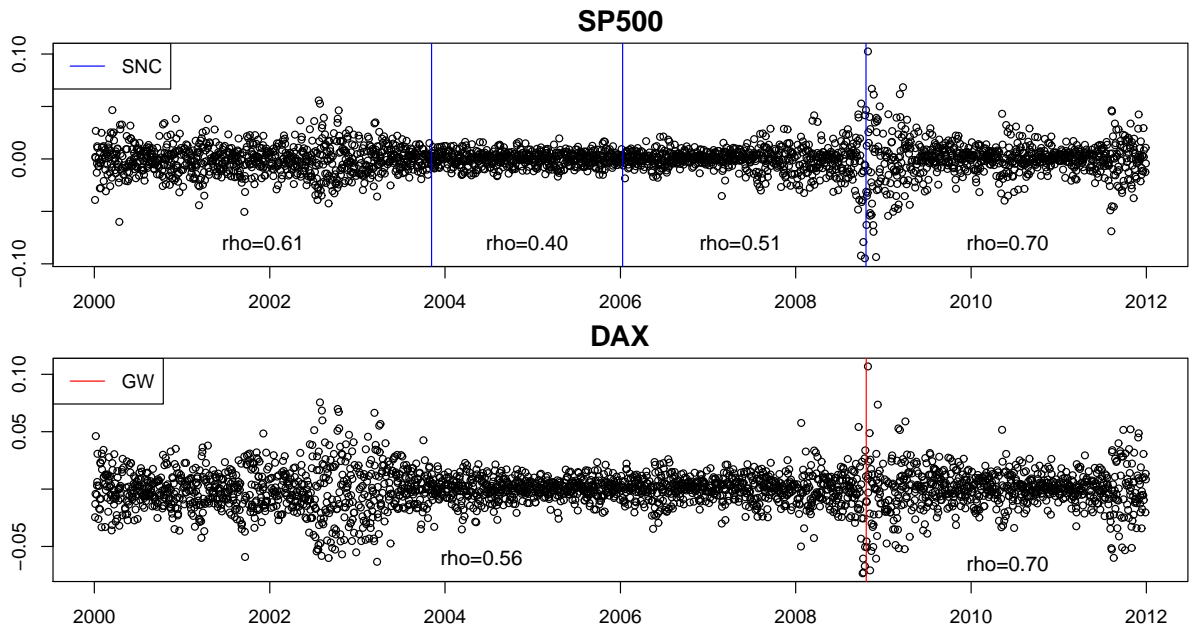}}
\caption{\it Estimated change-points in correlation between the S\&P 500 index and the DAX index by SNC (upper panel) and GW (lower panel) from January 2000 to December 2012.}
\label{fig: stocks_correlation}
\end{figure}

\clearpage
\section{{Verification of Assumptions in Smooth Function Model}}\label{sec:verify}
In this section, we derive explicit formulas for the partial influence function based expansion \eqref{inf3} and \eqref{theta} in the main text, and provide detailed  verification of Assumptions \ref{ass_influence}-\ref{ass_theta} and \ref{ass_no1}-\ref{ass_no3} for  $\theta(\cdot)$ in the smooth function model, which includes a wide class of parameters such as mean, variance, (auto)-covariance and (auto)-correlation. 

\subsection{Verification of assumptions in \textbf{Mean}}\label{subsec:verify_assumption_mean}
By simple calculation, for any $\omega_{a,b}$, we have $\xi_{i}(Y_t^{(i)},\omega_{a,b})=Y_t^{(i)}-E(Y_t^{(i)})$, $i=1,2$ and $r_{a,b}(\omega_{a,b})\equiv 0$. Thus, Assumption \ref{ass_influence} holds if the invariance principle holds jointly for the two stationary segments. Assumptions \ref{ass_remainder} and \ref{ass_theta} hold trivially.  
Verification of Assumptions \ref{ass_no1}-\ref{ass_no3} is quite similar.  

\subsection{Verification of assumptions in \textbf{Smooth Function Model}}
The smooth function model is a broad framework that covers important functionals such as mean, variance, (auto)-covariance and (auto)-correlation~\citep{bhattacharya1978validity,hall2013bootstrap}. Roughly speaking, $\theta$ can be viewed as a smooth function model of a stationary process $\{Y_t\in\mathbb{R}^p\}$ if there exists a smooth function $H:\mathbb{R}^d\to\mathbb{R}$ and $d$ measurable functions $z_i(\cdot):\mathbb{R}^p\to\mathbb{R}, i=1,\cdots,d$ such that $\theta=H(\mu_z),$ where $\mu_z=(\mu_{z,1},\cdots,\mu_{z,d})^{\top}$ with $\mu_{z,i}=E\{z_i(Y_t)\}, i=1,\cdots,d.$ Denote $Z_t=Z(Y_t)=(z_1(Y_t),\cdots,z_d(Y_t))^\top$, we have $E(Z_t)=\mu_z.$

\subsubsection{Verification of Assumptions \ref{ass_influence}, \ref{ass_remainder}, \ref{ass_no1} and \ref{ass_no2}}

We detail the verification of Assumptions \ref{ass_no1} and \ref{ass_no2} (multiple change-point setting), which include Assumptions \ref{ass_influence} and \ref{ass_remainder}~(single change-point setting) as special cases.

Given the subsample $\{Y_t\}_{t=a}^{b}$, the estimator $\widehat{\theta}_{a,b}$ is given by $\widehat{\theta}_{a,b}=H(\widehat{\mu}_{a,b})$, where $\widehat{\mu}_{a,b}=\frac{1}{b-a+1}\sum_{t=a}^{b}Z_t$ is the unbiased estimator of $\mu_{a,b}=E(\widehat{\mu}_{a,b})$. By the smoothness of $H(\cdot)$ and a Taylor expansion, we have
\begin{flalign}\label{smooth}
\begin{split}
\widehat{\theta}_{a,b}=&H({\mu}_{a,b})+\frac{\partial H(\mu_{a,b})^\top}{\partial \mu}\frac{1}{b-a+1}\sum_{t=a}^{b}(Z_t-\mu_{a,b})+\frac{1}{2}(\widehat{\mu}_{a,b}-\mu_{a,b})^\top\frac{\partial^2 H(\tilde{\mu})}{\partial\mu\partial\mu^{\top}}(\widehat{\mu}_{a,b}-\mu_{a,b}),
\end{split}
\end{flalign}
where $\tilde{\mu}=u\widehat{\mu}_{a,b}+(1-u)\mu_{a,b}$ for some $u\in[0,1]$. Denote $\mu_z^{(i)}=E(Z_t^{(i)})=E\{Z(Y_t^{(i)})\}$, $i=1,\cdots, m_o+1$. Both expansions \eqref{theta_null} and \eqref{theta} in the main text can be naturally derived from \eqref{smooth}.

For \eqref{theta_null}, $(a,b)$ contains no change-points, then   we have $\xi_i(Y_t^{(i)})=\frac{\partial H(\mu_z^{(i)})^\top}{\partial \mu}(Z_t^{(i)}-\mu_z^{(i)}), i=1,\cdots,m_o+1$. For \eqref{theta}, let $k_{i-1}+1\leq a\leq k_i <\cdots<k_{j}+1\leq b\leq k_{j+1}$, where $k_i,\cdots,k_j$ are change-points between $a$ and $b$, 
we have $\xi_r(Y_t^{(r)},\omega_{a,b})= \frac{\partial H(\mu_{a,b})}{\partial \mu}^{\top}(Z_t^{(r)}-\mu_z^{(r)}), r=i,\cdots,j+1$ and $r_{a,b}(\omega_{a,b})=\frac{1}{2}(\widehat{\mu}_{a,b}-\mu_{a,b})^\top\frac{\partial^2 H(\tilde{\mu})}{\partial\mu\partial\mu^{\top}}(\widehat{\mu}_{a,b}-\mu_{a,b})$.

Thus, a sufficient condition for Assumption \ref{ass_no1}(i) to hold is that 
$$\frac{1}{\sqrt{n}}\sum_{t=1}^{\lfloor nr\rfloor}\left(Z_t^{(1)}-\mu_z^{(1)},\cdots,Z_t^{(m_o+1)}-\mu_z^{(m_o+1)}\right)\Rightarrow  (\Sigma_1^{1/2} \mathcal{B}_d^{(1)}(r),\cdots,\Sigma_{m_o+1}^{1/2}\mathcal{B}_d^{(m_o+1)}(r)), $$
where $\mathcal{B}_d^{(i)}(\cdot), i=1,\cdots,m_o+1$ are $m_o+1$ $d$-dimensional standard Brownian motions and $\Sigma_i, i=1,\cdots,m_o+1$ are $m_o+1$ positive definite matrices. This is a mild assumption and holds under suitable moment conditions of $Z_t^{(i)}=Z(Y_t^{(i)}),i=1,\cdots,m_o+1$ and mixing conditions of $\{Y_t^{(1)},\cdots,Y_t^{(m_o+1)}\}$~(see \cite{phillips1987time}).

Assumption \ref{ass_no1}(ii) can be easily verified based on the functional CLT assumed above for Assumption \ref{ass_no1}(i) if we further have the mild condition that $\sup_{1\leq a<b\leq n}\left\|\frac{\partial H(\mu_{a,b})}{\partial \mu}\right\|_2 <C$ for some $C>0$. In particular, for variance and (auto)-covariance functional, this condition holds when both variance and absolute mean are upper bounded; for correlation and auto-correlation functional, this condition holds if we further have that the variance is lower bounded.

Verification of Assumption \ref{ass_no2} for the remainder term is more involved. A sufficient condition is $\sup_{1\leq a<b\leq n}\Big\|\frac{\partial H^2(\mu_{a,b})}{\partial \mu \partial \mu^\top}\Big\|<C$ for some $C>0$ and $\sup_{1\leq a<b\leq n}\sqrt{b-a+1}\|\widehat{\mu}_{a,b}-\mu_{a,b}\|_2=o_p(n^{1/4})$, which in general can be verified via results in \cite{Shao1995} and \cite{Wu2011}. See \cite{dette2020a} for verification of such condition when $\theta(\cdot)$ is variance. The same technical arguments developed there can be applied to verify other smooth function models such as (auto)-covariance and (auto)-correlation.

By setting $m_o=1$ in the above arguments, we can automatically verify Assumptions \ref{ass_influence} and \ref{ass_remainder}.

\subsubsection{Verification of Assumption \ref{ass_theta}}\label{rem_smooth}
Verification of Assumption \ref{ass_theta} requires a case-by-case analysis. We provide sufficient conditions for Assumption \ref{ass_theta} for common functionals such as variance, (auto)-covariance and (auto)-correlation under mild conditions. In the following, denote $\delta=\theta_1-\theta_2.$ Given a mixture weight $\omega=(\omega^{(1)},\omega^{(2)})^{\top}$ with $\omega^{(i)}\in[0,1]$, $i=1,2$ and $\omega^{(1)}+\omega^{(2)}=1$, we define the mixture distribution of $F^{(1)}$ and $F^{(2)}$ as $F^{\omega}=\omega^{(1)}F^{(1)}+\omega^{(2)}F^{(2)}$.

\noindent\textbf{Example 1~(Variance change)}: In this case, $\theta(\cdot)$ is the variance functional.	Let $Y^{(1)}\sim F^{(1)}$ such that $EY^{(1)}=\mu_1$ and $\mathrm{Var}(Y^{(1)})=\sigma_1^2=\theta_1$, and $Y^{(2)}\sim F^{(2)}$ such that $EY^{(2)}=\mu_{2}$ and $\mathrm{Var}(Y^{(2)})=\sigma_2^2=\theta_2$.  Let $Y\sim F^{\omega}$, we have 	$\theta(\omega)=\omega^{(1)}\sigma_1^2+\omega^{(2)}\sigma_2^2+\omega^{(1)}\omega^{(2)}(\mu_1-\mu_2)^2$. Hence 	$\theta(\omega)-\theta_1=\omega^{(2)}(\theta_2-\theta_1)+\omega^{(1)}\omega^{(2)}(\mu_1-\mu_2)^2$, and $\theta(\omega)-\theta_2=\omega^{(1)}(\theta_1-\theta_2)+\omega^{(1)}\omega^{(2)}(\mu_1-\mu_2)^2$. Simple calculation shows that a sufficient condition for Assumption \ref{ass_theta} is $(\mu_1-\mu_2)^2<|\theta_1-\theta_2|=|\delta|$, in which case we can set $C_1=1-|\delta|^{-1}(\mu_1-\mu_2)^2$ and $C_2=1+|\delta|^{-1}(\mu_1-\mu_2)^2$.

\bigskip

In Examples 2 and 3, we further consider covariance and correlation functional for bivariate time series. In the following, let $\mathbf{Y}^{(1)}=(Y_1^{(1)},Y_2^{(1)})^{\top}\sim F^{(1)}$ such that $E\mathbf{Y}^{(1)}=\pmb{\mu}^{(1)}=(\mu_{1}^{(1)},\mu_{2}^{(1)})$ and $\mathrm{Cov}(Y_1^{(1)},Y_2^{(1)})=\gamma_1$, and $\mathbf{Y}^{(2)}=(Y_1^{(2)},Y_2^{(2)})^{\top}\sim F^{(2)}$ such that $E\mathbf{Y}^{(2)}=\pmb{\mu}^{(2)}=(\mu_{1}^{(2)},\mu_{2}^{(2)})$ and $\mathrm{Cov}(Y_1^{(2)},Y_2^{(2)})=\gamma_2$.  Furthermore, denote $\chi=\mu_{1}^{(1)}\mu_{2}^{(1)}+\mu_{1}^{(2)}\mu_{2}^{(2)}-\mu_{1}^{(2)}\mu_{2}^{(1)}-\mu_{1}^{(1)}\mu_{2}^{(2)}=(\mu_{1}^{(1)}-\mu_{1}^{(2)})(\mu_{2}^{(1)}-\mu_{2}^{(2)})$, which measures the effect of mean change.

\noindent\textbf{Example 2~(Covariance change)}: In this case, $\theta(\cdot)$ is the covariance functional. We have that 
\begin{flalign*}
\theta(\omega)=&\omega^{(1)}EY_{t1}^{(1)}Y_{t2}^{(1)}+\omega^{(2)}EY_{t1}^{(2)}Y_{t2}^{(2)}-[\omega^{(1)}EY_{t1}^{(1)}+\omega^{(2)}EY_{t1}^{(2)}][\omega^{(1)}EY_{t2}^{(1)}+\omega^{(2)}EY_{t2}^{(2)}]
\\=&\omega^{(1)}\theta_1+\omega^{(2)}\theta_2+\omega^{(1)}\omega^{(2)}[\mu_{1}^{(1)}\mu_{2}^{(1)}+\mu_{1}^{(2)}\mu_{2}^{(2)}-\mu_{1}^{(2)}\mu_{2}^{(1)}-\mu_{1}^{(1)}\mu_{2}^{(2)}]
\\=&\omega^{(1)}\theta_1+\omega^{(2)}\theta_2+\omega^{(1)}\omega^{(2)}\chi.
\end{flalign*}
Therefore, we have
\begin{flalign*}
\theta(\omega)-\theta_1=&\omega^{(2)}(\theta_2-\theta_1)+\omega^{(1)}\omega^{(2)}\chi,\\
\theta(\omega)-\theta_2=&\omega^{(1)}(\theta_1-\theta_2)+\omega^{(1)}\omega^{(2)}\chi.
\end{flalign*}
Simple calculation shows that a sufficient condition for Assumption \ref{ass_theta} is $|\chi|<|\theta_1-\theta_2|=|\delta|$, in which case we can set $C_1=1-|\delta|^{-1}|\chi|$ and $C_2=1+|\delta|^{-1}|\chi|$.

\noindent\textbf{Example 3~(Correlation change)}: In this case, $\theta(\cdot)$ is the correlation functional. We consider the following two scenarios:
\vspace{-0.4cm}
\begin{enumerate}[label={[\Alph*]}]
\item (Changing mean with constant variance) For notational simplicity, we assume the bivariate time series share the same variance such that $\mathrm{Var}(Y_1^{(i)})=\mathrm{Var}(Y_2^{(i)})=\sigma^2, i=1,2.$ The conditions under unequal variance can be derived using the same but more algebraically involved arguments. It can be shown that
\begin{flalign*}
&\theta(\omega)=\frac{\omega^{(1)}\gamma_1+\omega^{(2)}\gamma_2+\omega^{(1)}\omega^{(2)}\chi}{{\sigma^2+\omega^{(1)}\omega^{(2)}\chi}}\\=&
\left(1+\frac{\omega^{(1)}\omega^{(2)}\chi}{\sigma^2}\right)^{-1}\left[\omega^{(1)}\theta_1+\omega^{(2)}\theta_2+\frac{\omega^{(1)}\omega^{(2)}\chi}{\sigma^2}\right].
\end{flalign*}
Define $M=\chi/\sigma^2$, we have 
\begin{flalign*}
\theta(\omega)-\theta_1=&\frac{\omega^{(2)}(\theta_2-\theta_1)+\omega^{(1)}\omega^{(2)}M(1-\theta_1)}{1+\omega^{(1)}\omega^{(2)}M},\\\theta(\omega)-\theta_2=&\frac{\omega^{(1)}(\theta_1-\theta_2)+\omega^{(1)}\omega^{(2)}M(1-\theta_2)}{1+\omega^{(1)}\omega^{(2)}M},
\end{flalign*}
and Assumption \ref{ass_theta} is equivalent to
$$
C_1\leq \Big|\frac{1-\omega^{(1)}\delta^{-1}M(1-\theta_1)}{1+\omega^{(1)}\omega^{(2)}M}\Big|\leq C_2,\quad\text{and}\quad C_1\leq \Big|\frac{1+\omega^{(1)}\delta^{-1}M(1-\theta_2)}{1+\omega^{(1)}\omega^{(2)}M}\Big|\leq C_2.
$$
For $M>0$, a sufficient condition is $0<2M<|\delta|$, in which case we can set $C_1=\frac{1-2M/|\delta|}{1+M/4}$ and $C_2=1+2M/|\delta|$. 

For $M\leq 0$, a sufficient condition is $\max\{-8,-|\delta|\}<2M\leq 0$, in which case we can set $C_1=1+2M/|\delta|$ and $C_2=\frac{1-2M/|\delta|}{1+M/4}$.

\item (Changing variance with constant mean) Note that in this case, we have $\chi=0$. Same as in [A], for notational simplicity, we assume the bivariate time series share the same variance such that $\mathrm{Var}(Y_1^{(i)})=\mathrm{Var}(Y_2^{(i)})=\sigma_i^2, i=1,2.$ The conditions under unequal variance can be derived using the same but more algebraically involved arguments. It can be shown that
\begin{flalign*}
\theta(\omega)=\frac{\omega^{(1)}\gamma_1+\omega^{(2)}\gamma_2}{\omega^{(1)}\sigma^2_1+\omega^{(2)}\sigma^2_2},
\end{flalign*}
which implies that
\begin{flalign*}
\theta(\omega)-\theta_1=\frac{\omega^{(2)}\sigma_2^2(\theta_2-\theta_1)}{\omega^{(1)}\sigma_1^2+\omega^{(2)}\sigma_2^2},\quad\text{and}\quad
\theta(\omega)-\theta_2=\frac{\omega^{(1)}\sigma_1^2(\theta_1-\theta_2)}{\omega^{(1)}\sigma_1^2+\omega^{(2)}\sigma_2^2}.
\end{flalign*}
Therefore, for Assumption \ref{ass_theta} to hold, it suffices to let $\sigma_1^2$ and $\sigma_2^2$ to have the same order, i.e. $\frac{\sigma_1^2}{\sigma_2^2}+\frac{\sigma_2^2}{\sigma_1^2}<\infty$, in which case we can set $C_1=\min\{\frac{\sigma_1^2}{\sigma_2^2},\frac{\sigma_2^2}{\sigma_1^2}\}$ and $C_2=\max\{\frac{\sigma_1^2}{\sigma_2^2},\frac{\sigma_2^2}{\sigma_1^2}\}$.
\end{enumerate}
\bigskip

\noindent\textbf{Examples 4 and 5~(Autocovariance and Autocorrelation change)}: Note that for a univariate time series $\{X_t\}$, its autocovariance and autocorrelation functionals of lag-$d$ can be viewed as the covariance and variance functionals of the bivariate time series $\{\mathbf{Y}_t=(X_t,X_{t-d})^{\top}\}$. Thus, the conditions in Examples 2 and 3 can be applied.

\subsubsection{ Verification of Assumption \ref{ass_no3}}\label{subsec:Alter_assno3}
Assumption \ref{ass_no3star} is a natural extension of Assumption \ref{ass_theta} in the main text from the single change-point setting to multiple change-point setting. Assumption \ref{ass_no3star} generalizes Assumption \ref{ass_no3} in the main text in the sense that Assumption \ref{ass_no3} implies Assumption \ref{ass_no3star}. Thus, in this section, we provide verification for Assumption \ref{ass_no3star}.

\begin{thmbis}{ass_no3}\label{ass_no3star}
Assumption \ref{ass_theta} holds. Furthermore, there exist positive constants $0<C_1<C_2<\infty$ such that for any two consecutive change-points $k_i, k_{i+1}$ for $i=1,\cdots,m_o-1$, and any two time points satisfying $a<k_i\leq k_{i+1}<b$, we have 
\begin{flalign*}
&C_1\Big|\frac{k_i-a}{k_{i+1}-a}[\theta(\omega_{a+1,k_i})-\theta_{i+1}]+\frac{b-k_{i+1}}{b-k_i}[\theta_{i+1}-\theta(\omega_{k_{i+1}+1,b})]\Big|\leq 	\Big|\theta(\omega_{a+1,k_{i+1}})-\theta(\omega_{k_i+1,b})\Big|,\\
&\Big|\theta(\omega_{a+1,k_{i+1}})-\theta(\omega_{k_i+1,b})\Big|\leq C_2\Big\{\frac{k_i-a}{k_{i+1}-a}\Big|\theta(\omega_{a+1,k_i})-\theta_{i+1}\Big|+\frac{b-k_{i+1}}{b-k_i}\Big|\theta_{i+1}-\theta(\omega_{k_{i+1}+1,b})\Big|\Big\}.
\end{flalign*}
\end{thmbis}

Assumption \ref{ass_no3star} essentially regulates the behavior of the functional $\theta(\cdot)$ on mixture of three subsamples $\{Y_t\}_{t=a}^{k_i}$, $\{Y_t\}_{t=k_i+1}^{k_{i+1}}$ and $\{Y_t\}_{t=k_{i+1}+1}^b$. Intuitively, Assumption \ref{ass_no3star} requires that $\theta(\cdot)$ can distinguish the mixture of $\{Y_t\}_{t=a}^{k_i}$ and $\{Y_t\}_{t=k_i+1}^{k_{i+1}}$ from the mixture of $\{Y_t\}_{t=k_i+1}^{k_{i+1}}$ and $\{Y_t\}_{t={k_{i_1}}+1}^b$.

Similar to the verification of Assumption \ref{ass_theta}, the verification of Assumption \ref{ass_no3star} requires a case-by-case analysis and is much more tedious. Below we give an illustrative example of the verification of Assumption \ref{ass_no3star} for variance under the multiple change-point setting. The verification of Assumption \ref{ass_no3star} for (auto)-covariance and (auto)-correlation can be done in the same way with more tedious algebra.
\vskip 2mm

\noindent\textbf{Example~(Variance change)}:  The data $\{Y_t\}_{t=1}^n$ consists of $m_o+1$ stationary segments with $Y^{(i)}\sim F^{(i)}$ such that $EY^{(i)}=\mu_i$ and $\mathrm{Var}(Y^{(i)})=\sigma_i^2$, for $i=1,\cdots,m_o+1.$

In the case of $m_o(\geq 2)$ change-points, we provide a verification of Assumption \ref{ass_no3star} under two sufficient conditions: (1). monotonic change and (2). $\max_{i\neq j}|\mu_i-\mu_j|^2<\min_i |\sigma_{i}^2-\sigma_{i+1}^2|.$ Without loss of generality, we assume  variance is monotonically decreasing.


For a mixture of distributions, let $\omega_{r,s}$ be defined as in equation (\ref{theta}) of the main text, we have  
$$
\theta_{r,s}=\omega_{r,s}^{\top}\underline{\sigma}^2+\omega_{r,s}^{\top}(\underline{\mu})^2-(\omega_{r,s}^{\top}\underline{\mu})^2,
$$
where $\underline{\sigma}^2=(\sigma_1^2,\cdots,\sigma_{m_o+1}^2)^{\top}$ and $\underline{\mu}=(\mu_1,\cdots,\mu_{m_o+1})^{\top}$.
Denote $\sigma_{r,s}^2=\omega_{r,s}^{\top}\underline{\sigma}^2$, $\delta_{r,s}^2=\omega_{r,s}^{\top}\underline{\mu}^2$ and $\mu_{r,s}=\omega_{r,s}^{\top}\underline{\mu}$ for any $1\leq r<s\leq n$.   It is easy to see that $\delta_{r,s}^2\geq \mu_{r,s}^2.$  

Then, for any $(k_i,k_{i+1})$ such that $a<k_i<k_{i+1}<b$, we have 
\begin{flalign*}
[\theta_{a+1,k_{i+1}}-\theta_{k_i+1,b}]
=&\sigma_{a+1,k_{i+1}}^2+\delta_{a+1,k_{i+1}}^2-\mu^2_{a+1,k_{i+1}}-(\sigma_{k_i+1,b}^2+\delta_{k_i+1,b}^2-\mu^2_{k_i+1,b})\\
=&\frac{k_i-a}{k_{i+1}-a}\sigma^2_{a+1,k_i}+\frac{k_{i+1}-k_i}{k_{i+1}-a}\sigma^2_{k_i+1,k_{i+1}}-(\frac{k_{i+1}-k_i}{b-k_i}\sigma^2_{k_i+1,k_{i+1}}+\frac{b-k_{i+1}}{b-k_i}\sigma_{k_{i+1}+1,b}^2)
\\&
+(\delta_{a+1,k_{i+1}}^2-\mu^2_{a+1,k_{i+1}})-(\delta^2_{k_i+1,b}-\mu^2_{k_i+1,b})
\\   =&\frac{k_i-a}{k_{i+1}-a}(\sigma_{a+1,k_i}^2-\sigma_{k_i+1,k_{i+1}}^2)
+(\delta_{a+1,k_{i+1}}^2-\mu^2_{a+1,k_{i+1}})  \\&
+\frac{b-k_{i+1}}{b-k_i}(\sigma_{k_i+1,k_{i+1}}^2-\sigma_{k_{i+1}+1,b}^2)-(\delta^2_{k_i+1,b}-\mu^2_{k_i+1,b})\\:=&Z_1+Z_2.
\end{flalign*}
Here, we note that when $\max_{i\neq j}|\mu_i-\mu_j|^2<\min_i |\sigma_{i+1}^2-\sigma_{i}^2|,$ 
\begin{flalign*}
Z_2= &\frac{b-k_{i+1}}{b-k_i}(\sigma_{k_i+1,k_{i+1}}^2-\sigma_{k_{i+1}+1,b}^2)-(\delta^2_{k_i+1,b}-\mu^2_{k_i+1,b})\\
=&\frac{b-k_{i+1}}{b-k_i}(\sigma_{k_i+1,k_{i+1}}^2-\sigma_{k_{i+1}+1,b}^2)-\frac{(b-k_{i+1})(k_{i+1}-k_i)}{(b-k_i)^2}(\mu_{k_{i}+1,k_{i+1}}-\mu_{k_{i+1}+1,b})^2\\
\geq& \frac{b-k_{i+1}}{b-k_i}[(\sigma_{k_i+1,k_{i+1}}^2-\sigma_{k_{i+1}+1,b}^2)-(\mu_{k_{i}+1,k_{i+1}}-\mu_{k_{i+1}+1,b})^2]>0.
\end{flalign*}
Hence, both $Z_1$ and $Z_2$ are positive.

Similarly, one can show \begin{flalign*}
&\frac{k_i-a}{k_{i+1}-a}[\theta_{a+1,k_i}-\theta_{i+1}]+\frac{b-k_{i+1}}{b-k_i}[\theta_{i+1}-\theta_{k_{i+1}+1,b}]\\
=&\frac{k_i-a}{k_{i+1}-a}(\sigma_{a+1,k_i}^2-\sigma_{k_i+1,k_{i+1}}^2+\delta^2_{a+1,k_i}-\delta^2_{k_i+1,k_{i+1}})\\&
+\frac{b-k_{i+1}}{b-k_i}(\sigma_{k_i+1,k_{i+1}}^2-\sigma_{k_{i+1}+1,b}^2+\delta^2_{k_i+1,k_{i+1}}-\delta^2_{k_{i+1}+1,b})\\&+\frac{k_i-a}{k_{i+1}-a}(-\mu^2_{a+1,k_i}+\mu^2_{k_i+1,k_{i+1}})+
\frac{b-k_{i+1}}{b-k_i}(-\mu^2_{k_i+1,k_{i+1}}+\mu^2_{k_{i+1}+1,b})\\
= & [\theta_{a+1,k_{i+1}}-\theta_{k_i+1,b}]+\frac{k_i-a}{k_{i+1}-a}(-\mu^2_{a+1,k_i}+\mu^2_{k_i+1,k_{i+1}})+
\frac{b-k_{i+1}}{b-k_i}(-\mu^2_{k_i+1,k_{i+1}}+\mu^2_{k_{i+1}+1,b})\\&+\mu^2_{a+1,k_{i+1}}-\mu^2_{k_i+1,b}
\\\leq& [\theta_{a+1,k_{i+1}}-\theta_{k_i+1,b}]+(\mu^2_{a+1,k_{i+1}}- \frac{k_i-a}{k_{i+1}-a}\mu^2_{a+1,k_i}-\frac{k_{i+1}-k_i}{k_{i+1}-a}\mu^2_{k_i+1,k_{i+1}})
\\:=&Z_1+Z_2+R_1.
\end{flalign*}
where the last inequality holds by Cauchy-Schwartz inequality that $$\mu^2_{k_i+1,b}\geq  \frac{b-k_{i+1}}{b-k_i}\mu^2_{k_{i+1}+1,b}+\frac{k_{i+1}-k_i}{b-k_i}\mu^2_{k_i+1,k_{i+1}}.$$
Hence, to show for some $C_1>0$ such that
$$
C_1\left|\frac{k_i-a}{k_{i+1}-a}[\theta_{a+1,k_i}-\theta_{i+1}]+\frac{b-k_{i+1}}{b-k_i}[\theta_{i+1}-\theta_{k_{i+1}+1,b}]\right|\leq \left|\theta_{a+1,k_{i+1}}-\theta_{k_i+1,b}\right|,
$$
or equivalently,  $$
C_1(Z_1+Z_2+R_1)\leq  Z_1+Z_2,
$$
it suffices to show that for some $C_1\in(0,1)$,
$$
R_1< (1-C_1)Z_1.
$$
Note that  $\delta^2_{a+1,k_{i+1}}\geq \mu^2_{a+1,k_{i+1}}$, hence  a sufficient condition is 
$$
\left(\delta_{a+1,k_{i+1}}^2- \frac{k_i-a}{k_{i+1}-a}\mu^2_{a+1,k_i}-\frac{k_{i+1}-k_i}{k_{i+1}-a}\mu^2_{k_i+1,k_{i+1}}\right)< (1-C_1)\left(\frac{k_i-a}{k_{i+1}-a}\right)(\sigma_{a+1,k_i}^2-\sigma_{k_i+1,k_{i+1}}^2).
$$
Note that $\delta_{a+1,k_{i+1}}^2=\frac{k_i-a}{k_{i+1}-a}\delta^2_{a+1,k_{i}}+\frac{k_{i+1}-k_i}{k_{i+1}-a}\delta^2_{k_i+1,k_{i+1}}$.   Using $\delta_{k_i+1,k_{i+1}}^2=\mu_{k_i+1,k_{i+1}}^2$,  it suffices to note that  
$$
\delta^2_{a+1,k_{i}}-\mu^2_{a+1,k_{i}}\leq \max_{i\neq j }(\mu_i-\mu_j)^2<\min_{i}(\sigma_i^2-\sigma^2_{i+1})\leq  (\sigma_{a+1,k_i}^2-\sigma_{k_i+1,k_{i+1}}^2).
$$
That is, we can choose $$C_1=1-\frac{\max_{i\neq j}(\mu_i-\mu_j)^2}{\min_{i\neq j}|\sigma_i^2-\sigma^2_{j}|}.$$
The upper bound $C_2$ can be similarly chosen as $$C_2=1+\frac{\max_{i\neq j}(\mu_i-\mu_j)^2}{\min_{i\neq j}|\sigma_i^2-\sigma^2_{j}|}.$$

\newpage
\section{{Verification of Assumptions in Quantile}}\label{sec:verify_quantile}

This section verifies technical assumptions for quantiles. Section \ref{subsec:partial_quanttile} derives the partial influence functions for quantile functionals, Sections \ref{subsec:theta_quantile}-\ref{subsec:rem_quantile_multiple} provide detailed verification for Assumptions \ref{ass_influence}-\ref{ass_theta} and Assumptions \ref{ass_no1}-\ref{ass_no3}.

In particular, the verification of Assumptions \ref{ass_remainder} and \ref{ass_no2} is highly nontrivial, and their proofs are provided in Section \ref{subsec:proof_quant}. As a byproduct, in Section \ref{subsec:fluc}, we derive a local fluctuation rate for strongly mixing empirical processes of the type 
$$
P\left[\sup_{|x-y|\leq b_n}\Big| \hat{F}_{1,n}(x)-F_{1,n}(x)-\hat{F}_{1,n}(y)+F_{1,n}(y)\Big|>\epsilon\right],
$$
where $\hat{F}_{1,n}(x)=\frac{1}{n}\sum_{t=1}^n \mathbf{1}(Y_t\leq x)$ is the empirical CDF function. This rate may be of independent interest for other research.

\subsection{ Derivation of (partial) influence functions}\label{subsec:partial_quanttile}
Let $\theta=F^{-1}(q)$ be the $q$th quantile functional for a distribution function $F$. We consider $Y_t^{(i)}$ with continuous CDF $F^{(i)}$ and density $f^{(i)}$ for $i=1,\cdots,m_o+1$. Denote the mixture weight by $\omega_{a,b}$, and the mixture CDF by $F^{\omega_{a,b}}$, see the detailed definition of $F^{\omega_{a,b}}$ below equation \eqref{theta} in the main text.

By Definition \ref{def:PIF}, we have $\xi_i(y,\omega_{a,b})=\lim\limits_{\zeta\to0}\zeta^{-1}\Big[\theta\left((\delta_y-F^{(i)})\zeta+F^{\omega_{a,b}}\right)-\theta\left(F^{\omega_{a,b}}\right)\Big]$ is the G\^ateaux derivative of $\theta\big(F^{\omega_{a,b}}\big)$ in the direction $\delta_y-F^{(i)}$ for $i=1,\cdots,m_o+1$. A direct calculation gives that
$$\xi_i(y,\omega_{a,b})=\frac{F^{(i)}(\theta(F^{\omega_{a,b}}))-\mathbf{1}(y\leq \theta(F^{\omega_{a,b}}))}{f^{\omega_{a,b}}(\theta(F^{\omega_{a,b}}))}, ~ i=1,\cdots,m_o+1.$$
Thus, for expansion \eqref{theta_null} in the main text, we have $\xi_i(Y_t^{(i)})=\frac{q-\mathbf{1}(Y_t^{(i)}\leq \theta(F^{(i)}))}{f^{(i)}(\theta(F^{(i)}))}, i=1,\cdots,m_o+1,$ and for expansion \eqref{theta} in the main text, we have
$$\xi_i(Y_t^{(i)},\omega_{a,b})=\frac{F^{(i)}(\theta_{a,b})-\mathbf{1}(Y_t^{(i)}\leq \theta_{a,b})}{f^{\omega_{a,b}}(\theta_{a,b})},~ i=1,\cdots,m_o+1,$$
where $\theta_{a,b}=\theta(F^{\omega_{a,b}})$ and $r_{a,b}(\omega_{a,b})$ is defined implicitly. 

Similarly, we could derive the results for expansions \eqref{inf1} and \eqref{inf3} in the main text.

\subsection{Verification of Assumptions \ref{ass_influence} and \ref{ass_no1}}\label{subsec:theta_quantile}
Assumption \ref{ass_influence}(i) holds under mild mixing conditions of $\{Y_t^{(1)},Y_{t}^{(2)}\}$~(see \cite{phillips1987time}). Assumption \ref{ass_influence}(ii) holds if $\inf_{1\leq a<b\leq n}f^{\omega_{a,b}}(\theta_{a,b})>c$ for some $c>0$, which is true under the mild sufficient condition that $\inf_{\theta \in [\theta_1,\theta_2]}\min(f^{(1)}(\theta),f^{(2)}(\theta))>c$ with $\theta_i=\theta(F^{(i)}), i=1,2$. Similarly, Assumption \ref{ass_no1} hold under mild mixing conditions of $\{Y_t^{(1)},\cdots,Y_{t}^{(m_o+1)}\}$ and $\inf_{\theta \in [\min\theta_i,\max\theta_i]}\min_{1\leq i\leq m_o+1} f^{(i)}(\theta)>c$.

\subsection{ Verification of Assumptions \ref{ass_remainder} and \ref{ass_no2}}\label{subsec:inf_quantile}
Verification of Assumption \ref{ass_remainder} and \ref{ass_no2} is highly nontrivial, and related results can be found in \cite{Wu2005} and \cite{dette2020a}. However, the arguments in their papers are not directly applicable in our setting. Specifically, \cite{Wu2005} provides  the Bahadur representation of sample quantiles for linear and some nonlinear processes, but the result only holds for the full sample and no uniform result (in terms of subsample) is given. In addition, he requires the underlying process to be stationary, hence change-points are not allowed. \cite{dette2020a} extend the result in \cite{Wu2005} to hold uniformly for subsample, see Theorem 4.1  therein, but they only obtain the result under the i.i.d. setting, which excludes both temporal dependence and change-points.  

Theorems \ref{thm_single_quant} and \ref{thm_multiple_quant} below provide verification of Assumptions \ref{ass_remainder} and \ref{ass_no2} in strongly mixing processes, which give a uniform control for the reminder terms of the partial influence functions derived in \Cref{subsec:partial_quanttile} for all subsample quantiles. Therefore, we give a Bahadur representation that hold uniformly for all subsample quantiles. Note that our result allows for both temporal dependence and structural breaks, and thus improves the results in \cite{dette2020a}. To proceed, we first make the following assumptions.

\begin{ass}\label{ass_mixing}
The data $\{Y_t\}_{t=1}^n$ is $\alpha$-mixing with mixing coefficient $\alpha(k)=\exp(-c_0k)$ for some constant $c_0>0$.
\end{ass}

\begin{ass}\label{ass_bound}
For some positive constants $0<c_1,c_2,c_3<\infty$, the density function $f^{(i)}$ satisfies: (1).\ $\max_i \sup_x f^{(i)}(x)\leq c_1$; (2).\ for $x\in[\min_i\theta_i,\max_i\theta_i]$, $\min_i f^{(i)}(x)\geq c_2$; (3).\ $\max_i\sup_x|f^{(i)'}(x)|\leq c_3$, where $f^{(i)'}(x)$ denotes the first-order derivative of $f^{(i)}(x).$
\end{ass}
\begin{ass}\label{ass_tail}
$\max_i P(|Y_t^{(i)}|>x)\leq Cx^{-\lambda}$ with $\lambda>4/5$ for some constant $C>0$.
\end{ass}

Assumption \ref{ass_mixing} is needed for invoking the Bernstein type inequality for $\alpha$-mixing sequences in \cite{merlevede2009bernstein}, which plays an important role in our proof.  It  is satisfied by commonly used time series models, such as  ARMA models and GARCH models. Assumption \ref{ass_bound} and Assumption \ref{ass_tail} are adapted from \cite{dette2020a} to our change-point setting. 

\begin{thm} \label{thm_single_quant}
Under Assumptions \ref{ass_mixing}-\ref{ass_tail}, Assumption \ref{ass_remainder} holds.
\end{thm}

The proof of \Cref{thm_single_quant} can be found in \Cref{subsec:proof_quant}. When there are multiple change-points, we need a stronger assumption on the tail behavior of $Y_t^{(i)}$. This is because we need to control for both the starting and ending index of the subsample in Assumption \ref{ass_no2} as both $a$ and $b$ are free to move, while we only need to control for one of the two indices in Assumption \ref{ass_remainder}.
\begin{ass}\label{ass_tail2}
$\max_i P(|Y_t^{(i)}|>x)\leq Cx^{-\lambda}$ with $\lambda>18/5$ for some constant $C>0$.
\end{ass}
\begin{thm} \label{thm_multiple_quant}
Under Assumptions \ref{ass_mixing}, \ref{ass_bound} and \ref{ass_tail2}, Assumption \ref{ass_no2} holds.
\end{thm}
The proof of Theorem \ref{thm_multiple_quant} can be found in Section \ref{subsec:proof_quant}.

\subsection{Verification of Assumption \ref{ass_theta}}\label{subsec:rem_quantile}
Let $\theta=F^{-1}(q)$ be the $q$th quantile functional for a distribution function $F$. We consider $Y_t^{(i)}$ with continuous CDF $F^{(i)}$ and density $f^{(i)}$ for $i=1,2$. Denote the mixture weight by $\omega=(\omega^{(1)},\omega^{(2)})^{\top}$, and the mixture CDF by $F^{\omega}=\omega^{(1)} F^{(1)}+\omega^{(2)}F^{(2)}$ and denote $f^{\omega}$ as the density. By the mean value theorem, we have 
\begin{flalign*}
F^{\omega}(\theta(F^{\omega}))-F^{\omega}(\theta_1)=&f^{\omega}(\theta^*)[\theta(F^{\omega})-\theta_1],
\end{flalign*}
for some $\theta^*$ that lies between $\theta_1$ and $\theta(F^{\omega})$.
In addition, we have 
\begin{flalign*}
F^{\omega}(\theta(F^{\omega}))-F^{\omega}(\theta_1)=&q-[\omega^{(1)}F^{(1)}+\omega^{(2)}F^{(2)}](\theta_1)
\\=&\omega^{(2)}[F^{(1)}(\theta_1)-F^{(2)}(\theta_1)].
\end{flalign*}
Hence, provided $f^{\omega}(\theta^*)$ is positive, we have 
$$
\theta(F^{\omega})-\theta_1=f^{\omega}(\theta^*)^{-1}\omega^{(2)}[F^{(1)}(\theta_1)-F^{(2)}(\theta_1)].
$$
Similarly, we have
$$
\theta(F^{\omega})-\theta_2=f^{\omega}(\theta^{\dag})^{-1}\omega^{(1)}[F^{(2)}(\theta_2)-F^{(1)}(\theta_2)].
$$
for some $\theta^{\dag}$ that lies between $\theta_2$ and $\theta(F^{\omega})$.

By symmetry, we assume $\theta_1 < \theta_2$, which implies that $\theta_1 < \theta(F^{\omega})<  \theta_2$. 
Therefore, a sufficient condition for Assumption \ref{ass_theta} is that 
$c_1\leq f^{\omega}(\theta)\leq c_2$ on $[\theta_1,\theta_2]$ where $c_1$ and $c_2$ are two positive constants, which holds if $c_1\leq f^{(i)}(\theta)\leq c_2$, $i=1,2$ on $[\theta_1,\theta_2]$. In this case, we can set 
$$
C_1=\frac{\min_{i=1,2}|F^{(1)}(\theta_i)-F^{(2)}(\theta_i)|}{\max_{i=1,2}\sup_{\theta\in[\theta_1,\theta_2]}f^{(i)}(\theta)|\theta_1-\theta_2|},\quad C_2=\frac{\max_{i=1,2}|F^{(1)}(\theta_i)-F^{(2)}(\theta_i)|}{\min_{i=1,2}\inf_{\theta\in[\theta_1,\theta_2]}f^{(i)}(\theta)|\theta_1-\theta_2|}.
$$

\subsection{ Verification of Assumption \ref{ass_no3star}}\label{subsec:rem_quantile_multiple}

We verify Assumption \ref{ass_no3star}~(see definition in Section \ref{subsec:Alter_assno3}) under the sufficient conditions: (1).\ monotonic change and (2).\ $\inf_{\theta\in[\min\theta_i,\max\theta_i]}\min_if^{(i)}(\theta)>c>0$.

Let $a<k_i<k_{i+1}<b$, similar to Section \ref{subsec:rem_quantile}, we can obtain 
\begin{flalign*}
\theta_{a+1,k_{i+1}}-\theta_{i+1}=&\frac{k_i-a}{k_{i+1}-a}f^{\omega_{a+1,k_i}}(\xi_1)^{-1}[F_{k_i+1,k_{i+1}}(\theta_{i+1})-F_{a+1,k_i}(\theta_{i+1})],\\
\theta_{k_i+1,b}-\theta_{i+1}=&\frac{b-k_{i+1}}{b-k_i}f^{\omega_{k_i+1,b}}(\xi_2)^{-1}[F_{k_i+1,k_{i+1}}(\theta_{i+1})-F_{k_{i+1}+1,b}(\theta_{i+1})],
\end{flalign*}
where $\xi_1$ lies between $\theta_{a+1,k_i}$ and $\theta_{i+1}$,  $\xi_2$ lies between $\theta_{k_{i+1}+1,b}$ and $\theta_{i+1}$.

Without loss of generality, we assume the change is monotonically decreasing, hence Assumption \ref{ass_no3star} is equivalent to
\begin{flalign*}
&C_1\left|\frac{k_i-a}{k_{i+1}-a}[\theta_{a+1,k_i}-\theta_{i+1}]+\frac{b-k_{i+1}}{b-k_i}[\theta_{i+1}-\theta_{k_{i+1}+1,b}]\right|\\\leq &
\frac{k_i-a}{k_{i+1}-a}f^{\omega_{a+1,k_i}}(\xi_1)^{-1}[F_{k_i+1,k_{i+1}}(\theta_{i+1})-F_{a+1,k_i}(\theta_{i+1})]\\&+\frac{b-k_{i+1}}{b-k_i}f^{\omega_{k_i+1,b}}(\xi_2)^{-1}[F_{k_{i+1}+1,b}(\theta_{i+1})-F_{k_i+1,k_{i+1}}(\theta_{i+1})]\\
\leq &C_2\left\{\left|\frac{k_i-a}{k_{i+1}-a}[\theta_{a+1,k_i}-\theta_{i+1}]\right|+\left|\frac{b-k_{i+1}}{b-k_i}[\theta_{i+1}-\theta_{k_{i+1}+1,b}]\right|\right\}.
\end{flalign*}
Thus, for Assumption \ref{ass_no3star} to hold, we can choose
\begin{flalign*}
C_1=&\frac{\min_{i\neq j}|F^{(j)}(\theta_i)-F^{(i)}(\theta_i)|}{\max_{i}\sup_{\theta\in[\min_i\theta_i,\max_i\theta_i]}f^{(i)}(\theta)\max_{i\neq j}|\theta_i-\theta_j|},\\ C_2=&\frac{\max_{i\neq j}|F^{(j)}(\theta_i)-F^{(i)}(\theta_i)|}{\min_{i}\inf_{\theta\in[\min_i\theta_i,\max_i\theta_i]}f^{(i)}(\theta)\min_{i\neq j}|\theta_i-\theta_j|}.
\end{flalign*}

\subsection{ Local fluctuation rate of strongly mixing empirical processes}\label{subsec:fluc}
This section provides a local fluctuation rate of empirical processes for $\alpha$-mixing sequences in  Lemma \ref{lem_DKW}, which is of independent interest.   We remark that the key result of Lemma \ref{lem_DKW} and Lemma \ref{lem_berstein} holds even when the distribution functions are not identical, and hence is useful for verification of Assumption \ref{ass_no2} in the case of multiple change-points. 

\begin{lemma}\label{lem_DKW}
For a sequence of $\alpha$-mixing random variables $\{X_i\}_{i=1}^n$, suppose each $X_i$ has the marginal  cumulative distribution function $F^{(i)}(x)$ and density function $f^{(i)}(x)$ for which $\max_{i}\sup_x f^{(i)}(x)\leq c_1<\infty$. In addition, assume the mixing coefficient $\alpha(k)=\exp(-c_0k)$ for some  $c_0>0$, and  $\max_iP(|X_i|\geq x)\leq C|x|^{-\lambda}$ for some $\lambda>0$. Denote the empirical cdf as $\hat{F}_{1,n}(x)=\frac{1}{n}\sum_{i=1}^n\mathbf{1}(X_i\leq x)$ and the mean cdf as $F_{1,n}(x)=n^{-1}\sum_{i=1}^nF^{(i)}(x)$.
Then, for any $\tau>1$, and $\iota>0$,  there exists constants $C_{\tau}$, $C_{\tau,\iota}$ and $K$ such that,
$$
P\left[\sup_x\Big| \hat{F}_{1,n}(x)-F_{1,n}(x)\Big|>\frac{C_{\tau}\log^{1/2}(n)}{n^{1/2}}\right]\leq Kn^{-\tau},
$$
and 
$$
P\left[\sup_{|x-y|\leq b_n}\Big| \hat{F}_{1,n}(x)-F_{1,n}(x)-\hat{F}_{1,n}(y)+F_{1,n}(y)\Big|>\frac{C_{\tau,\iota}(b_n^{2/(2+\iota)}\log(n))^{1/2}}{n^{1/2}}\right]\leq Kn^{-\tau},
$$
where  $b_n$ is any positive, bounded sequence of real numbers such that $log^5(n)=o(nb_n^{2/(2+\iota)})$, and $b_n=o(1)$.
\end{lemma}

\begin{proof}
We only prove for the second, as the first can be proved similarly.

Denote $Y_i(x)=\mathbf{1}(X_i\leq x)-F^{(i)}(x)$, and let $J=n^{(\tau+3)/\lambda}$, $v_n=\sqrt{nb_n^{2/(2+\iota)}log(n)}$, $t_n=v_n/n$.

$$
\begin{aligned}
I_{n} &:=P\left[\sup _{|x-y| \leq b_{n}, x \leq-J}\Big| \hat{F}_{1,n}(x)-F_{1,n}(x)-\hat{F}_{1,n}(y)+F_{1,n}(y)\Big|>C t_{n}\right] \\
& \leq \sum_{i=1}^nP\left[\sup _{|x-y| \leq b_{n}, x \leq-J}\left|Y_{i}(y)-Y_{i}(x)\right|>C t_{n}\right] \\
& \leq \sum_{i=1}^n(Ct_n)^{-1} E\left[\sup _{|x-y| \leq b_{n}, x \leq-J}\left|Y_{i}(y)-Y_{i}(x)\right|\right] \\
&\leq \sum_{i=1}^n 2(Ct_n)^{-1} E\left[\sup _{x \leq-J+b_{n}}\left|Y_{i}(x)\right|\right] \\
&\leq Kn^2 v_{n}^{-1}\left(J-b_{n}\right)^{-\lambda} \leq Kn^{-1-\tau},
\end{aligned}
$$
where the second inequality holds by Markov's inequality, and the last  by
\begin{flalign*}
E\left[\sup _{x \leq-J+b_{n}}\left|Y_{i}(x)\right|\right]\leq& E\left[\sup _{x \leq-J+b_{n}} \Big(F^{(i)}(x)+\mathbf{1}(X_i\leq x)(1-F^{(i)}(x))\Big)\right]
\leq 2F^{(i)}(-J+b_n),
\end{flalign*}
and that $P(|X_i|\geq x)\leq C |x|^{-\lambda}.$

Similarly, we have 
$$
III_{n}:=P\left[\sup_{|x-y| \leq b_{n}, x \geq J}\Big| \hat{F}_{1,n}(x)-F_{1,n}(x)-\hat{F}_{1,n}(y)+F_{1,n}(y)\Big|>C t_{n}\right]\leq Kn^{-1-\tau}.
$$

Let $x_{i}=i b_{n} / n, i=-N-1, \ldots, N+1$, where $N=\left\lfloor J n / b_{n}\right\rfloor$, and
$$
I I_{n}:=P\left[\sup _{|x-y| \leq b_{n},-J<x<J}\Big| \hat{F}_{1,n}(x)-F_{1,n}(x)-\hat{F}_{1,n}(y)+F_{1,n}(y)\Big|>C t_{n}\right] \text {. }
$$

For any $x, y$ with $|x-y| \leq b_{n},|x| \leq J$ and $|y| \leq J$, choose $i$ and $j$ such that $x_{i} \leq x<x_{i+1}$ and $x_{j} \leq y<x_{j+1}$, we know that either $|x_{i+1}-x_j|\leq b_n$ or $|x_{j+1}-x_i|\leq b_n$. 
In addition, $F_{1,n}(x_{i+1})-F_{1,n}(x_i)\leq c_1(x_{i+1}-x_i)=c_1b_n/n$, and similarly $F_{1,n}(x_{j+1})-F_{1,n}(x_j)\leq c_1b_n/n$, we have
$$
\begin{aligned}
& n\Big[\hat{F}_{1,n}(x_i)-F_{1,n}(x_i)-\hat{F}_{1,n}(x_{j+1})+F_{1,n}(x_{j+1})\Big]-2 c_1 b_{n}\\ \leq&  n\Big[\hat{F}_{1,n}(x)-F_{1,n}(x)-\hat{F}_{1,n}(y)+F_{1,n}(y)\Big] \\
\leq &n\Big[\hat{F}_{1,n}(x_{i+1})-F_{1,n}(x_{i+1})-\hat{F}_{1,n}(x_j)+F_{1,n}(x_j)\Big]+2 c_1 b_{n}.
\end{aligned}
$$
Thus,  note  $2c_1b_n\ll v_n$, we obtain by Lemma \ref{lem_berstein},
$$
II_n\leq P\left[\max _{i, j=-N-1, \ldots, N+1:\left|x_{i}-x_{j}\right| \leq b_{n}} n\Big| \hat{F}_{1,n}(x_i)-F_{1,n}(x_i)-\hat{F}_{1,n}(x_j)+F_{1,n}(x_j)\Big|>(C-1) v_{n}\right]\leq 2N^2 n^{-\frac{C_B(C-1)^2}{C_V+1}},
$$
where $C_B$ and $C_V$ are constants dependent only on $c_0$ and $c_1$.
Using the fact that $b_n\gg n^{-(2+\iota)/2}$, we have $$
2N^2 n^{-\frac{C_B(C-1)^2}{C_V+1}}=o(n^{\frac{2\tau+6}{\lambda}+4+\iota-\frac{C_B(C-1)^2}{C_V+1}}).
$$
Choose $(C_{\tau,\iota}-1)^2=C_B^{-1}\lambda^{-1}(C_V+1)[2\tau+6+(4+\iota+\tau)\lambda]$ would suffice.

\end{proof}

\begin{lemma}\label{lem_berstein}
Under conditions of Lemma \ref{lem_DKW}, if $|x-y|\leq b_n$, then for any constant $C\in(0,\infty)$,
$$P\left[ n\Big| \hat{F}_{1,n}(x)-F_{1,n}(x)-\hat{F}_{1,n}(y)+F_{1,n}(y)\Big|>C v_{n}\right]\leq 2n^{-\frac{C_BC^2}{C_V+1}},$$
where $C_V$ and $C_B$ are constants that only depend on $c_0$ and $c_1$.
\end{lemma}
\begin{proof}
For any fixed pair $(x,y)$ such that $y\leq x\leq y+ b_{n}$, denote $\xi_i=Y_i(x)-Y_i(y)$ and $p_i=F^{(i)}(x)-F^{(i)}(y)$.  
Let \begin{flalign*}
V_i^2=\mathrm{Var}[\xi_i]+2\sum_{j>i}|\mathrm{Cov}[\xi_i,\xi_j]|.
\end{flalign*}
Then,
$$
n\Big[ \hat{F}_{1,n}(x)-F_{1,n}(x)-\hat{F}_{1,n}(y)+F_{1,n}(y)\Big]=\sum_{i=1}^n\mathbf{1}(y\leq X_i\leq x)-[F^{(i)}(x)-F^{(i)}(y)]=\sum_{i=1}^n \xi_i.
$$
Note that 
\begin{equation}\label{bound_xi}
E|\xi_i|^{(2+\iota)}\leq p_i\leq \max_i\sup_zf^{(i)}(z)|x-y|\leq c_1b_n. 
\end{equation}
By Davydov's inequality, we know that 
\begin{flalign*}
&|\mathrm{Cov}[\xi_i,\xi_j]|\leq 12 \alpha(|i-j|)^{\iota/(2+\iota)} (E|\xi_i|^{2+\iota})^{1/(2+\iota)}(E|\xi_j|^{2+\iota})^{1/(2+\iota)}\\
\leq & 12 \exp(-c_0|j-i|)|p_ip_j|^{2/(2+\iota)}\\\leq&
12 \exp(-c_0|j-i|)|c_1b_n|^{2/(2+\iota)}
\end{flalign*}
where the last inequality holds by (\ref{bound_xi}). Note that $b_n=o(1)$, we obtain 
$$
\max_iV_i^2\leq \sum_{i=0}^{\infty} 24 \exp(-c_0i)|c_1b_n|^{2/(2+\iota)}\leq C_V b_n^{2/(2+\iota)}
$$
for some constant $C_V$ dependent only on $c_1$ and $c_0$.

Thus, note $\max_i|\xi_i|\leq 1$, by Bernstein's inequality  for $\alpha$-mixing process, see Theorem 2 in \cite{merlevede2009bernstein}, we know that for some constant $C_B$ dependent only on $c_0$:
\begin{flalign*}
&P\left[\Big|\sum_{i=1}^n\xi_i\Big|>C v_{n}\right]\\ \leq&  \exp \left[\frac{-C_B C^{2} v_{n}^{2}}{\big(
\max_i V_i^2 n+ 1+ Cv_n\log^2(n)
\big)}\right] 
\\  \leq& \exp \left[\frac{-C_B C^{2} v_{n}^{2}}{(C_V+1)nb_n^{2/(2+\iota)}}\right]=n^{-\frac{C_BC^2}{C_V+1}}
\end{flalign*}
using the fact $2Cv_n\ll nb_n^{2/(2+\iota)}$ when $\log^5(n)=o(nb_n^{2/(2+\iota)})$.

Similarly, we can prove for the case $y-b_n\leq x<y$.
\end{proof}

\subsection{ Proofs of Theorem \ref{thm_single_quant} and \ref{thm_multiple_quant}}\label{subsec:proof_quant}

\subsubsection{Results regarding \Cref{thm_single_quant}}
Let 
$\hat{F}_{1,k}(x)=\frac{1}{k}\Big[  \sum_{t=1}^{k^*\wedge k} \mathbf{1}_{Y_{t}^{(1)} \leq x}+\mathbf{1}(k>k^*)\sum_{t=k^*+1}^{k} \mathbf{1}_{Y_{t}^{(2)} \leq x}\Big]$,  ${F}_{1,k}(x)=\frac{k^*\wedge k}{k}F^{(1)}(x)+\mathbf{1}(k>k^*)\frac{k-k^*}{k}F^{(2)}(x)$ and $f_{1,k}(x)=F_{1,k}(x)$. 
The following technical treatments are modified from the i.i.d setting in \cite{dette2020a} to accommodate the  change-point setting in $\alpha$-mixing time series.  Lemma \ref{lem_fluc} establishes a uniform fluctuation rate of empirical process based on results in Lemma \ref{lem_DKW}; Lemma \ref{lem_bound} shows that the sample quantile estimates are uniformly close to the true ones; Lemma \ref{lem_small} further improves Lemma \ref{lem_bound} and deals with the case when the subsample size is small; and Lemma \ref{lem_replace} indicates the fluctuation  of empirical process at the true quantile levels can be replaced by the value at the estimates. We defer their proofs after the proof of Theorem \ref{thm_single_quant}.

\begin{lemma}\label{lem_fluc}
For all $0<r<1$, and  $\vartheta>1$, there exists constants $C_{r, \vartheta}>0$ and $K$ such that 
$$P\left(\max _{n^r\leq  k\leq n } \sup _{x }\left|\hat{F}_{1,k}(x)-F_{1,k}(x)\right|>C_{r, \vartheta} \frac{\sqrt{ r\log (n)}}{n^{r / 2}}\right) \leq Kn^{-\vartheta} .$$
Furthermore, for any $\iota>0$, there exists a constant $C_{r, \vartheta,\iota}>0$
$$P\left(\max _{n^r\leq  k\leq n } \sup _{|x-y|\leq  b_{n,r} }\left|\hat{F}_{1,k}(x)-F_{1,k}(x)-\hat{F}_{1,k}(y)+F_{1,k}(y)\right|>C_{r, \vartheta,\iota}\frac{\sqrt{b_{n,r}^{2/(2+\iota)}r \log (n)}}{n^{r / 2}}\right) \leq Kn^{-\vartheta},$$
where $b_{n,r}$ satisfies $b_{n,r}=o(1)$ and  $\log^5(n)=o(n^rb_{n,r}^{2/(2+\iota)})$.
\end{lemma}

\begin{lemma}\label{lem_bound}
Suppose for some constant $c_2>0$,  $ f^{(i)}(x)\geq c_2$ for $x\in[\min_{i=1,2}\theta^{(i)},\max_{i=1,2}\theta^{(i)}]$, and $\sup_x|f^{(i)'}(x)|<\infty$ for $i=1,2$. Let $b_{n,r}=4C_{r,\vartheta}/c_2\sqrt{r\log(n)}/n^{r/2}$ with $\vartheta>1$, for all  $0<r<1$, $$P(  \max _{ n^{r}\leq k\leq n}\left|\hat{\theta}(\omega_{1,k})-{\theta}(\omega_{1,k})\right| \geq b_{n, r}, i.o.)=0.$$

In addition,  for all $0<r<1$ and $0<d_0<1$,
\begin{equation}\label{claim}
n^{d_0r/2} \max _{ n^{r}\leq k\leq n}\left|\hat{\theta}(\omega_{1,k})-{\theta}(\omega_{1,k})\right|=o_p(1).
\end{equation}

\end{lemma}
\begin{lemma}\label{lem_small} 
For $0<r<\frac{\lambda}{2(\lambda+1)}$, 
$$n^{-1/2}\max_{1\leq k< n^r} k |\hat{\theta}(\omega_{1,k})-{\theta}(\omega_{1,k})|=o_p(1).$$
\end{lemma}

\begin{lemma}\label{lem_replace}
For $2/9<r<1$, $$n^{-1/2}\max_{ n^r\leq k\leq  n} k\Big| \hat{F}_{1,k}(\hat{\theta}(\omega_{1,k}))-F_{1,k}(\hat{\theta}(\omega_{1,k}))-\hat{F}_{1,k}(\theta(\omega_{1,k}))+F_{1,k}(\theta(\omega_{1,k}))\Big|=o_p(1).$$
\end{lemma}

\bigskip
\noindent\textsc{\textbf{Proof Theorem \ref{thm_single_quant}}:}
\begin{proof}
Note that $\theta(\omega_{1,k})\in[\min_{i=1,2}\theta^{(i)},\max_{i=1,2}\theta^{(i)}] $, we have $f_{1,k}(\theta(\omega_{1,k}))\geq c_2$ for all $k$. By symmetry, we only need to show $\sup_k k |r(\omega_{1,k})|=o_p(n^{1/2})$.

It suffices to show 
$$
n^{-1/2}\max_{1\leq k\leq n}k |f_{1,k}(\theta(\omega_{1,k})) \big(\hat{\theta}(\omega_{1,k})-{\theta}(\omega_{1,k})\big)-\hat{F}_{1,k}(\hat{\theta}(\omega_{1,k}))+\hat{F}_{1,k}(\theta(\omega_{1,k}))|=o_p(1).
$$

Let $\lambda>4/5$, and $2/9<r<\frac{\lambda}{2(\lambda+1)}<1/2$.  Hence, using the boundedness of $\hat{F}_{1,k}(\cdot)$, $\sup_xf_{1,k}(x)\leq c_1$ and Lemma \ref{lem_small}, it suffices to show 
$$
n^{-1/2}\max_{n^{r}\leq k\leq n}k |f_{1,k}(\theta(\omega_{1,k})) \big(\hat{\theta}(\omega_{1,k})-{\theta}(\omega_{1,k})\big)-\hat{F}_{1,k}(\hat{\theta}(\omega_{1,k}))+\hat{F}_{1,k}(\theta(\omega_{1,k}))|=o_p(1).
$$
By Lemma \ref{lem_replace}, it suffices to show 
$$
n^{-1/2}\max_{n^{r}\leq k\leq n}k |f_{1,k}(\theta(\omega_{1,k})) \big(\hat{\theta}(\omega_{1,k})-{\theta}(\omega_{1,k})\big)-F_{1,k}(\hat{\theta}(\omega_{1,k}))+F_{1,k}(\theta(\omega_{1,k}))|=o_p(1).
$$
Using the Taylor's expansion, it suffices to show 
$$
n^{-1/2}\max_{n^{r}\leq k\leq n}k \sup_x f'_{1,k}(x) \big(\hat{\theta}(\omega_{1,k})-{\theta}(\omega_{1,k})\big)^2=o_p(1).
$$
Choose  $d_0$ such that $\frac{1}{2 d_{0}} \leq r d_{0}+1 / 2$,  $\delta$ such that $r<\frac{1}{2 d_{0}} \leq \delta \leq r d_{0}+1 / 2$.
Then \begin{flalign*}
\frac{1}{\sqrt{n}} \max _{n^{r}\leq k< n^{\delta}}k \big(\hat{\theta}(\omega_{1,k})-{\theta}(\omega_{1,k})\big)^2\leq &\max _{ n^r \leq k<n^{\delta}} n^{\delta-1 / 2}\big(\hat{\theta}(\omega_{1,k})-{\theta}(\omega_{1,k})\big)^2\\\leq& \max _{ n^r \leq k<n^{\delta}}  n^{rd_0}\big(\hat{\theta}(\omega_{1,k})-{\theta}(\omega_{1,k})\big)^2=o_p(1),
\end{flalign*}
by Lemma \ref{lem_bound}. Similarly, 
$$
\frac{1}{\sqrt{n}} \max _{ n^{\delta}\leq k\leq n}k \big(\hat{\theta}(\omega_{1,k})-{\theta}(\omega_{1,k})\big)^2\leq \max _{n^{\delta}\leq k\leq n} n^{1/ 2}\big(\hat{\theta}(\omega_{1,k})-{\theta}(\omega_{1,k})\big)^2\leq \max _{n^{\delta}\leq k\leq n}  n^{\delta d_0}\big(\hat{\theta}(\omega_{1,k})-{\theta}(\omega_{1,k})\big)^2=o_p(1).
$$
\end{proof}

\noindent\textsc{Proof of Lemma \ref{lem_fluc}}
\begin{proof}
Note that $\sqrt{\log(k)}/k^{1/2}\leq \sqrt{r\log(n)}/n^{r/2}$ for $k\geq n^r$, then
\begin{flalign*}
&P\left(\max _{n^r\leq  k\leq n } \sup _{x }\left|\hat{F}_{1,k}(x)-F_{1,k}(x)\right|>C_{r, \vartheta} \frac{\sqrt{ r\log (n)}}{n^{r / 2}}\right) \\\leq& \sum_{k=n^r}^n P\left( \sup _{x }\left|\hat{F}_{1,k}(x)-F_{1,k}(x)\right|>C_{r, \vartheta} \frac{\sqrt{ \log (k)}}{k^{1/ 2}}\right) 
\leq \sum_{k=n^r}^nKk^{-\tau}\leq Kn^{1-r\tau},
\end{flalign*}
where the second inequality holds by Lemma \ref{lem_DKW}.
Choose $\tau$ large enough such that $1-r\tau=-\vartheta$ then gives the result.

The second inequality holds similarly. To see this, note that 
\begin{flalign*}
&P\left(\max _{n^r\leq  k\leq n } \sup _{|x-y|\leq b_{n,r} }\left|\hat{F}_{1,k}(x)-F_{1,k}(x)-\hat{F}_{1,k}(y)+F_{1,k}(y)\right|>C_{r, \vartheta,\iota} \frac{\sqrt{b_{n,r}^{2/(2+\iota)} r\log (n)}}{n^{r / 2}}\right)\\
\leq & \sum_{k=n^r}^n P\left(\sup _{|x-y|\leq b_{n,r} }\left|\hat{F}_{1,k}(x)-F_{1,k}(x)-\hat{F}_{1,k}(y)+F_{1,k}(y)\right|>C_{r, \vartheta,\iota} \frac{\sqrt{b_{n,r}^{2/(2+\iota)} \log (k)}}{k^{1/ 2}}\right)\\
\leq&\sum_{k=n^r}^nKk^{-\tau}\leq Kn^{1-r\tau}.
\end{flalign*}
where the second inequality holds by  Lemma \ref{lem_DKW}, and   that for all $n^r\leq k\leq n$, $\frac{\log^5(k) }{k}\leq\frac{r^5\log^5(n) }{n^r} =o(b_{n,r}^{2/(2+\iota)})$. 
\end{proof}

\noindent\textsc{Proof of Lemma \ref{lem_bound}}
\begin{proof}
Recall $\hat{F}_{1,k}(\hat{\theta}(\omega_{1,k}))=q,$ it suffices to show $P(  \max _{ n^{r}\leq k\leq n}\hat{F}_{1,k}({\theta}(\omega_{1,k})-b_{n,r})-q>0, i.o.)=0,$ and $P(  \min _{ n^{r}\leq k\leq n}\hat{F}_{1,k}({\theta}(\omega_{1,k})+b_{n,r})-q<0, i.o.)=0,$ 

By symmetry, we only prove for the case 
$$P(  \max _{ n^{r}\leq k\leq n}\hat{F}_{1,k}({\theta}(\omega_{1,k})-b_{n,r})-q>0, i.o.)=0.$$

Recall also $F_{1,k}(\theta(\omega_{1,k}))=q$,
\begin{flalign*}
&\max _{ n^{r}\leq k\leq n}\hat{F}_{1,k}({\theta}(\omega_{1,k})-b_{n,r})-q
\\=&\max _{ n^{r}\leq k\leq n}F_{1,k}({\theta}(\omega_{1,k})-b_{n,r})-q+\hat{F}_{1,k}({\theta}(\omega_{1,k})-b_{n,r})-F_{1,k}({\theta}(\omega_{1,k})-b_{n,r})-\hat{F}_{1,k}(\theta(\omega_{1,k}))+q\\&+\hat{F}_{1,k}({\theta}(\omega_{1,k}))-F_{1,k}({\theta}(\omega_{1,k}))
\\\leq&\max _{ n^{r}\leq k\leq n} F_{1,k}({\theta}(\omega_{1,k})-b_{n,r})-F_{1,k}({\theta}(\omega_{1,k}))+\max _{ n^{r}\leq k\leq n}\sup_{|x-y|\leq b_{n,r}} |\hat{F}_{1,k}(x)-F_{1,k}(x)-\hat{F}_{1,k}(y)+F_{1,k}(y)|\\&+\max _{ n^{r}\leq k\leq n}\sup_x|\hat{F}_{1,k}(x)-F_{1,k}(x)|\\
:=&R_{1}+R_{2}+R_{3}.
\end{flalign*}

By Taylor's expansion,  $R_1\leq \max _{ n^{r}\leq k\leq n}\{-f_{1,k}(\theta(\omega_{1,k}))b_{n,r}+\sup_x|{f}'_{1,k}(x)| b_{n,r}^2$\}. Note that $\theta(\omega_{1,k})\in[\min_{i=1,2}\theta^{(i)},\max_{i=1,2}\theta^{(i)}] $, we have $f_{1,k}(\theta(\omega_{1,k}))\geq c_2$ for all $k$. Hence,   observe $b_{n,r}=o(1)$ and $\sup_x|f'_{1,k}(x)|<\infty$, we have  $R_1\leq -c_2b_{n,r}/2$.
Thus,
\begin{flalign*}&P(  \max _{ n^{r}\leq k\leq n}\hat{F}_{1,k}({\theta}(\omega_{1,k})-b_{n,r})-q>0, i.o.)
\leq P(R_2+R_3\geq c_2b_{n,r}/2, i.o.)
\\\leq &P(\max _{ n^{r}\leq k\leq n}\sup_x|\hat{F}_{1,k}(x)-F_{1,k}(x)|\geq c_2b_{n,r}/4,i.o.)\\&+P(\max _{ n^{r}\leq k\leq n}\sup_{|x-y|\leq b_{n,r}} |\hat{F}_{1,k}(x)-F_{1,k}(x)-\hat{F}_{1,k}(y)-F_{1,k}(y)|\geq c_2b_{n,r}/4,i.o.).
\end{flalign*}

Plug in that $b_{n,r}=4C_{r,\vartheta}/c_2\sqrt{r\log(n)}/n^{r/2}$ with $\vartheta>1,$ the first probability is equivalent to  $$
P(\max _{ n^{r}\leq k\leq n}\sup_x|\hat{F}_{1,k}(x)-F_{1,k}(x)|\geq C_{r,\vartheta}\sqrt{r\log(n)}/n^{r/2},i.o.),
$$
which is zero by Lemma \ref{lem_fluc} and the Borel-Cantelli lemma.

Similarly, note that $b_{n,r}=o(1)$, hence for any $\iota>0$, and $\vartheta=2$, $$C_{r,\vartheta,\iota}\frac{\sqrt{b_{n,r}^{2/(2+\iota)}r\log(n)}}{n^{r/2}}\ll C_{r,\vartheta}\sqrt{r\log(n)}/n^{r/2}=c_2b_{n,r}/4. $$ Therefore, by Lemma \ref{lem_fluc} and the Borel-Cantelli lemma, the second probability  is also zero.

The second argument is clear by noting $n^{d_0r/2}=o(n^{r/2}/\sqrt{\log(n)})$ for $d_0<1$.
\end{proof}

\noindent\textsc{Proof of Lemma \ref{lem_small}}
\begin{proof}
By the fact that $n^{-1/2}\max_{1\leq k< n^r} k{\theta}(\omega_{1,k})\leq n^{r-1/2}\max_{i=1,2}|\theta^{(i)}|=o(1)$, it suffices to show that 
$$
n^{r-1/2}|\hat{\theta}(\omega_{1,k})|\leq n^{r-1/2}\max_{1\leq i\leq n^r}|X_i|=o_p(1).$$

In fact, for any $\epsilon>0$,
\begin{flalign*}
P(n^{r-1/2}\max_{1\leq i\leq n^r}|X_i|>\epsilon)=&P(\max_{1\leq i\leq n^r}|X_i|>n^{1/2-r} \epsilon)]
\\\leq & \sum_{i=1}^{n^r} P(|X_i|>n^{1/2-r} \epsilon)]
\\\leq & {n^r}C[n^{1/2-r} \epsilon]^{-\lambda}
\\=& C n^{\lambda(r-1/2)+r}=o(1),
\end{flalign*}
by the assumption that $r(\lambda+1)-\lambda/2<0$.
\end{proof}

\noindent\textsc{Proof of Lemma \ref{lem_replace}}
\begin{proof}
For any $\epsilon>0$, choose $2/3<\delta<3/4r+1/2$ and that $r<\delta<1$, then \begin{flalign*}
&P(n^{-1/2}\max_{n^r\leq k\leq  n} k\Big| \hat{F}_{1,k}(\hat{\theta}(\omega_{1,k}))-F_{1,k}(\hat{\theta}(\omega_{1,k}))-\hat{F}_{1,k}(\theta(\omega_{1,k}))+F_{1,k}(\theta(\omega_{1,k}))\Big|>\epsilon)\\\leq& P(n^{-1/2}\max_{ n^{\delta}\leq k\leq  n} k\Big| \hat{F}_{1,k}(\hat{\theta}(\omega_{1,k}))-F_{1,k}(\hat{\theta}(\omega_{1,k}))-\hat{F}_{1,k}(\theta(\omega_{1,k}))+F_{1,k}(\theta(\omega_{1,k}))\Big|>\epsilon)\\&+
P(n^{-1/2}\max_{n^{r}\leq k< n^{\delta}} k\Big| \hat{F}_{1,k}(\hat{\theta}(\omega_{1,k}))-F_{1,k}(\hat{\theta}(\omega_{1,k}))-\hat{F}_{1,k}(\theta(\omega_{1,k}))+F_{1,k}(\theta(\omega_{1,k}))\Big|>\epsilon).
\end{flalign*}
We need to bound these two probabilities separately. 

Let $d_1\in(0,1)$ such that $\delta/2+d_1\delta/4>1/2$, and then fix a small $\iota_1>0$ such that $d_1(2+\iota_1)<2$. Let $a_{n,\delta}=n^{-d_1(2+\iota_1)\delta/4}$. The first part is bounded by 
\begin{flalign*}
&P(\max_{ n^{\delta}\leq k\leq  n}\sup_{|x-y|\leq a_{n,\delta}} \Big| \hat{F}_{1,k}(x)-F_{1,k}(x)-\hat{F}_{1,k}(y)+F_{1,k}(y)\Big|>n^{-1/2}\epsilon)+P(\max _{n^{\delta}\leq k\leq  n}\left|\hat{\theta}(\omega_{1,k})-{\theta}(\omega_{1,k})\right|\geq a_{n,\delta})\\
\leq& P(\max_{ n^{\delta}\leq k\leq  n}\sup_{|x-y|\leq a_{n,\delta}} \Big| \hat{F}_{1,k}(x)-F_{1,k}(x)-\hat{F}_{1,k}(y)+F_{1,k}(y)\Big|>C_{\delta,\vartheta,\iota_1}\frac{\sqrt{a_{n,\delta}^{2/(2+\iota_1)}\log (n^{\delta})}}{n^{\delta/2}})\\&+P(\max _{n^{\delta}\leq k\leq  n}\left|\hat{\theta}(\omega_{1,k})-{\theta}(\omega_{1,k})\right|\geq a_{n,\delta})\\
\leq &O(n^{-\vartheta})+o(1)=o(1),
\end{flalign*}
where the first inequality holds by noting $n^{-1/2}\epsilon>C_{\delta,\vartheta,\iota_1}\frac{\sqrt{\log(n^{\delta})}}{n^{\delta/2+d_1\delta/4}}=C_{\delta,\vartheta,\iota_1}\frac{\sqrt{a_{n,\delta}^{2/(2+\iota_1)}\log (n^{\delta})}}{n^{\delta/2}},$ and second by Lemma \ref{lem_fluc} with $\vartheta>1$, and (\ref{claim}) in Lemma \ref{lem_bound}.

The second part is similar. Note that we can choose a constant $0<d_2<1$ (close to 1) such that $\delta-1 / 2<\left(d_2 / 4+1 / 2\right) r$. In addition, choose $\iota_2$ such that $d_2(2+\iota_2)<2$.   Denote $a_{n,r}=n^{-d_2(2+\iota_2)r/4}$, then the second part is bounded by 

\begin{flalign*}
&P(\max_{ n^r\leq k\leq n^{\delta}}\sup_{|x-y|\leq a_{n,r}} \Big| \hat{F}_{1,k}(x)-F_{1,k}(x)-\hat{F}_{1,k}(y)+F_{1,k}(y)\Big|>n^{1/2-\delta}\epsilon)\\&+P(\max _{n^r\leq k\leq n^{\delta}}\left|\hat{\theta}(\omega_{1,k})-{\theta}(\omega_{1,k})\right|\geq a_{n,r})\\
\leq& P(\max_{ n^r\leq k\leq n^{\delta}}\sup_{|x-y|\leq a_{n,r}} \Big| \hat{F}_{1,k}(x)-F_{1,k}(x)-\hat{F}_{1,k}(y)+F_{1,k}(y)\Big|>C_{r,\vartheta,\iota_2}\frac{\sqrt{a_{n,r}^{2/(2+\iota_2)}\log (n^r)}}{n^{r/2}})\\+&P(\max _{n^r\leq k\leq n^{\delta}}\left|\hat{\theta}(\omega_{1,k})-{\theta}(\omega_{1,k})\right|> a_{n,r})\\
\leq &O(n^{-\vartheta})+o(1)=o(1),
\end{flalign*}
where the fact that $n^{1/2-\delta}\epsilon>C_{r,\vartheta,\iota_2}\frac{\sqrt{\log(n^r)}}{n^{r/2+d_2r/4}}=C_{r,\vartheta,\iota_2}\frac{\sqrt{a_{n,r}^{2/(2+\iota_2)}\log(n^r)}}{n^{r/2}}$ is used.
\end{proof}

\subsubsection{Results regarding \Cref{thm_multiple_quant}}

Denote $\hat{F}_{a,b}(x)$ be the empirical CDF based on observations $\{Y_t\}_{t=a}^b$, and $F_{a,b}(x)$ as the true mixture CDF with density function ${f}_a^b(x)$. The proof of Theorem \ref{thm_multiple_quant} also requires four lemmas, which are analogous to Lemma \ref{lem_fluc}-Lemma \ref{lem_replace} used in the proof of \Cref{thm_single_quant}.
\begin{lemma}\label{lem_fluc2}
For all $0<r<1$, and  $\vartheta>1$, there exists constants $C_{r, \vartheta}>0$ and $K$ such that 
$$P\left(\max _{\substack{1 \leq a<b \leq n \\|b-a| \geq n^{r}}} \sup _{x }\left|\hat{F}_{a,b}(x)-F_{a,b}(x)\right|>C_{r, \vartheta} \frac{\sqrt{ r\log (n)}}{n^{r / 2}}\right) \leq Kn^{-\vartheta} .$$
Furthermore, for any $\iota>0$, there exists a constant $C_{r, \vartheta,\iota}>0$
$$P\left(\max _{\substack{1 \leq a<b \leq n \\|b-a| \geq n^{r}}} \sup _{|x-y|\leq  b_{n,r} }\left|\hat{F}_{a,b}(x)-{F}_{a,b}(x)-\hat{F}_{a,b}(y)+{F}_{a,b}(y)\right|>C_{r, \vartheta,\iota}\frac{\sqrt{b_{n,r}^{2/(2+\iota)}r \log (n)}}{n^{r / 2}}\right) \leq Kn^{-\vartheta},$$
where $b_{n,r}$ satisfies $b_{n,r}=o(1)$ and  $\log^5(n)=o(n^rb_{n,r}^{2/(2+\iota)})$.
\end{lemma}
\begin{proof}
Similar to the proof of Lemma \ref{lem_fluc}, hence omitted.
\end{proof}

\begin{lemma}\label{lem_bound2}
Suppose for some constant $c_2>0$,  $f^{(i)}(x)\geq c_2$ for $x\in[\min_{i}\theta^{(i)},\max_{i}\theta^{(i)}]$, and $\sup_x|f^{(i)'}(x)|<\infty$ for $i=1,\cdots m_o+1$. Let $b_{n,r}=4C_{r,\vartheta}/c_2\sqrt{r\log(n)}/n^{r/2}$ with $\vartheta>1$, for all  $0<r<1$, $$P(  \max _{\substack{1 \leq a<b \leq n \\|b-a| \geq n^{r}}} \left|\hat{\theta}(\omega_{a,b})-{\theta}(\omega_{a,b})\right| \geq b_{n, r}, i.o.)=0.$$

In addition,  for all $0<r<1$ and $0<d_0<1$,
\begin{equation}\label{claim}
n^{d_0r/2} \max _{\substack{1 \leq a<b \leq n \\|b-a| \geq n^{r}}} \left|\hat{\theta}(\omega_{a,b})-{\theta}(\omega_{a,b})\right|=o_p(1).
\end{equation}
\end{lemma}
\begin{proof}
Similar to the proof of Lemma \ref{lem_bound}, hence omitted.
\end{proof}

\begin{lemma}\label{lem_small2} 
For $0<r<1/2-1/\lambda$, 
$$n^{-1/2} \max _{\substack{1 \leq a<b \leq n \\|b-a| < n^{r}}} (b-a) |\hat{\theta}(\omega_{a,b})-{\theta}(\omega_{a,b})|=o_p(1).$$
\end{lemma}
\begin{proof}
By the fact that $\theta(\omega_{a,b})\in[\min_i\theta_i,\max_i\theta_i]$, it suffices to show that 
$$
n^{r-1/2}|\hat{\theta}(\omega_{a,b})|\leq n^{r-1/2}\max_{1\leq i\leq n}|X_i|=o_p(1).$$
Here, because $a,b$ may take any values in $\{1,\cdots,n\}$, we need a union bound to control for both indices. In Lemma \ref{lem_small}, only the end index varies, and  is strictly smaller than $n^r$.

For any $\epsilon>0$,
\begin{flalign*}
P(n^{r-1/2}\max_{1\leq i\leq n}|X_i|>\epsilon)=&P(\max_{1\leq i\leq n}|X_i|>n^{1/2-r} \epsilon)]
\\\leq & \sum_{i=1}^{n} P(|X_i|>n^{1/2-r} \epsilon)]
\\\leq & {n}C[n^{1/2-r} \epsilon]^{-\lambda}
\\=& C n^{\lambda(r-1/2)+1}=o(1),
\end{flalign*}
by the assumption that $\lambda(r-1/2)<-1$.
\end{proof}

\begin{lemma}\label{lem_replace2}
For $2/9<r<1$, $$n^{-1/2}\max _{\substack{1 \leq a<b \leq n \\|b-a| < n^{r}}} (b-a)\Big| \hat{F}_{a,b}(\hat{\theta}(\omega_{a,b}))-F_{a,b}(\hat{\theta}(\omega_{a,b}))-\hat{F}_{a,b}(\theta(\omega_{a,b}))+F_{a,b}(\theta(\omega_{a,b}))\Big|=o_p(1).$$
\end{lemma}
\begin{proof}
Similar to the proof of Lemma \ref{lem_replace}, hence omitted.
\end{proof}

\bigskip

\noindent\textsc{\textbf{Proof of Theorem \ref{thm_multiple_quant}}:}
\begin{proof}
It is similar to the proof of Theorem \ref{thm_single_quant}. The only difference is that we need $2/9<r<1/2-1/\lambda$ instead of $2/9<r<\frac{\lambda}{2(\lambda+1)}$ when applying Lemma \ref{lem_small2}. Note $2/9<r<1/2-1/\lambda$ requires $\lambda >18/5$, hence Assumption \ref{ass_tail2} instead of Assumption \ref{ass_tail} is imposed.
\end{proof}


\newpage
\section{Consistency of SNCP for a general univariate parameter}\label{sec:consistencygeneral}

\subsection{Proof of Theorem \ref{thm_onechange}}
The proof of $\mathrm{(i)}$ is similar to the proof of Theorem 3.1 in \cite{shao2010testing}, and we can show that $\max_{k\in\{1,\cdots,n-1\}}T_n(k)=O_p^s(1)$. Therefore, we mainly focus on $\mathrm{(ii)}$. 

\noindent Let $M_{n1}=\{k|k_1-k>\iota_n\}$ and $M_{n2}=\{k|k-k_1>\iota_n\}$. 
Since
$$
\max_{k\in\{k|k_1-\iota_n\leq k\leq k_1+\iota_n\}}T_n(k)\geq T_n(k_1),
$$
by the symmetricity of $M_{n1}$ and $M_{n2}$, it suffices to show that  $n^{-1}\delta^{-2}\max_{k\in M_{n1}}T_n(k)=o_p(1)$ and $T_n(k_1)\geq O_p^s(n\delta^2)$,  respectively.

The proof lies in showing the following intermediate results:
\begin{enumerate}[(1)]
\item $\max_{k\in M_{n1}}D_n(k)^2 \leq Cn\delta^2+O_p(1)$;
\item $\max_{k\in M_{n1}}V_{n}(k)^{-1}\leq O_p(n[\delta\iota_n]^{-2})$;
\item $V_n(k_1)^{-1}=O_p^s(1)$ and $D_n(k_1)^2= Cn\delta^2+O_p(1)$.
\end{enumerate}

(1) For $k\leq k_1$, we have $\omega_{1,k}=(1,0)^{\top}$ and $\omega_{k+1,n}=(\frac{k_1-k}{n-k},\frac{n-k_1}{n-k})^{\top}$. Then
\begin{flalign*}
&D_n(k)\\=&\frac{k(n-k)}{n^{3/2}}(\widehat{\theta}_{1,k}-\widehat{\theta}_{k+1,n})
\\=&\frac{k(n-k)}{n^{3/2}} [\theta_1+\frac{1}{k}\sum_{t=1}^{k}\xi_1(Y_t^{(1)},\omega_{1,k})+r_{1,k}(\omega_{1,k})]
\\&-\frac{k(n-k)}{n^{3/2}}[\theta(\omega_{k+1,n})+\frac{1}{n-k}\sum_{t=k+1}^{k_1}\xi_1(Y_{t}^{(1)},\omega_{k+1,n})+\frac{1}{n-k}\sum_{t=k_1+1}^{n}\xi_2(Y_{t}^{(2)},\omega_{k+1,n})+r_{k+1,n}(\omega_{k+1,n})]
\\=&\frac{k(n-k)}{n^{3/2}}[\theta_1-\theta(\omega_{k+1,n})] 
\\&+\frac{k(n-k)}{n^{3/2}}[\frac{1}{k}\sum_{t=1}^{k}\xi_1(Y_t^{(1)},\omega_{1,k})-\frac{1}{n-k}\sum_{t=k+1}^{k_1}\xi_1(Y_{t}^{(1)},\omega_{k+1,n})-\frac{1}{n-k}\sum_{t=k_1+1}^{n}\xi_2(Y_{t}^{(2)},{\omega_{k+1,n}})]
\\&+\frac{k(n-k)}{n^{3/2}}[r_{1,k}(\omega_{1,k})-r_{k+1,n}(\omega_{k+1,n})]\\:=&D_{n1}(k)+D_{n2}(k)+D_{n3}(k).
\end{flalign*}

\noindent Note  
$
D_n(k)^2\leq 3\sum_{i=1}^{3}D_{ni}(k)^2,
$
where under Assumption \ref{ass_influence} and \ref{ass_remainder}, it is easy to see   $\max_{k\in M_{n1}}D_{n2}(k)^2\leq O_p(1)$ and $\max_{k\in M_{n1}}D_{n3}(k)^2=o_p(1)$. By Assumption \ref{ass_theta}, we have $\max_{k\in M_{n1}}D_{n1}(k)^2\leq Cn\delta^2$, the result is clear.

(2) We decompose 
$
V_n(k)=L_{n}(k)+R_{n1}(k)+R_{n2}(k),
$
where 
\begin{flalign*}
L_{n}(k)=&\sum_{i=1}^{k}\frac{i^2(k-i)^2}{n^2k^2}(\widehat{\theta}_{1,i}-\widehat{\theta}_{i+1,k})^2,\\
R_{n1}(k)=&\sum_{i=k+1}^{k_1}\frac{(n-i+1)^2(i-k-1)^2}{n^2(n-k)^2}(\widehat{\theta}_{i,n}-\widehat{\theta}_{k+1,i-1})^2,\\
R_{n2}(k)=&\sum_{i=k_1+1}^{n}\frac{(n-i+1)^2(i-k-1)^2}{n^2(n-k)^2}(\widehat{\theta}_{i,n}-\widehat{\theta}_{k+1,i-1})^2.
\end{flalign*}
For $k\leq k_1$, we further decompose $R_{n2}(k)$ such that
\begin{flalign*}
&R_{n2}(k)\\=&n^{-2}\sum_{i=k_1+1}^{n}\frac{(n-i+1)^2(i-k-1)^2}{(n-k)^2}\Big\{\Big[\theta_2-\theta(\omega_{k+1,i-1})\Big]\\&+\Big[\frac{1}{n-i+1}\sum_{t=i}^{n}\xi_2(Y_t^{(2)},\omega_{i,n})-\frac{1}{i-k-1}\sum_{t=k+1}^{k_1}\xi_1(Y_t^{(1)},\omega_{k+1,i-1})-\frac{1}{i-k-1}\sum_{t=k_1+1}^{i-1}\xi_2(Y_t^{(2)},\omega_{k+1,i-1})\Big]\\&+\Big[r_{i,n}(\omega_{i,n})-r_{k+1,i-1}(\omega_{k+1,i-1})\Big]\Big\}^2,
\end{flalign*}
where $\omega_{k+1,i-1}=(\frac{k_1-k}{i-k-1},\frac{i-k_1-1}{i-k-1})^{\top}$ and $\omega_{i,n}=(0,1)^{\top}$.

\noindent Note that $V_n(k)\geq R_{n2}(k)$, we have $R_{n2}(k)^{-1}\geq V_n(k)^{-1}$, hence it suffices to show $\max_{k\in M_{n1}}R_{n2}(k)^{-1}\leq O_p(n[\delta\iota_n]^{-2})$.

\noindent Denote $A_k=\sum_{t=k_1+1}^{n}a_t(k)^2$, $B_k=\sum_{t=k_1+1}^{n}[b_{t}(k)+c_{t}(k)]^2$ and $C_k=-2\sum_{t=k_1+1}^{n}a_t(k)[b_{t}(k)+c_{t}(k)]$, where 
\begin{flalign*}
a_t(k)=&\frac{(n-i+1)(i-k-1)}{(n-k)}\Big[\theta_2-\theta(\omega_{k+1,i-1})\Big],\\
b_{t}(k)=&\frac{(n-i+1)(i-k-1)}{(n-k)}\Big\{\frac{1}{n-i+1}\sum_{t=i}^{n}\xi_2(Y_t^{(2)},\omega_{i,n})\\&-\frac{1}{i-k-1}\Big[\sum_{t=k+1}^{k_1}\xi_1(Y_t^{(1)},\omega_{k+1,i-1})+\sum_{t=k_1+1}^{i-1}\xi_2(Y_t^{(2)},\omega_{k+1,i-1})\Big]\Big\},\\
c_{t}(k)=&\frac{(n-i+1)(i-k-1)}{(n-k)}\Big[r_{i,n}(\omega_{i,n})-r_{k+1,i-1}(\omega_{k+1,i-1})\Big],
\end{flalign*}
then we get 
$$
R_{n2}(k)=n^{-2}[A_k+B_k+C_k].
$$
By Hua's identity, we obtain 
\begin{flalign}\label{hua}
R_{n2}(k)^{-1}=n^2A_k^{-1}-n^{2}A_k^{-1}[1+(B_k+C_k)^{-1}A_k]^{-1}.
\end{flalign}
Note that there exists some constant $0<c_1<c_2<\infty$ independent of $k$ such that 
$
c_1n^3\leq \sum_{i=k_1+1}^{n}(n-i+1)^2\leq c_2n^3,
$
and by Assumption \ref{ass_theta}, we have 
$
C_1\frac{(k_1-k)^2}{(i-k-1)^2}\leq \Big[\theta_2-\theta(\omega_{k+1,i-1})\Big]^2\leq C_2\frac{(k_1-k)^2}{(i-k-1)^2},
$
then it follows that 
$$
c_1C_1\frac{(k_1-k)^2n^3\delta^2}{(n-k)^2}\leq A_k\leq c_2C_2\frac{(k_1-k)^2n^3\delta^2}{(n-k)^2},
$$
or 
equivalently, \begin{equation}\label{supAinv}
\max_{k\in M_{n1}}(A_k^{-1})\leq Cn^{-1}[\delta\iota_n]^{-2}.
\end{equation}
In addition, by (\ref{hua}), we have 
$
R_{n2}(k)^{-1}=\frac{n^2A_k^{-1}}{1+A_k^{-1}[B_k+C_k]},
$
that is 
\begin{flalign}\label{supV3inv}
\max_{k\in M_{n1}}R_{n2}(k)^{-1}\leq  \frac{Cn[\delta\iota_n]^{-2}}{1+\min_{k\in M_{n1}}\{A_k^{-1}[B_k+C_k]\}}\leq  \frac{Cn[\delta\iota_n]^{-2}}{1-2\max_{k\in M_{n1}}\{A_k^{-1/2}B_k^{1/2}\}}.
\end{flalign}
where the second inequality holds by Cauchy-Schwarz inequality, that $-2A_k^{-1/2}B_k^{1/2}\leq A_k^{-1}C_k\leq2 A_k^{-1/2}B_k^{1/2}$ and $A_k^{-1}B_k\geq 0$.
Using (\ref{supAinv}), we have 
$$
\max_{k\in M_{n1}}\{A_k^{-1}B_k\}\leq \max_{k\in M_{n1}} (A_k^{-1})(\max_{k\in M_{n1}}B_k)\leq Cn[\delta\iota_n]^{-2}\max_{k\in M_{n1}}n^{-2}B_k.
$$
Under Assumptions  \ref{ass_influence} and \ref{ass_remainder}, we can show that 
\begin{flalign*}
n^{-2}\max_{k\in M_{n1}}B_k=O_p(1).
\end{flalign*}
Hence, using the fact that $n\delta^{-2}\iota_n^{-2}=o(1)$, we obtain $\max_{k\in M_{n1}}\{A_k^{-1}B_k\}=o_p(1)$. Therefore, in view of (\ref{supV3inv}), we obtain 
\begin{equation}\label{V3invbound}
\max_{k\in M_{n1}}R_{n2}(k)^{-1}\leq O_p(n[\delta\iota_n]^{-2}).
\end{equation}

(3)
Under Assumption \ref{ass_influence} and \ref{ass_remainder}, we can show that $V_{n}(k_1)\to_{D} V(\tau_1),$
where 
$
V(\tau_1)=\sigma_1^2 \int_0^{\tau_1} \{B^{(1)}(r)-r/u B^{(1)}(\tau_1)\}^2dr +\sigma_2^2 \int_{\tau_1}^{1} \{B^{(2)}(1)-B^{(2)}(r)-(1-r)/(1-\tau_1)(B^{(2)}(1)-B^{(2)}(\tau_1))\}^2dr.
$
In addition, the similar arguments used in proving (1) give us the second part of (3).\qed

\subsection{Proof of Theorem \ref{thm1}}\label{subsec: outline_proof}
To ease the presentation, we assume $\theta_{a,b}=(\theta_1,\cdots,\theta_{m_o+1})\omega_{a,b}$ in Assumption \ref{ass_no3} as the residual $o(1/\sqrt{n})$ is asymptotically negligible as long as $\iota_n^{-2}\delta^{-2}n\to 0$ as $n\to\infty$. The proof under Assumption \ref{ass_no3star} is similar based on the argument in the proof of Theorem \ref{thm_onechange}.


Let $M_n=\{k||k-k_i|>\iota_n, \forall i=1,\cdots, m_o\}$ denote the set of time points that are at least of $\iota_n$ points away from the true change-point locations. 
The basic idea of the consistency proof is as follows. Based on the location of $k$ and its local window $(t_1,t_2)\in H_{1:n}(k)$, our analysis of $T_n(t_1,k,t_2)$ boils down to the following three scenarios:

\begin{enumerate}
\item When $k$ is a true change-point $k_i$, using the fact that $\epsilon<\epsilon_o$, we have that the smallest local-window $(k_i-h+1,k_i+h)$ contains only one single change-point $k_i$. Thus, the test statistic $T_{1,n}(k_i)$ at $k_i$ is at least as large as $T_n(k_i-h+1,k_i,k_i+h)$, which is shown to be of order $O_p^s(n\delta^2)$ due to the inflation of the contrast statistic $D_n(k_i-h+1,k_i,k_i+h)$;
\item When $k\in M_n$ and the local-window $(t_1,t_2)$ contains no change-points, we use the invariance principle to show that  $T_n(t_1,k,t_2)$ is of order $O_p(1)$, a smaller order compared with $O_p(n\delta^2)$;
\item When $k\in M_n$ and  the local-window $(t_1,t_2)$ contains some change-points, we further show that the presence of the change-points causes the inflation of the self-normalizer $V_n(t_1,k,t_2)$, which in turn causes $T_n(t_1,k,t_2)$ to take a smaller order compared with $O_p(n\delta^2)$.
\end{enumerate}
Note that the three scenarios imply that asymptotically the distance between the estimated change-point $\widehat{k}_i$ and the corresponding true location $k_i$ is of order $O(\iota_n)$. Together with the fact that $\epsilon < \epsilon_o$ and $\iota_n=o(n)$, it can be shown that the impact of the estimation error $|\widehat{k}_i-k_i|$ is negligible for the subsequent change-point estimation, which ensures the same convergence rate $\iota_n$ for all the estimated change-points by SNCP.


The major part of the proof focuses on establishing the result in scenario 3, where $k\in M_n$ and the local-window $(t_1,t_2)$ contains some change-points. Different analysis is required depending on the number and relative locations of the change-points contained in the local-window $(t_1, t_2)$, which makes the proof rather complicated.


The proof proceeds step-by-step, from the case of single change-point to the case of two change-points, then to the case of three  or more change-points.  Every step builds upon the result obtained in the previous step. For example, when there are two change-points, but if the local-window  $(t_1,t_2)$ of $k$ only contains one change-point or both change-points appear on the same side of  $k$ (such as $k_1<k_2<k$), then we can show the analysis reduce to the case of single change-point.

\subsubsection{No change-point}
When $m_0=0$, note $\omega_{a,b}\equiv1$ (as we only have one stationary segment), it follows that 
$$
\widehat{\theta}_{a,b}=\theta_1+\bar{\xi}_{a,b}(1)+r_{a,b}(1).$$
Then we have  
\begin{flalign*}
D_n(t_1,k,t_2)=&\frac{(k-t_1+1)(t_2-k)}{(t_2-t_1+1)^{3/2}}\Big([\bar{\xi}_{t_1,k}(1)-\bar{\xi}_{k+1,t_2}(1)]+[r_{t_1,k}(1)-r_{k+1,t_2}(1)]\Big),\\
L_n(t_1,k,t_2)=&\sum_{i=t_1}^{k}\frac{(i-t_1+1)^2(k-i)^2}{(t_2-t_1+1)^{2}(k-t_1+1)^2}\Big([\bar{\xi}_{t_1,i}(1)-\bar{\xi}_{i+1,k}(1)]+[r_{t_1,i}(1)-r_{i+1,k}(1)]\Big)^2,\\
R_n(t_1,k,t_2)=&\sum_{i=k+1}^{t_2}\frac{(t_2-i+1)^2(i-1-k)^2}{(t_2-t_1+1)^{2}(t_2-k)^2}\Big([\bar{\xi}_{i,t_2}(1)-\bar{\xi}_{k+1,i-1}(1)]+[r_{i,t_2}(1)-r_{k+1,i-1}(1)]\Big)^2,\\
V_n(t_1,k,t_2)=&L_n(t_1,k,t_2)+R_n(t_1,k,t_2).
\end{flalign*}
Note that under Assumption \ref{ass_no2}, we can show 
$$
\frac{(k-t_1+1)(t_2-k)}{(t_2-t_1+1)^{3/2}}\Big|r_{t_1,k}(1)-r_{k+1,t_2}(1)\Big|\leq \frac{(k-t_1+1)}{(t_2-t_1+1)^{1/2}}\Big|r_{t_1,k}(1)\Big|+\frac{(t_2-k)}{(t_2-t_1+1)^{1/2}}\Big|r_{k+1,t_2}(1)\Big|=o_p(1),
$$
and 
$$\frac{(i-t_1+1)(k-i)}{(t_2-t_1+1)(k-t_1+1)}\Big|r_{t_1,i}(1)-r_{i+1,k}(1)\Big|\leq \Big|\frac{(i-t_1+1)}{(t_2-t_1+1)}r_{t_1,i}(1)\Big|+\Big|\frac{(k-i)}{(t_2-t_1+1)}r_{i+1,k}(1)\Big|=o_p(n^{-1/2})$$
uniformly for $t_1\leq i\leq k$, and 
$$\frac{(t_2-i+1)(i-1-k)}{(t_2-t_1+1)(t_2-k)}\Big|r_{i,t_2}(1)-r_{k+1,i-1}(1)\Big|\leq \Big|\frac{(t_2-i+1)}{(t_2-t_1+1)}r_{i,t_2}(1)\Big|+\Big|\frac{(i-1-k)}{(t_2-t_1+1)}r_{k+1,i-1}(1)\Big|=o_p(n^{-1/2}),$$
uniformly for $k+1 \leq i \leq t_2$.

\noindent Continuous mapping theorem and Assumption \ref{ass_no1} indicates that 
$$\max_{k=h,\ldots,n-h} T_{1,n}(k) \to_D \sup_{u\in(\epsilon,1-\epsilon)}\max_{(u_1,u_2)\in H_{\epsilon}(u)}\frac{D(u_1,u,u_2)^2}{V(u_1,u,u_2)}=O_p^s(1),$$
where \begin{align*}
D(u_1,u,u_2)&=\frac{1}{(u_2-u_1)^{1/2}}\left\{B(u)-B(u_1)-\frac{u-u_1}{u_2-u_1}(B(u_2)-B(u_1)) \right\},\\
L(u_1,u,u_2)&=\frac{1}{(u_2-u_1)^2}\int_{u_1}^{u}\left[B(s)-B(u_1)-\frac{s-u_1}{u-u_1}(B(u)-B(u_1)) \right]^2ds,\\
R(u_1,u,u_2)&=\frac{1}{(u_2-u_1)^2}\int_{u}^{u_2}\left[B(u_2)-B(s)-\frac{u_2-s}{u_2-u}(B(u_2)-B(u)) \right]^2ds, \\
V(u_1,u,u_2)&=L(u_1,u,u_2)+R(u_1,u,u_2).
\end{align*}
Therefore, 
\begin{equation}\label{eq:first_round_SNCP_0}
P(\widehat{m}=0)=P(\max_{k=h,\ldots,n-h} T_{1,n}(k)<K_n)\to1.
\end{equation}

In the following, we first analyze the behavior of $\widehat{k}=\arg\max_{k=h,\ldots,n-h} T_{1,n}(k)$ generated by applying SNCP to the time series $\{Y_t\}_{t=1}^{n}$ when $m_o>0$ and prove that
\begin{align}
\label{eq:first_round_SNCP}
P\left(\max_{k=h,\ldots,n-h} T_{1,n}(k)>K_n \text{ and }\min_{1\leq i\leq m_o}|k_i-\widehat{k}|<\iota_n\right) \to 1.
\end{align}
In other words, when $m_o>0$, SNCP can detect the change and the estimated change-point $\widehat{k}$ converges to one of the true change-points with rate $\iota_n$.

\noindent Note $\min_{1\leq i\leq m_o+1}(k_i-k_{i-1})>\lfloor\epsilon n\rfloor=h$, thus we can easily show that
\begin{align*}
\frac{T_{1,n}(k_i)}{n\delta^2} \geq \frac{D_{n}(k_i-h+1,k_i,k_i+h)^2}{V_{n}(k_i-h+1,k_i,k_i+h) n\delta^2}
\to_D \frac{\epsilon c_i^2}{8V(\tau_i-\epsilon,\tau_i,\tau_i+\epsilon)}=O_p^{s}(1), \text{ for } i=1,\ldots,m_o.
\end{align*}
Note that $K_n/(n\delta^2)\to 0$, thus we have $P(\max_{k=h,\ldots,n-h} T_{1,n}(k)>K_n) \to 1.$

\noindent Therefore, to prove \eqref{eq:first_round_SNCP}, we only need to focus on the set $$M_n=\{k| \min_{i=1,\ldots,m_o}|k-k_i|>\iota_n \},$$
and show that
\begin{align}
\label{eq:first_round_SNCP_null}
\max_{k \in M_n} \frac{T_{1,n}(k)}{n\delta^2}=o_p(1).
\end{align}
In the following, we prove \eqref{eq:first_round_SNCP_null} progressively for one~($m_o=1$), two~($m_o=2$) and multiple change-points~($m_o\geq 3$) cases, by building proofs gradually upon previous arguments.

\subsubsection{One change-point}
We can decompose $M_n=M_{n1}\bigcup M_{n2}$, where $M_{n1}=\{k|k_1-k>\iota_n\}$ and $M_{n2}=\{k|k-k_1>\iota_n\}$. By symmetry, we only need to prove the result for $M_{n1}$. We  choose $c_n$ such that
\begin{equation}\label{cn}
c_n^{-3}n^4\delta^{-2}\iota_n^{-2}\to 0 \quad\text{and}\quad n^{-1}c_n\to 0 \quad\text{as}\quad n\to \infty.
\end{equation}
We mention here that such choice of $c_n$ is always possible, for example, we can choose $c_n=n^{5/4}(\iota_n\delta)^{-1/2}$.

\noindent Then we decompose $H_{1:n}(k)$ as:
\begin{flalign*}
&(1)~H_{1:n}^0(k)=H_{1:n}(k)\cap \{(t_1,t_2)|t_2\leq k_1\},\\
&(2)~H_{1:n}^1(k,c_n)=H_{1:n}(k)\cap \{(t_1,t_2)|k_1+1\leq t_2\leq k_1+c_n\},\\
&(3)~H_{1:n}^2(k,c_n)=H_{1:n}(k)\cap \{(t_1,t_2)| t_2> k_1+c_n\}.
\end{flalign*}

$(\mathrm{1})$ On $H_{1:n}^0(k)$, there is no change-point, hence it follows
\begin{flalign*}
\max_{k\in M_{n1}}\max_{(t_1,t_2)\in H_{1:n}^0(k)} \frac{D_n(t_1,k,t_2)^2}{V_n(t_1,k,t_2)n\delta^2}=\frac{O_p(1)}{n\delta^2}=o_p(1).
\end{flalign*}

On $H_{1:n}^1(k,c_n)$ and $H_{1:n}^2(k,c_n)$ , we have  
\begin{flalign*}
\begin{split}
&D_n(t_1,k,t_2)\\=&\frac{(k-t_1+1)(t_2-k)}{(t_2-t_1+1)^{3/2}}(\widehat{\theta}_{t_1,k}-\widehat{\theta}_{k+1,t_2})
\\=&\frac{(k-t_1+1)(t_2-k)}{(t_2-t_1+1)^{3/2}}\Big[\theta_1-(\frac{k_1-k}{t_2-k}\theta_1+\frac{t_2-k_1}{t_2-k}\theta_2)\Big]\\&+\frac{(k-t_1+1)(t_2-k)}{(t_2-t_1+1)^{3/2}}\Big[\bar{\xi}_{t_1,k}(\omega_{t_1,k})-\bar{\xi}_{k+1,t_2}(\omega_{k+1,t_2})\Big]\\&+\frac{(k-t_1+1)(t_2-k)}{(t_2-t_1+1)^{3/2}}\Big[r_{t_1,k}(\omega_{t_1,k})-r_{k+1,t_2}(\omega_{k+1,t_2})\Big]
\\:=&D_{n1}(t_1,k,t_2)+D_{n2}(t_1,k,t_2)+D_{n3}(t_1,k,t_2),
\end{split}
\end{flalign*}
where $\omega_{t_1,k}=(1,0)^{\top}$ and $\omega_{t_1,t_2}=(\frac{k_1-k}{t_2-k},\frac{t_2-k_1}{t_2-k})^{\top}$.

\noindent Note that \begin{align*}
\max_{k\in M_{n1}}\max_{(t_1,t_2)\in \{H_{1:n}^1(k,c_n)\cup H_{1:n}^2(k,c_n)\}}\frac{D_{n}(t_1,k,t_2)^2}{n\delta^2}
\leq3\sum_{i=1}^{3}\max_{k\in M_{n1}}\max_{(t_1,t_2)\in \{H_{1:n}^1(k,c_n)\cup H_{1:n}^2(k,c_n)\}}\frac{D_{ni}(t_1,k,t_2)^2}{n\delta^2}.
\end{align*}
Under Assumptions \ref{ass_no1} and \ref{ass_no2}, on $H_{1:n}^1(k,c_n)\cup H_{1:n}^2(k,c_n)$, it is easy to see that 
\begin{flalign}\label{Dn23}
\begin{split}
&\max_{k\in M_{n1}}\max_{(t_1,t_2)\in \{H_{1:n}^1(k,c_n)\cup H_{1:n}^2(k,c_n)\}}\frac{D_{n2}(t_1,k,t_2)^2}{n\delta^2}=o_p(1),\\
&\max_{k\in M_{n1}}\max_{(t_1,t_2)\in \{H_{1:n}^1(k,c_n)\cup H_{1:n}^2(k,c_n)\}}\frac{D_{n3}(t_1,k,t_2)^2}{n\delta^2}=o_p(1).
\end{split}
\end{flalign}

$(\mathrm{2})$ On $H_{1:n}^1(k,c_n)$, 
we have 
\begin{align*}
&\max_{k\in M_{n1}}\max_{(t_1,t_2)\in H_{1:n}^1(k,c_n)}\frac{D_{n}(t_1,k,t_2)^2}{V_{n}(t_1,k,t_2)n\delta^2}\leq
\max_{k\in M_{n1}}\max_{(t_1,t_2)\in H_{1:n}^1(k,c_n)}\frac{D_{n}(t_1,k,t_2)^2}{L_{n}(t_1,k,t_2)n\delta^2}
\\\leq&3\sum_{i=1}^{3}\max_{k\in M_{n1}}\max_{(t_1,t_2)\in H_{1:n}^1(k,c_n)}\frac{D_{ni}(t_1,k,t_2)^2}{L_{n}(t_1,k,t_2)n\delta^2}
\\\leq &3\sum_{i=1}^{3}[\max_{k\in M_{n1}}\max_{(t_1,t_2)\in H_{1:n}^1(k,c_n)}\frac{D_{ni}(t_1,k,t_2)^2}{n\delta^2}][\max_{k\in M_{n1}}\max_{(t_1,t_2)\in H_{1:n}^1(k,c_n)}\frac{1}{L_{n}(t_1,k,t_2)}]
\end{align*}
Note it is easy to see that $\max_{k\in M_{n1}}\max_{(t_1,t_2)\in H_{1:n}^1(k,c_n)}[{L_{n}(t_1,k,t_2)}]^{-1}=O_p^s(1)$ and by  (\ref{Dn23}), it suffices to show 
$$
\max_{k\in M_{n1}}\max_{(t_1,t_2)\in H_{1:n}^1(k,c_n)}\frac{D_{n}^{1}(t_1,k,t_2)^2}{n\delta^2}=o(1).
$$
Note that
\begin{flalign*}
\max_{k\in M_{n1}}\max_{(t_1,t_2)\in H_{1:n}^1(k,c_n)}\frac{D_{n1}(t_1,k,t_2)^2}{n\delta^2}\leq &C\max_{k\in M_{n1}}\max_{(t_1,t_2)\in H_{1:n}^1(k,c_n)}	\frac{(t_2-k_1)^2(k-t_1+1)^2 }{n(t_2-t_1+1)^{3}}\\\leq &C\max_{k\in M_{n1}}\max_{(t_1,t_2)\in H_{1:n}^1(k,c_n)}\frac{(c_n)^2n^2 }{n(2\epsilon n)^{3}}=O(\frac{c_n^2}{n^2})=o(1),
\end{flalign*}
where the last equality holds by (\ref{cn}), the result follows. 

$(\mathrm{3})$
On the set $H_{1:n}^2(k,c_n)$,
we focus on $R_n(t_1,k,t_2)$, where 
\begin{flalign*}
R_n(t_1,k,t_2)=&\Big[
\sum_{i=k+1}^{k_1}+\sum_{i=k_1+1}^{t_2}
\Big]\frac{(t_2-i+1)^2(i-1-k)^2}{(t_2-t_1+1)^{2}(t_2-k)^2}(\widehat{\theta}_{i,t_2}-\widehat{\theta}_{k+1,i-1})^2\\:=&R_{n1}(t_1,k,t_2)+R_{n2}(t_1,k,t_2).
\end{flalign*}
Since, \begin{align*}
&\max_{k\in M_{n1}}\max_{(t_1,t_2)\in H_{1:n}^2(k,c_n)}\frac{D_{n}(t_1,k,t_2)^2}{V_{n}(t_1,k,t_2)n\delta^2}\leq
\max_{k\in M_{n1}}\max_{(t_1,t_2)\in H_{1:n}^2(k,c_n)}\frac{D_{n}(t_1,k,t_2)^2}{R^{2}_{n}(t_1,k,t_2)n\delta^2}
\\\leq&3\sum_{i=1}^{3}\max_{k\in M_{n1}}\max_{(t_1,t_2)\in H_{1:n}^2(k,c_n)}\frac{D_{ni}(t_1,k,t_2)^2}{R_{n2}(t_1,k,t_2)n\delta^2}
\\\leq &3\sum_{i=1}^{3}[\max_{k\in M_{n1}}\max_{(t_1,t_2)\in H_{1:n}^2(k,c_n)}\frac{D_{ni}(t_1,k,t_2)^2}{n\delta^2}][\max_{k\in M_{n1}}\max_{(t_1,t_2)\in H_{1:n}^2(k,c_n)}\frac{1}{R_{n2}(t_1,k,t_2)}]
\end{align*}
and 
$
\max_{k\in M_{n1}}\max_{(t_1,t_2)\in H_{1:n}^{2}(k,c_n)}\frac{D_{n1}(t_1,k,t_2)^2}{n\delta^2}\leq O(1),
$
combined with Lemma \ref{lem_1} and (\ref{Dn23}), the result follows.

\subsubsection{Two change-points}
We now consider the case where there are two change-points at $(k_1,k_2)$. We can decompose $M_n=M_{n1}\cup M_{n2}\cup M_{n3}$, where $M_{n1}=\{k|k_1-k>\iota_n\}$, $M_{n2}=\{k|k-k_1>\iota_n \text{ and } k_2-k >\iota_n\}$ and $M_{n3}=\{k|k-k_2>\iota_n\}$. By symmetry, we only need to show the result for $M_{n1}$ and $M_{n2}$.

We first consider the set $M_{n1}=\{k|k_1-k>\iota_n\}$ and analyze the behavior of $T_n(t_1,k,t_2)$ on two subsets of $(t_1,t_2)\in H_{1:n}(k)$. Define
\begin{align*}
H_{1:n}^{1,0}(k)&=H_{1:n}(k)\cap \{(t_1,t_2)|t_2\leq k_2\},\\
H_{1:n}^{1,1}(k)&=H_{1:n}(k)\cap \{(t_1,t_2)|t_2>k_2\}.
\end{align*}

For the behavior of $T_{n}(t_1,k,t_2)$ on $(t_1,t_2)\in H_{1:n}^{1,0}(k)$, it reduces to one change-point case. 

For the behavior of $T_{n}(t_1,k,t_2)$ on $(t_1,t_2)\in H_{1:n}^{1,1}(k)$, we have 
\begin{flalign*}
R_n(t_1,k,t_2)=&\Big[\sum_{i=k+1}^{k_1}+\sum_{i=k_1+1}^{k_2}+\sum_{i=k_2+1}^{t_2}\Big]\frac{(t_2-i+1)^2(i-1-k)^2}{(t_2-t_1+1)^{2}(t_2-k)^2}(\widehat{\theta}_{i,t_2}-\widehat{\theta}_{k+1,i-1})^2\\:=&R_{n1}(t_1,k,t_2)+R_{n2}(t_1,k,t_2)+R_{n3}(t_1,k,t_2).
\end{flalign*}

Similar arguments used in the proof of one change-point case indicates that 
$$
\max_{k\in M_{n1}}\max_{(t_1,t_2)\in H_{1:n}^{1,1}(k)}\frac{[D_n(t_1,k,t_2)]^2}{n\delta^2}\leq O_p(1),
$$
the result follows by Lemma \ref{lem_2} that 
$$
\max_{k\in M_{n1}}\max_{(t_1,t_2)\in H_{1:n}^{1,1}(k)}\frac{T_{n}(t_1,k,t_2)}{n\delta^2}\leq \max_{k\in M_{n1}}\max_{(t_1,t_2)\in H_{1:n}^{1,1}(k)}\frac{[D_n(t_1,k,t_2)]^2}{n\delta^2}\max_{k\in M_{n1}}\max_{(t_1,t_2)\in H_{1:n}^{1,1}(k)}R_{n2}(t_1,k,t_2)^{-1}.
$$

We now focus on the set $M_{n2}=\{k|k-k_1>\iota_n \text{ and } k_2-k >\iota_n\}$ and analyze the behavior of $T_{n}(t_1,k,t_2)$ on three subsets of $(t_1,t_2)\in H_{1:n}(k)$. Define
\begin{align*}
H_{1:n}^{2,0}(k)&=H_{1:n}(k)\cap \{(t_1,t_2)|t_1>k_1\},\\
H_{1:n}^{2,1}(k)&=H_{1:n}(k)\cap \{(t_1,t_2)|t_1\leq k_1,t_2\leq k_2\},\\
H_{1:n}^{2,2}(k)&=H_{1:n}(k)\cap \{(t_1,t_2)|t_1\leq k_1,t_2> k_2\}.
\end{align*}

For the behavior of $T_n(t_1,k,t_2)$ on $(t_1,t_2)\in H_{1:n}^{2,0}(k)$ and $(t_1,t_2)\in H_{1:n}^{2,1}(k)$, it reduces to one change-point case. 
For $H_{1:n}^{2,2}(k)$, we further split it into two subsets
\begin{align*}
&(1)~H_{1:n}^{2,21}(k,c_n)=H_{1:n}^{2,2}(k)\cap \{(t_1,t_2)|t_1\leq k_1-c_n \text{ or }t_2> k_2+c_n\},\\
&(2)~H_{1:n}^{2,22}(k,c_n)=H_{1:n}^{2,2}(k)\cap \{(t_1,t_2)|t_1>k_1-c_n \text{ and }t_2\leq k_2+c_n\},
\end{align*}
where  $c_n$ satisfies (\ref{cn}).

(1) For $k\in M_{n2}$ and $(t_1,t_2)\in H_{1:n}^{2,21}(k,c_n)$, without loss of generality, we assume $t_2>k_2+c_n$, then we have 
\begin{flalign*}
R_n(t_1,k,t_2)=&\Big[\sum_{i=k+1}^{k_2}+\sum_{i=k_2+1}^{t_2}\Big]\frac{(t_2-i+1)^2(i-1-k)^2}{(t_2-t_1+1)^{2}(t_2-k)^2}(\widehat{\theta}_{i,t_2}-\widehat{\theta}_{k+1,i-1})^2\\:=&R_{n1}(t_1,k,t_2)+R_{n2}(t_1,k,t_2).
\end{flalign*}
A similar argument used in Lemma \ref{lem_1} would yield that 
$$
\max_{k\in M_{n2}}\max_{(t_1,t_2)\in H_{1:n}^{2,21}(k,c_n)} [R_{n2}(t_1,k,t_2)]^{-1}\leq (C+o_p(1))\frac{n^4}{\delta^2c_n^3\iota_n^2}=o_p(1).
$$
This is true since $\theta_{i,t_2}=\theta_3$ and $\theta_{k+1,i-1}=\frac{k_2-k}{i-1-k}\theta_2+\frac{i-1-k_2}{i-1-k}\theta_3$, and 
\begin{flalign*}
&\min_{k\in M_{n2}}\min_{(t_1,t_2)\in H_{1:n}^{2,21}(k,c_n)}\sum_{i=k_2+1}^{t_2}\frac{(t_2-i+1)^2(i-1-k)^2}{(t_2-t_1+1)^{2}(t_2-k)^2}({\theta}_{i,t_2}-{\theta}_{k+1,i-1})^2\\=&\min_{k\in M_{n2}}\min_{(t_1,t_2)\in H_{1:n}^{2,21}(k,c_n)}\sum_{i=k_2+1}^{t_2}\frac{(t_2-i+1)^2(k_2-k)^2}{(t_2-t_1+1)^{2}(t_2-k)^2}\delta_2^2>C\frac{(t_2-k_2)^3\iota_n^2\delta^2}{n^4}>C\frac{\iota_n^2c_n^3\delta^2}{n^4}
\end{flalign*}
It is easy to see that $\max_{k\in M_{n2}}\max_{(t_1,t_2)\in H_{1:n}^{2,21}(k,c_n)}[D_n(t_1,k,t_2)]^2=O_p(n\delta^2)$, the result follows.

(2) For $k\in M_{n2}$ and $(t_1,t_2)\in H_{1:n}^{2,22}(k,c_n)$, we have that 
\begin{flalign}\label{Dn2}
\begin{split}
D_n(t_1,k,t_2)=&\frac{(k-t_1+1)(t_2-k)}{(t_2-t_1+1)^{3/2}}(\widehat{\theta}_{t_1,k}-\widehat{\theta}_{k+1,t_2})
\\=&\frac{(k-t_1+1)(t_2-k)}{(t_2-t_1+1)^{3/2}}\Big[\frac{k_1-t_1+1}{k-t_1+1}\theta_1+\frac{k-k_1}{k-t_1+1}\theta_2-\frac{k_2-k}{t_2-k}\theta_2-\frac{t_2-k_2}{t_2-k}\theta_3\Big]\\&+\frac{(k-t_1+1)(t_2-k)}{(t_2-t_1+1)^{3/2}}\Big[\bar{\xi}_{t_1,k}(\omega_{t_1,k})-\bar{\xi}_{k+1,t_2}(\omega_{k+1,t_2})\Big]\\&+\frac{(k-t_1+1)(t_2-k)}{(t_2-t_1+1)^{3/2}}\Big[r_{t_1,k}(\omega_{t_1,k})-r_{k+1,t_2}(\omega_{k+1,t_2})\Big]
\\:=&D_{n1}(t_1,k,t_2)+D_{n2}(t_1,k,t_2)+D_{n3}(t_1,k,t_2),
\end{split}
\end{flalign}
where $\omega_{t_1,k}=(\frac{k_1-t_1+1}{k-t_1+1},\frac{k-k_1}{k-t_1+1},0)^{\top}$ and $\omega_{k+1,t_2}=(0,\frac{k_2-k}{t_2-k},\frac{t_2-k_2}{t_2-k})^{\top}$.

\noindent Since  $(k_1+1-t_1)\leq c_n$ and $t_2-k_2<c_n$ on $H_{1:n}^{2,22}(k,c_n)$, we have
\begin{flalign*}
\max_{k\in M_{n2}}\max_{(t_1,t_2)\in H_{1:n}^{2,22}(k,c_n)}\frac{[D_{n1}(t_1,k,t_2)]^2}{n\delta^2}\leq&\frac{2}{n\delta^2} \frac{(k-t_1+1)^2}{(t_2-t_1+1)^3}\Big[\frac{(t_2-k)^2(k_1+1-t_1)^2\delta_1^2}{(k-t_1+1)^2}+(t_2-k_2)^2\delta_2^2\Big]\\&
\leq C\frac{c_n^2}{n^2}=o(1).
\end{flalign*}
Note that
\begin{flalign*}
R_n(t_1,k,t_2)=&\Big[\sum_{i=k+1}^{k_2}+\sum_{i=k_2+1}^{t_2}\Big]\frac{(t_2-i+1)^2(i-1-k)^2}{(t_2-t_1+1)^{2}(t_2-k)^2}(\widehat{\theta}_{i,t_2}-\widehat{\theta}_{k+1,i-1})^2\\:=&R_{n1}(t_1,k,t_2)+R_{n2}(t_1,k,t_2),
\end{flalign*}
where in Lemma \ref{lem_3} we can show that 
$$
\max_{k\in M_{n2}}\max_{(t_1,t_2)\in H_{1:n}^{2,22}(k,c_n)}[R_{n1}(t_1,k,t_2)]^{-1}=O_p^s(1),
$$
the result is clear.

\subsubsection{Three or more change-points}\label{sec:three}
We now consider the case where there are multiple change-points at $(k_1,k_2,\ldots,k_{m_o})$. The proof for three or more change-points can largely be built based on the proof for one and two change-points case. We proceed as follows. For each $k\in M_n$, denote $k_U(k)=\min\{k_i|k_i>k, i=0,\ldots,m_o+1\}$ and $k_L(k)=\max\{k_i|k_i<k,i=0,\ldots,m_o+1\}$. We decompose $H_{1:n}(k)$ into five sets:
\begin{align*}
&H_{1:n}^1(k)=H_{1:n}(k)\cap \{(t_1,t_2)|t_1>k_L(k) \text{ and } t_2\leq k_U(k) \},\\
&H_{1:n}^2(k)=H_{1:n}(k)\cap \{(t_1,t_2)|k_{L-1}(k)<t_1\leq k_L(k) \text{ and } t_2\leq k_U(k) \},\\
&H_{1:n}^3(k)=H_{1:n}(k)\cap \{(t_1,t_2)|t_1> k_L(k) \text{ and } k_U(k)<t_2\leq k_{U+1}(k) \},\\
&H_{1:n}^4(k)=H_{1:n}(k)\cap \{(t_1,t_2)|k_{L-1}(k)<t_1\leq k_L(k) \text{ and } k_U(k) <t_2 \leq k_{U+1}(k) \},\\
&H_{1:n}^5(k)=H_{1:n}(k)\cap \{(t_1,t_2)|t_1\leq k_{L-1}(k), |k_{L-1}(k)|<\infty \text{ or } t_2 > k_{U+1}(k), |k_{U+1}(k)|<\infty  \},
\end{align*}
where $k_{L-1}=-\infty$ if $k_L=0$ and  $k_{U+1}=\infty$ if $k_U=n$.

For $H_{1:n}^1(k)$, it reduces to the no change-point case.

For $H_{1:n}^2(k)$ and $H_{1:n}^3(k)$, it reduces to the one change-point case.

For $H_{1:n}^4(k)$, it reduces to the two change-point case.

For $H_{1:n}^5(k)$, without loss of generality, we assume $t_2>k_{U+1}(k)$, $k_{U+1}(k)<\infty$.

Then, it follows that 
$$
R_n(t_1,k,t_2)\geq \sum_{i=k_U(k)+1}^{k_{U+1}(k)}\frac{(t_2-i+1)^2(i-1-k)^2}{(t_2-t_1+1)^{2}(t_2-k)^2}(\widehat{\theta}_{i,t_2}-\widehat{\theta}_{k+1,i-1})^2.
$$
Note that we can write 
\begin{flalign*}
{\theta}_{i,t_2}=&\frac{k_{U+1}(k)-i+1}{t_2-i+1}\theta_{U+1}+\frac{t_2-k_{U+1}(k)}{t_2-i+1}\theta_{k_{U+1}(k)+1,t_2},\\
{\theta}_{k+1,i-1}=&\frac{k_{U}(k)-k}{i-1-k}\theta_{U}+\frac{i-1-k_U(k)}{i-1-k}\theta_{U+1},
\end{flalign*}
then we have 
$$
{\theta}_{i,t_2}-{\theta}_{k+1,i-1}=\frac{t_2-k_{U+1}(k)}{t_2-i+1}[\theta_{k_{U+1}(k)+1,t_2}-\theta_{U+1}]+\frac{k_{U}(k)-k}{i-1-k}(\theta_{U+1}-\theta_{U}),
$$
similar arguments used in Lemma \ref{lem_2} will yield 

$$
\max_{k\in M_{n2}}\max_{(t_1,t_2)\in H_{1:n}^{5}(k)} [R_n(t_1,k,t_2)]^{-1}\leq (C+o_p(1))\frac{n^4}{\delta^2c_n^3\iota_n^2}=o_p(1),
$$
the result follows.

\subsubsection{Consistency}\label{sec:con}

To finish the consistency proof, we need to show that \eqref{eq:first_round_SNCP_0} and \eqref{eq:first_round_SNCP} still hold for applying SNCP on $\{Y_t\}_{t=1}^{\widehat{k}}$ and $\{Y_t\}_{t=\widehat{k}+1}^{n}$. WLOG, we prove the result for $\{Y_t\}_{t=1}^{\widehat{k}}$. In other words, when there are no other change-points in $\{Y_t\}_{t=1}^{\widehat{k}}$, we need to show that
\begin{align}
\label{eq:consistency_null}
P\left(\max_{k=h,\ldots,\widehat{k}-h} T_{1,\widehat{k}}(k)<K_n\right)\to 1,
\end{align}
and when there still exist other change-points, say $(k_1,\ldots,k_{m_o})$, in $Y_{1:\widehat{k}}$, we need to show that
\begin{align}
\label{eq:consistency_notnull}
P\left(\max_{k=h,\ldots,\widehat{k}-h} T_{1,\widehat{k}}(k)>K_n \text{ and }\min_{1\leq i\leq m_o}|k_i-\tilde{k}|<\iota_n\right) \to 1,
\end{align}
where $\tilde{k}=\arg\max_{k=h,\ldots,\widehat{k}-h} T_{1,\widehat{k}}(k).$

\noindent We first prove \eqref{eq:consistency_null}. Note that if there is no other change-points in $\{Y_t\}_{t=1}^{\widehat{k}}$, we know that $P(|\widehat{k}-k_1|<\iota_n)\to 1.$  By \eqref{eq:first_round_SNCP} we have
\begin{align*}
&P\left(\max_{k=h,\ldots,\widehat{k}-h} T_{1,\widehat{k}}(k)<K_n\right)\geq
P\left(\max_{k=h,\ldots,k_1+\iota_n-h} T_{1,k_1+\iota_n}(k)<K_n\right)-o(1)\\
=&P\left(\max_{k=h,\ldots,k_1+\iota_n-h} \max\limits_{(t_1,t_2)\in H_{1:(k_1+\iota_n)}(k)} \frac{D_{n}(t_1,k,t_2)^2}{V_{n}(t_1,k,t_2)}<K_n\right)-o(1)\\
\geq &P\left(\max_{k=h,\ldots,k_1+\iota_n-h} \max\limits_{(t_1,t_2)\in H_{1:(k_1+\iota_n)}(k)} \frac{D_{n}(t_1,k,t_2)^2}{L_{n}(t_1,k,t_2)}<K_n\right)-o(1).
\end{align*}
Using the same argument as the one used in one change-point case by expanding the $D_{n}(t_1,k,t_2)^2$ term in the numerator, it is straightforward to show that
\begin{align*}
&\max_{k=h,\ldots,k_1+\iota_n-h} \max\limits_{(t_1,t_2)\in H_{1:(k_1+\iota_n)}(k)}\frac{D_{n}(t_1,k,t_2)^2}{L_{n}(t_1,k,t_2)K_n}\\
\leq& \frac{C\iota_n^2\delta^2}{nK_n}\max_{k=h,\ldots,n-h} \max\limits_{(t_1,t_2)\in H_{1:n}(k)}\frac{1}{L_{n}(t_1,k,t_2)}+o_p(1)=o_p(1),
\end{align*}
where the last equality holds noting that we can choose $\iota_n$ small enough, say  $\iota_n\leq \left(n\delta^{-2}\log(n\delta^2)\right)^{1/2}$,  such that assumptions in Theorem \ref{thm1} are satisfied. Thus we have proved \eqref{eq:consistency_null}.

\noindent We now prove \eqref{eq:consistency_notnull}. Note that if there still exist other change-points, say $(k_1,\ldots,k_{m_o'})$ in $\{Y_t\}_{t=1}^{\widehat{k}}$, we know that either $(\mathrm{i})$ $P(0\leq k_{m_o'+1}-\widehat{k}<\iota_n)\to 1$ or $(\mathrm{ii})$ $P(0\leq \widehat{k}-k_{m_o'}<\iota_n)\to 1$.

$(\mathrm{i})$ 
Since  $k_{m_o'+1}-k_{m_o'}>n\epsilon_o>n\epsilon$  and $\iota_n=o(n)$, we have $n(\epsilon_o-\epsilon)>\iota_n$, hence $k_{m_o'+1}-k_{m_o'}>\iota_n+h$ and
\begin{align*}
&P(k_{m_o'}+h<\widehat{k})\geq P(k_{m_o'+1}-\iota_n<\widehat{k})\to 1.
\end{align*}
Thus, we have 
\begin{align*}
P\left(\max_{k=h,\ldots,\widehat{k}-h} T_{1,\widehat{k}}(k)>K_n \right) \geq 
P\left( \frac{D_{n}(k_{m_o'}-h+1,k_{m_o'},k_{m_o'}+h)^2}{V_{n}(k_{m_o'}-h+1,k_{m_o'},k_{m_o'}+h) }>K_n \right)-o(1) \to 1.
\end{align*}

$(\mathrm{ii})$ 
It is easy to see that $k_{m_o'-1}+h<k_{m_o'}<\widehat{k}$, then,
\begin{align*}
P\left(\max_{k=h,\ldots,\widehat{k}-h} T_{1,\widehat{k}}(k)>K_n \right) \geq 
P\left( \frac{D_{n}(k_{m_o'-1}-h+1,k_{m_o'-1},k_{m_o'-1}+h)^2}{V_{n}(k_{m_o'-1}-h+1,k_{m_o'-1},k_{m_o'-1}+h) }>K_n \right)-o(1) \to 1.
\end{align*}
Define $M_{\widehat{k}}=\{k|\min_{i=1,\ldots,m_o'}|k_i-k|>\iota_n, h\leq k\leq\widehat{k}-h\}$,  it suffices to show that
\begin{align*}
\max_{k\in M_{\widehat{k}}}\frac{T_{1,\widehat{k}}(k)}{n\delta^2}=o_p(1).
\end{align*}
Note that $P(M_{\widehat{k}}\subseteq M_n)\to 1$, thus
\begin{align*}
\max_{k\in M_{\widehat{k}}}\frac{T_{1,\widehat{k}}(k)}{n\delta^2}\leq \max_{k\in M_{n}}\frac{T_{1,n}(k)}{n\delta^2} =o_p(1).
\end{align*}

Thus, we have proved \eqref{eq:consistency_notnull}. The argument then goes on similarly till SNCP stops. Since there are finite number of change-points, SNCP will eventually stop.

\qed
\vspace{4mm}

\subsection{Lemmas}
\begin{lemma}\label{lem_1}
For the one-change point case,
$$
\max_{k\in M_{n1}}\max_{(t_1,t_2)\in H_{1:n}^{2}(k,c_n)}
R_{n2}(t_1,k,t_2)^{-1}=O_p(n^{4}\delta^{-2}c_n^{-3}\iota_n^{-2})=o_p(1).
$$
\end{lemma}
\noindent\textsc{Proof of Lemma \ref{lem_1}}
We can see that 
\begin{flalign*}
R_{n2}(t_1,k,t_2)=&\sum_{i=k_1+1}^{t_2}\frac{(t_2-i+1)^2(i-1-k)^2}{(t_2-t_1+1)^{2}(t_2-k)^2}\Big([\theta_2-\theta_{k+1,i-1}]	+[\bar{\xi}_{i,t_2}(\omega_{i,t_2})-\bar{\xi}_{k+1,i-1}(\omega_{k+1,i-1})]\\&
+[r_{i,t_2}(\omega_{i,t_2})-r_{k+1,i-1}(\omega_{k+1,i-1})]\Big)^2,
\end{flalign*}
where $\omega_{i,t_2}=(0,1)^{\top}$ and $\omega_{k+1,i-1}=(\frac{k_1-k}{i-1-k},\frac{i-1-k_1}{i-1-k})^{\top}$.
Denote
\begin{flalign*}
A_n(t_1,k,t_2)=&\sum_{i=k_1+1}^{t_2}\frac{(t_2-i+1)^2(i-1-k)^2}{(t_2-k)^2}[\theta_2-\theta_{k+1,i-1}]^2,\\ B_n(t_1,k,t_2)=&\sum_{i=k_1+1}^{t_2}\frac{(t_2-i+1)^2(i-1-k)^2}{(t_2-k)^2}\Big([\bar{\xi}_{i,t_2}(\omega_{i,t_2})-\bar{\xi}_{k+1,i-1}(\omega_{k+1,i-1})]\\&
+[r_{i,t_2}(\omega_{i,t_2})-r_{k+1,i-1}(\omega_{k+1,i-1})]\Big)^2,\\
C_n(t_1,k,t_2)=&2\sum_{i=k_1+1}^{t_2}\frac{(t_2-i+1)^2(i-1-k)^2}{(t_2-k)^2}\Big([\bar{\xi}_{i,t_2}(\omega_{i,t_2})-\bar{\xi}_{k+1,i-1}(\omega_{k+1,i-1})]\\&
+[r_{i,t_2}(\omega_{i,t_2})-r_{k+1,i-1}(\omega_{k+1,i-1})]\Big)[\theta_2-\theta_{k+1,i-1}],
\end{flalign*} 
and by Hua's identity, we obtain 
\begin{flalign*}
R_{n2}(t_1,k,t_2)^{-1}=\frac{(t_2-t_1+1)^2[A_n(t_1,k,t_2)]^{-1}}{1+[B_n(t_1,k,t_2)+C_n(t_1,k,t_2)][A_n(t_1,k,t_2)]^{-1}}.
\end{flalign*}
Note that $(\theta_2-\theta_{k+1,i-1})=\frac{k_1-k}{i-1-k}\delta$, and 
we can find some constants $0<c_1<c_2<\infty$ independent of $\{t_1,t_2,k,k_1\}$ such that $c_1(t_2-k_1)^3<\sum_{i=k_1+1}^{t_2}{(t_2-i+1)^2}<c_2(t_2-k_1)^3$,
we have 
\begin{flalign*}
[A_n(t_1,k,t_2)]>c_1\frac{(k_1-k)^2(t_2-k_1)^{3}}{(t_2-k)^2}\delta^2.
\end{flalign*}
Hence, \begin{equation}\label{A}
\min_{k\in M_{n1}}\min_{(t_1,t_2)\in H_{1:n}^{2}(k,c_n)}[A_n(t_1,k,t_2)]\geq  c_1\frac{\iota_n^2c_n^3\delta^2}{n^2}.
\end{equation}
Under Assumptions \ref{ass_no1} and \ref{ass_no2}, we can show that 
\begin{flalign}\label{B}
(t_2-t_1+1)^{-2}\max_{k\in M_{n1}}\max_{(t_1,t_2)\in H_{1:n}^{2}(k,c_n)}B_n(t_1,k,t_2)\leq O_p(1),
\end{flalign}
and by Cauchy-Schwartz inequality, we have $$-2B_n(t_1,k,t_2)[A_n(t_1,k,t_2)]^{-1/2}\leq  C_n(t_1,k,t_2)[A_n(t_1,k,t_2)]^{-1}\leq 2[B_n(t_1,k,t_2)]^{1/2}[A_n(t_1,k,t_2)]^{-1/2}.$$ 

\noindent Hence, by (\ref{A}) and (\ref{B}), we obtain
\begin{flalign*}
&\max_{k\in M_{n1}}\max_{(t_1,t_2)\in H_{1:n}^{2}(k,c_n)}
R_{n2}(t_1,k,t_2)^{-1}\\\leq &\frac{(t_2-t_1+1)^2\max_{k\in M_{n1}}\max_{(t_1,t_2)\in H_{1:n}^{2}(k,c_n)}[A_n(t_1,k,t_2)]^{-1}}{1+\min_{k\in H_{1:n}^2(k,c_n)}[B_n(t_1,k,t_2)+C_n(t_1,k,t_2)][A_n(t_1,k,t_2)]^{-1}}\\\leq&  \frac{(t_2-t_1+1)^2\max_{k\in M_{n1}}\max_{(t_1,t_2)\in H_{1:n}^{2}(k,c_n)}[A_n(t_1,k,t_2)]^{-1}}{1-2\max_{k\in M_{n1}}\max_{(t_1,t_2)\in H_{1:n}^{2}(k,c_n)}[B_n(t_1,k,t_2)]^{1/2}[A_n(t_1,k,t_2)]^{-1/2}}
\\\leq& C\frac{n^4}{\iota_n^2c_n^3\delta^2}(1+o_p(1))=o_p(1),
\end{flalign*}
where the last inequality holds by (\ref{cn}).

\qed
\vspace{4mm}

\begin{lemma}\label{lem_2}
For the two change-point case,
$$
\max_{k\in M_{n1}}\max_{(t_1,t_2)\in H_{1:n}^{1,1}(k)}
R_{n2}(t_1,k,t_2)^{-1}=O_p(n\delta^{-2}\iota_n^{-2})=o_p(1)
$$
\end{lemma}

\noindent\textsc{Proof of Lemma \ref{lem_2}}

For $i\in [k_1+1,k_2]$ we decompose  $\widehat{\theta}_{i,t_2}$ and $\widehat{\theta}_{k+1,i-1}$  as 
\begin{flalign*}
\widehat{\theta}_{i,t_2}=&\theta_{i,t_2}+\bar{\xi}_{i,t_2}(\omega_{i,t_2})+r_{i,t_2}(\omega_{i,t_2}),\\
\widehat{\theta}_{i,i-1}=&\theta_{k+1,i-1}+\bar{\xi}_{k+1,i-1}(\omega_{k+1,i-1})+r_{k+1,i-1}(\omega_{k+1,i-1}),
\end{flalign*}
where $\omega_{i,t_2}=(0,\frac{k_2-i+1}{t_2-i+1},\frac{t_2-k_2}{t_2-i+1})^{\top}$ and $\omega_{k+1,i-1}=(\frac{k_1-k}{i-1-k},\frac{i-1-k_1}{i-1-k},0)^{\top}$.

\noindent Note that,   $$\theta_{k+1,i-1}-\theta_{i,t_2}=\frac{k_1-k}{i-1-k}(\theta_1-\theta_2)+\frac{(t_2-k_2)}{(t_2-i+1)}(\theta_2-\theta_3),$$ 
and we denote 
\begin{flalign*}
A_n(t_1,k,t_2)=&\sum_{i=k_1+1}^{k_2}\frac{(t_2-i+1)^2(i-1-k)^2}{(t_2-k)^2}[\frac{(k_1-k)}{(i-1-k)}\delta_1+\frac{(t_2-k_2)}{(t_2-i+1)}\delta_2]^2
\end{flalign*} 

\noindent If $\delta_1\delta_2\geq 0$, then 
$$
A_n(t_1,k,t_2)\geq \sum_{i=k_1+1}^{k_2}(t_2-i+1)^2[\frac{k_1-k}{t_2-k}\delta_1]^2>C\frac{(k_2-k_1)^3(k_1-k)^2\delta_1^2}{(t_2-k)^2}=C{n\iota_n^2\delta_1^2}.
$$
If $\delta_1\delta_2<0$,  denote $x=(t_2-k_2)\delta_2$ and $y=(k_1-k)\delta_1$, then 
\begin{flalign*}
A_n(t_1,k,t_2)=&\frac{1}{(t_2-k)^2}\sum_{i=k_1+1}^{k_2}[(i-1-k)x+(t_2-i+1)y]^2\\>&\frac{1}{n^2}\sum_{i=k_1+1}^{k_2}\left[(x-y)i+(t_2+1)y-(1+k)x\right]^2\\
=&\frac{(x-y)^2}{n^2}\sum_{i=0}^{k_2-k_1}\left[i+\frac{(k_1-k)x+(t_2-k_1)y}{x-y}\right]^2\\
\geq& \frac{(x-y)^2}{n^2} \min_a\sum_{i=0}^{\epsilon n}(i+a)^2=Cn{(x-y)^2}>Cny^2>C{n\iota_n^2\delta_1^2}.
\end{flalign*}
Then, the rest follows from similar arguments (below (\ref{A})) used in Lemma \ref{lem_1}. \qed
\vspace{4mm}

\begin{lemma}\label{lem_3}
For the two change-point case,	$$\max_{k\in M_{n2}}\max_{(t_1,t_2)\in H_{1:n}^{2,22}(k,c_n)}[R_{n1}(t_1,k,t_2)]^{-1}=O_p^s(1).$$
\end{lemma}

\noindent\textsc{Proof of Lemma \ref{lem_3}}
\begin{flalign*}
R_{n1}(t_1,k,t_2)=&\sum_{i=k+1}^{k_2}\frac{(t_2-i+1)^2(i-1-k)^2}{(t_2-t_1+1)^{2}(t_2-k)^2}(\widehat{\theta}_{i,t_2}-\widehat{\theta}_{k+1,i-1})^2\\=&
\sum_{i=k+1}^{k_2}\frac{(t_2-i+1)^2(i-1-k)^2}{(t_2-t_1+1)^{2}(t_2-k)^2}\Big([{\theta}_{i,t_2}-{\theta}_{k+1,i-1}]+[\bar{\xi}_{i,t_2}(\omega_{i,t_2})-\bar{\xi}_{k+1,i-1}(\omega_{k+1,i-1})]\\&+[r_{i,t_2}(\omega_{i,t_2})-r_{k+1,i-1}(\omega_{k+1,i-1})]\Big)^2.
\end{flalign*}
where $\omega_{i,t_2}=(0,\frac{k_2-i+1}{t_2-i+1},\frac{t_2-k_2}{t_2-i+1})^{\top}$ and $\omega_{k+1,i-1}=(0,1,0)^{\top}$.

\noindent Note that uniformly on $M_{n2}(t_1,t_2)$ and $H_{1:n}^{2,22}(k,c_n)$, we have  by Assumption \ref{ass_no2}
$$
\frac{(t_2-i+1)(i-1-k)}{(t_2-t_1+1)(t_2-k)}\Big|r_{i,t_2}(\omega_{i,t_2})-r_{k+1,i-1}(\omega_{k+1,i-1})\Big|=o_p(n^{-1/2}).
$$
Hence, it suffices to consider 
$$
\sum_{i=k+1}^{k_2}\frac{(t_2-i+1)^2(i-1-k)^2}{(t_2-t_1+1)^{2}(t_2-k)^2}[{\theta}_{i,t_2}-{\theta}_{k+1,i-1}+\bar{\xi}_{i,t_2}(\omega_{i,t_2})-\bar{\xi}_{k+1,i-1}(\omega_{k+1,i-1})]^2.
$$
Since $t_2-k>\epsilon n$ and $t_2-k_2<c_n=o(n)$, we have $k_2-k>\epsilon n/2$, thus we have 
\begin{flalign*}
&\sum_{i=k+1}^{k_2}\frac{(t_2-i+1)^2(i-1-k)^2}{(t_2-t_1+1)^{2}(t_2-k)^2}[{\theta}_{i,t_2}-{\theta}_{k+1,i-1}+\bar{\xi}_{i,t_2}(\omega_{i,t_2})-\bar{\xi}_{k+1,i-1}(\omega_{k+1,i-1})]^2\\
>&\sum_{i=k+1}^{k+\epsilon n/2}\frac{(t_2-i+1)^2(i-1-k)^2}{(t_2-t_1+1)^{2}(t_2-k)^2}[{\theta}_{i,t_2}-{\theta}_{k+1,i-1}+\bar{\xi}_{i,t_2}(\omega_{i,t_2})-\bar{\xi}_{k+1,i-1}(\omega_{k+1,i-1})]^2
\\=& \sum_{i=k+1}^{k+\epsilon n/2}\frac{(t_2-i+1)^2(i-1-k)^2}{(t_2-t_1+1)^{2}(t_2-k)^2}[\frac{(t_2-k_2)\delta_2}{t_2-i+1}+\bar{\xi}_{i,t_2}(\omega_{i,t_2})-\bar{\xi}_{k+1,i-1}(\omega_{k+1,i-1})]^2
\\\geq &\min_{a}\sum_{i=k+1}^{k+\epsilon n/2}\frac{(t_2-i+1)^2(i-1-k)^2}{(t_2-t_1+1)^{2}(t_2-k)^2}[\frac{a}{t_2-i+1}+\bar{\xi}_{i,t_2}(\omega_{i,t_2})-\bar{\xi}_{k+1,i-1}(\omega_{k+1,i-1})]^2\\:=&R_{n1}^*(t_1,k,t_2).
\end{flalign*}
The above quadratic function  will be minimized at $$a^*= -\frac{\sum_{i=k+1}^{k+\epsilon n/2} (i-1-k)^2(t_2-i+1)[\bar{\xi}_{i,t_2}(\omega_{i,t_2})-\bar{\xi}_{k+1,i-1}(\omega_{k+1,i-1})]}{\sum_{i=k+1}^{k+\epsilon n/2}(i-1-k)^2}.$$

So, if $\bar{\xi}_{i,t_2}(\omega_{i,t_2})-\bar{\xi}_{k+1,i-1}(\omega_{k+1,i-1})=O_p^s(n^{-1/2})$, we can show that $R^*_{n1}=O_p^{s}(1)$. 

For example, in the case of smooth function model,  we have  
$$
\bar{\xi}_{i,t_2}^{(2,3)}=\frac{\partial H(\mu_{i,t_2})}{\partial\mu}^{\top}\frac{1}{t_2-i+1}[\sum_{t=i}^{k_2}(Z_t^{(2)}-\mu_{i,t_2})+\sum_{t=k_2+1}^{t_2}(Z_t^{(3)}-\mu_{i,t_2})].
$$

Note that $\mu_{i,t_2}=\frac{k_2-i+1}{t_2-i+1}\mu_z^{(2)}+\frac{t_2-k_2}{t_2-i+1}\mu_z^{(3)}$, hence for each $k<i\leq k+\epsilon n/2$, 
\begin{flalign*}
\frac{\partial H(\mu_{i,t_2})}{\partial\mu}=&\frac{\partial H(\mu_z^{(2)})}{\partial\mu}+\frac{1}{2}(\frac{t_2-k_2}{t_2-i+1})^2 (\mu_z^{(2)}-\mu_z^{(3)})^{\top}\frac{\partial^2H(\tilde{\mu})}{\partial\mu\partial\mu^{\top}}
(\mu_z^{(2)}-\mu_z^{(3)})
\\=&\frac{\partial H(\mu_z^{(2)})}{\partial\mu}+O(\frac{c_n^2}{n^2})
\end{flalign*}
for some $\tilde{\mu}=u\mu_z^{(2)}+(1-u)\mu_{i,t_2}$.
Hence, as long as $\|\frac{\partial^2H(\tilde{\mu})}{\partial\mu\partial\mu^{\top}}\|$ is bounded, since  $|t_2-k_2|<c_n$, we can show that  
\begin{flalign*}
\bar{\xi}_{i,t_2}=[\frac{\partial H(\mu_z^{(2)})}{\partial \mu}+o(1)]^{\top}\{\frac{1}{t_2-i+1}[\sum_{t=i}^{k_2}(Z_t^{(2)}-\mu_{z}^{(2)}-o(1)]+O_p(n^{-1}c_n^{1/2})\}
\end{flalign*}
Hence, it follows that 
$$
\frac{t_2-i+1}{\sqrt{n}}\bar{\xi}_{i,t_2}=\frac{1}{\sqrt{n}}\sum_{t=i}^{t_2}\xi_2(Y_t^{(2)})+o_p(1),
$$
where $o_p(1)$ holds uniformly for $k<i\leq k+\epsilon n/2$ and $\xi_2(Y_t^{(2)})$ is defined in Assumption \ref{ass_no1}.

Assuming that $t_1/n\to u_1, k/n\to u, t_2/n\to u_2$~(by definition $u_1\leq u-\epsilon$ and $u_2\geq u+\epsilon$), by Assumption \ref{ass_no1} it is straightforward to show that
\begin{align*}
\frac{a^*}{n^{1/2}} \to_D -\frac{24}{\epsilon^3}\int_{u}^{u+\epsilon/2}(s-u)\sigma_2\left[(s-u)({B}^{(2)}(u_2)-{B}^{(2)}(s))-(u_2-s)({B}^{(2)}(s)-{B}^{(2)}(u))\right]ds:=A^*(u,u_2,\epsilon),
\end{align*}
here the $\sum_{t=k_2+1}^{t_2}Z_t^{(3)}-\mu_{i,t_2}$ will not contribute to the asymptotic distribution since $|t_2-k_2|<c_n$ while $c_n/n\to0$.
Therefore, we have
\begin{align*}
&R_{n1}(t_1,k,t_2)>R_{n1}^{*}(t_1,k,t_2)
\to_D\frac{1}{(u_2-u_1)^2(u_2-u)^2}\cdot\\
&\int_{u}^{u+\epsilon/2}\left[(s-u)\sigma_2(B^{(2)}(u_2)-B^{(2)}(s))-(u_2-s)(B^{(2)}(s)-B^{(2)}(u))+(s-u)A^*(u,u_2,\epsilon) \right]^2ds.
\end{align*}
The result follows by the fact that $$\min_{k\in M_{n2}}\min_{(t_1,t_2)\in H_{1:n}^{2,22}(k,c_n)}R_{n1}(t_1,k,t_2)>\min_{k\in M_{n2}}\min_{(t_1,t_2)\in H_{1:n}^{2,22}(k,c_n)}R_{n1}^{*}(t_1,k,t_2)=O_p^s(1).$$ \qed

\section{Consistency of SNCP for multivariate mean change}\label{sec:consistencymean}
In this section, we provide detailed proof of Theorem \ref{thm_onechange_multi} in the main text, which gives the consistency of SNCP for multivariate mean change.

The proof essentially follows the same logic as the one in Section \ref{subsec: outline_proof} for Theorem \ref{thm1}. However, compared to the univariate proof, substantial technical complication arises due to the vector/multivariate nature of the parameter, which makes the self-normalizer $V^*_n(t_1,k,t_2)$ a matrix in $\mathbb{R}^{d\times d}$. Thus, to establish scenarios 1-3 listed at the beginning of Section \ref{subsec: outline_proof}, the technical argument needed is significantly different, which is indeed much more challenging than the univariate proof in Section \ref{subsec: outline_proof}, as it requires the analysis of random matrix and its inverse. Two main technical tools that will be used repeatedly in the proof are a matrix Cauchy-Schwartz inequality in \cite{tripathi1999matrix} (restated in Lemma \ref{lem_tripathi}) and the Sherman-Morrison formula, which quantifies the impact of a rank-one update to a matrix.

\subsection{No change-point case}
The proof of Theorem \ref{thm_onechange_multi}(i) follows standard arguments using the invariance principle. In addition, in this case~(i.e. the no change-point scenario), we have $\max_{k=1,\cdots,n}T_{1,n}(k)=O_p(1)$. Thus, for any threshold $K_n\to\infty$, we have 
$$
\lim\limits_{n\to\infty} P\Big(\max_{k=1,\cdots,n}T_{1,n}(k)<K_n\Big)=1.
$$
In the following, we focus on the proof of Theorem \ref{thm_onechange_multi}(ii). In what follows, denote $S^{X}_{a,b}=\sum_{t=a}^{b}X_t$.

\subsection{One change-point case}\label{sec:meanone}
We can decompose $M_n=M_{n1}\cup M_{n2}$, where $M_{n1}=\{k|k_1-k>\iota_n\}$ and $M_{n2}=\{k|k-k_1>\iota_n\}$. By symmetry, we only need to prove the result for $M_{n1}$. 
Let $c_n$ satisfy (\ref{cn}),   and recall that $\delta_1=\delta\eta_1=n^{-\kappa}\eta_1$, where $\eta_1\in\mathbb{R}^d/\{\bf{0}\}$.

\noindent We decompose $H_{1:n}(k)$ as:
\begin{flalign*}
&(1)~H_{1:n}^0(k)=H_{1:n}(k)\cap \{(t_1,t_2)|t_2\leq k_1\},\\
&(2)~H_{1:n}^1(k,c_n)=H_{1:n}(k)\cap \{(t_1,t_2)|k_1+1\leq t_2\leq k_1+c_n\},\\
&(3)~H_{1:n}^2(k,c_n)=H_{1:n}(k)\cap \{(t_1,t_2)| t_2> k_1+c_n\}.
\end{flalign*}

$(\mathrm{1})$ On $H_{1:n}^0(k)$, there is no change-point. Hence using the invariance principle, it follows
\begin{flalign*}
(n\delta^2)^{-1}\max_{k\in M_{n1}}\max_{(t_1,t_2)\in H_{1:n}^0(k)} D_n(t_1,k,t_2)^{\top}V_n(t_1,k,t_2)^{-1}D_n(t_1,k,t_2)\leq (n\delta^2)^{-1}\max_{k=1,\cdots,n}T_{1,n}(k)=o_p(1).
\end{flalign*}

\noindent Now, on $H_{1:n}^1(k,c_n)$ and  $H_{1:n}^2(k,c_n)$, by simple calculation, we have
\begin{flalign*}
D_n(t_1,k,t_2)=&-\frac{(k-t_1+1)(t_2-k_1)}{(t_2-t_1+1)^{3/2}}\delta\eta_1+\frac{(k-t_1+1)(t_2-k)}{(t_2-t_1+1)^{3/2}}\Big[\frac{S^{X}_{t_1,k}}{k-t_1+1}-\frac{S^{X}_{k+1,t_2}}{t_2-k}\Big]\\:=&-\frac{(k-t_1+1)(t_2-k_1)}{(t_2-t_1+1)^{3/2}}\delta\eta_1+D_n^{X}(t_1,k,t_2).
\end{flalign*}

$(\mathrm{2})$ On $H_{1:n}^1(k,c_n)$,
Cauchy Schwarz inequality indicates that
\begin{flalign*}
&\max_{k\in M_{n1}}\max_{(t_1,t_2)\in H_{1:n}^1(k,c_n) }(n\delta^2)^{-1}D_n(t_1,k,t_2)^{\top}V_n(t_1,k,t_2)^{-1}D_n(t_1,k,t_2)\\
\leq &2\max_{k\in M_{n1}}\max_{(t_1,t_2)\in H_{1:n}^1(k,c_n) }\frac{(t_2-k_1)^2(k-t_1+1)^2}{n(t_2-t_1+1)^3}{\eta_1^{\top}V_n(t_1,k,t_2)^{-1}\eta_1}
\\&+2\max_{k\in M_{n1}}n^{2\kappa-1}{D_n^{X}(t_1,k,t_2)^{\top}V_n(t_1,k,t_2)^{-1}D_n^{X}(t_1,k,t_2)}.
\end{flalign*}
Note that $V_n(t_1,k,t_2)^{-1}\leq L_n(t_1,k,t_2)^{-1}$, and since $(t_1,k)$ contains no change-points when $k\in M_{n1}$, one can easily verify that 
$$\max_{k\in M_{n1}}\max_{(t_1,t_2)\in H_{1:n}^1(k,c_n) }{\eta_1^{\top}L_n(t_1,k,t_2)^{-1}\eta_1}=O_p(1),$$ 
using the invariance principle.  Recall $t_2-k_1\leq c_n$ when $(t_1,t_2)\in H_{1:n}^{1}(k,c_n)$, hence 
\begin{flalign}\label{hn1}
\begin{split}
&\max_{k\in M_{n1}}\max_{(t_1,t_2)\in H_{1:n}^1(k,c_n) }(n\delta^2)^{-1}D_n(t_1,k,t_2)^{\top}L_n(t_1,k,t_2)^{-1}D_n(t_1,k,t_2)\\
\leq &2c_n^2 \frac{n^2\epsilon^2}{n^4}{\eta_1^{\top}L_n(t_1,k,t_2)^{-1}\eta_1}+2\max_{k\in M_{n1}}n^{2\kappa-1}{D_n^{X}(t_1,k,t_2)^{\top}L_n(t_1,k,t_2)^{-1}D_n^{X}(t_1,k,t_2)}.
\\=&C\frac{c_n}{n}^2{\eta_1^{\top}L_n(t_1,k,t_2)^{-1}\eta_1}+O_p(n^{2\kappa-1})=o_p(1),
\end{split}
\end{flalign}
where the last equality holds by the fact that $n^{2\kappa-1}=o(1)$ and ${D_n^{X}(t_1,k,t_2)^{\top}L_n(t_1,k,t_2)^{-1}D_n^{X}(t_1,k,t_2)}=O_p(1)$.

(3)  On $H_{1:n}^2(k,c_n)$, similarly, we have 
\begin{flalign*}
&\max_{k\in M_{n1}}\max_{(t_1,t_2)\in H_{1:n}^2(k,c_n) }(n\delta^2)^{-1}D_n(t_1,k,t_2)^{\top}V_n(t_1,k,t_2)^{-1}D_n(t_1,k,t_2)\\
\leq &2\max_{k\in M_{n1}}\max_{(t_1,t_2)\in H_{1:n}^2(k,c_n) }\frac{(t_2-k_1)^2(k-t_1+1)^2}{n(t_2-t_1+1)^3}{\eta_1^{\top}V_n(t_1,k,t_2)^{-1}\eta_1}
\\&+2\max_{k\in M_{n1}}n^{2\kappa-1}{D_n^{X}(t_1,k,t_2)^{\top}V_n(t_1,k,t_2)^{-1}D_n^{X}(t_1,k,t_2)}
\\\leq& 2\max_{k\in M_{n1}}\max_{(t_1,t_2)\in H_{1:n}^2(k,c_n) }C{\eta_1^{\top}V_n(t_1,k,t_2)^{-1}\eta_1}
+2\max_{k\in M_{n1}}n^{2\kappa-1}{D_n^{X}(t_1,k,t_2)^{\top}L_n(t_1,k,t_2)^{-1}D_n^{X}(t_1,k,t_2)}
\end{flalign*}
By Lemma \ref{lem_mean1},    we can see that 
$$
\max_{k\in M_{n1}}\max_{(t_1,t_2)\in H_{1:n}^2(k,c_n) }(n\delta^2)^{-1}D_n(t_1,k,t_2)^{\top}V_n(t_1,k,t_2)^{-1}D_n(t_1,k,t_2)\leq O_p(\frac{n^{4+2\kappa}}{c_n^3\iota_n^2})+o_p(1)=o_p(1).
$$


Therefore, using the result in (1)-(3),   when there is one-change point, $$\max_{k\in M_{n1}}\max_{(t_1,t_2)\in H_{1:n}}(n\delta^2)^{-1}D_n(t_1,k,t_2)^{\top}V_n(t_1,k,t_2)^{-1}D_n(t_1,k,t_2)=o_p(1).$$

\subsection{Two change-points}
We now consider the case where there are two change-points at $(k_1,k_2)$. We can decompose $M_n=M_{n1}\cup M_{n2}\cup M_{n3}$, where $M_{n1}=\{k|k_1-k>\iota_n\}$, $M_{n2}=\{k|k-k_1>\iota_n \text{ and } k_2-k >\iota_n\}$ and $M_{n3}=\{k|k-k_2>\iota_n\}$. By symmetry, we only need to show the result for $M_{n1}$ and $M_{n2}$.

(i) We first consider the set $M_{n1}=\{k|k_1-k>\iota_n\}$ and analyze the behavior of $T_n(t_1,k,t_2)$ on two subsets of $(t_1,t_2)\in H_{1:n}(k)$. Define
\begin{align*}
H_{1:n}^{1,0}(k)=&H_{1:n}(k)\cap \{(t_1,t_2)|t_2\leq k_2\},\\
H_{1:n}^{1,1}(k)=&H_{1:n}(k)\cap \{(t_1,t_2)|t_2>k_2\}.
\end{align*}

For the behavior of $T_{n}(t_1,k,t_2)$ on $(t_1,t_2)\in H_{1:n}^{1,0}(k)$, it reduces to the one change-point case. 

For the behavior of $T_{n}(t_1,k,t_2)$ on $(t_1,t_2)\in H_{1:n}^{1,1}(k)$, we first notice that 
\begin{flalign*}
&D_n(t_1,k,t_2)\\=&\frac{(k-t_1+1)(t_2-k)}{(t_2-t_1+1)^{3/2}}\Big[\frac{t_2-k_1}{t_2-k}\delta_1+\frac{t_2-k_2}{t_2-k}\delta_2\Big]+\frac{(k-t_1+1)(t_2-k)}{(t_2-t_1+1)^{3/2}}\Big[\frac{S^{X}_{t_1,k}}{(k-t_1+1)}-\frac{S^{X}_{k+1,t_2}}{(t_2-k)}\Big]
\\=&\frac{(k-t_1+1)(t_2-k_1)}{(t_2-t_1+1)^{3/2}}\delta_1+\frac{(k-t_1+1)(t_2-k_2)}{(t_2-t_1+1)^{3/2}}\delta_2+\frac{(k-t_1+1)(t_2-k)}{(t_2-t_1+1)^{3/2}}\Big[\frac{S^{X}_{t_1,k}}{(k-t_1+1)}-\frac{S^{X}_{k+1,t_2}}{(t_2-k)}\Big]
\\:=&\sum_{i=1}^{3}D_n^{(i)}(t_1,k,t_2).
\end{flalign*}
where
\begin{flalign*}
D_n^{(1)}(t_1,k,t_2)=&\frac{(k-t_1+1)}{(t_2-t_1+1)^{3/2}}(t_2-k_1)\eta_1\delta,\\ D_n^{(2)}(t_1,k,t_2)=&\frac{(k-t_1+1)}{(t_2-t_1+1)^{3/2}}(t_2-k_2)\eta_2\delta,\\ D_n^{(3)}(t_1,k,t_2)=&\frac{(k-t_1+1)(t_2-k)}{(t_2-t_1+1)^{3/2}}\Big[\frac{S^{X}_{t_1,k}}{(k-t_1+1)}-\frac{S^{X}_{k+1,t_2}}{(t_2-k)}\Big].	
\end{flalign*} 
Note that by Cauchy Schwarz inequality, we have 
$$
D_n(t_1,k,t_2)^{\top}V_n(t_1,k,t_2)^{-1}D_n(t_1,k,t_2)\leq 3\sum_{i=1}^{3}D_n^{(i)}(t_1,k,t_2)^{\top}V_n(t_1,k,t_2)^{-1}D_n^{(i)}(t_1,k,t_2)
$$
Hence, it suffices to show that $$\max_{k\in M_{n1}}\max_{(t_1,t_2)\in H_{1:n}^{1,1}(k) }D_n^{(i)}(t_1,k,t_2)^{\top}V_n(t_1,k,t_2)^{-1}D_n^{(i)}(t_1,k,t_2)=o_p(n\delta^2)$$
for $i=1,2,3$.

(1) We first show that when $(t_1,t_2)\in H_{1:n}^{1,1}(k)$, 
$$\max_{k\in M_{n1}}\max_{(t_1,t_2)\in H_{1:n}^{1,1}(k) }D_n^{(1)}(t_1,k,t_2)^{\top}L_n(t_1,k,t_2)^{-1}D_n^{(1)}(t_1,k,t_2)=o_p(n\delta^2).$$  
Note $k-t_1=O(n)$ and $t_2-t_1=O(n)$, so it suffices to show that 
\begin{flalign*}
&\max_{k\in M_{n1}}\max_{(t_1,t_2)\in H_{1:n}^{1,1}(k) }\eta_1^{\top}V_n(t_1,k,t_2)^{-1}\eta_1=o_p(1),
\end{flalign*}
which is implied by Lemma \ref{lem_mean2}. 

(2) 
We further split $H_{1:n}^{1,1}(k)$ into $H_{1:n}^{1,1}(k)=H_{1:n}^{1,11}(k)\cup H_{1:n}^{1,12}(k)$ such that 
\begin{align*}
H_{1:n}^{1,11}(k)&=H_{1:n}^{1,1}(k)\cap \{t_2-k_2=o(n)\},\\
H_{1:n}^{1,12}(k)&=H_{1:n}^{1,1}(k)\cap \{t_2-k_2=O(n)\}.
\end{align*}
For $(t_1,t_2)\in H_{1:n}^{1,11}(k)$, using the fact that $t_2-k_2=o(n)$, we have that 
\begin{flalign*}
&\max_{k\in M_{n1}}\max_{(t_1,t_2)\in H_{1:n}^{1,11}(k) }(n\delta^2)^{-1}D_n^{(2)}(t_1,k,t_2)^{\top}V_n(t_1,k,t_2)^{-1}D_n^{(2)}(t_1,k,t_2)\\\leq& \max_{k\in M_{n1}}\max_{(t_1,t_2)\in H_{1:n}^{1,11}(k) }(n\delta^2)^{-1}D_n^{(2)}(t_1,k,t_2)^{\top}L_n(t_1,k,t_2)^{-1}D_n^{(2)}(t_1,k,t_2)\\=&\max_{(t_1,t_2)\in H_{1:n}^{1,11}(k) }O_p(\frac{(t_2-k_2)^2}{n^2})=o_p(1).
\end{flalign*}
For $(t_1,t_2)\in H_{1:n}^{1,12}(k)$, using similar arguments used in Lemma \ref{lem_mean2}, we can obtain that 
$$\max_{k\in M_{n1}}\max_{(t_1,t_2)\in H_{1:n}^{1,12}(k) }\eta_2^{\top}V_n(t_1,k,t_2)^{-1}\eta_2\leq O_p(\frac{n^{1+2\kappa}}{n^2}).$$
Hence, results above indicate that 
$$\max_{k\in M_{n1}}\max_{(t_1,t_2)\in H_{1:n}^{1,1}(k) }D_n^{(2)}(t_1,k,t_2)^{\top}V_n(t_1,k,t_2)^{-1}D_n^{(2)}(t_1,k,t_2)=o_p(n\delta^2).$$ 

(3) Using $V_n(t_1,k,t_2)^{-1}\leq L_n(t_1,k,t_2)^{-1}$ and the invariance principle, we  see that $$\max_{k\in M_{n1}}\max_{(t_1,t_2)\in H_{1:n}^{1,1}(k) }D_n^{(3)}(t_1,k,t_2)^{\top}L_n(t_1,k,t_2)^{-1}D_n^{(3)}(t_1,k,t_2)=O_p(1)=o_p(n\delta^2).$$

Therefore, using the result in (1)-(3),  we have  , $$\max_{k\in M_{n1}}\max_{(t_1,t_2)\in H_{1:n}^{1,1}}(n\delta^2)^{-1}D_n(t_1,k,t_2)^{\top}V_n(t_1,k,t_2)^{-1}D_n(t_1,k,t_2)=o_p(1).$$

(ii)  We now focus on the set $M_{n2}=\{k|k-k_1>\iota_n \text{ and } k_2-k >\iota_n\}$ and analyze the behavior of $T_{n}(t_1,k,t_2)$ on three subsets of $(t_1,t_2)\in H_{1:n}(k)$. Define
\begin{align*}
H_{1:n}^{2,0}(k)&=H_{1:n}(k)\cap \{(t_1,t_2)|t_1>k_1\},\\
H_{1:n}^{2,1}(k)&=H_{1:n}(k)\cap \{(t_1,t_2)|t_1\leq k_1,t_2\leq k_2\},\\
H_{1:n}^{2,2}(k)&=H_{1:n}(k)\cap \{(t_1,t_2)|t_1\leq k_1,t_2> k_2\}.
\end{align*}

For the behavior of $T_n(t_1,k,t_2)$ on $(t_1,t_2)\in H_{1:n}^{2,0}(k)$ and $(t_1,t_2)\in H_{1:n}^{2,1}(k)$, it reduces to one change-point case.

Forthe behavior of $T_n(t_1,k,t_2)$ on $(t_1,t_2)\in H_{1:n}^{2,2}(k)$, we have that  \begin{flalign*}
&D_n(t_1,k,t_2)\\=&-\frac{(k_1-t_1+1)(t_2-k)}{(t_2-t_1+1)^{3/2}}\delta\eta_1-\frac{(k-t_1+1)(t_2-k_2)}{(t_2-t_1+1)^{3/2}}\delta\eta_2+\frac{(k-t_1+1)(t_2-k)}{(t_2-t_1+1)^{3/2}}[\frac{S^{X}_{t_1,k}}{(k-t_1+1)}-\frac{S^{X}_{k+1,t_2}}{(t_2-k)}]\\:=&\sum_{i=1}^{3}D_n^{(i)}(t_1,k,t_2),
\end{flalign*}
and 
\begin{flalign*}
L_n(t_1,k,t_2)=&\Big[\sum_{i=t_1}^{k_1}+\sum_{i=k_1+1}^{k}\Big]\frac{(i-t_1+1)^2(k-i)^2}{(t_2-t_1+1)^{2}(k-t_1+1)^2}\Big\{\frac{S_{t_1,i}}{i-t_1+1}-\frac{S_{i+1,k}}{k-i}\Big\}^{\otimes 2}\\
:=&L_{n1}(t_1,k,t_2)+L_{n2}(t_1,k,t_2),\\
R_n(t_1,k,t_2)=&\Big[\sum_{i=k+1}^{k_2}+\sum_{i=k_2+1}^{t_2}\Big]\frac{(t_2-i+1)^2(i-1-k)^2}{(t_2-t_1+1)^{2}(t_2-k)^2}\Big\{\frac{S_{i,t_2}}{t_2-i+1}-\frac{S_{k+1,i-1}}{i-k-1}\Big\}^{\otimes 2}\\:=&R_{n1}(t_1,k,t_2)+R_{n2}(t_1,k,t_2).
\end{flalign*}
For $k-k_1$ and $k_2-k$, since $k_2-k+k-k_1=k_2-k_1=O(n)$, without loss of generality, we can assume $k_2-k=O(n)$, and for now we focus on $R_{n1}(t_1,k,t_2)$, where
$$
R_{n1}(t_1,k,t_2)=\sum_{i=k+1}^{k_2}\frac{(t_2-i+1)^2(i-1-k)^2}{(t_2-t_1+1)^{2}(t_2-k)^2}\Big\{\frac{S^X_{i,t_2}}{t_2-i+1}-\frac{S^X_{k+1,i-1}}{i-k-1}+\frac{(t_2-k_2)}{t_2-i+1}\delta_2\Big\}^{\otimes 2}
$$
Let $x_i=\frac{(i-1-k)}{(t_2-t_1+1)(t_2-k)}$, $y_i=\frac{(t_2-i+1)(i-1-k)}{(t_2-t_1+1)(t_2-k)}\Big\{\frac{S^X_{i,t_2}}{t_2-i+1}-\frac{S^X_{k+1,i-1}}{i-k-1}\Big\}$, and $z=(t_2-k_2)\delta_2$, and invoke Lemma \ref{lem_lambda}, we have that 
\begin{flalign*}
&R_{n1}(t_1,k,t_2)\\\geq& \sum_{i=k+1}^{k_2}\frac{(t_2-i+1)^2(i-1-k)^2}{(t_2-t_1+1)^{2}(t_2-k)^2}\Big\{\frac{S^X_{i,t_2}}{t_2-i+1}-\frac{S^X_{k+1,i-1}}{i-k-1}\Big\}^{\otimes 2}-
\Big(\frac{1}{n}\sum_{i=k+1}^{k_2}n^2\frac{(i-1-k)^2}{(t_2-t_1+1)^{2}(t_2-k)^2}\Big)^{-1}\\&\times
\Big\{\frac{n^{4}}{(t_2-t_1+1)^2(t_2-k)^2}\frac{1}{n}\sum_{i=k+1}^{k_2}\frac{(t_2-i+1)}{n}\Big[\frac{i-k-1}{n}\frac{S^X_{i,t_2}}{\sqrt{n}}-\frac{t_2-i+1}{n}\frac{S^X_{k+1,i-1}}{\sqrt{n}}\Big]\Big\}^{\otimes 2}
\\:=&A_n.
\end{flalign*}
Note here $k_2-k=O(n)$, hence we can show that the RHS of the above inequality will converge in distribution to 
\begin{flalign*}
&\Sigma_X^{1/2}\int_{u}^{\tau_2}\frac{1}{(u_2-u_1)^2(u_2-u)^2}\Big\{(s-u)[\mathcal{B}_d(u_2)-\mathcal{B}_d(s)]-(u_2-s)[\mathcal{B}_d(s)-\mathcal{B}_d(u)]\Big\}^{\otimes 2}ds\Sigma_X^{1/2}\\-& 
\Sigma_X^{1/2}	\Big(\int_{u}^{\tau_2}\frac{(s-u)^2}{(u_2-u_1)^2(u_2-u)^2}ds\Big)^{-1}\\&\times
\Big(\frac{1}{(u_2-u_1)^2(u_2-u)^2}\int_{u}^{\tau_2}(u_2-s)\Big\{(s-u)[\mathcal{B}_d(u_2)-\mathcal{B}_d(s)]-(u_2-s)[\mathcal{B}_d(s)-\mathcal{B}_d(u)]\Big\}ds\Big)^{\otimes 2}\Sigma_X^{1/2},
\end{flalign*}
which is of order $O_p(1)$ and positive definite almost surely (by the fact that independent Brownian Motions are linearly uncorrelated almost surely).  This implies that 
$$
V_n(t_1,k,t_2)^{-1}\leq R_{n1}(t_1,k,t_2)^{-1}\leq A_n^{-1}=O_p(1).
$$
Note that $D_n^{(3)}(t_1,k,t_2)=O_p(1)$, hence
$$
\max_{k\in M_{n1}}\max_{(t_1,t_2)\in H_{1:n}^{2,2}(k,c_n) }D_n^{(3)}(t_1,k,t_2)^{\top}V_n(t_1,k,t_2)^{-1}D_n^{(3)}(t_1,k,t_2)
\leq D_n^{(3)}(t_1,k,t_2)^{\top}A_n^{-1}D_n^{(3)}(t_1,k,t_2)=O_p(1),
$$
this implies that $\max_{k\in M_{n1}}\max_{(t_1,t_2)\in H_{1:n}^{2,2}(k,c_n) }(n\delta^2)^{-1}D_n^{(3)}(t_1,k,t_2)^{\top}V_n(t_1,k,t_2)^{-1}D_n^{(3)}(t_1,k,t_2)=o_p(1)$.

\noindent We further split $H_{1:n}^{2,2}(k)$ into four subsets:
\begin{align*}
&(1)~H_{1:n}^{2,21}(k,c_n)=H_{1:n}^{2,2}(k)\cap \{(t_1,t_2)|t_1\leq k_1-c_n \text{ and }t_2> k_2+c_n\},\\
&(2)~H_{1:n}^{2,22}(k,c_n)=H_{1:n}^{2,2}(k)\cap \{(t_1,t_2)|t_1>k_1-c_n \text{ and }t_2\leq k_2+c_n\},\\
&(3)~H_{1:n}^{2,23}(k,c_n)=H_{1:n}^{2,2}(k)\cap \{(t_1,t_2)|t_1\leq k_1-c_n \text{ and }t_2\leq k_2+c_n\},\\
&(4)~H_{1:n}^{2,24}(k,c_n)=H_{1:n}^{2,2}(k)\cap \{(t_1,t_2)|t_1> k_1-c_n \text{ and }t_2> k_2+c_n\},
\end{align*}
where  $c_n$ satisfies (\ref{cn}).

\noindent Using the Cauchy-Schwarz inequality, it suffices to show that 
for $j=1,2,3,4$, and $i=1,2$, 
$$
\max_{k\in M_{n1}}\max_{(t_1,t_2)\in H_{1:n}^{2,2j}(k,c_n) }D_n^{(i)}(t_1,k,t_2)^{\top}V_n(t_1,k,t_2)^{-1}D_n^{(i)}(t_1,k,t_2)=o_p(n\delta^2).
$$
To proceed, we remark here that 
\begin{flalign*}
L_{n1}(t_1,k,t_2)=&\sum_{i=t_1}^{k_1}\frac{(i-t_1+1)^2(k-i)^2}{(t_2-t_1+1)^{2}(k-t_1+1)^2}\Big\{\frac{S^X_{t_1,i}}{i-t_1+1}-\frac{S^X_{i+1,k}}{k-i}-\frac{(k-k_1)}{k_1-i}\delta_1\Big\}^{\otimes 2},\\
R_{n2}(t_1,k,t_2)=&\sum_{i=k_2+1}^{t_2}\frac{(t_2-i+1)^2(i-1-k)^2}{(t_2-t_1+1)^{2}(t_2-k)^2}\Big\{\frac{S^X_{i,t_2}}{t_2-i+1}-\frac{S^X_{k+1,i-1}}{i-k-1}+\frac{(k_2-k)}{i-k-1}\delta_2\Big\}^{\otimes 2}.
\end{flalign*}

(1) For $k\in M_{n2}$ and $(t_1,t_2)\in H_{1:n}^{2,21}(k,c_n)$, $(t_1,k)$ contains one change-point $k_1$. Hence, 
a similar argument as in Lemma \ref{lem_mean1} would yield that 
\begin{flalign}\label{L1}
\max_{k\in M_{n1}}\max_{(t_1,t_2)\in H_{1:n}^{2,21}(k,c_n) }\eta_1^{\top}V_{n}(t_1,k,t_2)^{-1}\eta_1\leq \max_{k\in M_{n1}}\max_{(t_1,t_2)\in H_{1:n}^{2,21}(k,c_n) }\eta_1^{\top}L_{n1}(t_1,k,t_2)^{-1}\eta_1= O_p(\frac{n^{4+2\kappa}}{c_n^3\iota_n^2}).
\end{flalign}
Similarly, $(k,t_2)$ contains one change-point $k_2$, hence 
\begin{flalign}\label{R2}
\max_{k\in M_{n1}}\max_{(t_1,t_2)\in H_{1:n}^{2,21}(k,c_n) }\eta_2^{\top}R_{n2}(t_1,k,t_2)^{-1}\eta_2\leq O_p(\frac{n^{4+2\kappa}}{c_n^3\iota_n^2}).
\end{flalign}
Therefore, recall $c_n$ satisfies (\ref{cn}), so that $\frac{n^{4+2\kappa}}{c_n^3\iota_n^2}\to 0$  as $n\to\infty$. Hence,  we have that $$(n\delta^2)^{-1}\max_{k\in M_{n1}}\max_{(t_1,t_2)\in H_{1:n}^{2,21}(k,c_n) }\sum_{i=1}^{2}D_n^{(i)}(t_1,k,t_2)^{\top}V_n(t_1,k,t_2)^{-1}D_n^{(i)}(t_1,k,t_2)=o_p(1).$$

(2)
Note that in this case, 
\begin{flalign}\label{L2}
\begin{split}
&(n\delta^2)^{-1}\max_{k\in M_{n1}}\max_{(t_1,t_2)\in H_{1:n}^{2,22}(k,c_n) }D_n^{(1)}(t_1,k,t_2)^{\top}V_n(t_1,k,t_2)^{-1}D_n^{(1)}(t_1,k,t_2)\\\leq&\max_{k\in M_{n1}}\max_{(t_1,t_2)\in H_{1:n}^{2,22}(k,c_n) } C\frac{(k_1-t_1)^2}{n^2}\eta_1 A_n^{-1}\eta_1=O_p(\frac{c_n^2}{n^2})=o_p(1)
\end{split}
\end{flalign}
and 
\begin{flalign}\label{R1}
\begin{split}
&(n\delta^2)^{-1}\max_{k\in M_{n1}}\max_{(t_1,t_2)\in H_{1:n}^{2,22}(k,c_n) }D_n^{(2)}(t_1,k,t_2)^{\top}V_n(t_1,k,t_2)^{-1}D_n^{(2)}(t_1,k,t_2)\\\leq&\max_{k\in M_{n1}}\max_{(t_1,t_2)\in H_{1:n}^{2,22}(k,c_n) } C\frac{(t_2-k_2)^2}{n^2}\eta_1 A_n^{-1}\eta_1=O_p(\frac{c_n^2}{n^2})=o_p(1).
\end{split}
\end{flalign}

(3)
Using (\ref{L1}) and (\ref{R1}), the result follows.

(4)
Using (\ref{R2}) and (\ref{L2}), the result follows.

\subsection{Three or more change-points and consistency arguments}

By  using the results shown in the cases of  one change-point and two change-points, the argument follows from subsection \ref{sec:three} and \ref{sec:con} with minor modifications, and thus is omitted.

\subsection{Lemmas}
\begin{lemma}\label{lem_tripathi}
Let $x_i\in\mathbb{R}^p$ and $y_i\in\mathbb{R}^q$, $i=1,\cdots,n$ such that, almost surely, $\|x_i\|^2<\infty$ and $\|y_i\|^2<\infty$. If $\sum_{i=1}^{n}y_iy_i^{\top}$ is non-singular, then, almost surely,
$$
\Big[\sum_{i=1}^{n}x_iy_i^{\top}\Big]\Big[\sum_{i=1}^{n}y_iy_i^{\top}\Big]^{-1}\Big[\sum_{i=1}^{n}x_iy_i^{\top}\Big]\leq \sum_{i=1}^{n}x_ix_i^{\top}.
$$
The equality holds if and only if $x^{\top}a_i+y^{\top}b_i=0$ almost surely, $i=1,\cdots,n$, for some $(a,b)\in\mathbb{R}^{p}\times \mathbb{R}^{q}$.
\end{lemma}

\noindent\textsc{Proof of Lemma \ref{lem_tripathi}}

See \cite{tripathi1999matrix} and remarks therein.

\begin{lemma}\label{lem_lambda}
Let $x_i\in\mathbb{R}$ such that $\sum_{i=1}^{n}x_i^2\neq 0$, $y_i\in\mathbb{R}^d$, $i=1,\cdots,n$, and $z\in\mathbb{R}^d$. Furthermore, let $$A_n=\sum_{i=1}^{n}\Big(y_iy_i^{\top}+x_izy_i^{\top}+x_iy_iz^{\top}+x_i^2zz^{\top}\Big),$$
and 
$$B_n=\sum_{i=1}^{n}y_iy_i^{\top}-\frac{\Big(\sum_{i=1}^{n}x_iy_i\Big)\Big(\sum_{i=1}^{n}x_iy_i\Big)^{\top}}{\sum_{i=1}^{n}x_i^2}.$$
Then,
$$
A_n-B_n\geq0,\quad\mbox{and}\quad B_n\geq 0,
$$
here a square matrix $A\geq 0$ indicates that $A$ is semi-positive definite.

\end{lemma}

\noindent\textsc{Proof of Lemma \ref{lem_lambda}}

For any $a\in\mathbb{R}^d$, we have 
\begin{flalign*}
&a^{\top} A_{n}(t_1,k,t_2)a\\=&	a^{\top}\Big(\sum_{i=1}^{n}y_iy_i^{\top}\Big)a+2(a^{\top}z)\Big(\sum_{i=1}^{n}x_iy_i^{\top}a\Big)+\Big(\sum_{i=1}^{n}x_i^2\Big)(a^{\top}z)^2
\\=&a^{\top}\Big(\sum_{i=1}^{n}y_iy_i^{\top}\Big)a+\Big[\big(\sum_{i=1}^{n}x_i^2\big)^{1/2}a^{\top}z+\big(\sum_{i=1}^{n}x_i^2\big)^{-1/2}\big(\sum_{i=1}^{n}x_iy_i^{\top}\big)a\Big]^2-\big(\sum_{i=1}^{n}x_i^2\big)^{-1}\big(\sum_{i=1}^{n}x_iy_i^{\top}a\big)^2\\
=&a^{\top}B_na+\Big[\big(\sum_{i=1}^{n}x_i^2\big)^{1/2}a^{\top}z+\big(\sum_{i=1}^{n}x_i^2\big)^{-1/2}\big(\sum_{i=1}^{n}x_iy_i^{\top}\big)a\Big]^2
\\\geq& a^{\top}B_na.
\end{flalign*}
Note the above results hold for any $a\in\mathbb{R}^d$, this implies that $A_n-B_n$ is semi-positive definite. 

In addition, by Lemma \ref{lem_tripathi}, since $\sum_{i=1}^{n}x_i^2>0$, we have 
$$
\Big(\sum_{i=1}^{n}x_iy_i\Big)\Big(\sum_{i=1}^{n}x_i^2\Big)^{-1}\Big(\sum_{i=1}^{n}x_iy_i\Big)^{\top}\leq \sum_{i=1}^{n}y_iy_i^{\top}.
$$

$\hfill \square$

\begin{lemma}\label{lem_mean1}
For the one change-point case, 
\begin{flalign}\label{eq:suffice1}
\max_{k\in M_{n1}}\max_{(t_1,t_2)\in H_{1:n}^2(k,c_n) }{\eta_1^{\top}V_n(t_1,k,t_2)^{-1}\eta_1}\leq O_p(\frac{n^{4+2\kappa}}{c_n^3\iota_n^2})=o_p(1).
\end{flalign}
\end{lemma}

\noindent\textsc{Proof of Lemma \ref{lem_mean1}}

Note that $V_n(t_1,k,t_2)=L_n(t_1,k,t_2)+R_n(t_1,k,t_2)$, and we further decompose 
\begin{flalign*}
R_n(t_1,k,t_2)=&[\sum_{i=k+1}^{k_1}+\sum_{i=k_1+1}^{t_2}]\frac{(t_2-i+1)^2(i-1-k)^2}{(t_2-t_1+1)^2(t_2-k)^2}\Big\{\frac{S_{i,t_2}}{t_1-i+1}-\frac{S_{k+1,i-1}}{i-k-1}\Big\}^{\otimes 2}
\\:=&R_{n1}(t_1,k,t_2)+R_{n2}(t_1,k,t_2),
\end{flalign*}
where 
\begin{flalign*}
R_{n2}(t_1,k,t_2)=&\sum_{i=k_1+1}^{t_2}\frac{(t_2-i+1)^2(i-1-k)^2}{(t_2-t_1+1)^2(t_2-k)^2}\Big\{\frac{S^X_{i,t_2}}{(t_2-i+1)}-\frac{S^X_{k+1,i-1}}{(i-k-1)}+\frac{k_1-k}{i-k-1}\delta\eta_1\Big\}^{\otimes 2}\\
=&(t_2-k_1)\frac{1}{t_2-k_1}\sum_{i=k_1+1}^{t_2}y_i(t_1,k,t_2)y_i(t_1,k,t_2)^{\top}
\\:=&(t_2-k_1)Z_{yy}(k_1,t_2)
\end{flalign*}
with
\begin{flalign*}
y_i(t_1,k,t_2)=\frac{(t_2-i+1)(i-1-k)}{(t_2-t_1+1)(t_2-k)}\Big\{\frac{S^X_{i,t_2}}{(t_2-i+1)}-\frac{S^X_{k+1,i-1}}{(i-k-1)}+\frac{k_1-k}{i-k-1}\delta\eta_1\Big\}.
\end{flalign*}
Now, let 
\begin{flalign*}
Z_{y}(k_1,t_2)=&\frac{1}{t_2-k_1}\sum_{i=k_1+1}^{t_2}y_i(t_1,k,t_2)
\\=&\frac{n^{3/2}}{(t_2-k_1)(t_2-t_1+1)(t_2-k)}\sum_{i=k_1+1}^{t_2}\Big\{\frac{(i-k-1)}{n}\frac{S^X_{i,t_2}}{\sqrt{n}}-\frac{(t_2-i+1)}{n}\frac{S^X_{k+1,i-1}}{\sqrt{n}}\Big\}
\\&+\frac{(k_1-k)n^{-\kappa} }{(t_2-k_1)(t_2-t_1+1)(t_2-k)}\sum_{i=k_1+1}^{t_2}(t_2-i+1)\eta_1
\end{flalign*}
Recall that $t_2-t_1+1=O(n)$ and $t_2-k=O(n)$, we have that 
\begin{flalign*}
Z_{y}(k_1,t_2)=&\frac{C}{\sqrt{n}}\frac{1}{t_2-k_1}\sum_{i=k_1+1}^{t_2}\Big\{\frac{(i-k-1)}{n}\frac{S^X_{i,t_2}}{\sqrt{n}}-\frac{(t-i+1)}{n}\frac{S^X_{k+1,i-1}}{\sqrt{n}}\Big\}
+\frac{C(k_1-k)n^{-\kappa} }{n^2}(t_2-k_1)\eta_1\\
:=&H_n(t_1,k,t_2)+\frac{C(k_1-k)n^{-\kappa} }{n^2}(t_2-k_1)\eta_1
\end{flalign*}
In addition, by Lemma \ref{lem_tripathi}, we have 
$$
Z_{y}^{\top} Z_{yy}^{-1}Z_{y}\leq 1.
$$
This implies that 
$$
(t_2-k_1)Z_{y}^{\top}R_{n2}(t_1,k,t_2)^{-1}Z_{y}\leq 1.
$$
By the inequality that $\frac{3a^2}{4}-3b^2\leq (a+b)^2$, we have 
\begin{flalign*}
\frac{3C^2(t_2-k_1)^3(k_1-k)^2n^{-2\kappa}}{4n^4}\eta_1^{\top}V_n(t_1,k,t_2)^{-1}\eta_1-3(t_2-k_1)H(t_1,k,t_2)^{\top}V_n(t_1,k,t_2)^{-1}H_n(t_1,k,t_2)\leq 1
\end{flalign*}
That is, 
\begin{flalign}\label{e1Ve1}
\eta_1^{\top}V_n(t_1,k,t_2)^{-1}\eta_1\leq \frac{4n^{4+2\kappa}}{3C^2(t_2-k_1)^3(k_1-k)^2}\Big[3(t_2-k_1)H(t_1,k,t_2)^{\top}V_n(t_1,k,t_2)^{-1}H(t_1,k,t_2)+1\Big]
\end{flalign}
By the invariance principle, we can show that  
$$
\frac{1}{t_2-k_1}\sum_{i=k_1+1}^{t_2}\Big\{\frac{(i-k-1)}{n}\frac{S^X_{i,t_2}}{\sqrt{n}}-\frac{(t-i+1)}{n}\frac{S^X_{k+1,i-1}}{\sqrt{n}}\Big\}=O_p(1),$$
hence 
$$
\sqrt{t_2-k_1}H(t_1,k,t_2)=O_p(\frac{C\sqrt{t_2-k_1}}{\sqrt{n}}).
$$
Therefore, note that $V_{n}(t_1,k,t_2)^{-1}\leq L_n(t_1,k,t_2)^{-1}$, using (\ref{e1Ve1}), the fact that $c_n<t_2-k_1\leq n$ when $(t_1,t_2)\in H^{2}_{1:n}(k,c_n)$, and $k_1-k>\iota_n$, we have 
$$
\max_{k\in M_{n1}}\max_{(t_1,t_2)\in H_{1:n}^2(k,c_n) }\eta_1^{\top}V_n(t_1,k,t_2)^{-1}\eta_1\leq\frac{4n^{4+2\kappa}}{3C^2(t_2-k_1)^3(k_1-k)^2} \Big[1+O_p(1)\Big]=O_p(\frac{n^{4+2\kappa}}{\iota_n^2c_n^3})=o_p(1),
$$
where the last equality holds by (\ref{cn}). 

$\hfill \square$

\begin{lemma}\label{lem_mean2}
For the two change-point case, 
\begin{flalign}
\max_{k\in M_{n1}}\max_{(t_1,t_2)\in H_{1:n}^{1,1}(k) }\eta_1^{\top}V_n(t_1,k,t_2)^{-1}\eta_1\leq O_p(\frac{n^{1+2\kappa}}{\iota_n^2})=o_p(1).
\end{flalign}
\end{lemma} 
\noindent\textsc{Proof of Lemma \ref{lem_mean2}}

Recall
\begin{flalign*}
R_n(t_1,k,t_2)=\Big[\sum_{i=k+1}^{k_1}+\sum_{k_1+1}^{k_2}+\sum_{k_2+1}^{t_2}\Big]\frac{(t_2-i+1)^2(i-1-k)^2}{(t_2-t_1+1)^2(t_2-k)^2}\Big\{\frac{S_{i,t_2}}{(t_2-i+1)}-\frac{S_{k+1,i-1}}{(i-k-1)}\Big\}^{\otimes 2},
\end{flalign*}
we focus on 
$$
R_{n2}(t_1,k,t_2)=\sum_{k_1+1}^{k_2}\frac{(t_2-i+1)^2(i-1-k)^2}{(t_2-t_1+1)^2(t_2-k)^2}\Big\{\frac{S^X_{i,t_2}}{(t_2-i+1)}-\frac{S^X_{k+1,i-1}}{(i-k-1)}+\frac{t_2-k_2}{t_2-i+1}\delta_2+\frac{k_1-k}{i-k-1}\delta_1\Big\}^{\otimes 2}.
$$

Define $x_i=\frac{(i-1-k)}{(t_2-t_1+1)(t_2-k)}$ , $y_i=\frac{(t_2-i+1)(i-1-k)}{(t_2-t_1+1)(t_2-k)}\Big\{\frac{S^X_{i,t_2}}{(t_2-i+1)}-\frac{S^X_{k+1,i-1}}{(i-k-1)}+\frac{k_1-k}{i-k-1}\delta_1\Big\}$ and $z=(t_2-k_2)\delta_2$, and invoke Lemma \ref{lem_lambda}, we obtain that 
\begin{flalign*}
\frac{1}{k_2-k_1}R_{n2}(t_1,k,t_2)\geq& \frac{1}{k_2-k_1}\sum_{k_1+1}^{k_2}y_iy_i^{\top}-\frac{\Big(\frac{1}{k_2-k_1}\sum_{k_1+1}^{k_2}x_iy_i\Big)\Big(\frac{1}{k_2-k_1}\sum_{k_1+1}^{k_2}x_iy_i\Big)^{\top}}{\frac{1}{k_2-k_1}\sum_{k_1+1}^{k_2}x_i^2}\\:=&Z_{yy}-(Z_{xx})^{-1}Z_{xy}Z_{xy}^{\top},
\end{flalign*}
Here 
\begin{flalign*}
&Z_{yy}\\=&\frac{1}{k_2-k_1}\sum_{k_1+1}^{k_2}y_iy_i^{\top}\\=&\frac{n^3}{(k_2-k_1)(t_2-t_1+1)^2(t_2-k)^2}\sum_{k_1+1}^{k_2}\Big\{\frac{(i-k-1)}{n}\frac{S^X_{i,t_2}}{\sqrt{n}}-\frac{(t_2-i+1)}{n}\frac{S^X_{k+1,i-1}}{\sqrt{n}}+\frac{(k_1-k)(t_2-i+1)}{n^{3/2}}\delta_1\Big\}^{\otimes 2},\\
&Z_{xy}\\=&\frac{1}{k_2-k_1}\sum_{k_1+1}^{k_2}x_iy_i\\=&\frac{n^{5/2}}{(k_2-k_1)(t_2-t_1+1)^2(t_2-k)^2}\sum_{k_1+1}^{k_2}\frac{i-1-k}{n}\Big\{\frac{(i-k-1)}{n}\frac{S^X_{i,t_2}}{\sqrt{n}}-\frac{(t_2-i+1)}{n}\frac{S^X_{k+1,i-1}}{\sqrt{n}}\Big\}
\\&+\frac{k_1-k}{k_2-k_1}\sum_{k_1+1}^{k_2}\frac{(t_2-i+1)(i-1-k)}{(t_2-t_1+1)^2(t_2-k)^2}\delta_1,\\
&Z_{xx}\\=&\frac{1}{(k_2-k_1)(t_2-t_1+1)^2(t_2-k)^2}\sum_{k_1+1}^{k_2}(i-1-k)^2
\end{flalign*}
To simplify the notation, we let $Q_i=\frac{(i-k-1)}{n}\frac{S^X_{i,t_2}}{\sqrt{n}}-\frac{(t_2-i+1)}{n}\frac{S^X_{k+1,i-1}}{\sqrt{n}}$, so that 
\begin{flalign*}
&(k_2-k_1)[Z_{yy}-(Z_{xx})^{-1}Z_{xy}Z_{xy}^{\top}]\\=&
\frac{n^3}{(t_2-t_1+1)^2(t_2-k)^2}\Big\{\sum_{k_1+1}^{k_2}Q_iQ_i^{\top}-\frac{n^2}{\sum_{k_1+1}^{k_2}(i-1-k)^2}\Big(\sum_{k_1+1}^{k_2}\frac{(i-k-1)}{n}Q_i\Big)\Big(\sum_{k_1+1}^{k_2}\frac{(i-k-1)}{n}Q_i\Big)^{\top}\Big\}\\
&+\frac{(k_1-k)^2}{(t_2-t_1+1)^2(t_2-k)^2}\Big\{{\sum_{k_1+1}^{k_2}(t_2-i+1)^2}-\frac{\Big(\sum_{k_1+1}^{k_2}(i-k-1)(t_2-i+1)\Big)^2}{\sum_{k_1+1}^{k_2}(i-k-1)^2}\Big\}\delta_1\delta_1^{\top}\\&+\frac{n^{3/2}(k_1-k)}{(t_2-t_1+1)^2(t_2-k)^2}\Big\{\sum_{k_1+1}^{k_2}(t_2-i+1)Q_i-\frac{n\Big(\sum_{k_1+1}^{k_2}\frac{i-1-k}{n}Q_i\Big)\Big(\sum_{k_1+1}^{k_2}(i-1-k)(t_2-i+1)\Big)}{\sum_{k_1+1}^{k_2}(i-1-k)^2}\Big\}
\delta_1^{\top}
\\&+\delta_1\frac{n^{3/2}(k_1-k)}{(t_2-t_1+1)^2(t_2-k)^2}\Big\{\sum_{k_1+1}^{k_2}(t_2-i+1)Q_i-\frac{n\Big(\sum_{k_1+1}^{k_2}\frac{i-1-k}{n}Q_i\Big)\Big(\sum_{k_1+1}^{k_2}(i-1-k)(t_2-i+1)\Big)}{\sum_{k_1+1}^{k_2}(i-1-k)^2}\Big\}^{\top}
\\:=&A_n^X(t_1,k,t_2)+\delta^2\zeta_n(t_1,k,t_2)\eta_1\eta_1^{\top}+\delta\eta_1H_n(t_1,k,t_2)^{\top}+\delta H_n(t_1,k,t_2)\eta_1^{\top}:=A_n(t_1,k,t_2)
\end{flalign*}
By Lemma \ref{lem_inverse}, and let 
\begin{flalign*}
C_n(t_1,k,t_2)=&\Big(1+\delta\eta_1^{\top}A_{n}^X(t_1,k,t_2)^{-1}H_n(t_1,k,t_2)\Big)^2\\&-\delta^2H_n(t_1,k,t_2)^{\top}A_{n}^X(t_1,k,t_2)^{-1}H_n(t_1,k,t_2)\eta_1^{\top}A_{n}^X(t_1,k,t_2)^{-1}\eta_1,  
\end{flalign*}
we obtain that 
\begin{flalign*}
\eta_1^{\top}A_n(t_1,k,t_2)^{-1}\eta_1
=&\frac{\eta_1^{\top}A_{n}^X(t_1,k,t_2)^{-1}\eta_1}{C_n(t_1,k,t_2)+\eta_1^{\top}A_{n}^X(t_1,k,t_2)^{-1}\eta_1\delta^2\zeta_n(t_1,k,t_2)}
\\=&\frac{n^{1+2\kappa}}{(k_1-k)^2}\frac{\eta_1^{\top}A_{n}^X(t_1,k,t_2)^{-1}\eta_1}{\frac{n^{1+2\kappa}}{(k_1-k)^2}C_n(t_1,k,t_2)+\eta_1^{\top}A_{n}^X(t_1,k,t_2)^{-1}\eta_1\frac{n}{(k_1-k)^2}\zeta_n(t_1,k,t_2)}
\end{flalign*}
Here, using the invariance principle, we can see that 
\begin{flalign*}
&\Big\{A_n^X(\lfloor nu_1\rfloor,\lfloor nu\rfloor,\lfloor nu_2\rfloor)\Big\}_{(0<u_1<u<u_2<1)}\\\Rightarrow &\Big\{\frac{1}{(u_2-u_1)^2(u_2-u)^2}\Sigma_X^{1/2}\Big[
\int_{\tau_1}^{\tau_2}\Big((s-u)(\mathcal{B}_d(u_2)-\mathcal{B}_d(s))-(u_2-s)(\mathcal{B}_d(s)-\mathcal{B}_d(u))\Big)^{\otimes 2}ds\\&-\Big({\int_{\tau_1}^{\tau_2}(s-u)^2ds}\Big)^{-1}\\&\times\Big(\int_{\tau_1}^{\tau_2}(s-u)^2(\mathcal{B}_d(u_2)-\mathcal{B}_d(s))-(u_2-s)(s-u)(\mathcal{B}_d(s)-\mathcal{B}_d(u))ds\Big)^{\otimes2}
\Big]\Sigma_X^{1/2}\Big\}_{(0<u_1<u<u_2<1)}\\
:=&\Big\{A^X(u_1,u,u_2)\Big\}_{(0<u_1<u<u_2<1)},\\
&\Big\{\frac{n}{(k_1-k)^2}\zeta_n(\lfloor nu_1\rfloor,\lfloor nu\rfloor,\lfloor nu_2\rfloor\Big\}_{(0<u_1<u<u_2<1)}\\\Rightarrow&\Big\{ \frac{1}{(u_2-u_1)^2(u_2-u)^2}\Big[\Big(\int_{\tau_1}^{\tau_2}(u_2-s)^2ds\Big)-\Big(\int_{\tau_1}^{\tau_2}(s-u)^2ds\Big)^{-1}\Big(\int_{\tau_1}^{\tau_2}(s-u)(u_2-s)ds\Big)^2\Big]\Big\}_{(0<u_1<u<u_2<1)}\\
:=&\Big\{\zeta(u_1,u,u_2)\Big\}_{(0<u_1<u<u_2<1)},\\
&\frac{n^{1/2}}{(k_1-k)}\Big\{H_n(\lfloor nu_1\rfloor,\lfloor nu\rfloor,\lfloor nu_2\rfloor)\Big\}_{(0<u_1<u<u_2<1)}\\\Rightarrow &\Big\{\frac{1}{(u_2-u_1)^2(u_2-u)^2}\Sigma_X^{1/2}\Big[
\Big(\int_{\tau_1}^{\tau_2}(s-u)(u_2-s)(\mathcal{B}_d(u_2)-\mathcal{B}_d(s))-(u_2-s)^2(\mathcal{B}_d(s)-\mathcal{B}_d(u))ds\Big)\\&-\Big({\int_{\tau_1}^{\tau_2}(s-u)^2ds}\Big)^{-1}\Big(\int_{\tau_1}^{\tau_2}(s-u)^2(\mathcal{B}_d(u_2)-\mathcal{B}_d(s))-(u_2-s)(s-u)(\mathcal{B}_d(s)-\mathcal{B}_d(u))ds\Big)\\&\times\Big(\int_{\tau_1}^{\tau_2}(s-u)(u_2-s)ds\Big)
\Big]\Big\}_{(0<u_1<u<u_2<1)}\\
:=&\Big\{H(u_1,u,u_2)\Big\}_{(0<u_1<u<u_2<1)}.
\end{flalign*}
Note that $\frac{n^{1/2+\kappa}}{k_1-k}\leq \frac{n^{1/2+\kappa}}{\iota_n}=o(1)$, this and above convergence results imply that 
\begin{flalign*}
&\frac{n^{1+2\kappa}}{(k_1-k)^2}\Big\{C_n(\lfloor nu_1\rfloor,\lfloor nu\rfloor,\lfloor nu_2\rfloor)\Big\}_{(0<u_1<u<u_2<1)}\\\Rightarrow& \Big\{\big(\eta_1^{\top}A^X(u_1,u,u_2)H(u_1,u,u_2)\big)^2-\big(\eta_1^{\top}A^X(u_1,u,u_2)\eta\big)\big(H(u_1,u,u_2)^{\top}A^X(u_1,u,u_2)H(u_1,u,u_2)\big)\Big\}_{(0<u_1<u<u_2<1)}
\end{flalign*}
Hence, we have $$\max_{k\in M_{n1}}\max_{(t_1,t_2)\in H_{1:n}^{1,1}(k) }\eta_1^{\top}A_n(t_1,k,t_2)^{-1}\eta_1=O_p(\frac{n^{1+2\kappa}}{(k_1-k)^2})\leq O_p(\frac{n^{1+2\kappa}}{\iota_n^2})=o_p(1).$$

$\hfill \square$

\begin{lemma}\label{lem_inverse}
For the invertible matrices $A_n(t_1,k,t_2)$ and $B_n(t_1,k,t_2)$ defined by 
\begin{flalign*}
A_n(t_1,k,t_2)= &B_{n}(t_1,k,t_2)+\delta^2\zeta_n(t_1,k,t_2)\eta\eta^{\top},\\
B_{n}(t_1,k,t_2)=&A^X_n(t_1,k,t_2)+\delta\eta_1 H_n(t_1,k,t_2)^{\top}+  \delta H_n(t_1,k,t_2) \eta^{\top},
\end{flalign*}
where $\eta\in\mathbb{R}^d/\{\bf{0}\}$ and $A^X_n(t_1,k,t_2)$ is invertible. Then
\begin{flalign*}
&\eta_1^{\top} A_{n}(t_1,k,t_2)^{-1}\eta_1=\frac{1}{\frac{1}{\eta_1^{\top}B_n(t_1,k,t_2)^{-1}\eta_1}+\zeta_n(t_1,k,t_2)\delta^2},\\
&\eta_1^{\top} B_{n}(t_1,k,t_2)^{-1}\eta_1	\\=&\frac{\eta_1^{\top}A_{n}^X(t_1,k,t_2)^{-1}\eta_1}{(1+\delta\eta_1^{\top}A_{n}^X(t_1,k,t_2)^{-1}H_n(t_1,k,t_2))^2-\delta^2H_n(t_1,k,t_2)^{\top}A_{n}^X(t_1,k,t_2)^{-1}H_n(t_1,k,t_2)\eta_1^{\top}A_{n}^X(t_1,k,t_2)^{-1}\eta_1}.
\end{flalign*}
\end{lemma}

\noindent\textsc{Proof of Lemma \ref{lem_inverse}}

By the well-known Sherman-Morrison formula~(see Lemma \ref{lem_sherman}), we have that
\begin{flalign}\label{A_n}
\eta_1^{\top}A_{n}(t_1,k,t_2)^{-1}\eta_1=\frac{\eta_1^{\top}B_{n}(t_1,k,t_2)^{-1}\eta_1}{1+\eta_1^{\top}B_n(t_1,k,t_2)^{-1}\eta_1 \zeta_n(t_1,k,t_2)\delta^2}.
\end{flalign}
By further applying the Sherman-Morrison formula twice on $B_n(t_1,k,t_2)$, we can establish that
\begin{flalign*}
&\eta_1^{\top}B_{n}(t_1,k,t_2)^{-1}\eta_1\\=&\frac{\eta_1^{\top}A^X_{n}(t_1,k,t_2)^{-1}\eta_1}{(1+\delta\eta_1^{\top}A_{n}^X(t_1,k,t_2)^{-1}H_n(t_1,k,t_2))^2-\delta^2\eta_1^{\top}A_{n}^X(t_1,k,t_2)^{-1}\eta_1 H_n(t_1,k,t_2)^{\top}A_{n}^X(t_1,k,t_2)^{-1}H_n(t_1,k,t_2)}.
\end{flalign*}

Specifically, let $C_n(t_1,k,t_2)=A_{n}^X(t_1,k,t_2)+ H_n(t_1,k,t_2)\delta\eta_1^{\top}$, applying the Sherman-Morrison formula, we have
$$
C_{n}(t_1,k,t_2)^{-1}=A_{n}^X(t_1,k,t_2)^{-1}-\frac{A_{n}^X(t_1,k,t_2)^{-1}H_n(t_1,k,t_2)\delta\eta_1 ^{\top}A_{n}^X(t_1,k,t_2)^{-1}}{1+\delta\eta_1^{\top}A_{n}^X(t_1,k,t_2)^{-1}H_n(t_1,k,t_2)},
$$
and thus
\begin{flalign*}
\eta_1^{\top}C_{n}(t_1,k,t_2)^{-1}\eta_1=\frac{\eta_1^{\top}A^X_{n}(t_1,k,t_2)^{-1}\eta_1}{1+\delta\eta_1^{\top}A^X_{n}(t_1,k,t_2)^{-1}H_n(k)}.
\end{flalign*}

Plug the above equation into,
\begin{flalign}\label{B_n}
\begin{split}
&\eta_1^{\top}B_{n}(t_1,k,t_2)^{-1}\eta_1=\frac{\eta_1^{\top}C_{n}(t_1,k,t_2)^{-1}\eta_1}{1+H_n(k)^{\top}C_n(t_1,k,t_2)^{-1}\delta\eta_1},
\end{split}
\end{flalign}
the result follows.
$\hfill \square$

\begin{lemma}\label{lem_sherman}
\textsc{(Sherman-Morrison Formula)} Suppose $A\in\mathbb{R}^{d\times d}$ is an invertible square matrix and $u,v\in\mathbb{R}^d$  are column vectors, then $A+uv^{\top}$ is invertible if and only if $1+v^{\top}A^{-1}u\neq 0$, and in this case 
$$
(A+uv^{\top})^{-1}=A^{-1}-\frac{A^{-1}uv^{\top}A^{-1}}{1+v^{\top}A^{-1}u}.
$$ 
\end{lemma}

\section{{Local Refinement of SNCP via CUSUM}}\label{sec:local_refine}
In this section, we further propose a simple and intuitive local refinement procedure for SNCP, which improves the localization error rate of SNCP from $O_p(n^{-1/2})$ to the optimal rate $O_p(n^{-1}).$

\Cref{sec:local_proc} proposes the local refinement procedure and further provides theoretical justification for the case of mean functional. \Cref{sec:local_simu} presents numerical evidence for the improvement brought by local refinement. Technical proofs are collected in \Cref{sec:local_theory}. All notations follow the definitions given in the main text unless otherwise noted.

\subsection{ A local refinement procedure for SNCP}\label{sec:local_proc}
For multiple change-point estimation in a univariate functional $\theta(\cdot)\in \mathbb{R}$, by \Cref{thm1}(ii), with probability going to 1, SNCP detects the correct number of change-points $m_o$ and achieves the localization error rate $\iota_n$ such that
\begin{align}\label{eq:sncp_local_rate}
\max_{1\leq i\leq m_o}|\widehat{k}_i-k_i|\leq \iota_n,
\end{align}
where $\iota_n$ is any sequence such that $\iota_n/n\to 0$ and $\iota_n^{-2}\delta^{-2}n\to 0$ as $n\to \infty$ and $\delta$ is the change size. Given a constant change size $\delta$, this implies that the best localization error rate delivered by SNCP is $O_p(n^{-1/2})$, which is slower than the optimal localization rate $O_p(n^{-1})$ for mean change~\citep{bai1994}.

Based on \eqref{eq:sncp_local_rate}, for any sequence $\widetilde{\iota}_n$ such that $\iota_n=O(\widetilde{\iota}_n)$ and $\widetilde{\iota}_n=o(n)$, we have that with probability going to 1, each local interval
\begin{align}\label{eq:local_interval}
[\widehat{k}_{i-1}+\widetilde{\iota}_n, \widehat{k}_{i+1}-\widetilde{\iota}_n] \text{ for } i=1,2,\cdots, m_o,
\end{align}
is of length at least $\epsilon_0 n$~(where $\epsilon_0$ is the minimum spacing between true change-points) and contains \textit{one and only one} change-point. (By convention, we define $\widehat{k}_0=-\widetilde{\iota}_n+1$ and $\widehat{k}_{m_o+1}=n+\widetilde{\iota}_n.$) In other words, SNCP can asymptotically isolate every single true change-point.

We then apply a CUSUM type statistic to each local interval in \eqref{eq:local_interval} to refine the estimated change-points $\{\widehat{k}_i\}_{i=1}^{m_o}$ by SNCP. Specifically, given a generic interval $[s,e]$, define the CUSUM statistic based on the subsample $\{Y_t\}_{t=s+1}^e$ as
\begin{align*}
T(k;s,e)=\frac{\sqrt{(k-s+1)(e-k)}}{e-s+1}\left[\widehat{\theta}_{k+1,e}-\widehat{\theta}_{s,k}\right].
\end{align*}

For each estimated change-point $\widehat{k}_i$ by SNCP, denote $[s_i^*,e_i^*]=[\widehat{k}_{i-1}+\widetilde{\iota}_n, \widehat{k}_{i+1}-\widetilde{\iota}_n]$ as its local interval in \eqref{eq:local_interval}, we define its locally refined estimator as
\begin{align}\label{eq:local_cusum}
\widetilde{k}_i^*={\arg\max}_{s_i^*+\widetilde{\iota}_n\leq k\leq e_i^*-\widetilde{\iota}_n}T^2(k;s_i^*, e_i^*), ~\text{ for } i=1,2,\cdots,m_o.
\end{align}

\textbf{Remark 1}: Note that the local refinement procedure works for any sequence $\widetilde{\iota}_n$ such that $\iota_n=O(\widetilde{\iota}_n)$ and $\widetilde{\iota}_n=o(n)$. In other words, it is essentially robust to the first stage localization error rate $\iota_n$ from SNCP. In practice, we can set $\widetilde{\iota}_n=\epsilon n/\log n$ to ensure that $\widetilde{\iota}_n$ is a valid sequence for the local refinement procedure, where $\epsilon$ is the window size parameter of SNCP.

\textbf{Remark 2}: Note that in \eqref{eq:local_cusum}, we define $\widetilde{k}_i^*$ as the argmax of CUSUM from $s_i^*+\widetilde{\iota}_n\leq k\leq e_i^*-\widetilde{\iota}_n$, i.e. we trim the first and last $\widetilde{\iota}_n$ data points in the local interval $[s_i^*,e_i^*]=[\widehat{k}_{i-1}+\widetilde{\iota}_n, \widehat{k}_{i+1}-\widetilde{\iota}_n]$. This is a technical consideration needed to control the random nature of the local interval, as it is based on the SNCP estimation $\{\widehat{k}_i\}_{i=1}^{m_o}$.

In the following, we show in \Cref{thm_refine} that the locally refined change-point estimator $\{\widetilde{k}_i^*\}_{i=1}^{m_o}$ achieves the optimal localization error rate for the mean functional, which is established by \cite{bai1994} and \cite{lavielle2000least} under the following data generating process for $\{Y_t \in \mathbb{R}\}_{t=1}^n$. Specifically, we have that
$$
Y_t=X_t+\mu_{i},\quad k_{i-1}+1\leq t\leq k_{i}, \quad \text{for}~ i=1,\cdots,m_o+1,
$$
where $\{X_t\}_{t=1}^n$ is a one dimensional stationary time series with $E(X_t)=0$, $k_0:=0<k_1<\cdots<k_{m_o}<k_{m_o+1}:=n$ denote the $m_o$ change-points, and $\mu_i\in\mathbb{R}$ denotes the mean of the $i$th segment. For $i=1,\cdots, m_o$, we assume that, $\mu_{i+1}-\mu_i=\delta_i=c_i\delta$ where $c_i\neq 0$ is a fixed constant. Thus, the overall change size is controlled by $\delta.$

Assuming the knowledge of $m_o=1$, \cite{bai1994} establishes the optimal $O_p(n^{-1})$ rate of an OLS-based change-point estimator when $\{X_t\}_{t=1}^n$ is a linear process. \cite{lavielle2000least} generalizes \cite{bai1994} by allowing $\{X_t\}_{t=1}^n$ to be a general short-range or long-range dependent process. Both works use the OLS-based estimator, which is equivalent to our CUSUM based statistic $T(k;s,e)$ for mean change.

To establish the optimal localization error rate, besides the assumptions required by SNCP, we introduce an additional assumption in \Cref{ass_moment}.
\begin{ass}\label{ass_moment}
There exists a fixed constant $0<C<\infty$ and $1<\phi<2$ such that for all $1\leq i,j \leq n$,
$$
E\left[\sum_{t=i}^j X_i\right]^2\leq C|j-i+1|^{\phi}.
$$
\end{ass}
\Cref{ass_moment} is adopted from Assumption H1 in \cite{lavielle2000least}, and holds for commonly used time series models such as linear processes~(e.g.  ARMA process), strongly mixing processes and certain long-range dependent processes. We refer to \cite{lavielle2000least} for more detailed examples.

A key consequence of \Cref{ass_moment} is a H\'ajek-R\'enyi type inequality~(see Lemma \ref{lem_HRin} in \Cref{sec:local_theory}), which is adopted from Theorem 1 and Lemma 2.2 in \cite{lavielle2000least} and plays a key role for controlling random fluctuations of the CUSUM statistic $T(k;s,e)$ around its expectation in the proof of \Cref{thm_refine}.

\begin{thm}\label{thm_refine}
Under the conditions of Theorem \ref{thm1} and \Cref{ass_moment}, we have that  
$$ |\widetilde{k}_i^*-k_i|=O_p(\delta^{2/(\phi-2)}), \text{ for } i=1,\cdots,m_o,$$
as long as $n^{2-\phi}\delta^{2}\to \infty$, where $1<\phi<2$ is defined in Assumption \ref{ass_moment}.
\end{thm}

Several remarks are in order. First, Theorem \ref{thm_refine} implies that for a constant change size $\delta$, the locally refined estimator $\widetilde{k}_i^*/n$ achieves the $O_p(n^{-1})$ localization error rate, which is the optimal rate~(e.g.\ \cite{bai1994}, \cite{Bai1998} and \cite{lavielle2000least}). Second, with the further assumption that $\{X_t\}_{t=1}^n$ is a linear process~(see condition (B) in \cite{bai1994}), \Cref{thm_refine} holds for $\phi=1$, which states that 
$|\widetilde{k}_i^*-k_i|=O_p(\delta^{-2})$ given that $n\delta^{2}\to \infty$, and recovers the result in \cite{bai1994}.

The proof of Theorem \ref{thm_refine} builds on the arguments in \cite{bai1994} and \cite{lavielle2000least}, where the key component is the use of a H\'ajek-R\'enyi type inequality. However, our proof requires additional technical arguments as we further need to control the randomness of the local interval brought by $\{\widehat{k}_i\}_{i=1}^{m_o}$ of the first stage SNCP.  

\textbf{Remark 3}: Beyond the mean case, it is more difficult to establish the $O_p(n^{-1})$ rate for the CUSUM-based local refinement. The reason is as follows. Roughly speaking, the proof for the optimal rate of CUSUM for the mean case consists of two conditions: (A).\ the population version of CUSUM is maximized at the true change-point and has non-zero left and right derivatives, and (B).\ the random fluctuations around the population CUSUM is uniformly small. To establish (B), the H\'ajek-R\'enyi type inequality is needed. For (A), it is trivially true for the mean case due to linearity, but is much more difficult to be verified for a general functional $\theta(\cdot)$ without linearity.

Nevertheless, when dealing with a functional other than mean, other types of single change-point estimator available in the literature can be used in the local refinement step to help achieve a provably $O_p(n^{-1})$ rate. For example, for the quantile case, we may use the change-point estimator in \cite{Oka2011}. Given the true number of change-points, it estimates the change-point location via minimization of an objective function based on the quantile check loss (see equation (6) therein) and achieves the $O_p(n^{-1})$ rate under additional assumptions. To conserve space, we do not further pursue this direction in our current manuscript.

\textbf{Extension to vector-valued functionals}: The proposed local refinement procedure for SNCP can be easily extended to a vector-valued functional where $\bftheta(\cdot)\in \mathbb{R}^d$ with $d>1.$ Specifically, we use the same CUSUM statistic $T(k;s,e)$ defined above and modify $\widetilde{k}_i^*$ in \eqref{eq:local_cusum} as 
\begin{align*}
\widetilde{k}_i^*={\arg\max}_{s_i^*+\widetilde{\iota}_n\leq k\leq e_i^*-\widetilde{\iota}_n}\|T(k;s_i^*, e_i^*)\|_2, ~\text{ for } i=1,2,\cdots,m_o,
\end{align*}
where $\|\cdot\|_2$ is the $l_2$-norm. The theoretical result in \Cref{thm_refine} can be easily established for the case of mean change in a multivariate time series $\{Y_t \in \mathbb{R}^d\}_{t=1}^n$ for any fixed dimension $d>1.$ We omit the details to conserve space.



\textbf{An informal procedure to identify the functionals that changed}: For change-point estimation in vector-valued functionals, once a change-point is detected, one may want to further identify which features actually changed. One informal strategy for this task is to further conduct a subsequent SN-test. Specifically, for each estimated change-point $\widehat{k}_i$ by SNCP, based on its local interval $[\widehat{k}_{i-1}+\widetilde{\iota}_n, \widehat{k}_{i+1}-\widetilde{\iota}_n]$ as defined above in \eqref{eq:local_interval}, we can further conduct a single change-point SN test (see equation \eqref{statgeneral} in the main text) for each feature and determine if it is changed at this very change-point. 

This is an informal procedure as it is obviously subject to the multiple testing issues as we are conducting SN tests simultaneously for multiple features at multiple estimated change-points. However, we believe this procedure can informally shed some light on which feature may have actually changed.

\subsection{ Numerical evidence for local refinement}\label{sec:local_simu}
In this section, we conduct numerical experiments to illustrate the improvements brought by the local refinement procedure proposed in \Cref{sec:local_proc} for mean change in univariate and multivariate time series and for multi-parameter change.

We generate a $d$-dimensional time series $\{Y_t\}_{t=1}^n$ with piecewise constant mean under both single and multiple change-point settings and different change sizes.
\begin{align*}
&\text{(LR1)}: n=600, ~~ Y_t=\begin{cases} 
0+ X_{t}, & t\in [1,300],\\
{0.5/\sqrt{d}+ X_{t}}, & t\in [301,600].
\end{cases}\\
&\text{(LR2)}: n=600, ~~ Y_t=\begin{cases} 
0+ X_{t}, & t\in [1,300],\\
{1/\sqrt{d}+ X_{t}}, & t\in [301,600].
\end{cases}\\
&\text{(LR3)}: n=1000, ~~ Y_t=\begin{cases} 
0+ X_{t}, & t\in [1,333], [668,1000]\\
{0.5/\sqrt{d}+ X_{t}}, & t\in [334,667].
\end{cases}\\
&\text{(LR4)}: n=1000, ~~ Y_t=\begin{cases} 
0+ X_{t}, & t\in [1,333], [668,1000]\\
{1/\sqrt{d}+ X_{t}}, & t\in [334,667].
\end{cases}
\end{align*}
where $\{X_t\}$ is \textit{i.i.d.}\ standard $d$-variate normal $N(0,\mathbf I_d)$. 

(LR1) and (LR2) have one change-point where (LR2) has a larger change size. (LR3) and (LR4) has two change-points where (LR4) has a larger change size. For each simulation setting, we repeat the experiments 1000 times. Note that to be fair, the local refinement procedure is performed for all experiments where the estimated number of change-points by SNCP $\widehat{m}\geq 1$, as in practice we do not know the true number of change-points $m_o.$

\Cref{tab_local_d1} and \Cref{tab_local_d5} summarize the performance of SNCP and the local refinement (SN-LR) under change in mean for (LR1)-(LR4) with $d=1$ and $d=5$ respectively. \Cref{tab_local_mul1} further provides the performance of SNCP and the local refinement (SN-LR) under change in multi-parameter (mean + median) for (LR1)-(LR4) with $d=1.$ As can be seen, the local refinement procedure does improve the estimation accuracy of SNCP~(e.g.\ higher ARI and lower Hausdorff distance) to some extent across all simulation settings. Note that SN and SN-LR have exactly the same $\hat{m}-m_o$ by design.


\begin{table}[H]
\centering
\begin{tabular}{rlrrrrrrrrrrr}
\hline
\hline & & \multicolumn{7}{c}{$\hat{m}-m_o$} & &&& \\
\hline
Method & Model & $\leq -3$ & $-2$ & $-1$ & $0$ & $1$ & $2$ & $\geq 3$ & ARI & $d_1 \times 10^2$ & $d_2 \times 10^2$ & $d_H \times 10^2$ \\ \hline
SN & $(LR1)$ & 0 & 0 & 46 & 914 & 39 & 1 & 0 & 0.868 & 3.04 & 4.39 & 5.34 \\ 
SN-LR & & 0 & 0 & 46 & 914 & 39 & 1 & 0 & 0.885 & 2.51 & 3.89 & 4.81 \\ 
SN & $(LR2)$ & 0 & 0 & 0 & 952 & 46 & 2 & 0 & 0.961 & 1.93 & 0.71 & 1.93 \\ 
SN-LR &  & 0 & 0 & 0 & 952 & 46 & 2 & 0 & 0.971 & 1.65 & 0.47 & 1.65 \\ \hline
SN & $(LR3)$ & 0 & 6 & 132 & 847 & 15 & 0 & 0 & 0.859 & 2.38 & 6.64 & 6.84 \\ 
SN-LR &  & 0 & 6 & 132 & 847 & 15 & 0 & 0 & 0.878 & 1.88 & 6.24 & 6.46 \\ 
SN & $(LR4)$ & 0 & 0 & 0 & 974 & 26 & 0 & 0 & 0.964 & 1.20 & 0.84 & 1.20 \\ 
SN-LR &  & 0 & 0 & 0 & 974 & 26 & 0 & 0 & 0.980 & 0.88 & 0.47 & 0.88 \\
\hline\hline
\end{tabular}
\caption{Performance of SNCP and the local refinement (SN-LR) under change in mean for $d=1$.} 
\label{tab_local_d1}
\end{table}

\begin{table}[H]
\centering
\begin{tabular}{rlrrrrrrrrrrr}
\hline
\hline & & \multicolumn{7}{c}{$\hat{m}-m_o$} & &&& \\
\hline
Method & Model & $\leq -3$ & $-2$ & $-1$ & $0$ & $1$ & $2$ & $\geq 3$ & ARI & $d_1 \times 10^2$ & $d_2 \times 10^2$ & $d_H \times 10^2$ \\ \hline
SN & $(LR1)$ & 0 & 0 & 97 & 862 & 39 & 2 & 0 & 0.811 & 3.39 & 7.17 & 8.24 \\ 
SN-LR &  & 0 & 0 & 97 & 862 & 39 & 2 & 0 & 0.841 & 2.37 & 6.26 & 7.22 \\ 
SN & $(LR2)$ & 0 & 0 & 0 & 948 & 50 & 2 & 0 & 0.961 & 2.18 & 0.73 & 2.18 \\ 
SN-LR &  & 0 & 0 & 0 & 948 & 50 & 2 & 0 & 0.973 & 1.74 & 0.40 & 1.74 \\  \hline
SN & $(LR3)$ & 0 & 26 & 212 & 740 & 22 & 0 & 0 & 0.803 & 2.60 & 10.30 & 10.57 \\ 
SN-LR &  & 0 & 26 & 212 & 740 & 22 & 0 & 0 & 0.829 & 1.77 & 9.67 & 9.94 \\
SN & $(LR4)$ & 0 & 0 & 0 & 969 & 30 & 1 & 0 & 0.963 & 1.31 & 0.86 & 1.31 \\ 
SN-LR &  & 0 & 0 & 0 & 969 & 30 & 1 & 0 & 0.980 & 0.89 & 0.46 & 0.89 \\ 
\hline\hline
\end{tabular}
\caption{Performance of SNCP and the local refinement (SN-LR) under change in mean for $d=5$.} 
\label{tab_local_d5}
\end{table}

\begin{table}[H]
\centering
\begin{tabular}{rlrrrrrrrrrrr}
\hline
\hline & & \multicolumn{7}{c}{$\hat{m}-m_o$} & &&& \\
\hline
Method & Model & $\leq -3$ & $-2$ & $-1$ & $0$ & $1$ & $2$ & $\geq 3$ & ARI & $d_1 \times 10^2$ & $d_2 \times 10^2$ & $d_H \times 10^2$ \\ \hline
SN & $(LR1)$ & 0 & 0 & 83 & 879 & 38 & 0 & 0 & 0.826 & 3.12 & 6.41 & 7.27 \\ 
SN-LR &  & 0 & 0 & 83 & 879 & 38 & 0 & 0 & 0.846 & 2.52 & 5.81 & 6.67 \\ 
SN & $(LR2)$ & 0 & 0 & 0 & 953 & 47 & 0 & 0 & 0.962 & 1.90 & 0.72 & 1.90 \\ 
SN-LR &  & 0 & 0 & 0 & 953 & 47 & 0 & 0 & 0.969 & 1.75 & 0.54 & 1.75 \\  \hline
SN & $(LR3)$ & 0 & 13 & 191 & 768 & 27 & 1 & 0 & 0.826 & 2.59 & 8.87 & 9.23 \\ 
SN-LR &  & 0 & 13 & 191 & 768 & 27 & 1 & 0 & 0.838 & 2.28 & 8.71 & 9.10 \\ 
SN & $(LR4)$ & 0 & 0 & 0 & 962 & 37 & 1 & 0 & 0.963 & 1.39 & 0.83 & 1.39 \\ 
SN-LR &  & 0 & 0 & 0 & 962 & 37 & 1 & 0 & 0.975 & 1.15 & 0.58 & 1.15 \\ 
\hline\hline
\end{tabular}
\caption{Performance of SNCP and the local refinement (SN-LR) under change in multi-parameter (mean + median).} 
\label{tab_local_mul1}
\end{table}


\subsection{ Theoretical results for the mean functional}\label{sec:local_theory}
In this section, we provide the proof of \Cref{thm_refine} for the mean functional. Recall that $\widetilde{\iota}_n$ is any sequence such that $\iota_n=O(\widetilde{\iota}_n)$ and $\widetilde{\iota}_n=o(n)$, with $\iota_n$ being the localization error rate of SNCP.


Under the mean change setting, the CUSUM statistic on a generic interval $[s_i,e_i]$ takes the form:
\begin{align*}
T(k;s_i,e_i)=\frac{\sqrt{(k-s_i+1)(e_i-k)}}{e_i-s_i+1}\left[\frac{1}{e_i-k}\sum_{t=k+1}^{e_i}Y_t-\frac{1}{k-s_i+1}\sum_{t=s_i}^{k}Y_t\right],
\end{align*}
and we define
$\tilde{k}_i(s_i,e_i)=\arg\max_{s_i+\widetilde{\iota}_n\leq k\leq e_i-\widetilde{\iota}_n}T^2(k;s_i,e_i)$.

The locally refined estimator for the $i$th change-point $k_i$ is thus
\begin{equation}\label{refiner}
\tilde{k}_i^*=\tilde{k}_i(s_i^*,e_i^*),
\end{equation}
with $[s_i^*,e_i^*]=[\widehat{k}_{i-1}+\widetilde{\iota}_n, \widehat{k}_{i+1}-\widetilde{\iota}_n] \text{ for } i=1,2,\cdots, m_o.$

\Cref{lem_HRin} gives a H\'ajek-R\'enyi type inequality, which is adopted from Theorem 1 and Lemma 2.2 in \cite{lavielle2000least} and plays a key role for controlling random fluctuations of the CUSUM statistic $T(k;s_i,e_i)$ around its expectation in the proof of \Cref{thm_refine}.

\begin{lemma}\label{lem_HRin}
Suppose $\{X_t\}_{t\in\mathbb{Z}}$ satisfies Assumption \ref{ass_moment}, there then exists a uniform constant $C(\phi,X)$ that only depends on the form of $\{X_t\}$ and $\phi$, such that for any $\varepsilon>0$,
\begin{flalign*}
\text{(i)   }&P\left(\max _{1 \leqslant k\leqslant n} \frac{\left|S^X_{1, k}\right|}{\sqrt{k}} \geqslant \varepsilon\right) \leqslant C(\phi, X) \frac{\log (n)n^{\phi-1}}{\varepsilon^{2}},\\
\text{(ii)   }&\sup _{i \in \mathbb{Z}} P\left(\max _{i+1 \leqslant k \leqslant i+n}\left|S^X_{i, k}\right| \geqslant \varepsilon\right) \leqslant C(\phi, X) \frac{n^{\phi}}{\varepsilon^{2}}, \\
\text{(iii)   }&\sup _{i \in \mathbb{Z}} P\left(\max _{k \geqslant m+i-1} \frac{\left|S^X_{i, k}\right|}{k} \geqslant \varepsilon\right) \leqslant C(\phi, X) \frac{m^{\phi-2}}{\varepsilon^{2}},
\end{flalign*}
where $S^{X}_{i,k}=\sum_{t=i}^{k}X_t$.
\end{lemma}
\begin{proof}
$(i)$ holds by letting $b_k=k^{-1/2}$ and notice that $\sum_{t=1}^{n}b_t^2=O(\log (n))$ in Theorem 1 in \cite{lavielle2000least}. (ii) and (iii) holds by Lemma 2.2 in \cite{lavielle2000least} directly.
\end{proof}
\medskip

\noindent\textbf{Proof of \Cref{thm_refine}}: Our proof consists of two steps: (1) the consistency of $\tilde{k}_i^*$; and (2) the localization error rate of $\tilde{k}_i^*$. The proof builds on the arguments in \cite{bai1994} and \cite{lavielle2000least}, where the key component is the use of the H\'ajek-R\'enyi type inequality in \Cref{lem_HRin}.

However, our proof requires additional technical arguments as we further need to control the randomness of the local interval brought by $\{\widehat{k}_i\}_{i=1}^{m_o}$ of the first stage SNCP.  

\begin{proof}
In the following, we provide the proof for a generic $i\in \{1,2,\cdots, m_o\}.$ For $M>0$, denote \begin{flalign*}
&A_i(M)=\{k: M<|k-k_i|\},\\
&B_{i}=\left\{[s_i,e_i]: k_{i-1}<s_i\leq k_{i-1}+2\iota_n, k_{i+1}-2\iota_n<e_i\leq  k_{i+1} \right\}.
\end{flalign*}
Recall by Theorem \ref{thm1}, we have $$
P(|\hat{k}_{j}-k_{j}|<\iota_n, \text{ for } j=i-1,i+1)=1-o(1).
$$
Thus, it follows $P([s_i^*,e_i^*] \in B_i)=1-o(1)$, and 
\begin{flalign}
\notag P(\tilde{k}_i^*\in A_i(M))\leq& P\left(\{\tilde{k}_i^*\in A_i(M)\}\cap \{[s_i^*,e_i^*]\in B_i\}\right)+o(1)
\\\leq &  \label{suff_AiM}P\left(\bigcup_{[s_i,e_i]\in B_i}\left\{\tilde{k}_i(s_i,e_i)\in A_i(M)\right\}\right)+o(1).
\end{flalign}

Without loss of generality, we assume $\mu_{i+1}-\mu_i=\delta_i>0$. In this case, $ET(k;s_i,e_i)\geq 0$ for $s_i\leq k\leq e_i$ and $[s_i,e_i]\in B_i$, and attains its maximum value at $k_i$. Otherwise, one can replace $T(k;s_i,e_i)$ by $-T(k;s_i,e_i)$ in the following analysis. 
The key difference, compared with proof of \cite{bai1994} and \cite{lavielle2000least}, is that we further need to control for CUSUM statistics evaluated on all intervals $[s_i,e_i]\in B_i$ uniformly while the referenced papers only have to deal with a single interval, as the CUSUM statistics there do not have a first stage randomness involved.


\textbf{Part (1)}. We first show the consistency of $\tilde{k}_i^*$. By \eqref{suff_AiM}, it suffices to show that, for any $\varepsilon>0$, 
$$
P\left(\bigcup_{[s_i,e_i]\in B_i}\left\{\tilde{k}_i(s_i,e_i)\in A_i(n\varepsilon)\right\}\right)=o(1),
$$
as $n\to\infty$.

Note  that
\begin{flalign}\label{suffdiffT}
&P\left(\bigcup_{[s_i,e_i]\in B_i}\left\{\tilde{k}_i(s_i,e_i)\in A_i(n\varepsilon)\right\}\right)\\\notag\leq& P\left(\bigcup_{[s_i,e_i]\in B_i}\left\{\max_{k\in[s_i+\widetilde{\iota}_n,e_i-\widetilde{\iota}_n]\cap A_i(n\varepsilon)}\left|T(k;s_i,e_i)\right|\geq\left|T(k_i;s_i,e_i)\right| \right\}\right).
\end{flalign}
By triangle inequality, we know that 
$$
\left|T(k;s_i,e_i)\right|-\left|T(k_i;s_i,e_i)\right|\leq 2\max_{s_i+\widetilde{\iota}_n\leq k\leq e_i-\widetilde{\iota}_n}\left|T(k;s_i,e_i)-ET(k;s_i,e_i) \right|+|ET(k;s_i,e_i)|-|ET(k_i;s_i,e_i)|.
$$
Using equation (14) in \cite{bai1994}, for some constant $C>0$ independent of $n$ and $\varepsilon$,
\begin{flalign}\label{bai14}
ET(k_i;s_i,e_i)-ET(k;s_i,e_i)\geq C(e_i-s_i+1)^{-1}|k-k_i|\delta_i.
\end{flalign}
Furthermore, 
$$
T(k;s_i,e_i)-ET(k;s_i,e_i)=\frac{\sqrt{(k-s_i+1)(e_i-k)}}{(e_i-s_i+1)}\left[\frac{\sum_{t=k+1}^{e_i}X_t}{(e_i-k)}-\frac{\sum_{t=s_i}^{k}X_t}{(k-s_i+1)}\right],
$$
and $$
\left|T(k;s_i,e_i)-ET(k;s_i,e_i)\right|\leq (e_i-s_i+1)^{-1/2}\left\{(k-s_i+1)^{-1/2}\left|\sum_{t=s_i}^{k}X_t\right|+(e_i-k)^{-1/2}\left|\sum_{t=k+1}^{e_i}X_t\right|\right\}.
$$
Hence, in view of (\ref{suffdiffT}), and that $e_i-s_i+1\leq n$, we have for some constant $C>0$ that is independent of $n$ and $\varepsilon$,
\begin{flalign}
&\notag P\left(\bigcup_{[s_i,e_i]\in B_i}\left\{\tilde{k}_i(s_i,e_i)\in A_i(n\varepsilon)\right\}\right)\\
\notag\leq &P\left(\bigcup_{[s_i,e_i]\in B_i}\left\{\max_{s_i+\widetilde{\iota}_n\leq k\leq e_i-\widetilde{\iota}_n}\frac{1}{\sqrt{k-s_i+1}}\left|\sum_{t=s_i}^k X_t\right|\geq (C/2) \delta_{i} (e_i-s_i+1)^{-1/2}n\varepsilon \right\}\right)\\&\notag+P\left(\bigcup_{[s_i,e_i]\in B_i}\left\{\max_{s_i+\widetilde{\iota}_n\leq k\leq e_i-\widetilde{\iota}_n}\frac{1}{\sqrt{e_i-k}}\left|\sum_{t=k+1}^{e_i} X_t\right| \geq (C/2) \delta_{i} (e_i-s_i+1)^{-1/2}n\varepsilon \right\}\right)
\\\notag\leq &P\left(\bigcup_{[s_i,e_i]\in B_i}\left\{\max_{s_i+\widetilde{\iota}_n\leq k\leq e_i-\widetilde{\iota}_n}\frac{1}{\sqrt{k-s_i+1}}\left|\sum_{t=s_i}^k X_t\right|\geq (C/2) \delta_{i} n^{1/2}\varepsilon \right\}\right)\\&\notag+P\left(\bigcup_{[s_i,e_i]\in B_i}\left\{\max_{s_i+\widetilde{\iota}_n\leq k\leq e_i-\widetilde{\iota}_n}\frac{1}{\sqrt{e_i-k}}\left|\sum_{t=k+1}^{e_i} X_t\right| \geq (C/2) \delta_{i} n^{1/2}\varepsilon \right\}\right)
\\\label{ineq_LM}:=&P_s+P_e.
\end{flalign}
Next, note that when $[s_i,e_i]\in B_i$ and $k-s_i\geq \widetilde{\iota}_n$, we have  $k-s_i\geq (s_i-k_{i-1})/2$ and $k-s_i\geq (k-k_{i-1})/3$, hence
\begin{flalign}
\notag
&\max_{s_i+\widetilde{\iota}_n\leq k\leq e_i-\widetilde{\iota}_n}\frac{1}{\sqrt{k-s_i+1}}\left|\sum_{t=s_i}^k X_t\right|\\\notag=&\max_{s_i+\widetilde{\iota}_n\leq k\leq e_i-\widetilde{\iota}_n}\frac{1}{\sqrt{k-s_i+1}}\left|\sum_{t=k_{i-1}+1}^{k} X_t-\sum_{t=k_{i-1}+1}^{s_i-1}X_t\right|\\\notag\leq& \max_{s_i+\widetilde{\iota}_n\leq k\leq e_i-\widetilde{\iota}_n}\frac{1}{\sqrt{k-s_i+1}}\left|\sum_{t=k_{i-1}+1}^{k} X_t\right|+\max_{s_i+\widetilde{\iota}_n\leq k\leq e_i-\widetilde{\iota}_n}\frac{1}{\sqrt{k-s_i+1}}\left|\sum_{t=k_{i-1}+1}^{s_i-1}X_t\right|
\\\notag\leq&\max_{s_i+\widetilde{\iota}_n\leq k\leq e_i-\widetilde{\iota}_n}\frac{\sqrt{3}}{\sqrt{k-k_{i-1}}}\left|\sum_{t=k_{i-1}+1}^{k} X_t\right|+\max_{s_i+\widetilde{\iota}_n\leq k\leq e_i-\widetilde{\iota}_n}\frac{\sqrt{2}}{\sqrt{s_i-k_{i-1}-1}}\left|\sum_{t=k_{i-1}+1}^{s_i-1}X_t\right|
\\\notag\leq&\max_{k_{i-1}+1\leq m\leq k_{i+1}}\frac{\sqrt{3}}{\sqrt{m-k_{i-1}}}\left|\sum_{t=k_{i-1}+1}^{m} X_t\right|+\max_{ k_{i-1}+1\leq m\leq k_{i+1} }\frac{\sqrt{2}}{\sqrt{m-k_{i-1}}}\left|\sum_{t=k_{i-1}+1}^{m}X_t\right|\\\leq& \label{bound_doublemax}\max_{k_{i-1}+1\leq m\leq k_{i+1}}\frac{4}{\sqrt{m-k_{i-1}}}\left|\sum_{t=k_{i-1}+1}^{m}X_t\right|.
\end{flalign}
Note (\ref{bound_doublemax}) holds for any pair $[s_i,e_i]\in B_i$,  hence, by Lemma \ref{lem_HRin} (i), 
\begin{flalign*}
P_s\leq &P\Big(\max_{k_{i-1}+1\leq m\leq k_{i+1}}\frac{4}{\sqrt{m-k_{i-1}}}\left|\sum_{t=k_{i-1}+1}^{m}X_t\right|\geq C/2\delta_in^{1/2}\varepsilon\Big)\\\leq& C(\phi,X)\frac{64(k_{i+1}-k_{i-1})^{\phi-1}\log(k_{i+1}-k_{i-1})}{C^2n\delta_i^2\varepsilon^2}\leq C n^{\phi-2}\log(n)\delta_i^{-2}.
\end{flalign*}
Similarly, we can show that
$$
P_e\leq C n^{\phi-2}\log(n)\delta_i^{-2}.
$$
By (\ref{ineq_LM}), this implies the consistency of $\tilde{k}_i^*$ given that $n^{\phi-2}\delta_{i}^{-2}\log(n)\to 0$.


\textbf{Part (2)}. Next, we derive the localization error rate  $|\tilde{k}_i^*-k_i|=O_p(\delta_{i}^{2/(\phi-2)}).$ By \eqref{suff_AiM}, it suffices to show that, for any $\varrho>0$, there exists $M_{\varrho}$ and $n_{\varrho}$ large enough, such that for $n>n_{\varrho}$ and $M>M_{\varrho}$, $$
P\left(\bigcup_{[s_i,e_i]\in B_i}\left\{\tilde{k}_i(s_i,e_i)\in A_i(M\delta_{i}^{2/(\phi-2)})\right\}\right)<\varrho.
$$

For any $\varepsilon\in(0,\epsilon_o/2)$,
\begin{flalign*}
&P\left(\bigcup_{[s_i,e_i]\in B_i}\left\{\tilde{k}_i(s_i,e_i)\in A_i(M\delta_{i}^{2/(\phi-2)})\right\}\right)\\
\leq & P\left(\bigcup_{[s_i,e_i]\in B_i}\left\{\max_{k\in[s_i+\varepsilon n,e_i-\varepsilon n]\cap A_i(M\delta_{i}^{2/(\phi-2)})}T^2(k;s_i,e_i)\geq T^2(k_i;s_i,e_i) \right\}\right)\\&+P\left(\bigcup_{[s_i,e_i]\in B_i}\left\{\max_{k\in\{[s_i,s_i+\varepsilon n)\bigcup (e_i-\varepsilon n,e_i]\}\cap A_i(M\delta_{i}^{2/(\phi-2)})}T^2(k;s_i,e_i)\geq T^2(k_i;s_i,e_i)\right\}\right),\end{flalign*}
by the proof of consistency, we know the second term can be bounded by $\varrho/5$ for $n$ large enough.  

It suffices to bound the first term by $4\varrho/5$. Note that $[s_i+\varepsilon n,e_i-\varepsilon n]\subset [k_{i-1}+\varepsilon n/2,k_{i+1}-\varepsilon n/2]$, denote $D_i(M)=\left\{[k_{i-1}+\varepsilon n/2,k_{i+1}-\varepsilon n/2]\cap A_i(M\delta_{i}^{2/(\phi-2)})\right\}$, we have
\begin{flalign*}
& P\left(\bigcup_{[s_i,e_i]\in B_i}\left\{\max_{k\in[s_i+\varepsilon n,e_i-\varepsilon n]\cap A_i(M\delta_{i}^{2/(\phi-2)})}T^2(k;s_i,e_i)\geq T^2(k_i;s_i,e_i)\right\}\right)\\
\leq & P\left(\bigcup_{[s_i,e_i]\in B_i}\left\{\max_{k\in D_i(M)}T^2(k;s_i,e_i)\geq T^2(k_i;s_i,e_i)\right\}\right)\\ \leq&
P\left(\bigcup_{[s_i,e_i]\in B_i}\left\{\max_{k\in D_i(M) }[T(k;s_i,e_i)-T(k_i;s_i,e_i)]\geq 0\right\}\right)\\&+
P\left(\bigcup_{[s_i,e_i]\in B_i}\left\{\min_{k\in D_i(M)}[T(k;s_i,e_i)+T(k_i;s_i,e_i)]\leq 0\right\}\right):=P_1+P_2.
\end{flalign*}

Define $b(k;s_i,e_i)=\frac{(k-s_i+1)^{1/2}(e_i-k)^{1/2}}{(e_i-s_i+1)}$. For $P_1$, on the event $\{T(k;s_i,e_i)\geq T(k_i;s_i,e_i)\}$, we have 
\begin{flalign*}
&T(k;s_i,e_i)-ET(k;s_i,e_i)-T(k_i;s_i,e_i)+ET(k_i;s_i,e_i)\\=&\frac{b(k_i;s_i,e_i)}{k_i-s_i+1}\sum_{t=s_i}^{k_i}X_t-\frac{b(k;s_i,e_i)}{k-s_i+1}\sum_{t=s_i}^{k}X_t-
\left(\frac{b(k_i;s_i,e_i)}{e_i-k_i}\sum_{t=k_i+1}^{e_i}X_t-\frac{b(k;s_i,e_i)}{e_i-k}\sum_{t=k+1}^{e_i}X_t\right)\\\geq& ET(k_i;s_i,e_i)-ET(k;s_i,e_i)\geq C\delta_i\frac{|k_i-k|}{(e_i-s_i+1)},
\end{flalign*}
where the last inequality holds by (\ref{bai14}).

Therefore, we have for some constant $C>0$ independent of $n$ and $M$,
\begin{flalign*}
P_1\leq& P\left(\bigcup_{[s_i,e_i]\in B_i}\left\{\max_{k\in D_i(M)} \frac{e_i-s_i+1}{|k_i-k|}
\left|\frac{b(k_i;s_i,e_i)}{k_i-s_i+1}\sum_{t=s_i}^{k_i}X_t-\frac{b(k;s_i,e_i)}{k-s_i+1}\sum_{t=s_i}^{k}X_t\right|>C\delta_{i}/2
\right\} \right)
\\&+P\left(\bigcup_{[s_i,e_i]\in B_i}\left\{\max_{k\in D_i(M)} \frac{e_i-s_i+1}{|k_i-k|}
\left|\frac{b(k_i;s_i,e_i)}{e_i-k_i}\sum_{t=k_i+1}^{e_i}X_t-\frac{b(k;s_i,e_i)}{e_i-k}\sum_{t=k+1}^{e_i}X_t\right|>C\delta_{i}/2
\right\} \right)\\:=&P_{11}+P_{12}.
\end{flalign*}
We only deal with $P_{11}$ as $P_{12}$ is similar.   Denote 
$$
G(k;s_i,e_i)= 
\left(\frac{b(k_i;s_i,e_i)}{k_i-s_i+1}\sum_{t=s_i}^{k_i}X_t-\frac{b(k;s_i,e_i)}{k-s_i+1}\sum_{t=s_i}^{k}X_t\right).
$$
Without loss of generality, we let $k<k_i$, then
$$
G(k;s_i,e_i)=\frac{(k-k_i)b(k_i;s_i,e_i)}{(k-s_i+1)(k_i-s_i+1)}\sum_{t=s_i}^{k_i}X_t+\frac{b(k_i;s_i,e_i)-b(k;s_i,e_i)}{k-s_i+1}\sum_{t=s_i}^{k}X_t+\frac{b(k_i;s_i,e_i)}{k-s_i+1}\sum_{t=k+1}^{k_i}X_t.
$$
Using the fact that $0<b(k;s_i,e_i)\leq 1$,  $|b(k_i;s_i,e_i)-b(k;s_i,e_i)|\leq C|k-k_i|/(e_i-s_i+1)$ for some constant $C>0$, and that $k-s_i+1\geq \varepsilon n/4$ for $k\in D_i(M)$, we obtain that for $k\in D_i(M)$,
\begin{flalign*}
&\frac{e_i-s_i+1}{|k_i-k|}\left|G(k;s_i,e_i)\right|\\\leq& \frac{4n}{n\varepsilon(k_i-s_i+1)}\left|\sum_{t=s_i}^{k_i}X_t\right|+\frac{C}{k-s_i+1}\left|\sum_{t=s_i}^{k}X_t\right|+\frac{4(e_i-s_i+1)}{\varepsilon n(k_i-k)}\left|\sum_{t=k+1}^{k_i}X_t\right|\\\leq &
\frac{8}{\varepsilon\epsilon_o n}\left|\sum_{t=s_i}^{k_i}X_t\right|+\frac{4C}{\varepsilon n}\left(\left|\sum_{t=k_{i-1}+1}^{k}X_t\right|+\left|\sum_{t=k_{i-1}+1}^{s_{i}-1}X_t\right|\right)+\frac{4}{\varepsilon (k_i-k) }\left|\sum_{t=k+1}^{k_i}X_t\right|,
\end{flalign*}

Hence using $|k-k_i|\leq n$, we have 
\begin{flalign*}
P_{11}=&P\left(\bigcup_{[s_i,e_i]\in B_i}\left\{\max_{k\in D_i(M)}\frac{e_i-s_i+1}{|k_i-k|}|G(k;s_i,e_i)|>C\delta_i/2 \right\}\right)
\\\leq & P\left(\frac{8}{\varepsilon\epsilon_on}\max_{k_{i-1}<s_i\leq k_{i}-1}\left|\sum_{t=s_i}^{k_i}X_t\right|>C\delta_i /6\right)+
P\left(\frac{4C}{\varepsilon n }\max_{k_{i-1}<k\leq k_{i+1}}\left|\sum_{t=k_{i-1}+1}^{k}X_t\right|>C\delta_i /12\right)\\&+P\left(\frac{4C}{\varepsilon n }\max_{k_{i-1}<s_i\leq k_{i-1}}\left|\sum_{t=k_{i-1}+1}^{s_i-1}X_t\right|>C\delta_i /12\right)+P\left(\max_{k\leq k_i-M\delta_{i}^{2/(\phi-2)}}\frac{4}{\varepsilon(k_i-k)}\left|\sum_{t=k+1}^{k_i}X_t\right|>C\delta_i /6\right),
\\\leq &C(\phi,X) \left[\frac{C (k_i-k_{i-1})^{\phi}}{\varepsilon^2\epsilon_o^2n^2\delta_i^2}+\frac{C (k_{i+1}-k_{i-1})^{\phi}}{\varepsilon^2n^2\delta_i^2}+\frac{C[M\delta_i^{2/(\phi-2)}]^{\phi-2}}{\delta_i^2\varepsilon^2}\right]\leq C[n^{\phi-2}+M^{\phi-2}],
\end{flalign*}
where  the last inequality holds by applying Lemma \ref{lem_HRin} (ii) to the first three terms, and Lemma \ref{lem_HRin} (iii) to the forth term.
Therefore, when $M$ is large, we can see $P_{11}<\varrho/5$, similarly $P_{12}<\varrho/5.$ Hence $P_1<2\varrho/5$.

For $P_2$, recall we assume $\delta_i>0$ and $ET(k;s_i,e_i)\geq 0$ for all $k$.  On the event $\{T(k;s_i,e_i)+T(k_i;s_i,e_i)\leq0\}$, we have 
$$
T(k;s_i,e_i)-ET(k;s_i,e_i)+T(k_i;s_i,e_i)-ET(k_i;s_i,e_i)\leq -ET(k_i;s_i,e_i),
$$
which implies either $T(k;s_i,e_i)-ET(k;s_i,e_i)\leq -ET(k_i;s_i,e_i)/2$ or 
$T(k_i;s_i,e_i)-ET(k_i;s_i,e_i)\leq -ET(k_i;s_i,e_i)/2$. Hence
\begin{flalign*}
&\left \{\min_{k\in D_i(M)}T(k;s_i,e_i)+T(k_i;s_i,e_i)\leq 0 \right\}\\\subset& \left\{\min_{k\in D_i(M)}\Big[T(k;s_i,e_i)-ET(k;s_i,e_i)\Big]\leq -ET(k_i;s_i,e_i)/2\right\}\cup \Big\{T(k_i;s_i,e_i)-ET(k_i;s_i,e_i)\leq -ET(k_i;s_i,e_i)/2\Big\}\\\subset&
\left\{\max_{k\in D_i(M)}\Big|T(k;s_i,e_i)-ET(k;s_i,e_i)\Big|\geq ET(k_i;s_i,e_i)/2\right\}\cup \Big\{\Big|T(k_i;s_i,e_i)-ET(k_i;s_i,e_i)\Big|\geq ET(k_i;s_i,e_i)/2\Big\}.
\end{flalign*}
Thus, we have 
\begin{flalign*}
P_2\leq &2P\left(\bigcup_{[s_i,e_i]\in B_i}\left\{\max_{k\in D_i(M)}\left|T(k;s_i,e_i)-ET(k;s_i,e_i)\right|\geq ET(k_i;s_i,e_i)/2\right\}\right)\\= &
2P\left(\bigcup_{[s_i,e_i]\in B_i}\left\{\max_{k\in D_i(M)}\left|\frac{b(k;s_i,e_i)}{k-s_i+1}\sum_{t=s_i}^{k}X_t-\frac{b(k;s_i,e_i)}{e_i-k}\sum_{t=k+1}^{e_i}X_t\right|\geq ET(k_i;s_i,e_i)/2 \right\}
\right)\\\leq &
2P\left(\bigcup_{[s_i,e_i]\in B_i}\left\{\max_{k\in D_i(M)}\left|\frac{b(k_i;s_i,e_i)}{k-s_i+1}\sum_{t=s_i}^{k}X_t\right|\geq C\delta_i \right\}
\right)\\&+2P\left(\bigcup_{[s_i,e_i]\in B_i}\left\{\max_{k\in D_i(M)}\left|\frac{b(k;s_i,e_i)}{e_i-k}\sum_{t=k+1}^{e_i}X_t\right|\geq C\delta_i \right\}
\right)
\\\leq &2P\left(\bigcup_{[s_i,e_i]\in B_i}\left\{\max_{k_{i-1}+\varepsilon n /2 \leq k\leq k_{i+1}} \frac{1}{k-s_i+1}\left|\sum_{t=s_i}^{k}X_t\right|\geq C\delta_i\right\}\right)\\&+2P\left(\bigcup_{[s_i,e_i]\in B_i}\left\{\max_{k_{i-1}<k\leq k_{i+1}-\varepsilon n/2} \frac{1}{e_i-k}\left|\sum_{t=k+1}^{e_i}X_t\right|\geq C\delta_i\right\}\right)\\:=&2[P_{21}+P_{22}],
\end{flalign*}
where the second inequality holds by noting $ET[k_i;s_i,e_i]=\frac{(e_i-k_i)^{1/2}(k_i-s_i+1)^{1/2}}{e_i-s_i+1}\delta_i\geq 4C\delta_i$ for some constant $C$, and the third inequality holds by $0<b(k;s_i,e_i)\leq 1$. 

We only deal with $P_{21}$ as $P_{22}$ is similar. Note that for $[s_i,e_i]\in B_i,$  $k-k_{i-1}\leq k-s_i+1$. Hence 
\begin{flalign*}
P_{21}\leq &
P\left(\max_{k_{i-1} <s_i\leq k_{i-1}+2\widetilde{\iota}_n}\max_{k_{i-1}+\varepsilon n/2\leq k} \frac{1}{k-k_{i-1}}\left|\sum_{t=s_i}^{k}X_t\right|\geq C\delta_i\right)\\\leq &
P\left(\max_{k_{i-1}+\varepsilon n/2\leq k} \frac{1}{k-k_{i-1}}\left|\sum_{t=k_{i-1}+1}^{k}X_t\right|\geq C\delta_i/2\right)\\&+P\left(\max_{k_{i-1} <s_i\leq k_{i-1}+2\widetilde{\iota}_n} \frac{2}{\varepsilon n}\left|\sum_{t=k_{i-1}+1}^{s_{i}-1}X_t\right|\geq C\delta_i/2\right)\\\leq &
C(\phi,X)\left[\frac{4(\varepsilon n/2)^{\phi-2}}{C^2\delta_i^2}+\frac{16(2\widetilde{\iota}_n)^\phi}{C^2\delta_i^2n^2\varepsilon^2}\right]=Cn^{\phi-2}\delta_i^{-2}<\varrho/10,
\end{flalign*}
where the second inequality holds by noting 
$\left|\sum_{t=s_i}^{k}X_t\right|\leq \left|\sum_{t=k_{i-1}+1}^{k}X_t\right|+\left|\sum_{t=k_{i-1}+1}^{s_i-1}X_t\right|$ and that $k-k_{i-1}>\varepsilon n/2$, and the last inequality holds by Lemma \ref{lem_HRin} and that $\widetilde{\iota}_n<n$. 
Similarly, we have $P_{22}<\varrho/10$, hence $P_2<2\varrho/5$. This completes the proof.
\end{proof}
\newpage






\begin{spacing}{1.14}
\bibliographystyle{apalike}
\bibliography{reference}

\begin{thebibliography}{}

\bibitem[Andrews, 1991]{andrew1991}
Andrews, D. W.~K. (1991).
\newblock Heteroskedasticity and autocorrelation consistent covariance matrix
  estimation.
\newblock {\em Econometrica}, 59(3):817--858.

\bibitem[Andrews, 1993]{andrews1993}
Andrews, D. W.~K. (1993).
\newblock Tests for parameter instability and structural change with unknown
  change point.
\newblock {\em Econometrica}, 61(4):821--856.

\bibitem[Aue et~al., 2014]{Aue2014}
Aue, A., Cheung, R. C.~Y., Lee, T.~C., and Zhong, M. (2014).
\newblock Segmented model selection in quantile regression using the minimum
  description length principle.
\newblock {\em Journal of the American Statistical Association},
  109(507):1241--1256.

\bibitem[Aue and Horv\'{a}th, 2013]{aue:13}
Aue, A. and Horv\'{a}th, L. (2013).
\newblock Structural breaks in time series.
\newblock {\em Journal of Time Series Analysis}, 34(1):1--16.

\bibitem[Aue et~al., 2009]{Aue2009}
Aue, A., Hörmann, S., Horváth, L., and Reimherr, M. (2009).
\newblock Break detection in the covariance structure of multivariate time
  series models.
\newblock {\em Annals of Statistics}, 37(6B):4046--4087.

\bibitem[Bai, 1994]{bai1994}
Bai, J. (1994).
\newblock Least squares estimation of a shift in linear processes.
\newblock {\em Journal of Time Series Analysis}, 15(5):453--472.

\bibitem[Bai and Perron, 1998]{Bai1998}
Bai, J. and Perron, P. (1998).
\newblock Estimating and testing linear models with multiple structural
  changes.
\newblock {\em Econometrica}, 66(1):47--78.

\bibitem[Bai and Perron, 2003]{Bai2003}
Bai, J. and Perron, P. (2003).
\newblock Computation and analysis of multiple structural change models.
\newblock {\em Journal of Applied Econometrics}, 18(1):1--22.

\bibitem[Baranowski et~al., 2019]{baranowski2019narrowest}
Baranowski, R., Chen, Y., and Fryzlewicz, P. (2019).
\newblock Narrowest-over-threshold detection of multiple change points and
  change-point-like features.
\newblock {\em Journal of the Royal Statistical Society: Series B},
  81(3):649--672.

\bibitem[Betken and Wendler, 2018]{Betken2018}
Betken, A. and Wendler, M. (2018).
\newblock Subsampling for general statistics under long range dependence with
  application to change point analysis.
\newblock {\em Statistica Sinica}, 28(3):1199--1224.

\bibitem[Bhattacharya et~al., 1978]{bhattacharya1978validity}
Bhattacharya, R.~N., Ghosh, J.~K., et~al. (1978).
\newblock On the validity of the formal edgeworth expansion.
\newblock {\em Annals of Statistics}, 6(2):434--451.

\bibitem[Bibinger et~al., 2017]{Bibinger2017}
Bibinger, M., Jirak, M., and Vetter, M. (2017).
\newblock Nonparametric change-point analysis of volatility.
\newblock {\em The Annals of Statistics}, 45(4):1542--1578.

\bibitem[Billingsley, 1968]{Billingsley1968}
Billingsley, P. (1968).
\newblock {\em Convergence of Probability Measures}.
\newblock John Wiley \& Sons.

\bibitem[Brodsky and Darkhovsky, 2013]{brodsky2013nonparametric}
Brodsky, E. and Darkhovsky, B.~S. (2013).
\newblock {\em {Nonparametric Methods in Change Point Problems}}.
\newblock Springer Netherlands.

\bibitem[Casini et~al., 2021]{Casini2021theory}
Casini, A., Deng, T., and Perron, P. (2021).
\newblock Theory of low frequency contamination from nonstationarity and
  misspecification: Consequences for {HAR} inference.
\newblock {\em arXiv preprint arXiv:2103.01604}.

\bibitem[Casini and Perron, 2019]{Casini2019structural}
Casini, A. and Perron, P. (2019).
\newblock Structural breaks in time series.
\newblock In {\em Oxford Research Encyclopedia of Economics and Finance}.

\bibitem[Casini and Perron, 2021a]{Casini2021change}
Casini, A. and Perron, P. (2021a).
\newblock Change-point analysis of time series with evolutionary spectra.
\newblock {\em arXiv preprint arXiv:2106.02031}.

\bibitem[Casini and Perron, 2021b]{casini2021minimax}
Casini, A. and Perron, P. (2021b).
\newblock Minimax {MSE} bounds and nonlinear {VAR} prewhitening for long-run
  variance estimation under nonstationarity.
\newblock {\em arXiv preprint arXiv:2103.02235}.

\bibitem[Chan and Walther, 2013]{Chan2013}
Chan, H.~P. and Walther, G. (2013).
\newblock Detection with the scan and the average likelihood ratio.
\newblock {\em Statistica Sinica}, 23(1):409--428.

\bibitem[Chan et~al., 2021]{Chan2021}
Chan, N.~H., Ng, W.~L., and Yau, C.~Y. (2021).
\newblock A self-normalized approach to sequential change-point detection for
  time series.
\newblock {\em Statistica Sinica}, 31(1):491--517.

\bibitem[Cho and Fryzlewicz, 2012]{Cho2012}
Cho, H. and Fryzlewicz, P. (2012).
\newblock Multiscale and multilevel technique for consistent segmentation of
  nonstationary time series.
\newblock {\em Statistica Sinica}, 22(1):207--229.

\bibitem[Crainiceanu and Vogelsang, 2007]{cv2007}
Crainiceanu, C. and Vogelsang, T. (2007).
\newblock Nonmonotonic power for tests of mean shift in a time series.
\newblock {\em Journal of Statistical Computation and Simulation},
  77(6):457--476.

\bibitem[Csörgő and Horváth, 1997]{Csoergoe1997}
Csörgő, M. and Horváth, L. (1997).
\newblock {\em Limit Theorems in Change-Point Analysis}.
\newblock Wiley Series in Probability and Statistics. Wiley.

\bibitem[Davis et~al., 2006]{Davis2006}
Davis, R., Lee, T. C.~M., and Rodriguez-Yam, G. (2006).
\newblock Structural break estimation for nonstationary time series models.
\newblock {\em Journal of the American Statistical Association},
  101(473):223--239.

\bibitem[Dette and G\"osmann, 2020]{dette2020a}
Dette, H. and G\"osmann, J. (2020).
\newblock A likelihood ratio approach to sequential change point detection.
\newblock {\em Journal of the American Statistical Association},
  115(531):1361--1377.

\bibitem[Dette et~al., 2020]{Dette2020b}
Dette, H., Kokot, K., and Volgushev, S. (2020).
\newblock Testing relevant hypotheses in functional time series via
  self-normalization.
\newblock {\em Journal of Royal Statistical Society: Series B}, 82(3):629--660.

\bibitem[Eichinger and Kirch, 2018]{Eichinger2018}
Eichinger, B. and Kirch, C. (2018).
\newblock A {MOSUM} procedure for the estimation of multiple random change
  points.
\newblock {\em Bernoulli}, 24(1):526--564.

\bibitem[Elsner et~al., 2008]{Elsner2008}
Elsner, J.~B., Kossin, J.~P., and Jagger, T.~H. (2008).
\newblock Increasing intensity of the strongest tropical cyclones.
\newblock {\em Nature}, 455:92–95.

\bibitem[Embrechts et~al., 1997]{Embrechts1997}
Embrechts, P., Klüppelberg, C., and Mikosch, T. (1997).
\newblock {\em Modelling Extremal Events for Insurance and Finance}.
\newblock Springer-Verlag Berlin Heidelberg.

\bibitem[Frick et~al., 2014]{Frick2014}
Frick, K., Munk, A., and Sieling, H. (2014).
\newblock Multiscale change point inference.
\newblock {\em Journal of the Royal Statistical Society: Series B},
  76(3):495--580.

\bibitem[Fryzlewicz, 2014]{Fryzlewicz2014}
Fryzlewicz, P. (2014).
\newblock Wild binary segmentation for multiple change-point detection.
\newblock {\em Annals of Statistics}, 42(6):2243--2281.

\bibitem[Fryzlewicz, 2020]{Fryzlewicz2020}
Fryzlewicz, P. (2020).
\newblock Detecting possibly frequent change-points: Wild binary segmentation 2
  and steepest-drop model selection.
\newblock {\em Journal of the Korean Statistical Society (with discussion)},
  49(4):1027--1070.

\bibitem[Fryzlewicz and Subba-Rao, 2014]{Fryzlewicz2014a}
Fryzlewicz, P. and Subba-Rao, S. (2014).
\newblock Multiple‐change‐point detection for auto‐regressive conditional
  heteroscedastic processes.
\newblock {\em Journal of the Royal Statistical Society: Series B},
  76(5):903--924.

\bibitem[Galeano and Wied, 2017]{Galeano2017}
Galeano, P. and Wied, D. (2017).
\newblock Dating multiple change points in the correlation matrix.
\newblock {\em Test}, 26(2):331--352.

\bibitem[Hall, 2013]{hall2013bootstrap}
Hall, P. (2013).
\newblock {\em {The Bootstrap and Edgeworth Expansion}}.
\newblock Springer-Verlag New York.

\bibitem[Hampel et~al., 1986]{Hampel1986}
Hampel, F.~R., Ronchetti, E.~M., Rousseeuw, P.~J., and Stahel, W.~A. (1986).
\newblock {\em Robust Statistics: The Approach Based on Influence Functions}.
\newblock John Wiley, New York.

\bibitem[Harchaoui and L{\'e}vy-Leduc, 2010]{harchaoui2010multiple}
Harchaoui, Z. and L{\'e}vy-Leduc, C. (2010).
\newblock Multiple change-point estimation with a total variation penalty.
\newblock {\em Journal of the American Statistical Association},
  105(492):1480--1493.

\bibitem[Hoga, 2018]{Hoga2018}
Hoga, Y. (2018).
\newblock A structural break test for extremal dependence in $\beta$-mixing
  random vectors.
\newblock {\em Biometrika}, 105(3):627–643.

\bibitem[Jiang et~al., 2020]{jiang2020time}
Jiang, F., Zhao, Z., and Shao, X. (2020).
\newblock Time series analysis of {COVID-19} infection curve: A change-point
  perspective.
\newblock {\em Journal of Econometrics, to appear}.

\bibitem[Jiang et~al., 2022]{jiang2022modelling}
Jiang, F., Zhao, Z., and Shao, X. (2022).
\newblock Modelling the {COVID-19} infection trajectory: A piecewise linear
  quantile trend model.
\newblock {\em Journal of the Royal Statistical Society: Series B, to appear}.

\bibitem[Killick et~al., 2012]{Killick2012}
Killick, R., Fearnhead, P., and Eckley, I. (2012).
\newblock Optimal detection of change-points with a linear computational cost.
\newblock {\em Journal of the American Statistical Association},
  107(500):1590--1598.

\bibitem[Korkas and Fryzlewicz, 2017]{Korkas2017}
Korkas, K.~K. and Fryzlewicz, P. (2017).
\newblock Multiple change-point detection for non-stationary time series using
  wild binary segmentation.
\newblock {\em Statistica Sinica}, 27(1):287--311.

\bibitem[Kovacs et~al., 2020]{Kovacs2020}
Kovacs, S., Li, H., B\"uhlmann, P., and Munk, A. (2020).
\newblock A seeded binary segmentation: A general methodology for fast and
  optimal change point detection.
\newblock {\em arxiv: https://arxiv.org/abs/2002.06633}.

\bibitem[K\"{u}nsch, 1989]{Kunsch1989}
K\"{u}nsch, H.~R. (1989).
\newblock The jackknife and the bootstrap for general stationary observations.
\newblock {\em Annals of Statistics}, 17(3):1217--1241.

\bibitem[Lavielle and Moulines, 2000]{lavielle2000least}
Lavielle, M. and Moulines, E. (2000).
\newblock Least-squares estimation of an unknown number of shifts in a time
  series.
\newblock {\em Journal of Time Series Analysis}, 21(1):33--59.

\bibitem[Longin and Solnik, 2002]{Longin2002}
Longin, F. and Solnik, B. (2002).
\newblock Extreme correlation of international equity markets.
\newblock {\em Journal of Finance}, 56(2):649--676.

\bibitem[Matteson and James, 2014]{Matteson2014}
Matteson, D. and James, N. (2014).
\newblock A nonparametric approach for multiple change-point analysis of
  multivariate data.
\newblock {\em Journal of the American Statistical Association},
  109(505):334--345.

\bibitem[Merlev{\`e}de et~al., 2009]{merlevede2009bernstein}
Merlev{\`e}de, F., Peligrad, M., and Rio, E. (2009).
\newblock Bernstein inequality and moderate deviations under strong mixing
  conditions.
\newblock In {\em High dimensional probability V: the Luminy volume}, pages
  273--292. Institute of Mathematical Statistics.

\bibitem[Morey and Agresti, 1984]{Morey1984}
Morey, L.~C. and Agresti, A. (1984).
\newblock The measurement of classification agreement: An adjustment to the
  rand statistic for chance agreement.
\newblock {\em Educational and Psychological Measurement}, 44(1):33–37.

\bibitem[{National Research Council}, 2013]{Council2013}
{National Research Council} (2013).
\newblock {\em Frontiers in Massive Data Analysis}.
\newblock The National Academies Press, Washington, DC.

\bibitem[Niu et~al., 2016]{Niu2016}
Niu, Y.~S., Hao, N., and Zhang, H. (2016).
\newblock Multiple change-point detection: a selective overview.
\newblock {\em Statistical Science}, 31(4):611--623.

\bibitem[Niu and Zhang, 2012]{Niu2012}
Niu, Y.~S. and Zhang, H. (2012).
\newblock The screening and ranking algorithm to detect {DNA} copy number
  variations.
\newblock {\em Annals of Applied Statistics}, 6(3):1306--1326.

\bibitem[Oka and Qu, 2011]{Oka2011}
Oka, T. and Qu, Z. (2011).
\newblock Estimating structural changes in regression quantiles.
\newblock {\em Journal of Econometrics}, 162(2):248--267.

\bibitem[Olshen et~al., 2004]{Olshen2004}
Olshen, A.~B., Venkatraman, S., Lucito, R., and Wigler, M. (2004).
\newblock Circular binary segmentation for the analysis of array-based {DNA}
  copy number data.
\newblock {\em Biostatistics}, 5(4):557--572.

\bibitem[Phillips, 1987]{phillips1987time}
Phillips, P.~C. (1987).
\newblock Time series regression with a unit root.
\newblock {\em Econometrica}, 55(2):277--301.

\bibitem[Pires and Branco, 2002]{pires2002partial}
Pires, A.~M. and Branco, J.~A. (2002).
\newblock Partial influence functions.
\newblock {\em Journal of Multivariate Analysis}, 83(2):451--468.

\bibitem[Poon et~al., 2004]{Poon2004}
Poon, S.-H., Rockinger, M., and Tawn, J. (2004).
\newblock Extreme value dependence in financial markets: Diagnostics, models,
  and financial implications.
\newblock {\em Review of Financial Studies}, 17(2):581–610.

\bibitem[Preuss et~al., 2015]{Preuss2015}
Preuss, P., Puchstein, R., and Dette, H. (2015).
\newblock Detection of multiple structural breaks in multivariate time series.
\newblock {\em Journal of the American Statistical Association},
  110(510):654--668.

\bibitem[Qu, 2008]{Qu2008}
Qu, Z. (2008).
\newblock Testing for structural change in regression quantiles.
\newblock {\em Journal of Econometrics}, 146:170--184.

\bibitem[Shao, 1995]{Shao1995}
Shao, Q.-M. (1995).
\newblock On a conjecture of révész.
\newblock {\em Proceedings of the American Mathematical Society},
  123(2):575--582.

\bibitem[Shao, 2010]{shao2010self}
Shao, X. (2010).
\newblock A self-normalized approach to confidence interval construction in
  time series.
\newblock {\em Journal of the Royal Statistical Society: Series B},
  72(3):343--366.

\bibitem[Shao, 2015]{shao2015self}
Shao, X. (2015).
\newblock Self-normalization for time series: a review of recent developments.
\newblock {\em Journal of the American Statistical Association},
  110(512):1797--1817.

\bibitem[Shao and Zhang, 2010]{shao2010testing}
Shao, X. and Zhang, X. (2010).
\newblock Testing for change points in time series.
\newblock {\em Journal of the American Statistical Association},
  105(491):1228--1240.

\bibitem[Tartakovsky et~al., 2014]{tartakovsky2014sequential}
Tartakovsky, A., Nikiforov, I., and Basseville, M. (2014).
\newblock {\em {Sequential Analysis: Hypothesis Testing and Change-point
  Detection}}.
\newblock CRC Press.

\bibitem[Tripathi, 1999]{tripathi1999matrix}
Tripathi, G. (1999).
\newblock A matrix extension of the {Cauchy-Schwarz} inequality.
\newblock {\em Economics Letters}, 63(1):1--3.

\bibitem[Truong et~al., 2020]{truong2020}
Truong, C., Oudre, L., and Vayatis, N. (2020).
\newblock Selective review of offline change point detection methods.
\newblock {\em Signal Processing}, 167.

\bibitem[Vanegas et~al., 2021]{Vanegas2020}
Vanegas, L.~J., Behr, M., and Munk, A. (2021).
\newblock Multiscale quantile segmentation.
\newblock {\em Journal of the American Statistical Association, to appear}.

\bibitem[Verzelen et~al., 2020]{Verzelen2020}
Verzelen, N., Fromont, M., Lerasle, M., and Reynaud-Bouret, P. (2020).
\newblock Optimal change-point detection and localization.
\newblock {\em arXiv preprint arXiv:2010.11470}.

\bibitem[Vostrikova, 1981]{vostrikova1981detecting}
Vostrikova, L.~Y. (1981).
\newblock Detecting “disorder” in multidimensional random processes.
\newblock In {\em Doklady Akademii Nauk}, volume 259, pages 270--274. Russian
  Academy of Sciences.

\bibitem[Wang et~al., 2020]{wangEJS2020}
Wang, D., Yu, Y., and Rinaldo, A. (2020).
\newblock Univariate mean change point detection: Penalization, {CUSUM} and
  optimality.
\newblock {\em Electronic Journal of Statistics}, 14(1):1917--1961.

\bibitem[Wied et~al., 2012]{Wied2012}
Wied, D., Krämer, W., and Dehling, H. (2012).
\newblock Testing for a change in correlation at an unknown point in time using
  an extended functional delta method.
\newblock {\em Econometric Theory}, 28(3):570--589.

\bibitem[Wu and Zhou, 2019]{wuzhou2019}
Wu, W. and Zhou, Z. (2019).
\newblock Multiscale jump testing and estimation under complex temporal
  dynamics.
\newblock {\em arXiv preprint arXiv:1909.06307}.

\bibitem[Wu, 2005]{Wu2005}
Wu, W.~B. (2005).
\newblock On the bahadur representation of sample quantiles for dependent
  sequences.
\newblock {\em Annals of Statistics}, 33:1934--1963.

\bibitem[Wu and Zhao, 2007]{wuzhao2007}
Wu, W.~B. and Zhao, Z. (2007).
\newblock Inference of trends in time series.
\newblock {\em Journal of the Royal Statistical Society: Series B},
  69(3):391--410.

\bibitem[Wu and Zhou, 2011]{Wu2011}
Wu, W.~B. and Zhou, Z. (2011).
\newblock Gaussian approximations for non-stationary multiple time series.
\newblock {\em Statistica Sinica}, 21(3):1397--1413.

\bibitem[Yau and Zhao, 2016]{Yau2016}
Yau, C.~Y. and Zhao, Z. (2016).
\newblock Inference for multiple change points in time series via likelihood
  ratio scan statistics.
\newblock {\em Journal of the Royal Statistical Society: Series B},
  78(4):895–916.

\bibitem[Zhang and Lavitas, 2018]{Zhang2018}
Zhang, T. and Lavitas, L. (2018).
\newblock Unsupervised self-normalized change-point testing for time series.
\newblock {\em Journal of the American Statistical Association},
  113(522):637--648.

\bibitem[Zhang and Wu, 2011]{Zhang2011}
Zhang, T. and Wu, W.~B. (2011).
\newblock Testing parametric assumptions of trends of a nonstationary time
  series.
\newblock {\em Biometrika}, 98(3):599--614.

\end{thebibliography}
\end{spacing}

\end{document}